\documentclass[10pt,english]{article}
\usepackage{lmodern}

\usepackage[T1]{fontenc}
\usepackage[latin9]{inputenc}
\usepackage{geometry}
\usepackage{array}
\usepackage{booktabs}
\usepackage{units}
\usepackage{mathtools}
\usepackage{multirow}
\usepackage{amsmath}
\usepackage{amsthm}
\usepackage{amssymb}
\usepackage[authoryear]{natbib}

\makeatletter

\newcommand{\lyxmathsym}[1]{\ifmmode\begingroup\def\b@ld{bold}
  \text{\ifx\math@version\b@ld\bfseries\fi#1}\endgroup\else#1\fi}

\providecommand{\tabularnewline}{\\}

\numberwithin{figure}{section}
\numberwithin{equation}{section}
\theoremstyle{definition}
\newtheorem*{problem*}{\protect\problemname}
\theoremstyle{plain}
\newtheorem{thm}{\protect\theoremname}[section]
\theoremstyle{definition}
\newtheorem{defn}[thm]{\protect\definitionname}
\theoremstyle{plain}
\newtheorem*{thm*}{\protect\theoremname}
\theoremstyle{plain}
\newtheorem{lem}[thm]{\protect\lemmaname}
\theoremstyle{remark}
\newtheorem{rem}[thm]{\protect\remarkname}
\theoremstyle{plain}
\newtheorem{prop}[thm]{\protect\propositionname}
\theoremstyle{plain}
\newtheorem{cor}[thm]{\protect\corollaryname}
\theoremstyle{definition}
\newtheorem{example}[thm]{\protect\examplename}
\theoremstyle{plain}
\newtheorem*{lem*}{\protect\lemmaname}
\theoremstyle{remark}
\newtheorem*{acknowledgement*}{\protect\acknowledgementname}
\theoremstyle{remark}
\newtheorem{claim}[thm]{\protect\claimname}

\numberwithin{equation}{section}
\numberwithin{figure}{section}
\theoremstyle{plain}

\theoremstyle{definition}

\theoremstyle{plain}
\newtheorem*{thm2*}{\protect\theoremname}
\theoremstyle{plain}

\theoremstyle{plain}

\theoremstyle{plain}

\usepackage{color}
\definecolor{ForestGreen}{rgb}{0.1333,0.5451,0.1333}
\usepackage{hyperref}
\hypersetup{
  colorlinks,
  linkcolor={red!50!black},
  citecolor={blue!70!black},
  urlcolor={blue!80!black}
}

\usepackage[lined,boxed,ruled,norelsize,linesnumbered,algo2e]{algorithm2e}
\@addtoreset{section}{part}
\usepackage{graphicx}
\usepackage{thmtools}
\usepackage{microtype}
\usepackage{mathrsfs}
\usepackage{ragged2e}
\usepackage{caption}
\usepackage{subcaption}
\usepackage{thm-restate}
\allowdisplaybreaks
\usepackage{natbib}
\bibpunct{(}{)}{;}{a}{,}{,}

\usepackage{multicol}
\usepackage{enumitem}
\setlist[itemize]{leftmargin=*, itemsep=1pt}

\usepackage{tikz}
\usetikzlibrary{calc}
\usetikzlibrary{positioning, arrows.meta, decorations.pathreplacing}

\makeatother

\usepackage{babel}
\providecommand{\acknowledgementname}{Acknowledgement}
\providecommand{\claimname}{Claim}
\providecommand{\corollaryname}{Corollary}
\providecommand{\definitionname}{Definition}
\providecommand{\examplename}{Example}
\providecommand{\lemmaname}{Lemma}
\providecommand{\problemname}{Problem}
\providecommand{\propositionname}{Proposition}
\providecommand{\remarkname}{Remark}
\providecommand{\theoremname}{Theorem}

\begin{document}
\global\long\def\bw{\mathsf{Ball\ walk}}%
\global\long\def\dw{\mathsf{Dikin\ walk}}%
\global\long\def\sw{\mathsf{Speedy\ walk}}%
\global\long\def\dws{\mathsf{Dikin\ walks}}%
\global\long\def\gcdw{\mathsf{GCDW}}%
\global\long\def\gc{\mathsf{Gaussian\ cooling}}%

\global\long\def\acal{\mathcal{A}}%
\global\long\def\bcal{\mathcal{B}}%
\global\long\def\ccal{\mathcal{C}}%
\global\long\def\dcal{\mathcal{D}}%
\global\long\def\ecal{\mathcal{E}}%
\global\long\def\fcal{\mathcal{F}}%
\global\long\def\gcal{\mathcal{G}}%
\global\long\def\hcal{\mathcal{H}}%
\global\long\def\ical{\mathcal{I}}%
\global\long\def\tcal{\mathbb{\mathcal{T}}}%
\global\long\def\mcal{\mathbb{\mathcal{M}}}%
\global\long\def\pcal{\mathcal{P}}%
\global\long\def\ncal{\mathcal{N}}%
\global\long\def\kcal{\mathcal{K}}%

\global\long\def\O{\mathcal{O}}%
\global\long\def\Otilde{\widetilde{\mathcal{O}}}%

\global\long\def\E{\mathbb{E}}%
\global\long\def\Z{\mathbb{Z}}%
\global\long\def\P{\mathbb{P}}%
\global\long\def\N{\mathbb{N}}%

\global\long\def\R{\mathbb{R}}%
\global\long\def\Rd{\mathbb{R}^{d}}%
\global\long\def\Rdd{\mathbb{R}^{d\times d}}%
\global\long\def\Rn{\mathbb{R}^{n}}%
\global\long\def\Rnn{\mathbb{R}^{n\times n}}%

\global\long\def\psd{\mathbb{S}_{+}^{d}}%
\global\long\def\pd{\mathbb{S}_{++}^{d}}%

\global\long\def\defeq{\stackrel{\mathrm{{\scriptscriptstyle def}}}{=}}%
\global\long\def\ind{\mathds{1}}%

\global\long\def\veps{\varepsilon}%
\global\long\def\lda{\lambda}%
\global\long\def\vphi{\varphi}%
\global\long\def\onu{\bar{\nu}}%
\global\long\def\og{\overline{g}}%
\global\long\def\del{\partial}%

\global\long\def\half{\frac{1}{2}}%
\global\long\def\nhalf{\nicefrac{1}{2}}%
\global\long\def\texthalf{{\textstyle \frac{1}{2}}}%

\global\long\def\kro{\otimes}%
\global\long\def\hada{\circ}%
\global\long\def\chooses#1#2{_{#1}C_{#2}}%

\global\long\def\vol{\textrm{vol}}%
\global\long\def\law{\textup{\textsf{law}}}%

\global\long\def\tr{\textup{\textsf{Tr}}}%
\global\long\def\diag{\textsf{\textup{diag}}}%
\global\long\def\Diag{\textup{\textsf{Diag}}}%
\global\long\def\vec{\textup{\textsf{vec}}}%
\global\long\def\svec{\textup{\textsf{svec}}}%
\global\long\def\inter{\textup{\textsf{int}}}%

\global\long\def\T{\mathsf{T}}%
\global\long\def\e{\mathrm{e}}%

\global\long\def\id{\mathrm{id}}%
\global\long\def\st{\mathrm{s.t.\ }}%
\global\long\def\nnz{\textup{\textsf{nnz}}}%
\global\long\def\lw{\textup{\textsf{Lw}}}%

\global\long\def\intk{\inter(K)}%

\global\long\def\range{\mathrm{Range}}%
\global\long\def\nulls{\mathrm{Null}}%
\global\long\def\spanning{\textup{\textsf{span}}}%
\global\long\def\rowspace{\textup{\textsf{row}}}%
\global\long\def\rank{\mathrm{rank}}%

\global\long\def\bs#1{\boldsymbol{#1}}%

\global\long\def\eu#1{\EuScript{#1}}%

\global\long\def\mb#1{\mathbf{#1}}%

\global\long\def\mbb#1{\mathbb{#1}}%

\global\long\def\mc#1{\mathcal{#1}}%

\global\long\def\mf#1{\mathfrak{#1}}%

\global\long\def\ms#1{\mathscr{#1}}%

\global\long\def\mss#1{\mathsf{#1}}%

\global\long\def\msf#1{\mathsf{#1}}%

\global\long\def\on#1{\operatorname{#1}}%

\global\long\def\textint{{\textstyle \int}}%
\global\long\def\Dd{\mathrm{D}}%
\global\long\def\D{\mathrm{d}}%
\global\long\def\grad{\nabla}%
 
\global\long\def\hess{\nabla^{2}}%
 
\global\long\def\lapl{\triangle}%
 
\global\long\def\deriv#1#2{\frac{\D#1}{\D#2}}%
 
\global\long\def\pderiv#1#2{\frac{\partial#1}{\partial#2}}%
 
\global\long\def\de{\partial}%
\global\long\def\lagrange{\mathcal{L}}%

\global\long\def\Gsn{\mathcal{N}}%
 
\global\long\def\BeP{\textnormal{BeP}}%
 
\global\long\def\Ber{\textnormal{Ber}}%
 
\global\long\def\Bern{\textnormal{Bern}}%
 
\global\long\def\Bet{\textnormal{Beta}}%
 
\global\long\def\Beta{\textnormal{Beta}}%
 
\global\long\def\Bin{\textnormal{Bin}}%
 
\global\long\def\BP{\textnormal{BP}}%
 
\global\long\def\Dir{\textnormal{Dir}}%
 
\global\long\def\DP{\textnormal{DP}}%
 
\global\long\def\Expo{\textnormal{Expo}}%
 
\global\long\def\Gam{\textnormal{Gamma}}%
 
\global\long\def\GEM{\textnormal{GEM}}%
 
\global\long\def\HypGeo{\textnormal{HypGeo}}%
 
\global\long\def\Mult{\textnormal{Mult}}%
 
\global\long\def\NegMult{\textnormal{NegMult}}%
 
\global\long\def\Poi{\textnormal{Poi}}%
 
\global\long\def\Pois{\textnormal{Pois}}%
 
\global\long\def\Unif{\textnormal{Unif}}%

\global\long\def\bpar#1{{\bigl(#1\bigr)}}%
\global\long\def\Bpar#1{{\Bigl(#1\Bigr)}}%

\global\long\def\snorm#1{{\|#1\|}}%
\global\long\def\bnorm#1{{\bigl\Vert#1\bigr\Vert}}%
\global\long\def\Bnorm#1{{\Bigl\Vert#1\Bigr\Vert}}%

\global\long\def\sbrack#1{{[#1]}}%
\global\long\def\bbrack#1{{\bigl[#1\bigr]}}%
\global\long\def\Bbrack#1{{\Bigl[#1\Bigr]}}%

\global\long\def\sbrace#1{\{#1\}}%
\global\long\def\bbrace#1{\bigl\{#1\bigr\}}%
\global\long\def\Bbrace#1{\Bigl\{#1\Bigr\}}%

\global\long\def\Abs#1{\left\lvert #1\right\rvert }%
\global\long\def\Par#1{\left(#1\right)}%
\global\long\def\Brack#1{\left[#1\right]}%
\global\long\def\Brace#1{\left\{  #1\right\}  }%

\global\long\def\inner#1{\langle#1\rangle}%
 
\global\long\def\binner#1#2{\left\langle {#1},{#2}\right\rangle }%

\global\long\def\norm#1{{\|#1\|}}%
\global\long\def\onenorm#1{\norm{#1}_{1}}%
\global\long\def\twonorm#1{\norm{#1}_{2}}%
\global\long\def\infnorm#1{\norm{#1}_{\infty}}%
\global\long\def\fronorm#1{\norm{#1}_{\text{F}}}%
\global\long\def\nucnorm#1{\norm{#1}_{*}}%
\global\long\def\staticnorm#1{\|#1\|}%
\global\long\def\statictwonorm#1{\staticnorm{#1}_{2}}%

\global\long\def\mmid{\mathbin{\|}}%

\global\long\def\otilde#1{\widetilde{\mc O}(#1)}%
\global\long\def\wtilde{\widetilde{W}}%
\global\long\def\wt#1{\widetilde{#1}}%

\global\long\def\KL{\msf{KL}}%
\global\long\def\dtv{d_{\textrm{\textup{TV}}}}%

\global\long\def\cov{\mathrm{Cov}}%
\global\long\def\var{\mathrm{Var}}%

\global\long\def\cred#1{\textcolor{red}{#1}}%
\global\long\def\cblue#1{\textcolor{blue}{#1}}%
\global\long\def\cgreen#1{\textcolor{green}{#1}}%
\global\long\def\ccyan#1{\textcolor{cyan}{#1}}%

\global\long\def\iff{\Leftrightarrow}%
 
\global\long\def\textfrac#1#2{{\textstyle \frac{#1}{#2}}}%

%--------------------------------------------------------------------------------------------------------------------------------
% Common differentials with a small space in front of them
%--------------------------------------------------------------------------------------------------------------------------------
\global\long\def\dee{\mathop{\mathrm{d}\!}}%
 
\global\long\def\dt{\,\dee t}%
 
\global\long\def\ds{\,\dee s}%
 
\global\long\def\dx{\,\dee x}%
 
\global\long\def\dy{\,\dee y}%
 
\global\long\def\dz{\,\dee z}%
  
\global\long\def\dr{\,\dee r}%
 
\global\long\def\dB{\,\dee B}%
 % Brownian motion
\global\long\def\dW{\,\dee W}%
 % Wiener process
\global\long\def\dmu{\,\dee\mu}%
 
\global\long\def\dnu{\,\dee\nu}%
 
\global\long\def\domega{\,\dee\omega}%

%--------------------------------------------------------------------------------------------------------------------------------
% Set notation
%--------------------------------------------------------------------------------------------------------------------------------
\global\long\def\smiddle{\mathrel{}|\mathrel{}}%
 % Well-spaced \middle | symbol
%--------------------------------------------------------------------------------------------------------------------------------
%Text with quads around it
%--------------------------------------------------------------------------------------------------------------------------------
\global\long\def\qtext#1{\quad\text{#1}\quad}%
 % Semidefinite orders
\global\long\def\psdle{\preccurlyeq}%
 
\global\long\def\psdge{\succcurlyeq}%
 
\global\long\def\psdlt{\prec}%
 
\global\long\def\psdgt{\succ}%

%--------------------------------------------------------------------------------------------------------------------------------
% Vectors and matrices
%--------------------------------------------------------------------------------------------------------------------------------
\global\long\def\boldone{\mbf{1}}%
 % Bold 1
\global\long\def\ident{\mbf{I}}%
 % Identity matrix
% \def\v#1{\mbi{#1}} % Vector notation

%--------------------------------------------------------------------------------------------------------------------------------
% Probability and statistics macros
%--------------------------------------------------------------------------------------------------------------------------------
\global\long\def\eqdist{\stackrel{d}{=}}%
 
\global\long\def\todist{\stackrel{d}{\to}}%
 
\global\long\def\eqd{\stackrel{d}{=}}%
 
\global\long\def\independenT#1#2{\mathrel{\rlap{$#1#2$}\mkern4mu  {#1#2}}}%

\title{Gaussian Cooling and Dikin Walks: The Interior-Point Method for Logconcave
Sampling\date{}\author{Yunbum Kook\\ Georgia Tech\\  \texttt{yb.kook@gatech.edu} \and Santosh S. Vempala\\ Georgia Tech\\ \texttt{vempala@gatech.edu}}}
\maketitle
\begin{abstract}
The connections between (convex) optimization and (logconcave) sampling
have been considerably enriched in the past decade with many conceptual
and mathematical analogies. For instance, the Langevin algorithm can
be viewed as a sampling analogue of gradient descent and has condition-number-dependent
guarantees on its performance. In the early 1990s, Nesterov and Nemirovski
developed the Interior-Point Method (IPM) for convex optimization
based on self-concordant barriers, providing efficient algorithms
for structured convex optimization, often faster than the general
method. This raises the following question: can we develop an analogous
IPM for structured sampling problems?

In 2012, Kannan and Narayanan proposed the Dikin walk for uniformly
sampling polytopes, and an improved analysis was given in 2020 by
Laddha-Lee-Vempala. The Dikin walk uses a local metric defined by
a self-concordant barrier for linear constraints. Here we generalize
this approach by developing and adapting IPM machinery together with
the Dikin walk for poly-time sampling algorithms. Our IPM-based sampling
framework provides an efficient warm start and goes beyond uniform
distributions and linear constraints. We illustrate the approach on
important special cases, in particular giving the fastest algorithms
to sample uniform, exponential, or Gaussian distributions on a truncated
PSD cone. The framework is general and can be applied to other sampling
algorithms.

\pagebreak{}
\end{abstract}
\tableofcontents{}

\addtocontents{toc}{\protect\setcounter{tocdepth}{2}} 

\newpage{}

\section{Introduction}

As a motivating example, consider the following problem: how can we
efficiently sample a $d\times d$ matrix from a distribution with
the following density?
\begin{align*}
\text{sample } & X\sim\exp\Bpar{-\bpar{\inner{A,X}+\snorm{X-B}_{F}+\snorm{X-C}_{F}^{2}-\log\det X}}\\
\text{s.t. } & X\succeq0,\,\inner{D_{i},X}\geq c_{i}\,,\quad\forall i\in[m]\,.
\end{align*}
This rather complicated looking distribution recovers as special cases
the problems of sampling from the Max-Cut semi-definite programming
relaxation and the set of minimum (or bounded) volume ellipsoids that
contain a given set of points. The above density is logconcave, so
we can use the $\bw$ (along with isotropic rounding) to sample the
distribution with $\mc O(d^{8}\log d)$ membership/evaluation queries
(\citet{lovasz2007geometry}). This ``general-purpose'' sampler
already gives a poly-time mixing algorithm. However, each term in
the density and constraints is ``structured'', which poses the following
natural question: can we leverage \emph{structure} inherent in the
problem to get more efficient algorithms?

The interior-point method (IPM) is a powerful optimization framework
suitable for solving convex optimization problems with structured
objectives and constraints: for proper convex functions $f_{i}$ and
$h_{j}$
\begin{align*}
\min & \sum_{i}f_{i}(x)\ \text{s.t. }h_{j}(x)\leq0\,.
\end{align*}
This leads us to our main question: is there a sampling analogue of
IPM that generates samples from the density proportional to $\exp(-\sum_{i}f_{i})$
restricted to the convex region defined by structured convex functions
$h_{j}$? This is the general problem we will address here and is
stated formally below.
\begin{problem*}
Let $f_{i}$ be a proper convex function and $h_{j}$ a convex function
on $\Rd$ for $i\in[I]$ and $j\in[J]$. Then the goal is:
\begin{align}
\text{sample } & x\sim\pi\propto\exp\Bpar{-\sum_{i}f_{i}}\tag{\ensuremath{\msf{strLC}}}\label{eq:problem}\\
\text{s.t. } & x\in K:=\bigcap_{j\in[J]}\{x\in\Rd:h_{j}(x)\leq0\}\,,\nonumber 
\end{align}
where we assume that $K$ has non-empty interior and $\pi$ has finite
second moment.
\end{problem*}
In this paper, we derive an IPM framework for structured logconcave
sampling. We use the $\dw$ as a sampler to implement the ``inner''
step of IPM. We provide a mixing time bound for the $\dw$, going
beyond uniform distributions (\S\ref{sec:mixing-Dikin}). This generalization
is necessary to be able to utilize the $\dw$ within the IPM framework.
In \S\ref{sec:IPM-framework}, we present the sampling IPM and derive
its guarantees. Our framework is suited for breaking down complicated
sampling problems into smaller structured problems. An important part
of this paper is \S\ref{sec:sc-theory-rules}, where we develop a
``calculus'' for combining multiple constraints and objectives,
and deriving the resulting theoretical guarantees (analogous to and
inspired by the work of \citet{nesterov1994interior} for optimization).
To provide concrete understanding and instances, we illustrate the
framework on some well-known families of constraints in \S\ref{sec:handbook-barrier},
in particular obtaining faster algorithms to sample uniform, exponential,
or Gaussian distributions on truncated PSD cones in \S\ref{sec:examples}.

\begin{figure}
\includegraphics[width=0.5\textwidth]{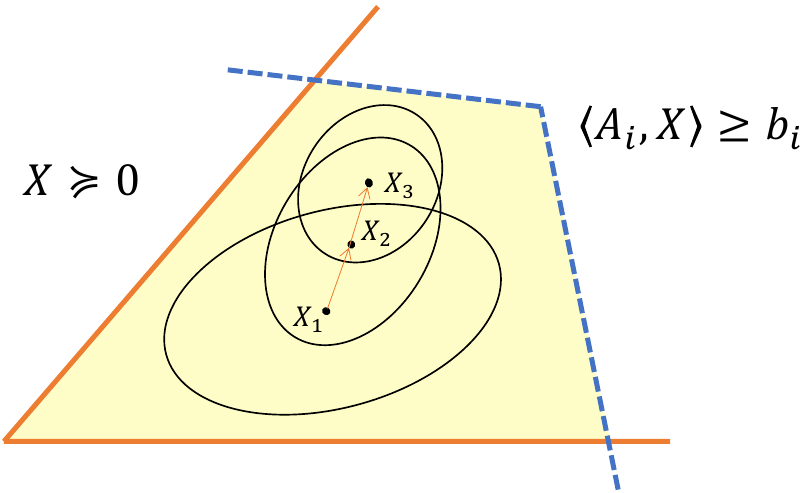}
\centering\caption{\label{fig:DW} Iterates of the $\protect\dw$ (Algorithm \ref{alg:DikinWalk}).
Solid lines centered at $X_{i}$ indicate Dikin ellipsoids, $\protect\mc D_{g}^{r}(X_{i})$.}
\end{figure}

\subsection{Warm-up: Dikin walk and self-concordance}

We use the same symbol for a distribution and its density w.r.t. the
Lebesgue measure. We use $\psd$ (and $\pd$) to denote the set of
$d\times d$ positive semidefinite (and definite) matrices, respectively.
For two matrices $A,B$, we use $A\asymp B$ to indciate $A\precsim B$
and $B\precsim A$. A \emph{local metric} $g$ defines at each point
$x\in K\subset\Rd$ a positive-definite inner product $\inner{\cdot,\cdot}_{g(x)}:\Rd\times\Rd\to\R$,
which induces the local norm $\snorm v_{g(x)}:=\sqrt{\inner{v,v}_{g(x)}}$.
We use $\snorm v_{x}$ to refer to $\snorm v_{g(x)}$ when the context
is clear. We abuse notation and use $g(x)$ to denote the $d\times d$
positive-definite matrix represented with respect to the canonical
basis $\{e_{1},\dots,e_{d}\}$. For a function $f$ defined on $K\subset\Rd$,
we let $\Dd^{i}f(x)[h_{1},\dotsc,h_{i}]$ denote the $i$-th directional
derivative of $f$ at $x$ in directions $h_{1},\dotsc,h_{i}\in\Rd$,
i.e., 
\[
\Dd^{i}f(x)[h_{1},\dotsc,h_{i}]=\frac{\D^{i}}{\D t_{1}\cdots\D t_{i}}f\Bpar{x+\sum_{j=1}^{i}t_{j}h_{j}}\Big|_{t_{1},\dotsc,t_{i}=0}\,.
\]
We let $\mc N_{g}^{r}(x):=\mc N(x,\frac{r^{2}}{d}g(x)^{-1})$ be the
normal distribution with mean $x$ and covariance $\frac{r^{2}}{d}g(x)^{-1}$.
See \S\ref{subsec:prelim} for full preliminaries and other notation.

\paragraph{Dikin walk.}

Given a local metric $g$ in $\Rd$, the Dikin ellipsoid of radius
$r$ at $x\in\Rd$ is defined as
\[
\mc D_{g}^{r}(x)\defeq\Big\{ y\in\Rd:\sqrt{(y-x)^{\T}g(x)(y-x)}=\snorm{y-x}_{g(x)}\leq r\Big\}\,,
\]
i.e., it is a norm ball of radius $r$ defined by the local metric.
From this perspective, the $\dw$ defined below is a natural generalization
of the $\bw$ to a local metric setting.

\begin{algorithm2e}[H]

\caption{$\dw(\pi_{0},\pi,g,rT)$}\label{alg:DikinWalk}

\SetAlgoLined

\textbf{Input:} Initial distribution $\pi_{0}$, target distribution
$\pi\propto\exp(-f)\cdot\mathbf{1}_{K}$, local metric $g$, step
size $r$, $\#$ iterations $T$.

\textbf{Output:} $x_{T}$

Draw an initial point $x_{0}\sim\pi_{0}$ at random. 

\For{$t=0,\cdots,T-1$}{

Sample $z\sim\mc N\bpar{x_{t},\frac{r^{2}}{d}g(x_{t})^{-1}}$.

$x_{t+1}\gets z$ w.p. $A_{x_{t}}(z):=\min\Bpar{1,\frac{p_{z}(x_{t})}{p_{x_{t}}(z)}\,\frac{\pi(z)}{\pi(x_{t})}}$,
where $p_{x}=\mc N_{g}^{r}(x)$.

Otherwise, $x_{t+1}\gets x_{t}$.

}

\end{algorithm2e}

\paragraph{Dikin metrics and self-concordance.}

The metric $g$ used to define the $\dw$ plays a crucial role in
its convergence. Our metrics will be defined by Hessians of convex
self-concordant barrier functions. We now collect definitions of these
functions; they will be important to state our general guarantees
for the mixing of the $\dw$. The concept we need is summarized by
the definition of a $(\nu,\onu)$\emph{-Dikin-amenable metric}.
\begin{figure}
\centering
\begin{subfigure}[b]{0.47\textwidth}          
	\centering          
	\includegraphics[width=\textwidth]{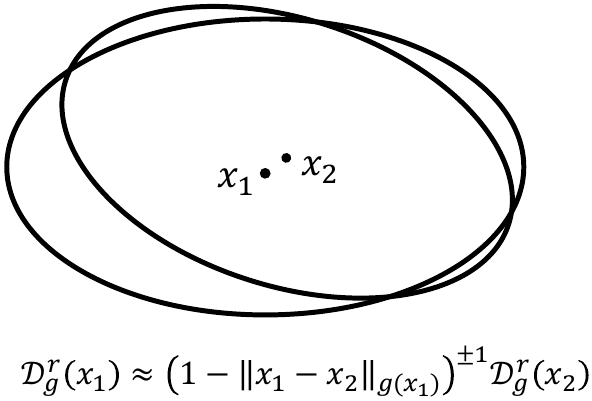}
	\caption{Self-concordance of barrier/metric}          
	\label{fig:sc}
\end{subfigure}      
\hfill      
\begin{subfigure}[b]{0.47\textwidth}
	\centering
	\includegraphics[width=\textwidth]{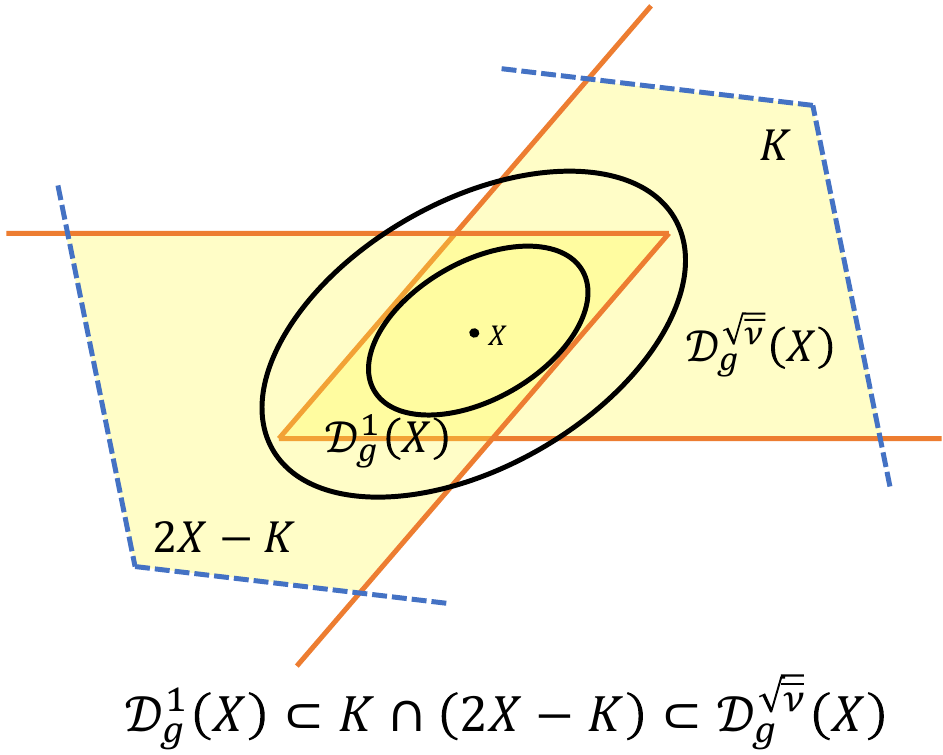}
	\caption{$\onu$-symmetry}
	\label{fig:symm}
\end{subfigure}\caption{\label{fig:sc-symm} (a) Self-concordance of barrier/metric (Definition
\ref{def:sc}) ensures that the Hessian (so Dikin ellipsoids) changes
smoothly. (b) $\protect\onu$-symmetry (Definition \ref{def:symm-param})
indicates how well a Dikin ellipsoid $\protect\mc D_{g}^{r}(X)$ approximates
the locally symmetrized convex body, $K\cap(2X-K)$.}
\end{figure}

\begin{defn}
[Self-concordance (brief version of Definition~\ref{def:sc})]
For convex $K\subset\Rd$, let $\phi:\intk\to\R$ be a smooth convex
function, $g(\cdot)\asymp\hess\phi(\cdot)$, and $\mc N_{g}^{r}(x):=\mc N\bpar{x,\frac{r^{2}}{d}g(x)^{-1}}$. 
\begin{itemize}
\item \emph{$\nu$-self-concordant barrier} (SC): (i) $|\Dd^{3}\phi(x)[h,h,h]|\leq2\snorm h_{\hess\phi(x)}^{3}$
for any $x\in\intk$ and $h\in\Rd$, (ii) $\lim_{x\to\de K}\phi(x)=\infty$,
and (iii) $\snorm{\nabla\phi(x)}_{[\hess\phi(x)]^{-1}}^{2}\leq\nu$
for any $x\in\intk$.
\item \emph{Highly SC} (HSC): $|\Dd^{4}\phi(x)[h,h,h,h]|\leq6\snorm h_{\hess\phi(x)}^{4}$
for any $x\in\intk$ and $h\in\Rd$, and $\lim_{x\to\de K}\phi(x)=\infty$.
\item \emph{Strong SC} (SSC): $\snorm{g(x)^{-\half}\Dd g(x)[h]\,g(x)^{-\half}}_{F}\leq2\snorm h_{g(x)}$
for any $x\in\intk$ and $h\in\Rd$. 
\item \emph{Strongly lower trace SC} (SLTSC): $\tr\bpar{\bpar{\bar{g}(x)+g(x)}^{-1}\Dd^{2}g(x)[h,h]}\geq-\snorm h_{g(x)}^{2}$
for any $\bar{g}:\intk\to\psd$, $x\in\intk$, and $h\in\Rd$. We
call it lower trace self-concordant (LTSC) if it is satisfied when
$\bar{g}=0$.
\item \emph{Strongly average SC} (SASC): For any $\veps>0$ and $\bar{g}:\intk\to\psd$,
there exists $r_{\veps}>0$ such that $\P_{z\sim\mc N_{g+\bar{g}}^{r}(x)}\bpar{\snorm{z-x}_{g(z)}^{2}-\snorm{z-x}_{g(x)}^{2}\leq\frac{2\veps r^{2}}{d}}\geq1-\veps$
for $r\leq r_{\veps}$. We call it average self-concordant (ASC) if
this is satisfied when $\bar{g}=0$. 
\end{itemize}
\end{defn}

SC imposes regularity on the eigenvalues of the directional derivative
$\Dd g[h]$ through its definition $-2\snorm h_{g}^{2}g\preceq\Dd g[h]\preceq2\snorm h_{g}^{2}g$
(or equivalently the largest magnitude of eigenvalues of $g^{-\nicefrac{1}{2}}\Dd g[h]\,g^{-\nicefrac{1}{2}}$),
and HSC does the same on the higher-order derivative $\Dd^{2}g[h,h]$.
SSC introduced by \citet{laddha2020strong} imposes \emph{stronger}
regularity on the eigenvalues of $\Dd g[h]$ by definition, as SSC
is stated in terms of the \emph{Frobenius norm} of $g^{-\half}\Dd g[h]\,g^{-\half}$.
LTSC relaxes `convexity of $\log\det g$' required by \citet{laddha2020strong}.
In particular, SSC and LTSC control the change of $\log\det g$, leading
to a refined analysis of the $\dw$. Lastly, ASC is pertinent to the
average of the squared local norm difference of $z-x$ computed at
$z$ and $x$, which controls the acceptance-probability of each iterate
of the $\dw$. 

These notions are sophisticated enough to carry out a tight mixing
analysis of the $\dw$, but also simple enough for us to develop a
``calculus'' for combining metrics for multiple constraints in \S\ref{sec:sc-theory-rules}.
Moreover, these conditions may look difficult to verify, but we show
that a proper scaling of (H)SC barriers immediately makes them satisfy
these properties.

Next, we recall a symmetry parameter of a self-concordant metric.
We will later see that it has a natural connection to the Cheeger
isoperimetry.
\begin{defn}
[$\onu$-symmetry] \label{def:symm-param} For convex $K\subset\Rd$,
a PSD matrix function $g:\intk\to\psd$ is said to be $\onu$-\emph{symmetric}
if $\mc D_{g}^{1}(x)\subseteq K\cap(2x-K)\subseteq\mc D_{g}^{\sqrt{\onu}}(x)$
for any $x\in K$.
\end{defn}

We note that $K\cap(2x-K)$ is the locally symmetrized convex body
with respect to $x$. Hence, $\onu$-symmetry measures how accurately
a Dikin ellipsoid approximates the locally symmetrized body. One can
show that $\onu=\mc O(\nu^{2})$ for any metric induced by a self-concordant
barrier.

Going forward, we call a PD matrix function $\onu$\emph{-Dikin-amenable}
if it is SSC, LTSC, ASC, and $\onu$-symmetric. We sometimes call
it $(\nu,\onu)$-Dikin-amenable to reveal its self-concordance parameter
$\nu$. For example, the Hessian of a logarithmic barrier is an $(m,m)$-Dikin-amenable
metric. We present more concrete examples after introducing Theorem~\ref{thm:IPM-sampling}.

\subsection{Results}

\subsubsection{Dikin walk ($\S$\ref{sec:mixing-Dikin})}

We begin with our analysis of the $\dw$ for general settings, going
beyond uniform distributions.

\begin{restatable}{thmre}{thmDikin} \label{thm:Dikin} Let $K\subset\Rd$
be convex and $0\leq\alpha\leq\beta<\infty$.
\begin{itemize}
\item (Local metric) Assume that a $C^{1}$-matrix function $g:\intk\to\pd$
is $\onu$-Dikin-amenable. 
\item (Distribution) Let $\pi_{0}$ and $\pi\propto e^{-f}\cdot\mathbf{1}_{K}$
be an initial and target distribution respectively, where $f$ is
$\alpha$-relatively strongly convex and $\beta$-smooth in $g$.
Let $\norm{\pi_{0}/\pi}=\E_{\pi_{0}}\big[\deriv{\pi_{0}}{\pi}\big]$
and $P$ be the transition kernel of the $\dw$ (Algorithm~\ref{alg:DikinWalk})
with the local metric $g$ and step size $r=\O(\min(1,\beta^{-1/2}))$.
\end{itemize}
Then for any $\veps>0$, it holds that $\dtv(\pi_{0}P^{(T)},\pi)\leq\veps$
for $T\gtrsim d\,\max(1,\beta)\,\min(\onu,\nicefrac{1}{\alpha})\,\log\frac{\norm{\pi_{0}/\pi}}{\veps}$.

\end{restatable}

This result serves as a unifying framework that recovers as special
cases previous works on the $\dw$ for uniform sampling (\citet{kannan2012random,narayanan2016randomized,chen2018fast,laddha2020strong}),
as seen later in \S\ref{sec:examples}. Our analysis extends beyond
uniform sampling, considering the $\dw$ under a more general setting
where the potential $f$ satisfies $\alpha g\preceq\hess f\preceq\beta g$
on $\intk$. This setting is a generalization of $\alpha I\preceq\hess f\preceq\beta I$
under a local metric $\hess\phi\asymp g$. We also note that the $\dw$
is the first \emph{implementable} algorithm that provides a clean
mixing guarantee under this general setting, which is a necessary
ingredient for theory of our sampling IPM. We refer readers to \S\ref{subsec:related-work}
for related work.

All previous analyses of the $\dw$ do not go through for general
distributions. The techniques either have gap (e.g., omit ASC) or
yield a wrong proof (e.g., for one-step coupling, the TV distance
bound from the triangle inequality is larger than $1$, so becomes
vacuous). Our analysis proceeds with the exact form of the TV distance,
additionally requiring the control of $\half\int|A_{x}(z)p_{x}(z)-A_{y}(z)p_{y}(z)|\,\D z$
for close points $x$ and $y$ (see Algorithm~\ref{alg:DikinWalk}).
As sketched in \S\ref{sec:mixing-Dikin}, this involved task simultaneously
quantifies closeness of acceptance probabilities $A_{x}(z)$ and $A_{y}(z)$
as well as that of the Gaussian densities $p_{x}(z)$ and $p_{y}(z)$.
This can be achieved through sophisticated conditioning on high-probability
events due to ASC, SSC, and symmetry of Gaussians.

\subsubsection{Sampling IPM: Gaussian cooling with the Dikin walk ($\protect\gcdw$)
($\S$\ref{sec:IPM-framework})}

We present $\gc$, essentially a sampling analogue of the optimization
IPM. The \emph{function counterpart} below refers to a self-concordant
barrier $\phi$ such that $\hess\phi\asymp g$ on $\intk$.

\begin{restatable}{thmre}{thmDikinannealing} \label{thm:Dikin-annealing}
For convex $K\subset\Rd$, suppose that $g:\intk\to\pd$ is $(\nu,\onu)$-Dikin-amenable
and $\phi$ is its function counterpart such that $\min_{K}\phi$
exists. $\gc$ with the $\dw$ (Algorithm~\ref{alg:IPM-sampling}
with the $\dw$ serving as a non-Euclidean sampler) generates a sample
that is $\veps$-close to $\exp(-f)\cdot\mathbf{1}_{K}$ in TV-distance
using $\mc O\bpar{d\,\max(d\frac{\nu\beta+d}{\nu\alpha+d},\nu,\onu)\log\frac{d\nu}{\veps}}$
iterations of the $\dw$ with $g$, where a $C^{2}$-function $f:\intk\to\R$
satisfies $\alpha\hess\phi\preceq\hess f\preceq\beta\hess\phi$ on
$K$ for $0\leq\alpha\leq\beta<\infty$. In particular, when $f(x)=\alpha^{\T}x$
or $c\phi(x)$ for $\alpha\in\Rd$ and $c\in\R_{+}$, the algorithm
uses $\widetilde{\mc O}(d\,\max(d,\nu,\onu))$ iterations of the $\dw$.

\end{restatable}

The inner loop of $\gc$ runs the $\dw$. The basic GC algorithm was
introduced in \cite{cousins2018gaussian} for efficient sampling and
volume computation. \cite{lee2018convergence} studied its extension
to Hessian manifolds for uniformly sampling polytopes. Our framework
is general in that it handles more general distributions through a
sophisticated annealing scheme.

This framework provides an efficient algorithm for generating a warm
start for constrained log-concave distributions. If we were to apply
Theorem~\ref{thm:Dikin} with initial distribution being a single
point at some distance from boundary, even for the simplest case of
uniform sampling, then an additional factor of $d$ would be incurred.
On the other hand, given that $\nu$ and $\onu$ are typically $\O(d)$,
our framework only has a logarithmic (in dimension) factor overhead
for generating a warm start. An important reason why this works is
the affine-invariance of the $\dw$. Samplers like the $\bw$ have
to apply isotropic transformation to achieve a warm start efficiently,
which requires a near-linear number of samples and thus have at least
a linear in dimension overhead.

\paragraph{Derivation of the algorithm.}

We describe this algorithm alongside its interpretation as a `sampling
analogue of the \emph{interior-point method}'. To this end, we revisit
ideas of IPM, derive its sampling version via a conceptual analogy
between optimization and sampling, and refine the derived sampling
IPM by highlighting the distinctions between the two methods. See
\S\ref{subsec:derivation-IPM-sampling} for details.

\subparagraph*{(1) Optimization IPM (Algorithm~\ref{alg:IPM}).}

In solving the optimization problem, $\min_{x\in K}f(x)$ for a real-valued
convex function $f$ on convex $K\subset\Rd$, IPM first replaces
$f$ by a new variable $t$ and appends the epigraph $\{(x,t)\in\R^{d+1}:f(x)\leq t\}$
to the constraint in addition to $x\in K$. Then summation of self-concordant
barriers for $K$ and the epigraph results in a $\nu$-self-concordant
barrier $\phi$ for the augmented constraints. This barrier $\phi$
allows one to convert the constrained problem to a unconstrained one,
$\min f_{\lda}(x,t):=t+\frac{1}{\lda}\phi(x,t)$ for a parameter $\lda>0$.
Then an optimization step (e.g., the Newtonian gradient descent) that
takes into account the local geometry given by $\hess\phi$ moves
a current point closer to an optimal point, with the barrier $\phi$
preventing escape from the constraints. Increasing $\lda\gets\lda(1+\frac{1}{\sqrt{\nu}})$,
IPM repeats this procedure with the updated point used as a starting
point. As $\lda$ increases (until $\lda\leq\nu/\veps$ for target
accuracy $\veps>0$), the effect of $\frac{1}{\lda}\,\phi(x,t)$ vanishes
in the regularized problem, which gradually brings us to a point sufficiently
closer to the minimum.

\subparagraph*{(2) Translation to sampling (Figure~\ref{fig:opt-samp-IPM}).}

We recall the following conceptual match between convex optimization
and logconcave sampling: for convex $K\subset\Rd$ and convex function
$f:K\to\R$
\begin{align*}
\min f(x) & \quad\longleftrightarrow\quad\text{sample }x\sim\pi\propto\exp(-f)\\
\text{s.t. }x\in K\,. & \qquad\qquad\quad\text{s.t. }x\in K\,.
\end{align*}
With the connection in mind, we can translate IPM's machinery into
the sampling context. As in IPM, we replace $f$ by a new variable
$t$, introduce the epigraph constraint, and attempt to sample a `regularized'
distribution $\mu_{\sigma^{2}}(x,t)\propto\exp\bigl(-f_{\sigma^{2}}(x,t)\bigr)=\exp\bpar{-\bpar{t+\frac{1}{\sigma^{2}}\,\phi(x,t)}}$,
where a parameter $\sigma^{2}$ corresponds to $\lda$ above. This
sampling step should be carried out with a sampler \emph{aware} of
the local geometry given by $\hess\phi$ (call it $\msf{NE\text{-}sampler}$,
which is the $\dw$ in our case). Then we increase $\sigma^{2}$ slightly,
and using the previous regularized distribution $\mu_{\sigma^{2}}$
as a warm start, we sample a next regularized distribution $\mu_{\sigma^{2}+\veps}$.
This iterative procedure continues until $\sigma^{2}$ reaches $\nu$.

\subparagraph*{(3) Refinements (Figure~\ref{fig:gc}).}

We now make this conceptual algorithm concrete in Algorithm~\ref{alg:IPM-sampling}.
The finalized sampling IPM\footnote{For the sake of exposition, we focus on just the exponential distribution
$e^{-t}$. Our algorithm can deal with more general potentials (relatively
convex and smooth).} consists of four phases --- Phase 1 for initialization, Phase 2
and 3 for increasing $\sigma^{2}$ with control, and Phase 4 for high-accuracy
sampling.

Phase 1 initializes the algorithm by a Gaussian truncated over a Dikin
ellipsoid of radius $\O(d^{-\Theta(1)})$. This Gaussian serves as
a good warm start for a regularized distribution with small $\sigma^{2}$.

The sampling IPM, in contrast to both the optimization IPM and the
basic GC algorithm, proceeds with a distinct annealing scheme. Phase
2 updates $\sigma^{2}\gets\sigma^{2}(1+\nicefrac{1}{d})$ until $\sigma^{2}$
reaches $\nu/d$, annealing not only $\phi$ but also the `modified'
potential $\nu t/d$. While $\frac{\nu}{d}\leq\sigma^{2}\leq\nu$,
Phase 3 updates $\sigma^{2}\gets\sigma^{2}(1+\nicefrac{\sigma}{\sqrt{\nu}})$
but only $\phi$ part with the potential $t$ now fixed. We note that
the basic GC anneals only regularization term throughout.

Lastly, the sampling IPM runs the $\dw$ once in Phase 4. If one stopped
after Phase 3 (when $\sigma^{2}$ reaches as the optimization version,
then the total iterates of the $\dw$ would be $\mc O(d\,(d\vee\nu)/\text{poly}(\veps))$.
This guarantee can avoid the symmetry parameter, but this comes at
the cost of low-accuracy of the sampler (i.e., dependence on $\text{poly}(\veps^{-1})$).
Hence, we finish up the algorithm with another execution of the $\dw$,
obtaining high-accuracy $\mc O(d\,(d\vee\nu\vee\onu)\,\log\frac{1}{\veps})$-mixing.

At the heart of the algorithm lies closeness of regularized distributions
in consecutive iterations. Closeness in the first two phases follows
from a property of logconcavity established by \citet{lovasz2006simulated},
while closeness in the last two phases is assured by the Brascamp-Lieb
inequality.

$\gcdw$ is exactly this refined algorithm with the $\dw$ used for
the $\msf{NE\text{-}sampler}$ (Algorithm~\ref{alg:IPM-sampling}).
Specifically in the inner loop, it runs the $\dw$ to sample regularized
exponential distributions of the form $\exp(-(c_{1}t+c_{2}\phi(x,t))$
subject to $x\in K$ and $\{(x,t)\in\R^{d+1}:f(x)\leq t\}$, where
the local metric therein consists of the Hessians of self-concordant
barriers for $K$ and the level set of $f$. Comparing with the $\bw$
for a general logconcave distribution \cite{lovasz2007geometry},
incorporating the geometry of a level set of $f$ (not $\hess f$)
is a natural approach to sampling from $e^{-f}$.

\begin{figure}[t]
\centering \begin{tikzpicture}%[>=stealth,every node/.style={shape=rectangle,draw,rounded corners, minimum width=4.5cm,},]
    \node[rectangle,draw,rounded corners, very thick, minimum width=2.8cm] at (-1.5, 0) (l1) { %fill=orange!30, 
    \begin{tabular}{c}
    	$\dw$'s mixing\\
		under self-concordance\\
		(Theorem~\ref{thm:Dikin}, \S\ref{sec:mixing-Dikin})
    \end{tabular}};

   \node[rectangle,draw,rounded corners, very thick, minimum width=3.3cm] at (-1.5, -4) (l2) {
    \begin{tabular}{c}
    	Sampling IPM\\
		(Algorithm~\ref{alg:IPM-sampling}, \S\ref{sec:IPM-framework})
    \end{tabular}};

   \node[rectangle,draw,rounded corners, very thick, minimum width=3.3cm] at (5, 0) (m1) {%fill=blue!30, 
    \begin{tabular}{c}
    	$\gc$\\
		with the $\dw$\\
		(Theorem~\ref{thm:Dikin-annealing}, \S\ref{sec:IPM-framework})
    \end{tabular}};

    \node[rectangle,draw,rounded corners, very thick, minimum width=2.8cm] at (5, -2) (m2) {%fill=blue!30, 
    \begin{tabular}{c}
    	Self-concordance theory\\
		(Theorem~\ref{thm:IPM-sampling}, \S\ref{sec:sc-theory-rules})
    \end{tabular}};

    \node[rectangle,draw,rounded corners, very thick, minimum width=2.8cm] at (5, -4) (m3) {%fill=blue!30, 
    \begin{tabular}{c}
    	Handbook for\\
		constraints \& epigraphs\\
		(\S\ref{sec:handbook-barrier})
    \end{tabular}};

	\node[rectangle,draw,rounded corners, very thick, minimum width=2.8cm] at (10.5, 0) (r1) {%fill=blue!30, 
    \begin{tabular}{c}
		Examples\\
		(\S\ref{sec:examples})
    \end{tabular}};
	
	\node[fill,circle,inner sep=0pt,minimum size=1pt] at (8.5, -3) (a1){};

    % connect the nodes
	\draw[->,  line width=.6mm] (l1) to[out=0,in=180] (m1);
	\draw[->,  line width=.6mm] (l2) to[out=0,in=180] (m1);

	\draw[->,  line width=.6mm] (m2) to[out=0,in=180] (a1);
	\draw[->,  line width=.6mm] (m3) to[out=0,in=180] (a1);

	\draw[->,  line width=.6mm] (a1) to[out=0,in=180] (r1);
	\draw[->,  line width=.6mm] (m1) to[out=0,in=180] (r1);\end{tikzpicture} \caption{Outline \label{fig:tikz}}
\end{figure}

\subsubsection{Self-concordance theory for combining barriers ($\S$\ref{sec:sc-theory-rules})}

From the earlier discussion, the sampling IPM allows us to focus on
the following reduced problem: Let $t_{1},\dots,t_{I}\in\R$ and $y=(x,t_{1},\dots,t_{I})\in\Rd\times\R^{I}=\R^{d+I}$.
We denote $E_{i}:=\{(x,t_{i})\in\R^{d+1}:f_{i}(x)\leq y_{n+i}\}$
for $i\in[I]$ and $K_{j}:=\{x\in\Rd:h_{j}(x)\leq0\}$ for $j\in[J]$,
whose convexity follows from convexity of $f_{i}$ and $h_{j}$. Denoting
the embeddings of $E_{i}$ and $K_{j}$ onto $\R^{d+I}$ by $\bar{E}_{i}$
and $\bar{K}_{j}$, we can reduce \eqref{eq:problem} to
\begin{align}
\text{sample } & y\sim\tilde{\pi}\propto\exp\bpar{-(\underbrace{0,\dotsc,0}_{d\text{ times}},\underbrace{1,\dotsc,1}_{I\text{ times}})^{\T}\,y}\tag{\ensuremath{\msf{redLC}}}\label{eq:reduced-problem}\\
\text{s.t. } & y\in K':=\bigcap_{i=1}^{I}\bar{E}_{i}\,\cap\,\bigcap_{j=1}^{J}\bar{K}_{j}\,,\nonumber 
\end{align}
where $K'$ is closed convex and has non-empty interior, and we are
given self-concordant barriers for each $E_{i}$ and $K_{j}$. As
the $x$-marginal of $\tilde{\pi}$ is $\pi$, we just project a drawn
sample from $\tilde{\pi}$ to the $x$-space. When $f_{i}(x)$ can
be written as $d$ separable terms (i.e., $f_{i}(x)=\sum_{l=1}^{d}f_{i,l}(x_{l})$),
it is more convenient to introduce $d$ many variables $t_{i,1},\dots t_{i,d}$
for $f_{i,1}(x_{1}),\dots,f_{i,d}(x_{d})$.

In \S\ref{sec:sc-theory-rules}, we study how to combine a self-concordant
metric and its parameters from each epigraph $E_{i}$ and convex set
$K_{j}$ (for the mixing estimation of $\dw$). As in the optimization
IPM, the addition of all barriers is actually a good candidate of
a barrier for $K'$, but under an appropriate scaling. However, the
sampling version requires not only just self-concordance parameters
but also symmetry parameters, SSC, and LTSC for final mixing time
guarantees. Notably, SSC and LTSC assume invertibility of a local
matrix function, but the Hessian of a barrier for a lower-dimensional
space is degenerate with respect to the augmented variable $y\subset\R^{d+I}$.
We address this technical issue by working with Definition~\ref{def:sc-along-subspace}
and several matrix lemmas to study how to maintain or update each
of the main properties such as symmetry, SSC, and LTSC under addition
and scaling.

Now we can state how to put together information of a barrier for
each constraint and epigraph. The readers can note the analogy to
Nesterov and Nemirovski's IPM theory for optimization.

\begin{restatable}{thmre}{thmIPMsampling} \label{thm:IPM-sampling}
In the reduced problem of \eqref{eq:reduced-problem}, let us assume
the following:
\begin{itemize}
\item For $i\in[I]$, the epigraph $E_{i}$ admits a PSD matrix function
$g_{i}^{e}(x,t_{i})$ (or $g_{i}^{e}(x,t_{i,1},\dots,t_{i,d})$) that
is a $(\nu_{i},\bar{\nu}_{i})$-SC barrier, SSC along some subspace,
SLTSC, and SASC.
\item For $j\in[J]$, the constraint $K_{j}$ admits a PSD matrix function
$g_{j}^{c}(x)$ that is a $(\eta_{j},\bar{\eta}_{j})$-SC barrier,
SSC along some subspace, SLTSC, and SASC.
\end{itemize}
For appropriate projections $\pi_{i}^{e}$ and $\pi^{c}$, a matrix
function $g$ on $y\in\inter(K')$ defined by
\[
\langle u,v\rangle_{g(y)}:=(I+J)\,\Bpar{\sum_{i=1}^{I}\langle\pi_{i}^{e}u,\pi_{i}^{e}v\rangle_{g_{i}^{e}(\pi_{i}^{e}(y))}+\sum_{j=1}^{J}\langle\pi^{c}u,\pi^{c}v\rangle_{g_{j}^{c}(\pi^{c}(y))}}\quad\text{for }u,v\in\Rd
\]
is $\bpar{(I+J)(\sum_{i=1}^{I}\nu_{i}+\sum_{j=1}^{J}\eta_{j}),\,(I+J)(\sum_{i=1}^{I}\onu_{i}+\sum_{j=1}^{J}\bar{\eta}_{j})}$-Dikin-amenable
on $K'$.

\end{restatable}

\subsubsection{Metrics for well-known structured instances ($\S$\ref{sec:handbook-barrier})}

\begin{table}
\begin{centering}
\begin{tabular}{ccccccccc}
\toprule 
\textbf{Constraints / Epigraphs} & \textbf{Barrier} & $\nu$ & $\onu$ & SSC & LTSC & SLTSC & ASC & SASC\tabularnewline
\midrule
\midrule 
\multirow{3}{*}{$Ax\geq b$} & $\phi_{\textup{log}}$ & $m$ & $m$ &  &  &  &  & \tabularnewline
\cmidrule{2-9} \cmidrule{3-9} \cmidrule{4-9} \cmidrule{5-9} \cmidrule{6-9} \cmidrule{7-9} \cmidrule{8-9} \cmidrule{9-9} 
 & $g_{\textup{Vaidya}}$ & $\sqrt{md}$ & $\sqrt{md}$ &  &  &  &  & \tabularnewline
\cmidrule{2-9} \cmidrule{3-9} \cmidrule{4-9} \cmidrule{5-9} \cmidrule{6-9} \cmidrule{7-9} \cmidrule{8-9} \cmidrule{9-9} 
 & $g_{\textup{Lw}}$ & $d$ & $d$ &  & $\sqrt{d}$ & $\sqrt{d}$ & $\sqrt{d}$ & $\sqrt{d}$\tabularnewline
\midrule 
$\norm{x-\mu}_{\Sigma}^{2}\leq1$ & $\phi_{\textup{ellip}}$ &  &  & $d$ &  &  & $d$ & $d$\tabularnewline
\midrule 
$\norm{x-\mu}_{\Sigma}^{2}\leq t$ & $\phi_{\textup{Gauss}}$ &  &  & $d$ &  &  & $d$ & $d$\tabularnewline
\midrule 
$\norm{x-\mu}_{\Sigma}\leq t$ & $\phi_{\textup{SOC}}$ &  &  & $d$ & $d$ & $d$ & $d$ & $d$\tabularnewline
\midrule 
$X\succeq0$ & $\phi_{\textup{PSD}}$ & $d$ & $d$ & $d$ &  &  & $d$ & $d^{2}$\tabularnewline
\midrule 
$-x_{i}\log x_{i}\leq t_{i}$ $\forall i\in[d]$ & $\phi_{\textup{ent}}$ & $d$ & $d$ &  &  &  & $d$ & $d$\tabularnewline
\midrule 
$\Abs{x_{i}}^{p}\leq t_{i}$ $\forall i\in[d]$ & $\phi_{\textup{power}}$ & $d$ & $d$ &  &  &  & $d$ & $d$\tabularnewline
\bottomrule
\end{tabular}
\par\end{centering}
\caption{\label{tab:scaling-table} Self-concordance and symmetry parameters,
and required scaling factors for a family of barriers. In this table,
we assume $A\in\protect\R^{m\times d},x\in\protect\Rd,$ and $X\in\protect\psd$.
Empty entries indicate $\protect\O(1)$-scalings.}
\end{table}

In \S\ref{sec:handbook-barrier}, we examine required parameters
and properties of a barrier for a structured constraint and potential,
such as linear, quadratic, entropy, $\ell_{p}$-norm, and PSD cone.
See Table~\ref{tab:scaling-table}.

\paragraph{(1) Linear constraints.}

We start with linear constraints given by $K:=\{x\in\Rd:Ax\geq b\}$
for $A\in\R^{m\times d}$ and $b\in\R^{m}$, where $A$ is assumed
to have no all-zero rows. For $x\in\intk$ and $i\in[m]$, let $a_{i}$
be the $i$-th row of $A$, and denote $S_{x}:=\Diag(a_{i}^{\T}x-b_{i})\in\R^{m\times m}$
and $A_{x}:=S_{x}^{-1}A\in\R^{m\times d}$.

These linear constraints admit efficiently computable self-concordant
barriers: logarithmic barrier, Vaidya metric, and Lewis-weight metric.
The logarithmic barrier is the simplest defined by 
\[
\phi_{\textup{log}}(x):=-\sum_{i=1}^{m}\log(a_{i}^{\T}x-b_{i})\,.
\]
When the number of constraints $m$ is large, one can use a self-concordant
metric due to \citet{vaidya1996new}. For a full-rank matrix $A$,
the resulting Vaidya metric takes advantage of the \emph{leverage
scores} $\sigma(A_{x})$ of $A_{x}$, the diagonal entries of the
orthogonal projection $P_{x}=A_{x}(A_{x}^{\T}A_{x})^{-1}A_{x}\in\R^{m\times m}$,
i.e., $[\sigma(A_{x})]_{i}:=(P_{x})_{ii}>0$ for $i\in[m]$. For $\Sigma_{x}=\Diag(\sigma(A_{x}))\in\R^{m\times m}$,
the Vaidya metric is defined by 
\[
g_{\textup{Vaidya}}(x):=\mc O(1)\sqrt{\frac{m}{d}}A_{x}^{\T}\bpar{\Sigma_{x}+\frac{d}{m}I_{m}}A_{x}\,,
\]
which satisfies $g_{\textup{Vaidya}}\asymp\hess\bpar{\sqrt{\frac{m}{d}}(\phi_{\vol}+\frac{d}{m}\phi_{\textup{log}})}$
for $\phi_{\vol}:=\half\log\det(\hess\phi_{\textup{log}})$.

Its self-concordance parameter is still polynomial in $m$, and it
is natural to ask if the dependence on $m$ can be removed or made
poly-logarithmic. This can be achieved by a Lewis-weight metric that
makes use of the \emph{Lewis weights} of $A_{x}$. The $\ell_{p}$-Lewis
weight of $A_{x}$ is the vector $w_{x}\in\R^{m}$ satisfying the
implicit equation $w_{x}=\sigma\bpar{\Diag(w_{x})^{1/2-1/p}A_{x}}$.
Note that the leverage scores can be recovered as the $\ell_{2}$-Lewis
weight of $A_{x}$. Then the Lewis-weight metric is defined by
\[
g_{\text{\textsf{Lw}}}(x):=\mc O(\log^{\mc O(1)}m)\,A_{x}^{\T}W_{x}A_{x}\,,
\]
which is an $\mc O\bpar{\log^{\mc O(1)}m}$-approximation of the Hessian
of $\phi_{\text{Lw}}(x):=\log\det(A_{x}^{\T}W_{x}^{1-2/p}A_{x})$.
With $p=\mc O(\log^{\Theta(1)}m)$, the self-concordance parameter
of this barrier and metric can be made $\mc O^{*}(d)$.

For the sampling purpose, we should look into other properties such
as symmetry, SSC, SLTSC, and SASC, going beyond just self-concordance
parameter. We note that the log-barrier and Vaidya metric fulfill
these properties without additional scaling, while the Lewis-weight
metric requires a $\sqrt{d}$-scaling for SLTSC and SASC. We summarize
these results below.
\begin{thm*}
[Linear constraints] We assume $m\geq d$ in the cases of the Vaidya
and Lewis-weight. Let $w_{x}$ be the $\ell_{p}$-Lewis weights with
$p=\O(\log^{\Theta(1)}m)$. 
\begin{itemize}
\item Log-barrier $\phi_{\textup{log}}$: $g=\hess\phi_{\textup{log}}$
satisfies $\nu,\onu\leq m$, SSC along $\rowspace(A)$, and $\Dd^{2}g(x)[h,h]\succeq0$
(so SLTSC), and SASC.
\item Vaidya metric $g_{\textup{Vaidya}}(x)=\sqrt{\frac{m}{d}}A_{x}^{\T}\bpar{\Sigma_{x}+\frac{d}{m}I_{m}}A_{x}$
and $\phi_{\textup{Vaidya}}=\sqrt{\frac{m}{d}}\bpar{\half\log\det(\hess\phi_{\textup{log}})+\frac{d}{m}\phi_{\textup{log}}}$
(with $m\geq d)$: $g=44g_{\textup{Vaidya}}$ satisfies $\nu,\onu=\mc O(\sqrt{md})$,
SSC, SLTSC, and SASC.
\item Lewis-weight metric $g_{\textup{Lw}}(x)=\mc O(\log^{\mc O(1)}m)\,A_{x}^{\T}W_{x}A_{x}$
and $\phi_{\textup{Lw}}=\log\det(A_{x}^{\T}W_{x}^{1-2/\mc O(\log m)}A_{x})$:
$g=\sqrt{d}g_{\textup{Lw}}$ satisfies $\nu,\onu=\mc O(d^{3/2}\log^{\mc O(1)}m)$,
SSC, SLTSC, and SASC.
\end{itemize}
\end{thm*}

\paragraph{(2) Quadratic potentials and constraints.}

Now we consider quadratic potential (i.e., Gaussian) and constraints
(i.e., ellipsoid and second-order cone). A self-concordant barrier
introduced by \citet{nesterov1994interior} serves as an efficient
barrier for each constraint or epigraph of a potential. We show that
all barriers are HSC, so the scaling of $d$ makes it satisfy SLTSC
and SASC.
\begin{thm}
[{[}Quadratic] Let $K_{1}=\{x\in\Rd:\half x^{\T}Qx+p^{\T}x+l\leq0\}$
with $p\in\Rd$ and $0\neq Q\in\psd$. Let $K_{2}=\{(x,t)\in\R^{d+1}:\half\norm{x-\mu}_{\Sigma}^{2}\leq t\}$
and $K_{3}=\{(x,t)\in\R^{d+1}:\norm{x-\mu}_{\Sigma}\leq t\}$ with
$\mu\in\Rd$ and $\Sigma\in\pd$. Let $x\in\inter(K_{i})$ and $h\in\R^{\dim(K_{i})}$. 
\begin{itemize}
\item Ellipsoid $\phi_{\textup{ellip}}(x)=-\log(-l-p^{\T}x-\half x^{\T}Qx)$
for $K_{1}$: $g=d\,\hess\phi_{\textup{ellip}}$ satisfies $\nu,\onu=\mc O(d)$,
SSC when $Q\in\pd$, $\Dd^{2}g(x)[h,h]\succeq0$ (so SLTSC), and SASC.
\item Gaussian $\phi_{\textup{Gauss}}(x,t)=-\log(t-\half\norm{x-\mu}_{\Sigma}^{2})$
for $K_{2}$: $g=d\,\hess\phi_{\textup{Gauss}}$ satisfies $\nu,\onu=\mc O(d)$,
SSC, and $\Dd^{2}g(x,t)[h,h]\succeq0$ (so SLTSC), and SASC.
\item Second-order cone $\phi_{\textup{SOC}}(x,t)=-\log(t^{2}-\norm{x-\mu}_{\Sigma}^{2})$
for $K_{3}$: $g=d\,\hess\phi_{\textup{SOC}}$ satisfies $\nu,\onu=\mc O(d)$,
SSC, SLTSC, and SASC.
\end{itemize}
\end{thm}

\paragraph{(3) PSD cone.}

Another fundamental constraint is the PSD cone. This convex region
admits a $d$-self-concordant barrier $\phi_{\textup{PSD}}(\cdot)=-\log\det(\cdot)$.
We show that it satisfies SLTSC, while the $d$-scaling further guarantees
SSC and ASC. In establishing ASC, we find an interesting connection
to the \emph{Gaussian orthogonal ensemble} (GOE), one of the main
objects studied in random matrix theory. However, we cannot prove
SASC, so we need the $\frac{d(d+1)}{2}$-scaling for SASC (due to
HSC of $\phi_{\textup{PSD}}$). 
\begin{thm}
[PSD cone] Let $K=\psd$, $X\in\intk$, and $H\in\mbb S^{d}$. Then,
$d\,\hess\phi_{\textup{PSD}}$ satisfies $\nu,\onu=\mc O(d^{2})$,
SSC, $\Dd^{2}g(X)[H,H]\succeq0$ (so SLTSC), and ASC. $\frac{d(d+1)}{2}\,\hess\phi_{\textup{PSD}}$
is SASC.
\end{thm}

\paragraph{(4) Entropy and $\ell_{p}$-norm.}

It is sometime more convenient to introduce $d$ many new variables
as seen in the following:
\begin{thm*}
[Entropy and $\ell_p$-norm] Let $K_{1}=\prod_{i=1}^{d}\{(x_{i},t_{i})\in\R^{2}:x_{i}\geq0,\,t_{i}\geq x_{i}\log x_{i}\}$
and $K_{2}=\prod_{i=1}^{d}\{(x_{i},t_{i})\in\R^{2}:\Abs{x_{i}}^{p}\leq t_{i}\}$.
\begin{itemize}
\item Entropy $\phi_{\textup{ent}}(x,t)=-\sum_{i=1}^{d}\bpar{\log(t_{i}-x_{i}\log x_{i})+36\log x_{i}}$
for $K_{1}$: $g=d\,\hess\phi_{\textup{ent}}$ satisfies $\nu,\onu=\mc O(d^{2})$,
SSC, SLTSC, and SASC.
\item The $p$-th power of $\ell_{p}$-norm $\phi_{\textup{power}}(x,t)=-\sum_{i=1}^{d}\bpar{\log(t_{i}^{2/p}-x_{i}^{2})+72\log t_{i}}$
for $K_{2}$: $g=d\,\hess\phi$ satisfies $\nu,\onu=\mc O(d^{2})$,
SSC, SLTSC, and SASC.
\end{itemize}
\end{thm*}

\subsubsection{Examples ($\S$\ref{sec:examples})}

Our theory (Theorem~\ref{thm:Dikin-annealing} and~\ref{thm:IPM-sampling})
with the study of barriers (Table~\ref{tab:scaling-table}) proposes
local metrics for structured instances. $\gcdw$ with them mixes in
poly-time faster than the $\bw$. For fair comparison, the complexity
of the $\bw$ refers to that of isotropic rounding\footnote{For general logconcave sampling, the $\bw$ needs isotropic rounding,
using $\Otilde(d^{4})$ queries, \textbf{after which} an $\O(1)$-warm
start and isotropy are provided, and then it mixes using additoinal
$\Otilde(d^{2})$ queries \citep{lovasz2007geometry}. Without rounding,
it is not necessarily poly-time mixing. %, so isotropic rounding is an inevitable procedure for fair comparison.
For uniform sampling \emph{only}, the complexity of obtaining isotropy
and an $\O(1)$-warm start was improved to $\Otilde(d^{3})$ by \citet{jia2021reducing}.} (see \S\ref{sec:examples}).

\paragraph{Motivating example.}

Let us introduce a variable for each of $\snorm{X-B}_{F}$ and $\snorm{X-C}_{F}^{2}$.
Then our theory suggests the following barrier: $4(\phi_{\textup{log}}+d^{2}\phi_{\textup{Gaussian}}+d^{2}\phi_{\textup{SOC}}+d^{2}\phi_{\textup{PSD}})$,
which is $\O(1)\,(m+d^{3},m+d^{3})$-self-concordant, SSC, LTSC, and
ASC. By Theorem~\ref{thm:Dikin-annealing} with $\alpha=0$ and $\beta=1$
(due to $\phi_{\textup{PSD}}$ in the potential), we need $\Otilde\bpar{d^{2}(m+d^{3})}$
iterations of the $\dw$ in total.

\paragraph{Uniform and exponential sampling.}

Let us first consider uniform sampling over linear constraints given
by $Ax\geq b$ for $A\in\R^{m\times d}$ and $b\in\R^{m}$. Recall
that for uniform sampling the $\bw$ mixes in $\Otilde(d^{3})$ iterations
(including isotropic rounding). On the other hand, $\Otilde(md)$
queries are enough for $\gcdw$ with the $(m,m)$-Dikin amenable metric
induced by $\phi_{\textup{log}}$. This recovers the mixing time of
\citet{kannan2012random} \emph{without} warmness. If we use the $(\sqrt{md},\sqrt{md})$-Dikin-amenable
Vaidya or $(d^{3/2},d^{3/2})$-Dikin-amenable Lewis-weight metric
instead, then $\gcdw$ with each metric recovers the $\Otilde(m^{1/2}d^{3/2})$
and $\Otilde(d^{5/2})$ mixing of the $\msf{Vaidya\ walk}$ and $\msf{Approximate\ John\ walk}$
\citep{chen2018fast} \emph{without} warmness. For a second-order
cone with linear constraints, we can use the Hessian of $2(\phi_{\textup{log}}+d\phi_{\textup{SOC}})$
that is $(m+d,m+d)$-Dikin-amenable, with which $\gcdw$ mixes in
$\Otilde(d\,(m+d))$ iterations in total. Lastly, for the PSD cone
with linear constraints, we can use the $(m+d^{3},m+d^{3})$-Dikin-amenable
$2\hess(\phi_{\textup{log}}+d^{2}\phi_{\textup{PSD}})$. $\gcdw$
with this needs $\Otilde(d^{2}(m+d^{3}))$ queries. For large $m$,
we use the $(d^{3},d^{3})$-Dikin-amenable $2(dg_{\textup{Lw}}+d^{2}\hess\phi_{\textup{PSD}})$,
with which $\gcdw$ mixes in $\Otilde(d^{5})$ iterations. In the
same setting, the $\bw$ needs $\Otilde(d^{6})$ queries.\\
For exponential sampling, $\gcdw$ requires the same number of iterations
of the $\dw$ for each case (i.e., polytope, second-order cone, PSD),
while the $\bw$ needs $\Otilde(d^{4})$ iterations for the polytope
and second-order cone, and $\Otilde(d^{8})$ iterations for the PSD
cone. Detailed statements on the mixing times and efficient per-step
implementation can be found in \S\ref{subsec:PSD-cone-sampling}.

\paragraph{Uniform sampling over hyperbolic cones.}

\citet{narayanan2016randomized} went beyond linear constraints and
analyzed the $\dw$ for uniform sampling over a convex region given
as the intersection of (1) linear constraints, (2) a hyperbolic cone
with a $\nu_{h}$-SC hyperbolic barrier $\phi_{h}$, and (3) a general
convex set with a $\nu_{s}$-SC barrier $\phi_{s}$. Using $\hess(\phi_{\textup{log}}+d\phi_{h}+d^{2}\phi_{s})$
as a local metric, this work shows that the $\dw$ mixes in $\O\bpar{d\bpar{m+d\nu_{h}+(d\nu_{s})^{2}}}$
steps from a warm start. The term $d(d\nu_{s})^{2}$ induced by self-concordance
alone is typically the largest one in the provable guarantee. Interesting
results of this work arise when $K$ is the intersection of (1) and
(2). Since a hyperbolic barrier is HSC \citep[Theorem 4.2]{guler1997hyperbolic},
the $d$-scaling of a HSC barrier makes it SSC, SLTSC, and SASC. Also,
as a $\nu_{h}$-SC hyperbolic barrier is $\O(\nu_{h})$-symmetric
(implied in \citet[\S4]{guler1997hyperbolic}), it follows that $d\phi_{h}$
is $(d\nu_{h},d\nu_{h})$-Dikin-amenable. Hence, $\phi_{\log}+d\phi_{h}$
induces an $(m+d\nu_{h},m+d\nu_{h})$-Dikin-amenable metric, and the
$\dw$ with this metric mixes in $\O(d\,(m+d\nu_{h}))$ iterations
from a warm start by Theorem~\ref{thm:Dikin}. Without warmness,
\citet{narayanan2016randomized} showed that the $\dw$ started at
$x\in K$, where $s\geq\nicefrac{|p|}{|q|}$ for any chord $\overline{pq}$
of $K$ passing through $x$, mixes in $\mc O\bpar{d(m+d\nu_{h})\bbrack{d\log\bpar{s(m+d\nu_{h})}+\log\frac{1}{\veps}}}$
steps. On the other hand, $\gcdw$ requires only $\O\bpar{d(m+d\nu_{h})\log\frac{d(m+d\nu_{h})}{\veps}}$
iterations.

\paragraph{Gaussian sampling.}

Going forward, we consider only logarithmic barriers for linear constraints.
The $\bw$ for general log-concave distributions mixes in $\Otilde(d^{4})$
iterations. As per our reduction, we first replace a quadratic potential
(coming from the Gaussian distribution) by a new variable, adding
its epigraph to a constraint. For a polytope, one can use the $(m+d,m+d)$-Dikin-amenable
$2\hess(\phi_{\textup{log}}+d\phi_{\textup{Gauss}})$, so $\gcdw$
needs $\Otilde(d\,(m+d))$ iterations of the $\dw$. For the second-order
cone with linear constraints, $\gcdw$ with the $(m+d,m+d)$-Dikin-amenable
metric $3\hess(\phi_{\textup{log}}+d\phi_{\textup{SOC}}+d\phi_{\textup{Gauss}})$
requires $\Otilde(d\,(m+d))$ iterations. For the PSD cone with linear
constraints, $\gcdw$ with the $(m+d^{3},m+d^{3})$-Dikin-amenable
metric $3\hess(\phi_{\textup{log}}+d^{2}\phi_{\textup{PSD}}+d^{2}\phi_{\textup{Gauss}})$
mixes in $\Otilde(d^{2}(m+d^{3}))$ iterations. The $\bw$ is much
slower, requiring $\Otilde(d^{8})$ iterations.

\paragraph{Entropy sampling. }

For a polytope, we use the $(m+d^{2},m+d^{2})$-Dikin-amenable $2\hess(\phi_{\textup{log}}+d\phi_{\textup{ent}})$
in $2d$-dimensional space. Thus, $\gcdw$ needs $\Otilde(d\,(m+d^{2}))$
iterations of the $\dw$. For the second-order cone with linear constraints,
$\gcdw$ with the $(m+d^{2},m+d^{2})$-Dikin-amenable $3\hess(\phi_{\textup{log}}+d\phi_{\textup{SOC}}+d\phi_{\textup{ent}})$,
requires in $\Otilde(d\,(m+d^{2}))$ iterations. Lastly, for the PSD
cone with linear constraints, $\gcdw$ with the $(m+d^{4},m+d^{4})$-Dikin-amenable
$3\hess(\phi_{\textup{log}}+d^{2}\phi_{\textup{PSD}}+d^{2}\phi_{\textup{ent}})$
mixes in $\Otilde(d^{2}(m+d^{4}))$ iterations. The $\bw$ mixes in
$\Otilde(d^{8})$ iterations in this setting.

\paragraph{Discussion.}

The inner loop of the sampling IPM samples from a distribution whose
potential is of the form $c^{\T}x+\alpha\phi(x)$. Thus, the study
of other non-Euclidean samplers for relatively convex and smooth potentials
will be interesting future work. Next, one question unanswered is
if the $d^{2}$-scaling of $\phi_{\textup{PSD}}$ can be improved,
which is mathematically interesting in its own right. The $d$-scaling
for ASC is shown through the random matrix theory, which is challenging
to extend to SASC (see Remark~\ref{rem:challenge-extension-SASC}). 

\subsection{Background and related work\label{subsec:related-work}}

Our problem \eqref{eq:problem} is a special case of \emph{logconcave
sampling}: sample from a distribution $\pi$ with density proportional
to $\exp(-V)$ for a convex function $V$ on $\Rd$. This problem
has spawned a long line of research in several communities, as it
captures various important distributions, including uniform distributions
over convex bodies and Gaussians.

A large body of recent work in machine learning and statistics makes
the assumption of $0\prec\alpha I\preceq\hess V\preceq\beta I$ on
$\Rd$ (i.e., $\alpha$-strong convexity and $\beta$-smoothness of
the potential $V$), where the strong-convexity assumption is sometimes
relaxed to isoperimetry assumptions such as log-Sobolev inequalities
(LSI), Poincaré inequality (PI), and Cheeger isoperimetry. See \citet{chewi2023log}
for a survey on this topic. The guarantees provided on the mixing
time of samplers under this assumption have polynomial dependence
on the condition number defined as $\beta/\alpha$ (or $\alpha$ is
replaced by the isoperimetric constant). These guarantees do not apply
to constrained sampling. For example, in uniform sampling, the simplest
constrained sampling problem, $V$ is set to be a constant within
the convex body and infinity outside the body, which leads to discontinuity
of $V$ and $\beta=\infty$. The sudden change of $V$ around the
boundary requires special consideration, such as small step size,
use of a Metropolis filter, projection, etc., making it a more challenging
problem.

\paragraph{Uniform sampling.}

Uniform sampling can be accomplished through the $\bw$ (\citet{lovasz1993random,kannan1997random})
and $\msf{Hit\text{-}and\text{-}Run}$ (\citet{smith1984efficient}),
both of which only require access to a function proportional to the
density. When a convex body $K\subset\Rd$ satisfies $B_{r}(x_{0})\subset K\subset B_{R}(x_{0})$
for some $x_{0}$, the $\bw$ mixes in $\Otilde\bpar{d^{2}(R/r)^{2}}$
steps from warm start (\citet{kannan1997random}) and $\textsf{Hit-and-Run}$
mixes in $\Otilde\bpar{d^{2}(R/r)^{2}}$ steps from any start\footnote{In this section, \emph{warm start} means polynomial dependence on
the warmness parameter $M$, while \emph{any start} means poly-logarithmic
dependency on $M$. We assume any start unless specified otherwise.} (\citet{lovasz1999hit,lovasz2006hit}). \citet{lovasz2007geometry}
further extended these results to general logconcave distributions.
These algorithms need to use a ``step size'' of $\Omega(1/\sqrt{d})$,
and their mixing is affected by the skewed geometry of the convex
body (i.e., when $R/r\gg1$). The latter can be addressed by first
\emph{rounding} the body, after which the $\bw$ and the $\textsf{Hit-and-Run}$
mix in $\Otilde(d^{2})$ steps from a warm start, due to bounds on
the KLS constant by \citet{chen2021almost,klartag2023logarithmic}
and stochastic localization by \citet{chen2022hit}. The fastest rounding
algorithm by \citet{jia2021reducing} requires $\Otilde(d^{3})$ queries
to a membership oracle, and uses the $\bw$.

\paragraph{Sampling with local geometry.}

The $\bw$ uses the same radius ball for every point in the convex
body. One might want to use a different radius depending on the distance
to the boundary. This by itself does not work as it simply makes the
current point converge to the boundary. However, replacing balls with
ellipsoids whose shape changes based on the proximity to the boundary
does work. Several sampling algorithms are motivated by the use of
local metrics: the $\dw$ (\citet{kannan2012random}), $\msf{Riemannian\ Hamiltonian\ Monte\ Carlo}$
(RHMC), $\msf{Riemannian\ Langevin\ algorithm}$ (\citet{girolami2011riemann}),
etc.

Which local metrics would be suitable candidates? It turns out that
a suitable metric can be derived from self-concordant barriers, a
concept dating back to the development of the interior-point method
in convex-optimization literature (\citet{nesterov1994interior}).
It is well-known that any convex body admits an $d$-self-concordant
barrier such as universal barrier (\citet{nesterov1994interior,lee2021universal})
and entropic barrier (\citet{bubeck2014entropic,chewi2021entropic}),
but these are computationally expensive. Moreover, as noted in \citet{laddha2020strong},
the symmetry parameter of these general barriers is $\Omega(d^{2})$
for $d$-dimensional bodies (even for second-order cones), and so
the resulting complexity for the $\dw$ on a PSD cone is $\Omega(d^{2}\cdot d^{4})=\Omega(d^{6})$.
Thus, there is a need to find barriers that are more closely aligned
with the structure of sets we wish to sample. 

\paragraph{Polytope sampling.}

Samplers such as the $\bw$ and $\msf{Hit\text{-}and\text{-}Run}$
can be used to sample polytopes, but they do not really use any special
properties of polytopes. 

For polytopes with $m$ linear constraints in $d$-dimension ($m>d$),
the first theoretical result via self-concordant barriers dates back
to \citet{kannan2012random} which proposed the $\dw$ with the $m$-self-concordant
logarithmic barrier and established the mixing rate of $\Otilde(md)$
for uniform sampling. \citet{chen2018fast} revisited the idea of
\citet{vaidya1996new} using the $\mc O(\sqrt{md})$-self-concordant
hybrid barrier, which is a hybrid of the volumetric barrier and the
log barrier and leads to a faster interior-point method. They presented
the $\dw$ with the hybrid barrier giving an $\Otilde(\sqrt{m}d^{3/2})$-mixing
guarantee. Lastly, \citet{laddha2020strong} proposed the $\dw$ with
a variant of the $\mc O^{*}(d)$-self-concordant LS barrier based
on Lewis weights, developed by \citet{lee2019solving}, and showed
a mixing rate of $\otilde{d^{2}}$. 

While the next point proposed by all these Markov chains is obtained
by a Euclidean straight line step, the $\msf{Geodesic\ walk}$ and
RHMC use curves (geodesics and Hamiltonian-preserving curves respectively).
\citet{lee2017geodesic} and \citet{lee2018convergence} showed that
for uniform sampling, the $\msf{Geodesic\ walk}$ and RHMC with the
log barrier mix in $\otilde{md^{3/4}}$ and $\otilde{md^{2/3}}$ steps
respectively. \citet{kook2022condition} extended theoretical analysis
of RHMC to truncated exponential distributions and showed that discretization
of Hamilton's equations by practical numerical integrators maintains
a fast mixing rate. \citet{gatmiry2023sampling} showed that just
as the $\dw$ enjoys faster mixing via a barrier with a better self-concordance
parameter, RHMC with a hybrid barrier consisting of the Lewis weights
and log barrier mixes in $\otilde{m^{1/3}d^{4/3}}$ steps. Their proof
is based on developing suitable properties and algorithmic bounds
for Riemannian manifolds.

\paragraph{Generalization of the approach.}

Extending these non-Euclidean methods to general domains (e.g., $\psd$)
and to more general densities (e.g., Gaussian, relatively strong convex
and smooth) to potentially improve the complexity of the problem significantly
beyond the bounds that follow from general convex body sampling, have
been open research directions and motivate our paper. 

\citet{narayanan2016randomized} explored the first direction, analyzing
the $\dw$ for uniform sampling over the intersection of linear constraints,
a hyperbolic cone with a hyperbolic barrier, and a general convex
set with a SC barrier. Our current understanding of the second direction
is rather limited. A line of work has focused on the analysis of first-order
non-Euclidean samplers, such as discretized $\msf{Mirror\ Langevin\ algorithm}$
(MLA) or $\msf{Riemannian\ Langevin\ algorithm}$ (RLA) but under
strong assumptions. For example, \citet{li2022mirror} provided mixing-rate
guarantees of MLA under the \emph{modified self-concordance} of $\phi$
in the setting $\alpha\hess\phi\preceq\hess f\preceq\beta\hess\phi$.
However, the modified self-concordance is not affine-invariant, so
it does not correctly capture affine-invariance of the algorithm.
\citet{ahn2021efficient,gatmiry2022convergence} avoid the modified
self-concordance, analyzing MLA and RLA under an alternative discretization
scheme that requires an exact simulation of the Brownian motion $\hess\phi(X_{t})^{-1/2}\,\D W_{t}$
which is not known to be achievable algorithmically. \citet{gopi2023algorithmic}
proposed a non-Euclidean version of the proximal sampler based on
the log-Laplace transformation (LLT) and analyzed its mixing when
a potential is strongly convex and \emph{Lipschitz} (not smooth) relatively
in $\hess\phi$. However, the LLT has no closed form in general. Recently,
\citet{srinivasan2023fast} analyzed the Metropolis-adjusted MLA under
the relative Lipschitzness of the potential  (i.e., $\norm{\nabla f}_{[\hess\phi]^{-1}}<\infty$)
in addition to the relative convex and smoothness.

Our study of the $\dw$ for general cones and general densities provides
a rather complete picture of zeroth-order non-Euclidean samplers.
It also provides a general framework and improved bounds as well as
a ``handbook'' for structured sampling.

\subsection{Preliminaries and notation\label{subsec:prelim}}

\paragraph{Basics.}

For $n\in\mathbb{N}$, let $[n]:=\{1,\cdots,n\}$. We use $f\lesssim g$
to denote $f\leq cg$ for some universal constant $c>0$. The $\widetilde{\mc O}$
complexity notation suppresses poly-logarithmic factors and dependence
on error parameters. For $a,b\in\Rd$, we denote $a\wedge b:=\min(a,b)$
and $a\vee b:=\max(a,b)$. For $v\in\Rd$, the Euclidean norm (or
$\ell_{2}$-norm) is denoted by $\snorm v_{2}\defeq\sqrt{\sum_{i\in[d]}v_{i}^{2}}$,
and the infinity norm is denoted by $\snorm v_{\infty}\defeq\max_{i\in[d]}|v_{i}|$.
A Gaussian distribution with mean $\mu\in\Rd$ and covariance $\Sigma\in\Rdd$
is denoted by $\mc N(\mu,\Sigma)$.

\paragraph{Matrices.}

We use $\mbb S^{d}$ to denote the set of symmetric matrices of size
$d\times d$. For $X\in\mbb S^{d}$, we call it \emph{positive semidefinite}
(PSD) (resp. \emph{positive definite} (PD)) if $h^{\T}Xh\geq0$ ($>0)$
for any $h\in\Rd$. We use $\psd$ to denote the set of positive definite
matrices of size $d\times d$. Note that their effective dimension
is $d_{s}:=d(d+1)/2$ due to symmetry. For a positive (semi) definite
matrix $X$, its \emph{square root} is denoted as $X^{\half}$, and
is the unique positive (semi) definite matrix satisfying $X^{\half}X^{\half}=X$.
For $A,B\in\mbb S^{d}$, we use $A\preceq B$ ($A\prec B$) to indicate
that $B-A$ is PSD (PD). For a matrix $A\in\Rdd$, its \emph{trace}
is denoted by $\tr(A)=\sum_{i=1}^{d}A_{ii}$. The \emph{operator norm}
and \emph{Frobenius norm} are denoted by $\snorm A_{2}\defeq\sup_{x\in\Rd}\snorm{Ax}_{2}/\snorm x_{2}$
and $\snorm A_{F}\defeq\bpar{\sum_{i,j=1}^{d}A_{ij}^{2}}^{1/2}=\sqrt{\tr(A^{\T}A)}$,
respectively.

\paragraph{Basic operations.}

For $X\in\mbb S^{d}$, its \emph{vectorization} $\vec{(}X)\in\R^{d^{2}}$
is obtained by stacking each column of $X$ vertically. Its symmetric
vectorization $\svec(X)\in\R^{d_{s}}$ is obtained by stacking the
lower triangular part in vertical direction. For a matrix $A\in\Rdd$
and vector $x\in\Rd$, we use $\diag(A)$ to denote the vector in
$\Rd$ with $[\diag(A)]_{i}=A_{ii}$ for $i\in[d]$, $\Diag(A)$ to
denote the diagonal matrix with $[\Diag(A)]_{ii}=A_{ii}$ for $i\in[d]$
and $\Diag(x)$ to denote the diagonal matrix in $\Rdd$ with $[\Diag(x)]_{ii}=x_{i}$
for $i\in[d]$.

\paragraph{Matrix operations.}

For matrices $A,B\in\Rdd$, their inner product is defined as the
inner product of $\vec{(}A)$ and $\vec{(}B)$, denoted by $\langle A,B\rangle=\tr(A^{\T}B)$.
Their \emph{Hadamard product} $A\circ B$ is the matrix of size $d\times d$
defined by $(A\hada B)_{ij}=A_{ij}B_{ij}$ (i.e., obtained by element-wise
multiplication). For $A\in\R^{p\times q}$ and $B\in\R^{r\times s}$,
their \emph{Kronecker product} $A\kro B$ is the matrix of size $pr\times qs$
defined by 
\[
A\otimes B=\left[\begin{array}{ccc}
A_{11}B & \cdots & A_{1q}B\\
\vdots &  & \vdots\\
A_{p1}B & \cdots & A_{pq}B
\end{array}\right]\,,
\]
where $A_{ij}B$ is the matrix of size $r\times s$ obtained by multiplying
each entry of $B$ by the scalar $A_{ij}$. 

\paragraph{Projection matrix, Leverage score and Lewis weights.}

For a full-rank matrix $A\in\R^{m\times d}$ with $m\geq d$, we recall
that $P(A):=A(A^{\T}A)^{-1}A^{\T}$ is the orthogonal projection matrix
onto the column space of $A$. The leverage scores of $A$ is denoted
by $\sigma(A):=\diag\bpar{P(A)}\in\R^{m}$. We let $\Sigma(A):=\Diag\bpar{\sigma(A)}=\Diag\bpar{P(A)}$
and $P^{(2)}(A):=P(A)\circ P(A)$. The $\ell_{p}$-Lewis weights of
$A$ is denoted by $w(A)$, the solution $w$ to the equation $w(A)=\diag\bpar{W^{\nicefrac{1}{2}-\nicefrac{1}{p}}A(A^{\T}W^{1-\nicefrac{2}{p}}A)^{-1}A^{\T}W^{\nicefrac{1}{2}-\nicefrac{1}{p}}}\in\R^{m}$
for $W=\Diag(w)$. When $m<d$ or $A$ is not full rank, both leverage
scores and Lewis weights can be generalized via the Moore-Penrose
inverse in place of the inverse in the definitions.

\paragraph{Derivatives.}

For a function $f:\Rd\to\R$, let $\grad f(x)\in\Rd$ denote the gradient
of $f$ at $x$ (i.e., $[\nabla f(x)]_{i}=\pderiv f{x_{i}}(x)$) and
$\hess f(x)\in\Rdd$ denote the Hessian of $f$ at $x$ (i.e., $[\hess f(x)]_{ij}=\frac{\de^{2}f}{\de x_{i}\de x_{j}}(x)$).
For a matrix function $g:\Rd\to\Rdd$ in $x$, we use $\Dd g$ and
$\Dd^{2}g$ to denote the third-order and fourth-order tensor defined
by $[\Dd g(x)]_{ijk}=\frac{\de[g(x)]_{ij}}{\de x_{k}}$ and $[\Dd^{2}g(x)]_{ijkl}=\frac{\de^{2}[g(x)]_{ij}}{\de x_{k}\de x_{l}}$.
We use the following shorthand notation: $g_{x,h}':=\Dd g(x)[h]$
and $g_{x,h}'':=\Dd^{2}g(x)[h,h]$, where $\Dd^{i}g(x)[h_{1},\dotsc,h_{i}]=\Dd^{i}g(x)[h_{1}\otimes\cdots\otimes h_{i}]$
denote the $i$-th directional derivative of $g$ at $x$ in directions
$h_{1},\dotsc,h_{i}\in\Rd$, i.e.,
\[
\Dd^{i}g(x)[h_{1},\dotsc,h_{i}]=\frac{\D^{i}}{\D t_{1}\cdots\D t_{i}}g\Bpar{x+\sum_{j=1}^{i}t_{j}h_{j}}\bigg|_{t_{1},\dotsc,t_{i}=0}\,.
\]

\paragraph*{Local norm.}

At each point $x$ in a set $K\subset\Rd$, a \emph{local metric}
$g$, denoted as $g_{x}$ or $g(x)$, is a positive-definite inner
product $g_{x}:\Rd\times\Rd\to\R$, which induces the local norm as
$\snorm v_{g(x)}:=\sqrt{g_{x}(v,v)}$. We use $\snorm v_{x}$ to refer
to $\snorm v_{g(x)}$ when the context is clear. When an ambient space
has an orthonormal basis as in our setting (e.g., $\{e_{1},\dots,e_{d}\}$),
the local metric $g_{x}$ can be represented as a positive-definite
matrix of size $d\times d$. In this case, we abuse notation by using
$g(x)$ to indicate the $d\times d$ positive-definite matrix represented
with respect to such an orthonormal basis. Also, the inner product
can be written as $g_{x}(v,w)=v^{\T}g(x)w$. Going forward, we use
$g_{x}=g(x)$ to denote a local metric (or positive definite matrix
of size $\dim(x)\times\dim(x)$) at each point $x\in K$. The local
metric $g$ is assumed to be at least twice differentiable.

\paragraph{Markov chains.}

We use the same symbol for a distribution and its density with respect
to the Lebesgue measure. Many sampling algorithms are based on \emph{Markov
chains}. A \emph{transition kernel} $P:\Rd\times\mc B(\Rd)\to\R_{\geq0}$
(or \emph{one-step distribution}) for the Borel $\sigma$-algebra
$\mc B(\Rd)$ quantifies the probability of the Markov chains transitioning
from one point to another measurable set. The next-step distribution
is defined by $P_{x}(A):=P(x,A)$, which is the probability of a step
from $x$ landing in the set $A$. The transition kernel characterizes
the Markov chain in the sense that if a current distribution is $\mu$,
then the distribution after $n$ steps can be expressed as $\mu P^{(n)}$,
where $\mu P^{(i)}(x):=\int_{\Rd}P(\cdot,x)\,\mu P^{(i-1)}$ is defined
recursively for $i\in[n]$ with the convention $\mu P^{(0)}=\mu$.
We call $\pi$ a \emph{stationary distribution} of the Markov chain
if $\pi=\pi P$. If the stationary distribution further satisfies
$\int_{A}P(x,B)\,\pi(\D x)=\int_{B}P(x,A)\,\pi(\D x)$ for any two
measurable subsets $A,B$, then the Markov chain is said to be \emph{reversible}
with respect to $\pi$.

It is expected that the Markov chain approaches the stationary distribution.
We measure this with the \emph{total variation distance} (TV-distance):
for two distributions $\mu$ and $\pi$ on $\Rd$, the TV-distance
is defined as $\dtv(\mu,\pi)\defeq\sup_{A\in\mc B(\Rd)}|\mu(A)-\pi(A)|=\half\textint_{\Rd}\big|\frac{\D\mu}{\D x}-\frac{\D\pi}{\D x}\big|\,\D x$,
where the last equality holds when the two distributions admit densities
with respect to the Lebesgue measure on $\Rd$. We also recall other
probabilistic distances: when $\mu\ll\nu$,
\begin{align*}
\text{The chi-squared divergence\ } & \chi^{2}(\mu\mmid\nu)\defeq\int\bpar{\frac{\D\mu}{\D\nu}-1}\,\D\nu\,,\\
L^{2}\text{-distance\ } & \snorm{\mu/\nu}\defeq\int\frac{\D\mu}{\D\nu}\,\D\mu=\chi^{2}(\mu\mmid\nu)+1\,.
\end{align*}
Moreover, the rate of convergence can be quantified by the \emph{mixing
time}: for an error parameter $\veps\in(0,1)$ and an initial distribution
$\pi_{0}$, the mixing time is defined as the smallest $n\in\N$ such
that $\dtv(\pi_{0}P^{(n)},\pi)\leq\veps$. In this paper, we consider
a \emph{lazy} Markov chain, which does not move with probability $\texthalf$
at each step, in order to avoid a uniqueness issue of a stationary
distribution. Note that this change worsens the mixing time by at
most a factor of $2$. One of the standard tools to control progress
made by each iterate is the \emph{conductance} $\Phi$ of the Markov
chain with its stationary distribution $\pi$, defined by
\[
\Phi\defeq\inf_{\text{measurable }S}\frac{\int_{S}P(x,S^{c})\,\pi(\D x)}{\pi(S)\wedge\pi(S^{c})}\,.
\]
Another crucial factor affecting the convergence rate is geometry
of the stationary distribution $\pi$, as measured by \emph{Cheeger
isoperimetry} 
\[
\psi_{\pi}\defeq\inf_{\text{measurable }S}\frac{\lim_{\delta\to0^{+}}\frac{1}{\delta}\pi\bpar{\{x:\,0<d(S,x)\leq\delta\}}}{\pi(S)\wedge\pi(S^{c})}\,,
\]
 where $d(S,x)$ is some distance between $x$ and the set $S$.

\paragraph{Full definition of self-concordance.}
\begin{defn}
[Self-concordance] \label{def:sc} For convex $K\subset\Rd$, let
$\phi:\intk\to\R$ be a convex function, $g:\intk\to\psd$ a PSD matrix
function, and $\mc N_{g}^{r}(x):=\mc N\bpar{x,\frac{r^{2}}{d}g(x)^{-1}}$.
\begin{itemize}
\item \emph{Self-concordance} (SC): A $C^{3}$-function $\phi$ is called
a self-concordant barrier if $|\Dd^{3}\phi(x)[h,h,h]|\leq2\snorm h_{\hess\phi(x)}^{3}$
for any $x\in\intk$ and $h\in\Rd$, and $\lim_{x\to\de K}\phi(x)=\infty$.
The first condition is equivalent to $-2\snorm h_{\hess\phi(x)}\hess\phi(x)\preceq\Dd^{3}\phi(x)[h]\preceq2\snorm h_{\hess\phi(x)}\hess\phi(x)$.
We call it a $\nu$-self-concordant barrier for $K$ if $\sup_{h\in\Rd}(2\langle\grad\phi(x),h\rangle-\snorm h_{\hess\phi(x)}^{2})\leq\nu$
for any $x\in\intk$ in addition to self-concordance. A $C^{1}$-PSD
matrix function $g:\intk\to\psd$ is called self-concordant if $-2\snorm h_{g(x)}g\preceq\Dd g(x)[h]\preceq2\snorm h_{g(x)}g$
for any $x\in\intk$ and $h\in\Rd$, and there exists a self-concordant
function $\phi:\intk\to\R$ such that $\hess\phi\asymp g$ on $\intk$.
We call it a $\nu$-self-concordant barrier for $K$ if its counterpart
$\phi$ is $\nu$-self-concordant.
\item \emph{Highly self-concordant function} (HSC): A $C^{4}$-function
$\phi$ is called highly self-concordant if $|\Dd^{4}\phi(x)[h,h,h,h]|\leq6\snorm h_{\hess\phi(x)}^{4}$
for any $x\in\intk$ and $h\in\Rd$, and $\lim_{x\to\de K}\phi(x)=\infty$.
\item \emph{Strong self-concordance} (SSC): A SC matrix function $g$ is
called strongly self-concordant if $g$ is PD on $\intk$ and $\snorm{g(x)^{-\nicefrac{1}{2}}\Dd g(x)[h]\,g(x)^{-\nicefrac{1}{2}}}_{F}\leq2\snorm h_{g(x)}$
for any $x\in\intk$ and $h\in\Rd$. We call a SC function $\phi$
strongly self-concordant if $\hess\phi(x)$ is strongly self-concordant.
\item \emph{Lower trace self-concordant matrix} (LTSC): A SC matrix function
$g$ is called lower trace self-concordant if $g$ is PD on $\intk$
and $\tr\bpar{g(x)^{-1}\Dd^{2}g(x)[h,h]}\geq-\snorm h_{g(x)}^{2}$
for any $x\in\intk$ and $h\in\Rd$. We call it strongly lower trace
self-concordant (SLTSC) if for any PSD matrix function $\bar{g}$
on $\intk$ it holds that $\tr\bpar{\bpar{\bar{g}(x)+g(x)}^{-1}\Dd^{2}g(x)[h,h]}\geq-\snorm h_{g(x)}^{2}$
for any $x\in\intk$ and $h\in\Rd$.
\item \emph{Average self-concordance} (ASC): A matrix function $g$ is called
average self-concordant if for any $\veps>0$ there exists $r_{\veps}>0$
such that $\P_{z\sim\mc N_{g}^{r}(x)}\bpar{\snorm{z-x}_{g(z)}^{2}-\snorm{z-x}_{g(x)}^{2}\leq\frac{2\veps r^{2}}{d}}\geq1-\veps$
for $r\leq r_{\veps}$. We call it strongly average self-concordant
(SASC) if for $\veps>0$ and any PSD matrix function $\bar{g}$ on
$\intk$ it holds that $\P_{z\sim\mc N_{g+\bar{g}}^{r}(x)}\bpar{\snorm{z-x}_{g(z)}^{2}-\snorm{z-x}_{g(x)}^{2}\leq\frac{2\veps r^{2}}{d}}\geq1-\veps$
for $r\leq r_{\veps}$.
\end{itemize}
\end{defn}

\section{Mixing of \texorpdfstring{$\dw$}{Dikin walks} \label{sec:mixing-Dikin}}

We follow a standard conductance based argument (see e.g., \citet{lovasz1993random,vempala2005geometric}).
A lower bound on the conductance of a Markov chain provides an upper
bound on the mixing time of the Markov chain due to the following
result.
\begin{lem}
[\citet{lovasz1993random}] \label{lem:conductanceBound} Let $\pi_{T}$
be the distribution obtained after $T$ steps of a lazy reversible
Markov chain of conductance at least $\Phi$ with stationary distribution
$\pi$ and initial distribution $\pi_{0}$. For $\snorm{\pi_{0}/\pi}=\E_{\pi_{0}}\bbrack{\deriv{\pi_{0}}{\pi}}$
and any $\veps>0$, we have $\dtv(\pi_{T},\pi)\leq\veps+\sqrt{\frac{\snorm{\pi_{0}/\pi}}{\veps}}\bpar{1-\frac{\Phi^{2}}{2}}^{T}$.
\end{lem}

A lower bound on the conductance follows from two ingredients: \textbf{(i)}
one-step coupling and \textbf{(ii)} isoperimetry. The first refers
to showing that the one-step distributions of the $\dw$ from two
nearby points have TV-distance bounded away from one. The second is
a purely geometry property about the expansion of the target distribution.
Combining these two leads to a lower bound on the conductance:
\begin{lem}
[\citet{kook2022condition}, Adapted from Proposition 9] \label{lem:conductance}
Let $\pi$ be the stationary distribution of a lazy reversible Markov
chain on $\mc M$ with a transition kernel $P_{x}$. Assume the isoperimetry
$\psi_{\mc M}$ under a Riemannian distance $d_{g}$ and the following
one-step coupling: if $\snorm{x-y}_{g(x)}\leq\Delta<1$ for $x,y\in\mc M$,
then $\dtv(P_{x},P_{y})\leq0.9$. Then the conductance $\Phi$ of
the Markov chain is bounded lower by $\Omega(\psi_{\mc M}\Delta)$.
\end{lem}

\subsection{One-step coupling and isoperimetry}

Recall that a $\onu$-Dikin-amenable metric is $\onu$-symmetric,
SSC, LTSC, and ASC. \citet{laddha2020strong} was the first to attempt
characterizing essential properties of $g$ (or $\phi$) that determine
mixing times of $\dws$ for uniform sampling. Their framework necessitates
that $g$ satisfies $\onu$-symmetric, SSC, convexity of $\log\det g(x)$,
and $x\in\mc D_{g}^{r}(z)$ w.h.p. (where $z\sim\text{Unif}\bpar{\dcal_{g}^{r}(x)}$). 

However, their framework encounters a challenge when further incorporating
the work of \citet{narayanan2016randomized}, which analyzes the $\dw$
for uniform sampling over a convex region given as the intersection
of various convex sets. The challenge arises from the difficulty of
verifying the convexity of $\log\det(g_{1}+g_{2})$ when $\log\det g_{i}$
is convex for each $i=1,2$.

To address this challenge and succinctly characterize essential characteristics
of a metric for one-step coupling, we relax the convexity of $\log\det$
to (S)LTSC and introduce the notion of ASC to account for the condition
``$x\in\mc D_{g}^{r}(z)$ w.h.p.''. We show that one-step coupling
lemma below, one of main proof ingredients in obtaining a mixing-time
guarantee of the $\dw$, can be established under Dikin-amenability
of a metric. Our characterization of a metric for achieving one-step
coupling is general and unifies previous work on $\dws$ (\citet{kannan2012random,narayanan2016randomized,chen2018fast,laddha2020strong}).

We now proceed to establish one-step coupling under the relative smoothness
in $\phi$.
\begin{lem}
[One-step coupling]\label{lem:one-step} For convex $K\subset\Rd$,
let $g:\intk\to\pd$ be SSC, ASC, LTSC, and $\phi:\intk\to\R$ be
its function counterpart. Suppose that the potential $f$ of the target
distribution $\pi$ is $\beta$-relatively smooth in $\phi$. Then
there exist constants $s_{1},s_{2}>0$ such that if $\snorm{x-y}_{g(x)}\leq s_{1}r/\sqrt{d}$
with $r=s_{2}\,(1\wedge\nicefrac{1}{\sqrt{\beta}})$ for $x,y\in\intk$,
then $\dtv(P_{x},P_{y})\leq\frac{3}{4}+0.01$. 
\end{lem}

We provide a sketch of the proof (see \S\ref{proof:onestep} for
the full proof). A key distinction when extending beyond uniform distributions
lies in establishing a lower bound for the ratio $\frac{\exp(f(x))}{\exp(f(z))}$
to ensure a high acceptance probability. To tackle this issue, we
use the symmetry of the proposal distribution, claiming $\nicefrac{\exp(f(x))}{\exp(f(z))}\geq1$
at the expense of $\texthalf$ probability. However, this $\texthalf$
probability loss is incompatible with previous proof techniques based
on the triangle inequality: for a transition kernel $T$ and proposal
kernel $P$, the triangle inequality leads to 
\[
\dtv(T_{x},T_{y})\leq\dtv(T_{x},P_{x})+\dtv(P_{x},P_{y})+\dtv(P_{y},T_{y})\,,
\]
and then bound the second term in the RHS by Pinsker's inequality,
making it arbitrarily small by taking $r=\O(1)$ small enough. However,
this approach yields a bound of $\texthalf+\veps$ for both $\dtv(T_{x},P_{x})$
and $\dtv(T_{y},P_{y})$, making the RHS vacuous.

We instead work with the exact formula for $\dtv(T_{x},T_{y})$: for
the Gaussian $p_{x}=\ncal(x,\frac{r^{2}}{d}g(x)^{-1})$, 
\[
R_{x}(z)=\frac{p_{z}(x)}{p_{x}(z)}\frac{\pi(z)}{\pi(x)}=\sqrt{\frac{\det g(z)}{\det g(x)}}\,\frac{\exp(f(x))}{\exp(f(z))},\qquad A_{x}(z)=\min\bpar{1,R_{x}(z)\,\mathbf{1}_{K}(z)}\,,
\]
the transition kernel $T_{x}$ of the $\dw$ started at $x$ can be
written as 
\[
T_{x}(dz)=\underbrace{\bpar{1-\E_{p_{x}}[A_{x}(\cdot)]}}_{\eqqcolon r_{x}}\,\delta_{x}(\D z)+A_{x}(z)\,p_{x}(\D z)\,.
\]
Then, 
\begin{align*}
\dtv(T_{x},T_{y}) & =\frac{r_{x}+r_{y}}{2}+\half\int|A_{x}(z)\,p_{x}(z)-A_{y}(z)\,p_{y}(z)|\,\D z\,.
\end{align*}

As for $r_{x}$ and $r_{y}$, we bound below $\sqrt{\nicefrac{\det g(z)}{\det g(x)}}$
by $1-\veps$ at the cost of $\veps$-probability through SSC, LTSC,
and ASC of $g$, following \citet{laddha2020strong} with convexity
of $\log\det$ replaced by LTSC. As mentioned earlier, we also deduce
$\nicefrac{\exp(f(x))}{\exp(f(z))}\geq1$ through the symmetry of
Gaussian distributions at the cost of $\texthalf$ probability. Combining
these results, we obtain upper bounds of $\texthalf+\veps$ for small
$\veps>0$ on $r_{x}$ and $r_{y}$.

Establishing a bound of $\nicefrac{1}{4}+\veps$ on the second term
is a more involved task. It requires the closeness of acceptance probabilities
$A_{x}(z)$ and $A_{y}(z)$ as well as the probability densities $g_{x}(z)$
and $g_{y}(z)$. This closeness can only be achieved through sophisticated
conditioning on high-probability events due to ASC, SSC, and symmetry
of Gaussian proposals. To be precise, define good events $G_{x}=\cap_{i=0,2,3}B_{x,i}^{c}$
and $G_{y}=\cap_{i=0,2,3}B_{y,i}^{c}$ such that $\P_{\ncal_{g}^{r}(x)}(G_{x}^{c})\leq3\veps$
and $\P_{\ncal_{g}^{r}(y)}(G_{y}^{c})\leq3\veps$, where 
\begin{align*}
B_{x,0} & =\{\norm{z-x}_{x}\geq cr\}\,\ \text{with }c\geq1+\frac{2}{\sqrt{d}}\,\log\frac{1}{\veps}\,,\quad\text{(Tail bound for Gaussian)}\\
B_{x,1} & =\{-\langle\nabla f(x),x-z\rangle\leq0\}\,,\quad\text{(Symmetry of Gaussian)}\\
B_{x,2} & =\{\snorm{z-x}_{z}^{2}-\snorm{z-x}_{x}^{2}>2\veps\frac{r^{2}}{d}\}\,,\quad\text{(ASC of }g)\\
B_{x,3} & =\bbrace{\langle\grad\vphi(x),z-x\rangle\leq-2\frac{r}{\sqrt{d}}\,\snorm{g(x)^{-1/2}\grad\vphi(x)}_{2}\,\log\frac{1}{\veps}}\,.\quad\text{(SSC \& tail bound for Gaussian)}
\end{align*}
We further denote $G:=G_{x}\cup G_{y}$ and a partition of $G$ by
\[
G_{x\backslash y}:=G_{x}\backslash G_{y},\qquad G_{x,y}:=G_{x}\cap G_{y},\qquad G_{y\backslash x}:=G_{y}\backslash G_{x}\,.
\]
Then,
\begin{align*}
\half\int\underbrace{|A(x,z)\,p_{x}(z)-A(y,z)\,p_{y}(z)|}_{\eqqcolon Q}\,\D z & \leq3\veps+\underbrace{\half\int_{G_{x\backslash y}}Q\,\D z}_{\eqqcolon\mc A}+\underbrace{\half\int_{G_{y\backslash x}}Q\,\D z}_{\eqqcolon\mc B}+\underbrace{\half\int_{G_{x,y}}Q\,\D z}_{\eqqcolon\mc C}\,.
\end{align*}
We can bound $\acal$ and $\bcal$ by $\mc O(\veps)$ by Pinsker's
inequality and a well-known formula for the $\KL$ divergence between
two Gaussians. As for $\mc C$, conditioning on $B_{x,1}$ and using
the triangle inequality lead to

\[
\mc C\leq\frac{1}{4}+2\veps+\half\int_{G_{x}\cap G_{y}\cap B_{x,1}^{c}}\Big|\min\Bpar{1,\underbrace{\frac{\exp f(x)}{\exp f(z)}\,\frac{p_{z}(x)}{p_{x}(z)}}_{\eqqcolon\msf U}}-\min\Bpar{\underbrace{\frac{p_{y}(z)}{p_{x}(z)}}_{\eqqcolon\msf V},\underbrace{\frac{\exp f(y)}{\exp f(z)}\,\frac{p_{z}(y)}{p_{x}(z)}}_{\eqqcolon\msf W}}\Big|\,p_{x}(z)\,\D z\,.
\]
The bound of $\log\msf U\ge-4\veps$ was already obtained when bounding
$r_{x}$. We then show that $\lvert\log\msf V\rvert\le5\veps$ and
$\log\msf W\ge-7\veps$ conditioned on $G_{x}\cap G_{y}\cap B_{x,1}^{c}$
via closeness of SSC (Lemma~\ref{lem:strongSC-closeness}). Using
these, 
\[
\int_{G_{x}\cap G_{y}\cap B_{x,1}^{c}}|1\wedge\msf U-\msf V\wedge\msf W|\,p_{x}(z)\,\D z\leq e^{5\veps}-e^{4\veps}\,,
\]
which results in $\mc C\le1/4+\O(\veps)$. Putting the bounds on $r_{x},r_{y},\mc A,\mc B$,
and $\mc C$ together, we conclude that the TV-distance is bounded
by $3/4+\O(\veps)$.
\begin{rem}
We further note that $\snorm{x-y}_{x}$ can be replaced by the Riemannian
distance $d_{\phi}(x,y)$ with the metric defined by $\hess\phi$,
since these two distance are within a constant factor of each other:
\begin{lem}
[\citet{nesterov2002riemannian}, Lemma 3.1] \label{lem:Riemann-Dikin-close}
Let $\phi:\intk\to\R$ be self-concordant, and $x,y\in\intk$ with
$\delta:=\snorm{x-y}_{x}<1$. Then,
\[
\delta-\half\delta^{2}\leq d_{\phi}(x,y)\leq-\log(1-\delta)\,.
\]
\end{lem}

\end{rem}

Next, we present two isoperimetric inequalities derived from distinct
sources: the first comes from the symmetry of a barrier, while the
second arises from strong convexity in a local metric.

\paragraph{Isoperimetry via barrier parameters.}

The first one states that isoperimetry of log-concave distributions
under distance $d_{g}(x,y)$ (or $\snorm{x-y}_{g(x)}$ due to Lemma~\ref{lem:Riemann-Dikin-close})
is $\Omega(1/\sqrt{\onu})$. The following lemma is an extension of
\citet{laddha2020strong} from uniform distributions (over a convex
body) to general log-concave distributions. We defer the proof to
\S\ref{proof:isoperimetry}.
\begin{lem}
\label{lem:symmetry-iso} Let $\phi$ be self-concordant and $d_{\phi}$
be the Riemannian distance induced by the Hessian metric $\hess\phi$.
For a log-concave distribution $\pi$, isoperimetry $\psi_{\pi}$
under distance $d_{\phi}$ is $\Omega(1/\sqrt{\onu})$.
\end{lem}

\paragraph{Isoperimetry from relative strong convexity.}

Another kind of isoperimetry comes from relative strong-convexity
of the potential of a distribution. For a scalar $\alpha>0$, isoperimetry
of $e^{-\alpha\phi}$ on a Hessian manifold equipped with the metric
$\hess\phi$ is $\Omega(\sqrt{\alpha})$ if $\Dd^{4}\phi(x)\Brack{h^{\otimes4}}\geq0$
for all $x\in K$ and $h\in\Rd$ (see \citet[Lemma 37]{lee2018convergence}).
\citet[Lemma 9]{gopi2023algorithmic} further generalizes this to
show that if $\phi$ is self-concordant and the potential $f$ is
$\alpha$-relatively strong convex, then its isoperimetry is $\Omega(\sqrt{\alpha})$.
We can adapt this lemma by restricting this to a convex set $K$ (not
necessarily bounded). See \S\ref{proof:isoperimetry} for the proof.
\begin{lem}
[\citet{gopi2023algorithmic}, Adapted from Lemma 9] \label{lem:sc-iso}
For a closed convex set $K\subset\Rd$, let a convex function $\phi:\intk\to\R$
be self-concordant on $K$, $f:\intk\to\R$ $\alpha$-relatively strongly
convex in $\phi$, and $\pi$ a log-concave distribution with $\pi\propto\exp(-f)\cdot\mathbf{1}_{K}$.
For a partition $\{S_{1},S_{2},S_{3}\}$ of $K$ and the Riemannian
distance $d_{\phi}$ induced by the inner product $\langle a,b\rangle_{x}:=a^{\T}\hess\phi(x)\,b$,
it holds that 
\[
\pi(S_{3})\gtrsim\sqrt{\alpha}\,d_{\phi}(S_{1},S_{2})\,\pi(S_{1})\,\pi(S_{2})\,.\qedhere
\]
\end{lem}

\subsection{Mixing time: Proof of Theorem~\ref{thm:Dikin}}

Putting all these components together, we obtain the following mixing-time
bounds for the $\dw$. 

\thmDikin*
\begin{proof}
Lemma~\ref{lem:conductance} ensures that $\Phi\gtrsim\frac{r}{\sqrt{d}}\psi$
due to the one-step coupling in Lemma~\ref{lem:one-step}. Lemma~\ref{lem:symmetry-iso}
leads to $\psi\gtrsim\frac{1}{\sqrt{\onu}}$, while Lemma~\ref{lem:sc-iso}
implies $\psi\gtrsim\sqrt{\alpha}$ due to $\hess\phi\asymp g$. Thus,
\[
\Phi\gtrsim\frac{1}{\sqrt{d}}\,\bpar{\sqrt{\alpha}\vee\frac{1}{\sqrt{\onu}}}\bpar{1\vee\frac{1}{\sqrt{\beta}}}\,,
\]
and using Lemma~\ref{lem:conductanceBound}, we can enforce $\dtv(\pi_{T},\pi)\leq\veps$
by solving $\sqrt{\Lambda}e^{-T\Phi^{2}/2}\leq\veps$ and $\frac{\veps}{2}+\sqrt{\frac{\Lambda}{\veps/2}}e^{-T\Phi^{2}/2}\leq\veps$
for $T$, which results in 
\[
T\gtrsim d\,(1\vee\beta)\,\bpar{\onu\wedge\frac{1}{\alpha}}\log\frac{\Lambda}{\veps}\,.\qedhere
\]
\end{proof}

\section{Gaussian cooling on manifolds revisited: IPM framework for sampling
\label{sec:IPM-framework}}

We derive a sampling analogue of the Interior-Point Method through
comparison with IPM in optimization, by extending \emph{Gaussian cooling
on manifolds} introduced in \citet{cousins2018gaussian,lee2018convergence}.
Combining the sampling IPM framework with the $\dw$ efficiently generates
a warm start for a target distribution $\pi\propto e^{-f}\cdot\mathbf{1}_{K}$
with finite second moment.

\subsection{Derivation of sampling IPM \label{subsec:derivation-IPM-sampling}}

Let us recall our setup. Let $K\subset\Rd$ be a closed convex set,
$g:\intk\to\pd$ a $(\nu,\onu)$-SC matrix function, and $\phi:\intk\to\R$
its (strictly convex) SC counterpart. We assume $\min_{x}\phi(x)=0$
by considering $\phi-\min_{x}\phi(x)$ (here, $\arg\min\phi(x)$ can
be efficiently found by the optimization IPM). We assume that $f$
is $\alpha$-relatively strongly convex and $\beta$-relatively smooth
in $\phi$ for $0\leq\alpha\leq\beta<\infty$, i.e., $0\preceq\alpha\hess\phi\preceq\hess f\preceq\beta\hess\phi$
on $\inter(K)$. We define $\bar{f}(\cdot):=\frac{\nu}{d}\,f(\cdot)$
and $g_{\phi}(\cdot):=\hess\phi(\cdot)$.

\begin{algorithm2e}[t]

\caption{Interior-Point Method} \label{alg:IPM}

\SetAlgoLined

\textbf{Input:} A $\nu$-self-concordant barrier $\phi$ for a constraint

\textbf{Output:} $y_{\lda}$

Denote $f_{\lda}(y):=c^{\T}y+\frac{1}{\lda}\,\phi(y)$.

\tcp{Phase 1: Starting feasible point}

Find $y_{0}=\arg\min\phi(y)$, set $\lda=\frac{1}{6}\,\snorm c_{[\hess\phi(y_{0})]^{-1}}^{-1}$,
and $\bar{y}_{\lda}\gets y_{0}$.

\tcp{Phase 2: Increasing $\lda$ until $\lda\leq\frac{\nu+1}{\veps}$}

\While{$\lda\leq\frac{\nu+1}{\veps}$}{

$\bar{y}_{\lda}\gets\bar{y}_{\lda}-[\hess f_{\lda}(\bar{y}_{\lda})]^{-1}\grad f_{\lda}(\bar{y}_{\lda})$
\tcp{``Opt. step'' (e.g., the Newton step)}

$\lda\gets(1+r)\,\lda$ with $r=\frac{1}{9\sqrt{\nu}}$. \tcp{Increase $\lda$}

}

\end{algorithm2e}

\paragraph{Interior-point method for optimization.}

A structural convex optimization problem is formulated as $\min_{x\in K}f(x)$,
where $f:\Rd\to\R$ is a convex function, and $K\subset\Rd$ is a
closed convex set. Also, both $K$ and $\{(x,t):f(x)\leq t\}$ admit
efficiently computable self-concordant barriers denoted by $\phi_{1}$
and $\phi_{2}$, respectively. We can simplify the problem by equivalently
solving $\min_{x\in K,\,\{(x,t):f(x)\leq t\}}t$ and in general focus
on $\min_{x\in K,\,\{(x,t):f(x)\leq t\}}c^{\T}(x,t)$ for a constant
$c\in\R^{d+1}$. 

IPM then regularizes $c^{\T}(x,t)$ by adding $\frac{1}{\lda}\,\phi(x,t)=\frac{1}{\lda}\bpar{\phi_{1}(x)+\phi_{2}(x,t)}$
for $\lda>0$. This regularization removes the hard constraint of
$K\cap\{f(x)\leq t\}$, and the resulting formulation becomes
\[
\min_{y=(x,t)\in\R^{d+1}}f_{\lda}(y):=c^{\T}y+\frac{1}{\lda}\,\phi(y)\,,
\]
where $\phi(y)$ blows up as $y$ approaches the boundary of the constraint.
For each fixed $\lda>0$, there exists a minimum $y_{\lda}$ of the
convex function $f_{\lda}(y)$. Intuitively, as $\lda\to\infty$ the
regularization term $\frac{1}{\lda}\,\phi(y)$ vanishes, so $y_{\lda}$
converges to $\arg\min_{y\in K\cap\{f(x)\le t\}}c^{\T}y$. The path
followed by $\{y_{\lda}\}_{\lda>0}$ is called the \emph{central path},
and IPM aims to approximately follow this central path as $\lda$
increases. 

To be precise, suppose that for $\lda_{1}>0$, an approximation solution
$\bar{y}_{\lda_{1}}$ maintained by IPM is close enough to $y_{\lda_{1}}$.
Then IPM takes an optimization step (e.g., a Newton step), which takes
into account the local geometry induced by the Hessian of the barrier
$\phi$, to find an approximate solution $\bar{y}_{\lda_{2}}$ when
$\lda_{2}>\lda_{1}$. As long as $\bar{y}_{\lda_{1}}$ is sufficiently
close to $y_{\lda_{1}}$, this approximate solution $\bar{y}_{\lda_{1}}$
serves a good starting point for the non-Euclidean optimizer, which
takes $\bar{y}_{\lda_{1}}$ to $\bar{y}_{\lda_{2}}$. IPM alternates
between increasing $\lda$ and updating $\bar{y}_{\lda}$, until $\lda$
reaches $\nu/\veps$. This is described formally in Algorithm~\ref{alg:IPM}.

The ideas behind IPM are justified by the following theoretical guarantee:
Algorithm~\ref{alg:IPM} returns $y$ in $\mc O\bpar{\sqrt{\nu}\,\log\bpar{\frac{\nu}{\veps}\snorm c_{[\hess\phi(y_{0})]^{-1}}}}$
iterations such that $c^{\T}y\leq c^{\T}y^{*}+\veps$ for $y^{*}=\arg\min_{y\in K\cap\{f(x)\le t\}}c^{\T}y$.

\paragraph{Translation to sampling.}

Now let us adapt each step of IPM into the sampling context with the
conceptual analogy between convex optimization and logconcave sampling
in mind: For convex $K\subset\Rd$ and convex function $f:K\to\R$
\begin{align*}
\min f(x) & \quad\longleftrightarrow\quad\text{sample }x\sim\exp(-f)\\
\text{s.t. }x\in K & \qquad\qquad\quad\text{s.t. }x\in K\,.
\end{align*}
Similar to the optimization IPM, we first replace $f(x)$ by a new
variable $t$ and add the constraint $\{f(x)\leq t\}$ (which is convex
due to convexity of $f$), resulting in the following sampling problem:
sample $(x,t)$ from a distribution with density proportional to $e^{-t}$
subject to $x\in K$ and $\{(x,t)\in\R^{d+1}:f(x)\leq t\}$. We note
that this is indeed an equivalent sampling problem, since the $x$-marginal
of the distribution is $\exp(-f)\cdot\mathbf{1}_{K}$:
\[
\int_{\{(x,t)\in\R^{d+1}:f(x)\leq t\}}\exp(-t)\cdot\mathbf{1}_{K}(x)\,\D t=\int_{f(x)}^{\infty}\exp(-t)\cdot\mathbf{1}_{K}(x)\,\D t=\exp(-f)\cdot\mathbf{1}_{K}\,.
\]

\begin{figure}[t]
\includegraphics[width=\textwidth]{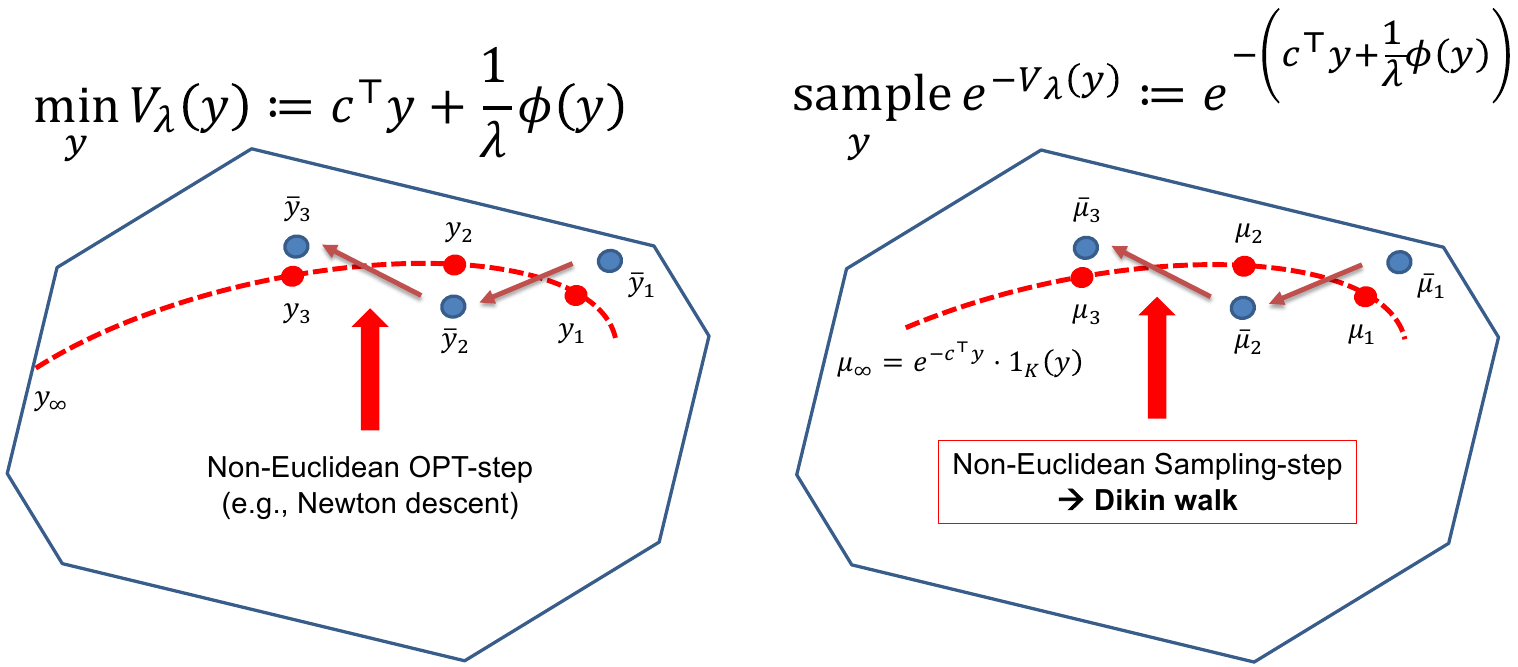}
\centering\caption{\label{fig:opt-samp-IPM} Comparison between the optimization IPM
and the sampling IPM.}
\end{figure}

Now assume that $K\cap\{f(x)\leq t\}$ admits a barrier $\phi$. Thus,
this motivates our focus on sampling from distributions of the form
$\exp(-c^{\T}y)$ subject to a convex region $K$ with a barrier $\phi$,
where $y:=(x,t)\in\R^{d+1}$ is a variable in the augmented space
and $c\in\R^{d+1}$ is a vector.

Regularizing the potential $c^{\T}y$ of the distribution by adding
$\frac{1}{\sigma^{2}}\,\phi(y)$ for some $\sigma^{2}>0$, we can
ignore the hard constraint $K$ and obtain the following formulation:
for $f_{\sigma^{2}}:=\langle c,\cdot\rangle+\frac{1}{\sigma^{2}}\,\phi$,
\[
\text{sample }y\sim\mu_{\sigma^{2}}\propto\exp(-f_{\sigma^{2}}(y))=\exp\Bpar{-\bpar{c^{\T}y+\frac{1}{\sigma^{2}}\,\phi(y)}}\,,
\]
where $\phi(y)$ goes to infinity as it approaches the boundary of
$K$. The regularization $\frac{1}{\sigma^{2}}\,\phi$ vanishes as
$\sigma^{2}\to\infty$, so we can expect $\mu_{\sigma^{2}}\to\pi\propto\exp(-\langle c,\cdot\rangle)\cdot\mathbf{1}_{K}$.
Comparing this with the optimization IPM, the path of measures $\{\mu_{\sigma^{2}}\}_{\sigma^{2}>0}$
can be viewed as the central path in the space of measures. In an
ideal scenario, a sampling IPM should closely follow this central
path while increasing $\sigma^{2}$ along the path. To this end, we
update the current distribution $\bar{\mu}_{\sigma^{2}}$, which is
already close to $\mu_{\sigma^{2}}$ on the central path. This update
should leverage a \emph{sampling step} that is aware of the local
geometry induced by $\hess\phi$, which may involve running a non-Euclidean
sampler such as the $\dw$. This update brings $\bar{\mu}_{\sigma^{2}}$
to a new distribution $\bar{\mu}_{\sigma^{2}+\delta}$ that should
be close to $\mu_{\sigma^{2}+\delta}$ for small $\delta>0$, while
$\bar{\mu}_{\sigma^{2}}$ serves a good starting point for this sampling
step to find $\bar{\mu}_{\sigma^{2}+\delta}$. This procedure is repeated
until $\sigma^{2}$ becomes large enough.

To use this sampling IPM, we further refine the framework via \emph{Gaussian
cooling on manifolds}.

\begin{figure}[t]
\centering
\includegraphics[width=\textwidth]{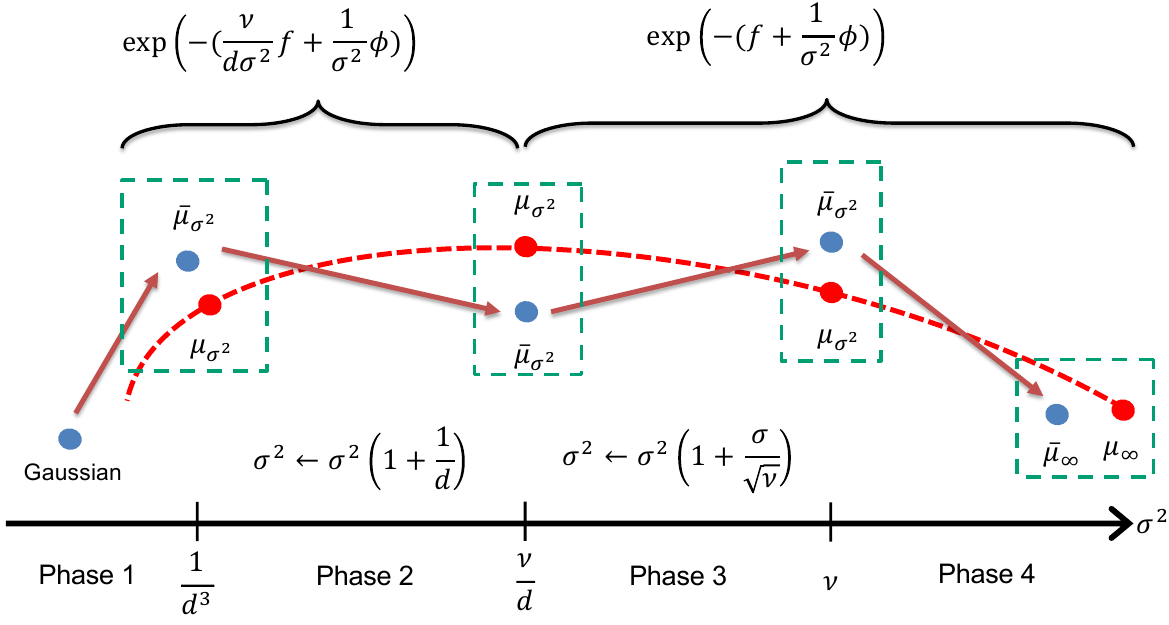}\caption{\label{fig:gc} We refine the derived sampling IPM to obtain the $\protect\gc$
on manifolds. The red dashed line indicates a centra path of measures.
The red dots are target probability measures appearing in the sampling
IPM, while blue dots are probability measures given by a non-Euclidean
sampler, which are approximately close to those target measures (red
dots). Closeness of two dots (bounded by the green dashed boxes) is
quantified by the TV-distance.}
\end{figure}

\paragraph{Comparison with the Gaussian cooling on manifolds (GCM).}

Gaussian Cooling introduced in \citet{cousins2018gaussian} was extended
to manifolds by \citet{lee2018convergence}. It was initially proposed
for volume computation but shares remarkable similarities with our
sampling IPM. In fact, GCM can be identified with the sampling IPM
with $c=0$ (i.e., uniform sampling) and the Riemannian Hamiltonian
Monte Carlo employed for the non-Euclidean sampling step.

Returning to the comparison with the optimization IPM, we note that
two algorithms use different rules for updating $\sigma^{2}$. While
the optimization IPM updates $\sigma^{2}\gets\bpar{1+\frac{1}{\sqrt{\nu}}}\sigma^{2}$,
GCM utilizes two distinct annealing schemes: 
\[
\sigma^{2}\gets\begin{cases}
\sigma^{2}\,\bpar{1+\frac{1}{\sqrt{d}}} & \text{if }\sigma^{2}\leq\frac{\nu}{d}\\
\sigma^{2}\,\bpar{1+\frac{\sigma}{\sqrt{\nu}}} & \text{o.w.}
\end{cases}
\]
While the first type of update in the small regime of $\sigma^{2}$
relies on a property of logconcavity of regularized distributions
$\mu_{\sigma^{2}}\propto\exp\bpar{-\bpar{s\phi(y)+c^{\T}y}}$, the
second type of update in the large regime of $\sigma^{2}$ is justified
by concentration of measure $e^{-s\phi}$ in a thin shell for $s>0$.
We note that the second type in fact accelerates the annealing process.

However, significant challenges remain for the sampling IPM. First,
we need to extend this annealing scheme to exponential distributions
(recall that GCM was proposed for uniform sampling). To be precise,
we must account for the linear term $c^{\T}y$ (in addition to the
$\phi$ term) when designing the annealing scheme. Unfortunately,
the previous update scheme (which is applied only to $\phi$ part)
with its analysis do not go through for this purpose. 

To address this issue, we introduce a further generalization of the
GCM annealing scheme in the small regime of $\sigma^{2}$, enabling
us to leverage logconcavity of $\mu_{\sigma^{2}}$. In the large regime
of $\sigma^{2}$, we use the same annealing scheme but employ a different
analytical approach, utilizing a functional inequality with no need
to quantify the thin-shell phenomenon of $\mu_{\sigma^{2}}$. 

To discuss another remaining issue, we note that a non-Euclidean sampler
used in the sampling step must have a provable mixing-time guarantee
for $\mu_{\sigma^{2}}$. We already provided this through Theorem~\ref{thm:Dikin}
in \S\ref{sec:mixing-Dikin} for the $\dw$, since the target potential
is $s$-relatively strongly convex and $s$-relatively smooth in $\phi$!

\subsection{IPM algorithm for sampling}

Our algorithm consists of four phases, where each phase updates a
current distribution in a different way. For generality, we present
this annealing process for a general potential $f$ instead of linear
functions, where $\alpha\hess\phi\preceq\hess f\preceq\beta\hess\phi$.

\begin{algorithm2e}[t]

\caption{Interior-Point Method for sampling} \label{alg:IPM-sampling}

\SetAlgoLined

\textbf{Input:} Target accuracy $\veps$, local metric $g$, its counterpart
$\phi$, non-Euclidean sampler $\textsf{NE-Sampler}(g,\veps)$, target
distribution $\pi\propto\exp(-f)$.

\textbf{Output:} $x'$

Let $\bar{f}=\frac{\nu}{d}\,f$ and $\mu_{\sigma^{2}}\propto\exp(-V_{\sigma^{2}})$,
where
\[
V_{\sigma^{2}}:=\begin{cases}
\frac{\bar{f}+\phi}{\sigma^{2}} & \text{if }\sigma^{2}\leq\frac{\nu}{d}\,,\\
f+\frac{1}{\sigma^{2}}\,\phi & \text{o.w}.
\end{cases}
\]

\tcp{Phase 1: Initial distribution}

Find $x^{*}=\arg\min_{x\in K}(\bar{f}+\phi)$ and let $D:=\dcal_{g}^{3\sigma_{0}\sqrt{d}}(x^{*})$
for $\sigma_{0}^{2}:=10^{-5}/d^{3}$.\label{line:min}

Draw $x_{0}\sim\textsf{NE-Sampler}\bpar{g,\frac{\veps}{\sqrt{d}}}$
with initial dist. $\ncal\bpar{x^{*},\frac{\sigma_{0}^{2}}{1+\nicefrac{\nu\beta}{d}}\,g(x^{*})^{-1}}\cdot\mathbf{1}_{D}$
and target $\mu_{\sigma_{0}^{2}}$.

\tcp{Phase 2 \& 3: Annealing until $\sigma^{2}\leq\nu$}

\While{$\sigma^{2}\leq\nu$}{

Update $\sigma^{2}$ by 
\[
\sigma^{2}\gets\begin{cases}
\sigma^{2}\,\bpar{1+\frac{1}{\sqrt{d}}} & \text{if }\sigma^{2}\leq\frac{\nu}{d}\text{ (Phase 2)}\\
\sigma^{2}\,\bpar{1+\frac{\sigma}{\sqrt{\nu}}} & \text{if }\frac{\nu}{d}\leq\sigma^{2}\leq\nu\text{ (Phase 3)},
\end{cases}
\]

Draw $x_{i+1}\sim\textsf{NE-Sampler}\bpar{g,\frac{\veps}{\sqrt{d}}}$
started at $x_{i}$ with target dist. $\mu_{\sigma^{2}}$, and increment
$i$.

}

\tcp{Phase 4: Sampling from $e^{-f}$} 

Draw $x'\sim\textsf{NE-Sampler}\bpar{g,\frac{\veps}{\sqrt{d}}}$ started
at $x_{i}$ with target dist. $\pi$.

\end{algorithm2e}

Going forward, we use the following notation: for $\bar{f}(x):=\frac{\nu}{d}\,f(x)$,
\begin{align*}
F(\sigma^{2}) & :=\begin{cases}
\int_{K}\exp\bpar{-\frac{\bar{f}(x)+\phi(x)}{\sigma^{2}}}\,\D x & \text{if }\sigma^{2}\leq\frac{\nu}{d}\,,\\
\int_{K}\exp\bpar{-f(x)-\frac{\phi(x)}{\sigma^{2}}}\,\D x & \text{if }\frac{\nu}{d}\leq\sigma^{2}\leq\nu\,.
\end{cases}
\end{align*}
We can show that $x^{*}=\arg\min_{K}(\bar{f}+\phi)$ exists in Line~\ref{line:min}
of Algorithm~\ref{alg:IPM-sampling} and that all distributions involved
in the algorithm are indeed integrable. We defer the proof to \S\ref{proof:IPM-welldefined}.
\begin{prop}
\label{prop:annealing-welldefined} Each probability density involved
in the algorithm is integrable.
\end{prop}

\subsubsection{Closeness of distributions in sampling IPM}

In this section, we demonstrate that within each phase a probability
distribution $\mu_{\sigma_{i}^{2}}$ serves as a good warm start for
sampling the subsequent distribution $\mu_{\sigma_{i+1}^{2}}$. While
Algorithm~\ref{alg:IPM-sampling} uses as an initial distribution
$\bar{\mu}_{\sigma^{2}}$ that is approximately close to $\mu_{\sigma^{2}}$,
we resolve this discrepancy through a coupling argument. We refer
readers to Remark~\ref{rem:divine-intervention} and to \citet[Proof of Lemma 4.2]{lovasz2006simulated}
for fuller details.

For the first two phases, closeness of consecutive distributions follow
purely from a property of log-concave distributions, which is independent
of local metrics.
\begin{lem}
[\citet{kalai2006simulated}, Lemma 3.2] \label{lem:adam-logconcave}
For a log-concave function $g:\Rd\to\R$, the function $a\mapsto a^{d}\int g(x)^{a}\,\D x$
is log-concave in $a$.
\end{lem}

In Phase 1, we leverage another fundamental property of log-concave
distributions. It allows us to establish that the Gaussian distribution
truncated over a small Dikin ellipsoid in Phase 1 provides an $\mc O\bpar{\bpar{\frac{\nu\beta+d}{\nu\alpha+d}}^{d}}$-warm
start for $\mu_{\sigma_{0}^{2}}$. Thus, the $\dw$ which has a log-dependency
on the warmness parameter introduces an additional factor of $d$. 
\begin{lem}
[\citet{lovasz2007geometry}, Lemma 5.16] \label{lem:mostMass-logconcave}
Let $X$ be a random point drawn from a log-concave distribution with
a density $g:\Rd\to\R$. If $\gamma\geq2$, then
\[
\P\bpar{g(X)\leq e^{-\gamma(d-1)}\,\max g}\leq(\gamma\,e^{1-\gamma})^{d-1}\,.
\]
\end{lem}

\begin{rem}
If we can show that the $\dw$ has a $\log\log$-dependency through
the \emph{blocking conductance} or \emph{Gaussian isoperimetry}, or
if we utilize a non-Euclidean sampler with a double-log dependency,
we can avoid the additional factor of $d$.
\end{rem}

We defer the proofs for closeness to \S\ref{proof:IPM-closeness}.
\begin{lem}
[Phase 1] \label{lem:phase1} Let $x^{*}=\arg\min_{K}(\bar{f}+\phi)$.
For $\sigma^{2}=10^{-5}/d^{3}$ and $g=\hess\phi$, let $\mu$ be
the Gaussian distribution $\ncal\bpar{x^{*},\frac{\sigma^{2}}{1+\nu\beta/d}\,g(x^{*})^{-1}}$
truncated over $\dcal_{g}^{3\sigma\sqrt{d}}(x^{*})$, and $\mu_{0}$
the initial distribution used in Phase 2 such that $\mu_{0}\propto\exp\bpar{-\frac{\bar{f}+\phi}{\sigma^{2}}}\cdot\mathbf{1}_{K}$.
Then $\snorm{\mu/\mu_{0}}\lesssim\bpar{\frac{\nu\beta+d}{\nu\alpha+d}}^{d}$.
\end{lem}

In the following lemmas, we show that within each phase of our algorithm
$\mu_{\sigma_{i}^{2}}$ serves as an $\mc O(1)$-warm start for the
following distribution $\mu_{\sigma_{i+1}^{2}}$. In Phase 2, for
$1/d^{3}\lesssim\sigma^{2}\leq\nu/d$ the multiplicative update of
$(1+1/\sqrt{d})$ allows us to achieve an $\mc O(1)$-warm start.
\begin{lem}
[Phase 2] \label{lem:phase2} In Phase 2 (i.e., $\sigma_{i}^{2}\leq\nu/d$
with the update $\sigma_{i+1}^{2}=(1+\nicefrac{1}{\sqrt{d}})\,\sigma_{i}^{2}$),
a previous distribution $\mu_{i}$ serves as an $\mc O(1)$-warm start
for the next distribution $\mu_{i+1}$, i.e., $\snorm{\mu_{i}/\mu_{i+1}}=\mc O(1)$.
\end{lem}

In the large regime of $\nu/d\leq\sigma^{2}\leq\nu$ during Phase
3, we leverage the Brascamp-Lieb inequality to show that the accelerated
update of $(1+\sigma/\sqrt{\nu})$ ensures an $\mc O(1)$-warm start.
Moreover, we employ the same technique along with a limiting argument
to show that in Phase 4 the final distribution of $\mu_{\nu}$ is
an $\mc O(1)$-warm start for the target distribution $\pi$.
\begin{lem}
[Phase 3 and 4] \label{lem:phase34} In Phase 3 (i.e., $\nu/d\leq\sigma_{i}^{2}\leq\nu$
with the update $\sigma_{i+1}^{2}=\sigma_{i}^{2}(1+\sigma_{i}/\sqrt{\nu})$,
a previous distribution $\mu_{i}$ serves as an $\mc O(1)$-warm start
for the next distribution $\mu_{i+1}$, i.e., $\snorm{\mu_{i}/\mu_{i+1}}=\mc O(1)$.
In Phase 4, the distribution $\mu\propto\exp\bpar{-(f+\phi/\nu)}\cdot\mathbf{1}_{K}$
is an $\mc O(1)$-warm start for the target distribution $\pi\propto\exp(-f)\cdot\mathbf{1}_{K}$.
\end{lem}

\subsubsection{Proof of Theorem \ref{thm:Dikin-annealing}}

We now prove Theorem~\ref{thm:Dikin-annealing}, Algorithm~\ref{alg:IPM-sampling}
with the $\dw$ employed for the non-Euclidean sampler. 

\thmDikinannealing*
\begin{proof}
By Theorem~\ref{thm:Dikin}, if the potential $V$ of a target distribution
satisfies $\alpha\hess\phi\preceq\hess V\preceq\beta\hess\phi$, the
mixing time of the $\dw$ is $d\,(1\vee\beta)\,(\onu\wedge\nicefrac{1}{\alpha})\,\log\frac{\Lambda}{\veps}$.
Let $\bar{\kappa}=\frac{\nu\beta+d}{\nu\alpha+d}$. 
\begin{itemize}
\item Phase 1: When a target distribution is $\exp\bpar{-\frac{\bar{f}+\phi}{\sigma^{2}}})$
with $\sigma^{2}=10^{-5}/d^{3}$, 
\begin{align*}
d^{2}\bpar{1+\frac{\nu\beta d^{-1}+1}{\sigma^{2}}}\,\min\bpar{\onu,\frac{\sigma^{2}}{1+\nu\alpha d^{-1}}}\,\log\bpar{\frac{\nu\beta+d}{\nu\alpha+d}} & \leq d^{2}\bar{\kappa}\log\bar{\kappa}\,.
\end{align*}
\item Phase 2 ($1/d^{3}\lesssim\sigma^{2}\leq\nu/d$): Note that we need
$\mc O^{*}(\sqrt{d})$-many iterations to double $\sigma^{2}$. Hence,
in this phase the number of iterations of the $\dw$ with a target
$\exp\bpar{-\frac{\bar{f}+\phi}{\sigma^{2}}}$ adds up to 
\begin{align*}
d\,\bpar{1+\frac{\nu\beta d^{-1}+1}{\sigma^{2}}}\,\min\bpar{\onu,\frac{\sigma^{2}}{1+\nu\alpha d^{-1}}}\cdot\sqrt{d} & \leq d^{1.5}\bar{\kappa}+\sqrt{d}\nu\,.
\end{align*}
\item Phase 3 ($\nu/d\leq\sigma^{2}\leq\nu$): We need $\mc O^{*}\bpar{\frac{\sqrt{\nu}}{\sigma}}$-many
iterations to double $\sigma^{2}$. Hence, in this phase the total
number of iterations of the $\dw$ with a target $\exp\bpar{-\bpar{f+\frac{\phi}{\sigma^{2}}}}$
is 
\[
d\,\bpar{1+\beta+\frac{1}{\sigma^{2}}}\,\min\bpar{\onu,\frac{1}{\alpha+\sigma^{-2}}}\cdot\frac{\sqrt{\nu}}{\sigma}\leq\frac{d\sqrt{\nu}}{\sigma}\bpar{\bar{\kappa}+\sigma^{2}}\leq(d^{1.5}\bar{\kappa}+\sqrt{d}\nu)\vee(d\bar{\kappa}+d\nu)\,.
\]
\item Phase 4: The $\dw$ takes $\mc O(d\onu)$ iterations. 
\end{itemize}
Adding up all iterations, we need $\Otilde(d\,(d\bar{\kappa}\vee\nu\vee\onu))$
iterations of the $\dw$ in total.
\end{proof}

\section{Self-concordance theory for sampling IPM \label{sec:sc-theory-rules}}

Theorem~\ref{thm:Dikin-annealing} shows that $\gcdw$  running with
a $(\nu,\onu$)-Dikin-amenable metric for exponential distributions
mixes in $\otilde{d\max\Par{d,\nu,\onu}}$ iterations. Since every
log-concave sampling problem can be reduced to an exponential sampling
problem (as shown in \eqref{eq:reduced-problem}), Theorem~\ref{thm:Dikin-annealing}
ensures a poly-time mixing algorithm that utilizes local geometry
if we have a $(\nu,\onu)$-Dikin-amenable metric for the reduced sampling
problem.

This poses a natural question of how to construct such an efficiently
computable Dikin-amenable metric for structured sampling problems.
Suppose that the structured sampling problems assume a Dikin-amenable
metric for each constraint and epigraph of potentials. Motivated by
self-concordance theory of the optimization IPM, we consider the sum
of each barrier (and thus, the sum of metrics) as a candidate for
the metric of the reduced sampling problem. In fact, this choice aligns
seamlessly with the $\dw$. However, obtaining a provable guarantee
of the sampling IPM with the $\dw$ necessitates a comprehensive understanding
not only of self-concordance but also of SSC, SLTSC, SASC, and $\onu$-symmetry
under the addition of barriers (or metrics). 

In this section, we develop a ``calculus'' for combining metrics
for multiple constraints and epigraphs, deriving the resulting theoretical
guarantees (Theorem~\ref{thm:IPM-sampling}). This leads to a consistent
analogy with the work of \citet{nesterov1994interior} for the optimization
IPM.

\subsection{Basic properties: Scaling, addition and closeness}

Self-concordance is a central notion in the theory of interior-point
methods for optimization (we refer interested readers to \citet{nesterov1994interior,nesterov2018lectures}).
We first recall basic properties of self-concordance and then investigate
those of strong self-concordance and lower trace self-concordance,
which are crucial to our analysis.

\paragraph{Self-concordance.}
\begin{lem}
[\citet{nesterov2003introductory}] Let $f_{i}$ be a $\nu_{i}$-self-concordant
function on a convex set $K_{i}\subset\Rd$ for $i\in[2]$, and $\alpha>0$
be a scalar.
\begin{itemize}
\item (Theorem 4.1.1 and 4.2.2) $f_{1}+f_{2}$ is $(\nu_{1}+\nu_{2})$-self-concordant
on $K_{1}\cap K_{2}$.
\item (Corollary 4.1.2) $g=\hess(\alpha f_{1})$ satisfies $\snorm{g(x)^{-1/2}\Dd g(x)[h]\,g(x)^{-1/2}}_{2}\leq\frac{2}{\sqrt{\alpha}}\,\snorm h_{g(x)}$
for $x\in\inter(K_{1}\cap K_{2})$ and $h\in\Rd$.
\item If $f_{1}$ is a $\nu$-self-concordant, then $cf_{1}$ is $(c\nu)$-self-concordant
for $c>1$.
\end{itemize}
\end{lem}

We can extend this to self-concordant matrices as well.
\begin{lem}
\label{lem:sc-addition} Let $g_{i}:\inter(K_{i})\to\psd$ be a PSD
matrix function on a convex set $K_{i}\subset\Rd$ for $i\in[2]$,
and $\alpha>0$ be a scalar.
\begin{itemize}
\item $g_{1}+g_{2}$ is $(\nu_{1}+\nu_{2})$-self-concordant on $K_{1}\cap K_{2}$.
\item If $g_{1}$ is self-concordant, then $\alpha g_{1}$ satisfies $\Dd(\alpha g_{1})(x)[h]\preceq\frac{2}{\sqrt{\alpha}}\,\snorm h_{\alpha g_{1}}(\alpha g_{1})$
for $x\in\inter(K_{1}\cap K_{2})$ and $h\in\Rd$.
\item If $g_{1}$ is $\nu$-self-concordant, then $cg_{1}$ is $(c\nu)$-self-concordant
for $c>1$.
\end{itemize}
\end{lem}

\begin{proof}
Let $\phi_{i}$ be a $\nu_{i}$-self-concordant function counterpart
of $g_{i}$ on $K_{i}$ for $i\in[2]$. Then for $x\in\inter(K_{1}\cap K_{2})$
and $h\in\Rd$
\begin{align*}
\Dd(g_{1}+g_{2})(x)[h] & \preceq2\,\bpar{\snorm h_{g_{1}}g_{1}+\snorm h_{g_{2}}g_{2}}\preceq2\,\bpar{\snorm h_{g_{1}+g_{2}}g_{1}+\snorm h_{g_{1}+g_{2}}g_{2}}=2\,\snorm h_{g_{1}+g_{2}}(g_{1}+g_{2})\,.
\end{align*}
Clearly, $\phi_{1}+\phi_{2}$ is a function counterpart of $g_{1}+g_{2}$.
Thus, $g_{1}+g_{2}$ is a $(\nu_{1}+\nu_{2})$-self-concordant matrix
function on $K_{1}\cap K_{2}$.

For $c>1$, if $g_{1}$ is self-concordant, then $\Dd(cg_{1})(x)[h]\preceq\frac{2}{\sqrt{c}}\,\snorm h_{cg_{1}}(cg_{1})\preceq2\,\snorm h_{cg_{1}}(cg_{1})$,
and its function counterpart $c\phi_{1}$ is $(c\nu)$-self-concordant
by Lemma~\ref{lem:sc-addition}. Hence, $cg_{1}$ is $(c\nu)$-self-concordant.
\end{proof}
The following lemma ensures that the $\dw$ stays inside the convex
body. This lemma was proven only for self-concordant function in \citet[Theorem 5.1.5]{nesterov2018lectures},
but it can be straightforwardly extended to self-concordant matrices
as well.
\begin{lem}
\label{lem:dikin-in-body} $\dcal_{g}^{1}(x)\subset K$ for a convex
set $K$ and self-concordant matrix function $g$ on $K$.
\end{lem}

\begin{proof}
Consider a matrix function $g_{\veps}$ from $\intk$ to $\pd$ defined
by $g_{\veps}(x):=g(x)+\veps I$. It is self-concordant with a function
counterpart $\phi(x)+\frac{\veps}{2}\,\snorm x^{2}$, where $\phi:\intk\to\R$
is a function counterpart of $g$. For fixed $x\in\intk$ and $h\in\Rd$,
let us define a function defined by $\psi(t):=\bpar{h^{\T}g_{\veps}(x+th)\,h}^{-1/2}$
for any feasible $t$. Then,
\[
\psi'(t)=-\frac{\Dd g_{\veps}(x+th)[h^{\otimes3}]}{2\snorm h_{g_{\veps}(x+th)}^{3}}\,,
\]
and the definition of self-concordance leads to $|\psi'(t)|\leq1$.
This function can be defined on the interval $\bpar{-\psi(0),\psi(0)}$
due to $\psi(t)\geq\psi(0)-|t|$ (see \citet[Corollary 5.14]{nesterov2018lectures}).
This implies that $K$ contains the set 
\[
\bbrace{x+th:|t|\leq\psi(0)=\snorm h_{g_{\veps}(x)}^{-1}}=\{x+th:\snorm{th}_{g_{\veps}(x)}\leq1\}\,.
\]
By sending $\veps\to0$, the claim follows.
\end{proof}
The following lemma states that self-concordant metrics are similar
for nearby points.
\begin{lem}
[\citet{nesterov2003introductory}, Theorem 4.1.6] \label{lem:scCloseness}
Given any self-concordant matrix function $g$ on $K\subset\Rd$ and
$x,y\in K$ with $\snorm{x-y}_{g(x)}<1$, we have 
\[
(1-\snorm{x-y}_{g(x)})^{2}g(x)\preceq g(y)\preceq(1-\snorm{x-y}_{g(x)})^{-2}g(x)\,.
\]
\end{lem}

\paragraph{Strong self-concordance.}

Strong self-concordance is additive up to a constant scaling. See
\S\ref{proof:ssc-basic} for the proof.
\begin{lem}
\label{lem:ssc-sum} If $g_{i}$ is a SSC matrix function on $K_{i}$
for $i\in[2]$, then $2\,(g_{1}+g_{2})$ is strongly self-concordant
on $K_{1}\cap K_{2}$. 
\end{lem}

Note that if we add $k$-many strongly self-concordant metrics, then
we need the scaling of $2^{\log_{2}k}=k$. We remark that the factor
of $2$ above might be redundant. Next, we recall an analogue of Lemma~\ref{lem:scCloseness}
for strong self-concordance.
\begin{lem}
[\citet{laddha2020strong}, Lemma 1.2] \label{lem:strongSC-closeness}Given
a strongly self-concordant matrix function $g$ on $K$, and any $x,y\in K$
with $\snorm{x-y}_{g(x)}<1$, 
\[
\snorm{g(x)^{-1/2}\bpar{g(y)-g(x)}\,g(x)^{-1/2}}_{F}\leq(1-\snorm{x-y}_{g(x)})^{-2}\snorm{x-y}_{g(x)}\,.
\]
\end{lem}

\paragraph{Symmetry.}

Recall that $\onu$-symmetry requires two-sided inclusion: the first
part is $\dcal_{g}^{1}(x)\subset K\cap(2x-K)$, and the second part
is $K\cap(2x-K)\subset\dcal_{g}^{\sqrt{\onu}}(x)$. The first part
immediately follows when a metric is induced by a self-concordant
function.
\begin{lem}
\label{lem:symmetricLeftpart} If $\phi$ is a self-concordant function
on $K$, then $\dcal_{g}^{1}(x)\subset K\cap(2x-K)$ for $g=\hess\phi$
and $x\in K$.
\end{lem}

\begin{proof}
Lemma~\ref{lem:dikin-in-body} ensures that $y\in K$ whenever $y\in\dcal_{g}^{1}(x)$.
Then $2x-y\in\dcal_{g}^{1}(x)$ and thus $2x-y\in K$. It implies
that $y\in2x-K$.
\end{proof}
When a metric is induced by a self-concordant barrier with a barrier
parameter $\nu$, it holds that $\onu=\O(\nu^{2})$.
\begin{lem}
\label{lem:bound-symmetry} For a self-concordant barrier $\phi$
with a barrier parameter $\nu$ on $K$ and $g=\hess\phi$, it follows
that $\onu=\O(\nu^{2})$.
\end{lem}

\begin{proof}
By \citet[Theorem 4.2.5]{nesterov2003introductory}, for any $x,y\in K$
with $\grad\phi(x)\cdot(y-x)\geq0$ it follows that $\snorm{y-x}_{g(x)}\leq\nu+2\sqrt{\nu}$.
Now, let $x\in K$ and $y\in K\cap(2x-K)$. The latter implies that
$y-x=x-z$ for some $z\in K$. 

If $\grad\phi(x)\cdot(y-x)\geq0$, then $\snorm{y-x}_{g(x)}\leq\nu+2\sqrt{\nu}.$
If $\grad\phi(x)\cdot(y-x)<0$, then $\grad\phi(x)\cdot(z-x)>0$ and
thus $\snorm{y-x}_{g(x)}=\snorm{z-x}_{g(x)}\leq\nu+2\sqrt{\nu}$.
From these two cases, it holds in general that $\snorm{y-x}_{g(x)}\leq\nu+2\sqrt{\nu}$
and thus $K\cap(2x-K)\subset\dcal_{g}^{\nu+2\sqrt{\nu}}(x)$. By Lemma~\ref{lem:symmetricLeftpart},
$\dcal_{g}^{1}(x)\subset K\cap(2x-K)$ and thus $\onu=\mc O(\nu^{2})$.
\end{proof}
For affine constraints $Ax\geq b$, the first inclusion above has
a useful equivalent description as follows:
\begin{lem}
\label{lem:symmforPolytope} Let $x\in K=\{Ax>b\}$. It holds that
$y\in K\cap(2x-K)$ if and only if $\snorm{A_{x}(y-x)}_{\infty}\leq1$.
\end{lem}

\begin{proof}
For $y\in K$, we have $Ay>b$ and thus $s_{x}=Ax-b>A(x-y)$ (elementwise
inequality). As $s_{x}>0$, we have $A_{x}(x-y)\leq1$. When $y\in(2x-K)$,
we can write $y=2x-z$ for some $z\in K$. Note that
\[
A(x-y)=A(z-x)>b-Ax=-s_{x}\,,
\]
and thus $A_{x}(x-y)\geq-1$. Therefore, $\snorm{A_{x}(y-x)}_{\infty}\leq1$.
\end{proof}
\begin{lem}
\label{lem:symmScaling} For $\alpha\geq1$, if $g$ is $\onu$-symmetric,
then $\alpha g$ is $\alpha\onu$-symmetric.
\end{lem}

Symmetry parameters and self-concordance parameters are additive.
\begin{lem}
\label{lem:symmetry-addition} If a PSD matrix function $g_{i}$ is
$\onu_{i}$-symmetric on $K_{i}$ for $i\in[2]$, then $g_{1}+g_{2}$
is $(\onu_{1}+\onu_{2})$-symmetric on $K_{1}\cap K_{2}$.
\end{lem}

\begin{proof}
For $g:=g_{1}+g_{2}$, let $y\in\dcal_{g}^{1}(x)$. It implies $y\in\dcal_{g_{1}}^{1}(x)\cap\dcal_{g_{2}}^{1}(x)$
and so $y\in K_{i}\cap(2x-K_{i})$. Due to $\cap_{i}\bpar{K_{i}\cap(2x-K_{i})}=K\cap(2x-K)$,
we have $y\in K\cap(2x-K)$ and so $\dcal_{g}^{1}(x)\subset K\cap(2x-K)$.

Now let $y\in K\cap(2x-K)$. It is obvious that $y\in K_{i}\cap(2x-K_{i})$
for $i=1,2$, and thus
\[
(y-x)^{\T}g_{1}(x)(y-x)\leq\nu_{1}\,,\qquad\text{and}\qquad(y-x)^{\T}g_{2}(x)(y-x)\leq\nu_{2}\,.
\]
By adding up these two, it follows that $\snorm{y-x}_{g(x)}^{2}\leq\nu_{1}+\nu_{2}$.
\end{proof}

\paragraph{Lower trace self-concordance.}

It readily follows that (strongly) LTSC holds under scaling by a scalar
greater than or equal to $1$. 

We provide a useful sufficient condition under which the sum of PSD
matrix functions is LTSC.
\begin{lem}
\label{lem:sltsc-additive} For a PSD matrix function $g_{i}$ on
$K_{i}$, let $g:=\sum_{i}g_{i}$ be PD on $\bigcap_{i}K_{i}$. If
$g_{i}$ is SLTSC on $K_{i}$, then $g$ is LTSC on $\bigcap_{i}K_{i}$.
\end{lem}

We note that $\Dd^{2}g_{i}(x)[h,h]\succeq0$ is a stronger condition
than $\tr\bpar{g(x)^{-1}\Dd^{2}g_{i}(x)[h,h]}\geq-\snorm h_{g_{i}(x)}^{2}$.
Thus, a special case of the lemma is that if $\Dd^{2}g_{1}[h,h]\succeq0$
and $\Dd^{2}g_{2}[h,h]\succeq0$, then $g_{1}+g_{2}$ is LTSC. Note
that this condition is \emph{additive.}

We also find that highly self-concordance is a handy sufficient condition
by which one can establish strongly lower trace self-concordance,
whose proof is deferred to \S\ref{proof:ltsc-basic}.
\begin{lem}
\label{lem:hsc-to-sltsc} For $K\subset\Rd$, let $\bar{g}:\intk\to\psd$
be a HSC matrix function, and define another matrix function by $g:=d\bar{g}$
on $K$. Then $g$ is SLTSC.
\end{lem}

\paragraph{Average self-concordance.}

Just as (S)LTSC, (S)ASC still holds under scaling by a scalar greater
than or equal to $1$. Also, the definition of SASC immediately leads
to the following additive condition:
\begin{lem}
\label{lem:sasc-additive} For a PSD matrix function $g_{i}$ on $K_{i}$
for $i\in[m]$, let $m=\mc O(1)$ and $g:=\sum_{i=1}^{m}g_{i}$ be
PD on $\bigcap_{i}K_{i}$. If $g_{i}$ is SASC on $K_{i}$, then $g$
is ASC on $\bigcap_{i}K_{i}$.
\end{lem}

\begin{proof}
Fix $\veps>0$. Each $g_{i}$ invokes $r_{i}(\veps)$ such that if
$r\leq r_{i}(\veps/m)$, then 
\[
\P_{z}\Bpar{\snorm{z-x}_{g_{i}(x)}^{2}-\snorm{z-x}_{g_{i}(x)}^{2}\leq\frac{2\veps}{m}\,\frac{r^{2}}{d}}\geq1-\frac{\veps}{m}\,.
\]
If $r\leq\bar{r}(\veps):=\min_{i}\,r_{i}(\veps/m)$, then the union
bound leads to ASC of $\sum g_{i}$ on $\bigcap_{i}K_{i}$.
\end{proof}
When does SASC hold? It is implied in \citet{narayanan2016randomized}
that HSC implies SASC. For completeness, we provide the proof in \S\ref{proof:sasc-basic}.
\begin{lem}
[HSC to SASC] \label{lem:hsc-to-sasc} If $\phi:\intk\to\R$ is HSC,
then $d\phi$ is SASC.
\end{lem}

\subsection{Collapse and embedding: Lifting up SSC, SLTSC, and SASC}

SSC, (S)LTSC, (S)ASC of a local metric do not carry over into an extended
space in the reduced sampling problem. For instance, SSC assumes the
invertibility of the local metric, which may become singular in the
extended space. To address this challenge, we introduce the notions
of \emph{collapse} and \emph{embedding}, based on which we can pass
those properties from the original sampling problem to the reduced
problem.
\begin{defn}
\label{def:sc-along-subspace} Let $K$ and $K'$ be convex sets in
$\Rd$ and in $\R^{m}$ with $d\leq m$, respectively. Let $g:\intk\to\psd$
be a PSD matrix function.
\begin{itemize}
\item We say $g$ is \emph{collapsed onto a linear subspace} $W\subset\Rd$
if $\inner{u,v}_{g(x)}=\inner{P_{W}u,P_{W}v}_{g(x)}$ for any $x\in\intk$
and $u,v\in\Rd$ where $P_{W}$ is the orthogonal projection onto
$W$.
\begin{itemize}
\item In other words, for an orthonormal basis $\{u_{1},\dots,u_{k}\}$
of $W$ there exists the PSD matrix function $g_{W}:\intk\to\mathbb{S}_{+}^{k}$
such that $\inner{e_{i},e_{j}}_{g_{W}(x)}=\inner{u_{i},u_{j}}_{g(x)}$
for $i,j\in[k]$ (i.e., $g_{W}(x)=U^{\T}g(x)U$ where the columns
of $U\in\R^{d\times k}$ are $\Brace{u_{1},\dots,u_{k}}$). 
\end{itemize}
\item For $g$ collapsed onto $W$, we say
\begin{itemize}
\item $g$ is PD along $W$ if $g_{W}$ is PD. In other words, $\norm h_{g(x)}=0$
implies $h\perp W$.
\item $g$ is SSC along $W$ if $g$ is a self-concordant matrix function
and $g_{W}\succ0$ satisfies
\[
\snorm{g_{W}(x)^{-1/2}\Dd g_{W}(x)[h]\,g_{W}(x)^{-1/2}}_{F}\leq2\snorm h_{g}\quad\text{for any }x\in\intk\ \text{and}\ h\in\Rd\,.
\]
\end{itemize}
\item \emph{Embedding} $\bar{g}$ of $g$ into $K'$
\begin{itemize}
\item Let $P:\R^{m}\to\Rd$ be the projection onto the set of coordinates
appearing in the variable $x$ of $g$. The embedding of $g$ onto
$K'$ is a PSD matrix function $\bar{g}(y):\inter(K')\to\mathbb{S}_{+}^{m}$
such that $\inner{u,v}_{\bar{g}(y)}=\inner{Pu,Pv}_{g(P(y))}$.
\end{itemize}
\end{itemize}
\end{defn}

We note that these notions are well-defined independently of the choice
of an orthonormal basis of $W$. The proof can be found in \S\ref{proof:collapse-embedding-welldefined}.
\begin{prop}
\label{prop:collapse-well-defined} Let $K\subset\Rd$ be convex and
$g:\intk\to\psd$ a PSD matrix function collapsed onto a subspace
$W\subset\Rd$. Then PD and SSC along $W$ are well-defined (i.e.,
the condition for each property holds for any orthonormal basis of
$W$).
\end{prop}

\paragraph{Affine transformation.}

Using these notions, we can make it precise that an inverse mapping
of affine transformations preserves SSC. We begin with a barrier version
and subsequently extend it to a matrix-function version. The detailed
proofs are deferred to \S\ref{proof:collap-affine}.
\begin{lem}
\label{lem:linear-trans} Let $T:\Rd\to\R^{m}$ be a linear operator
defined by $T(x)=Ax+b$ for $A\in\R^{m\times d}$ and $b\in\R^{m}$.
Let $\phi(y):\intk\subset\R^{m}\to\R$ be a self-concordant barrier
for $K$ and define $\psi(x):=\phi(T(x))=\phi(y)$ on $\bar{K}:=T^{-1}K\subset\Rd$.
\begin{itemize}
\item If $\phi$ is a $(\nu,\onu)$-self-concordant barrier for $K$, so
is $\psi$ for $\bar{K}$.
\item If $\Dd^{4}\phi(y)[v,v]\succeq0$ for $y\in\intk$ and $v\in\R^{m}$,
then $\Dd^{4}\psi(x)[u,u]\succeq0$ for $x\in\inter(\bar{K})$ and
$u\in\Rd$.
\item If $\phi$ is HSC, so is $\psi$.
\end{itemize}
\end{lem}

\begin{lem}
\label{lem:linear-trans-matrix} Let $g:\intk\subset\R^{m}\to\mathbb{S}_{+}^{m}$
be a self-concordant matrix function and $T(x)=Ax+b$ with $A\in\R^{m\times d}$
and $b\in\R^{m}$ be a linear operator. Let $\bar{g}(x):=A^{\T}g(Tx)A$
be a PSD matrix function from $\bar{K}:=T^{-1}K\subset\Rd$ to $\psd$.
\begin{itemize}
\item If $g$ is $(\nu,\onu)$-self-concordant barrier, so is $\bar{g}$
for $\bar{K}$.
\item If $g$ is SSC, then $\bar{g}$ is SSC along $W=\rowspace(A)$.
\item If $\Dd^{2}g(y)[h,h]\succeq0$ for $y\in\intk$ and $h\in\R^{m}$,
then $\Dd^{2}\bar{g}(x)[\bar{h},\bar{h}]\succeq0$ for $x\in\inter(\bar{K})$
and $\bar{h}\in\Rd$.
\item If $A$ is invertible and $g$ is SLTSC, then $\bar{g}$ is SLTSC.
\item If $A$ is invertible and $g$ is SASC, then $\bar{g}$ is SASC.
\end{itemize}
\end{lem}

Intuitively, embedding should not affect self-concordance and symmetry
parameter, which is indeed the case.
\begin{cor}
\label{cor:embedding-scness} Assume $K\subset\Rd$ is embeddable
into $K'\subset\R^{m}$. If $g:\intk\to\psd$ is a $(\nu,\onu)$-self-concordant
matrix function, then its embedding $\bar{g}:\inter(K')\to\mathbb{S}_{+}^{m}$
is a $(\nu,\onu)$-self-concordant matrix function.
\end{cor}

\begin{proof}
Since $K$ can be embedded into $K'$, there exists a projection matrix
$P\in\{0,1\}^{d\times m}$ such that $\bar{g}(y)=P^{\T}g(Py)P$ with
$x=Py\in\intk$ and $y\in\inter(K')$. As we can view $\bar{g}$ as
a matrix function induced by the inverse of the linear map $x=Py$,
Lemma~\ref{lem:linear-trans-matrix} shows that $\bar{g}$ is a $(\nu,\onu)$-self-concordant
matrix function for $K'=P^{-1}K$. 
\end{proof}

\paragraph{Lifting up SSC, SLTSC, and SASC via embedding.}

In reduction to the exponential sampling problem, passing essential
properties (e.g., SSC, SLTSC, and SASC) of metrics from the original
space to the extended space poses technical issues. We address these
issues in the following two lemmas, whose proofs are deferred to \S\ref{proof:lifting-ssc}.

As mentioned earlier, SSC in the original space does not automatically
imply SSC for its embedding $\bar{g}$, as SSC assumes invertibility.
However, there is a useful method for extending SSC from the original
space to the extended space.
\begin{lem}
\label{lem:embedding-ssc} For convex $K\subset\Rd$, let $g:\intk\to\psd$
be SSC along a subspace $W\subset\Rd$, and assume $K$ is embeddable
into convex $K'\subset\R^{m}$ with $m\geq d$. For the embedding
$\bar{g}:\inter(K')\to\mathbb{S}_{+}^{m}$ of $g$ into $K'$, it
holds that $\bar{g}+\veps I_{m}$ is SSC on $K'$ for any $\veps>0$.
\end{lem}

When extending SLTSC and SASC to the embedding space, we encounter
a different subtlety. The conditions in SLTSC and SASC of $\bar{g}$
consider every PSD matrix functions $g'$ such that $\bar{g}+g'$
is invertible in the extended space $\bar{K}$. However, the embedding
$\bar{g}$ of $g$ is collapsed onto the subspace corresponding to
the original space $K$. As SLTSC and SASC convolve $\bar{g}$ and
$g'$ by considering $(\bar{g}+g')^{-1}$ in their formulations, it
is not evident whether SLTSC and SASC can be transferred to the extended
space $\bar{K}$ from the original space $K$. However, by employing
with Schur complements we can show that these properties can indeed
carry over into the extended space.
\begin{lem}
\label{lem:embedding-sltsc} For convex $K\subset\Rd$, let $g:\intk\to\psd$
is SLTSC, and assume $K$ is embeddable into convex $K'\subset\R^{m}$
with $m\geq d$. Then its embedding $\bar{g}:\inter(K')\to\mathbb{S}_{+}^{m}$
is also SLTSC. The same is true for SASC.
\end{lem}

\subsection{Proof of Theorem \ref{thm:IPM-sampling}}

With our understanding of how to combine properties of barriers for
constraints and epigraphs, we are prepared to prove Theorem~\ref{thm:IPM-sampling}.
Let us revisit the reduced sampling problem in \eqref{eq:reduced-problem}:
\begin{align*}
\text{sample } & y\sim\tilde{\pi}\propto\exp\Bpar{-\inner{(\underbrace{0,\dots,0}_{d\text{ times}},\underbrace{1,\dots,1}_{I\text{ times}}),\cdot}}\\
\text{s.t. } & y\in\bigcap_{i=1}^{I}E_{i}\cap\underbrace{\bigcap_{j=1}^{J}K_{j}}_{\eqqcolon:K}\eqqcolon K'\,,
\end{align*}
where $E_{i}:=\bbrace{y=(x,t_{1},\dots,t_{I})\in\R^{d+I}:f_{i}(x)\leq y_{d+i}}$
for a proper closed convex function $f_{i}$ and $i\in[I]$, and $K_{j}:=\bbrace{y=(x,t_{1},\dots,t_{I})\in\R^{d+I}:h_{j}(x)\leq0}$
for a closed convex function $h_{j}$ and $j\in[J]$, and $K$ has
non-empty interior.

We begin with a useful geometric property of $K'$.
\begin{lem}
If the original sampling problem \eqref{eq:problem} is well-defined,
then the extended convex region $K'$ in the reduced sampling problem
\eqref{eq:reduced-problem} has non-empty interior and no straight
line.
\end{lem}

\begin{proof}
Since $f_{i}$ and $h_{j}$ are closed and convex, $K'$ is convex
and closed. Since $f_{i}$ is continuous on $\inter(K)$ due to convexity
(see \citet[Theorem 10.1]{rockafellar1997convex}), its epigraph has
non-empty interior. Thus, $K'$ has non-empty interior.

Since $K'$ is closed and convex, it can be written as $K'=\bigcap_{i}H_{i}$
where $H_{i}=\{x:a_{i}^{\T}x\geq b_{i}\}$ is any halfspace containing
$K'$. Suppose $K'$ contains a straight line $\ell:=\{p+th:t\in\R\}$
for some $p,h\in\Rd$. Then $\ell\subset H_{i}$ for any $i$, and
thus $\ell$ must be parallel to any halfspace $H_{i}$ (i.e., $h\perp a_{i}$). 

Fix $y\in\inter(K')$. The translated line $\ell_{y}$ of $\ell$
containing $y$ is still included in $H_{i}$ for all $i$. As $y\in\inter(K')$,
the distance from $y$ to $\de H_{i}$ is bounded lower by $\delta>0$
for all $i$. Hence, $\ell_{y}+B_{\delta}$ is fully contained in
$H_{i}$ and thus in $K'$.

Clearly, integration of the exponential distribution along the fiber
$\ell_{y}$ is infinite. Since $K'$ contains the cylinder $\ell_{y}+B_{\delta}$,
integration of the exponential distribution over $K'$ must be infinite,
leading to contradiction.
\end{proof}
The following is the extension of \citet[Theorem 5.1.6]{nesterov2018lectures}
to self-concordant matrix functions, which implies invertibility of
Dikin-amenable metrics in the reduced problem.
\begin{lem}
\label{lem:nondegenerate-no-straightline} For convex $K\subset\Rd$
containing no straight line, a self-concordant matrix function $g:\intk\to\psd$
is non-degenerate on $K$.
\end{lem}

\begin{proof}
Suppose $\snorm h_{g(x)}=0$ for some $0\neq h\in\Rd$ and $x\in\intk$.
Clearly, the line $x+th$ for $t\in\R$ is contained in $\dcal_{g}^{1}(x)$.
As $\dcal_{g}^{1}(x)\subset K$ due to Lemma~\ref{lem:dikin-in-body},
it implies that $K$ contains a straight line $x+th$, which leads
to contradiction. 
\end{proof}
\thmIPMsampling*
\begin{proof}
First of all, $\bar{g}_{i}^{e}$ is $(\nu_{i},\onu_{i})$-self-concordant
(Corollary~\ref{cor:embedding-scness}), and SLTSC and SASC on $K'$
(Lemma~\ref{lem:embedding-sltsc}). For fixed $\veps>0$, $\bar{g}_{i}^{e}+\veps I$
is SSC by Lemma~\ref{lem:embedding-ssc}. We can make similar arguments
for $\bar{g}_{j}^{c}$ regarding self-concordance, symmetry, SLTSC,
SASC, and SSC. Hence, $g+(I+J)\veps I$ is SSC by Lemma~\ref{lem:ssc-sum}.
Since $g$ is self-concordant on $K'$ by Lemma~\ref{lem:sc-addition}
and $K'$ contains no straight line, $g$ is PD by Lemma~\ref{lem:nondegenerate-no-straightline}.
Sending $\veps$ to $0$, we can obtain SSC of $g$. LTSC and ASC
of $g$ follows from Lemma~\ref{lem:sltsc-additive} and \ref{lem:sasc-additive}.
The symmetry parameter of $g$ follows from Lemma~\ref{lem:symmetry-addition}.
\end{proof}

\subsection{Direct product}

For $i\in[m]$ and domain $E_{i}\subset\R^{d_{i}}$, let $g_{i}(x_{i}):\inter(E_{i})\to\mathbb{S}_{++}^{d_{i}}$
be a self-concordant matrix. For $l:=\sum_{i}d_{i}$ and $E:=\prod_{i}E_{i}$,
we define a self-concordant matrix $g$ on $E\subset\R^{l}$ with
block diagonals being $g_{i}$. To be precise, we can write
\begin{align*}
g(x) & =g(x_{1},\dots,x_{m}):=\sum_{i}\bar{g}_{i}(x)\,,
\end{align*}
where $\bar{g}_{i}:\R^{l}\to\mathbb{S}_{+}^{l}$ is a matrix function
whose entry is all zero but the $i$-th block diagonal being $g_{i}$.

When handling the direct product of domains, it is common for each
domain to have an $\mc O(1)$-dimension. In such cases, scaling the
barriers by dimension worsens mixing time at most constant factors
while making the barriers SSC and SLTSC. We defer the proofs to \S\ref{proof:direct-ssc-sltsc}.
\begin{lem}
[SSC  under direct product] \label{lem:ssc-direct} For open $E_{i}\subset\R^{d_{i}}$,
let $g_{i}:E_{i}\to\mathbb{S}_{++}^{d_{i}}$ be SC. Then $g:=\sum d_{i}\bar{g}_{i}$
defined on $\prod E_{i}$ is SSC.
\end{lem}

\begin{lem}
[SLTSC  under direct product] \label{lem:sltsc-direct} For open
$E_{i}\subset\R^{d_{i}}$, let $g_{i}:E_{i}\to\mathbb{S}_{++}^{d_{i}}$
be HSC. Then $g:=\sum d_{i}\bar{g_{i}}$ defined on $\prod E_{i}$
is SLTSC.
\end{lem}

\subsection{Inverse images under non-linear mappings}

\citet{nesterov1994interior} introduced the notion of \emph{compatibility}
with a convex domain while constructing a self-concordant barrier
for a wider class of structured constraints. We generalize this notion
to the fourth order, by which we can easily construct a SSC, SLTSC,
and SASC barrier. For a convex cone $K$, we use $a\leq_{K}b$ to
denote $b-a\in K$.
\begin{defn}
[Compatibility] Let $\beta,\gamma\geq0$. Let $K$ be a convex cone
in $\R^{m}$ and $\Gamma$ be a closed convex domain in $\Rd$. A
mapping $\acal:\inter(\Gamma)\to\R^{m}$ of class $C^{4}$ is called
$(K,\beta,\gamma)$-compatible with the domain $\Gamma$ if 
\begin{itemize}
\item $\acal$ is concave with respect to $K$. That is, $t\acal(x)+(1-t)\,\acal(y)\leq_{K}\acal(tx+(1-t)\,y)$
for all $t\in[0,1]$ and $x,y\in\inter(\Gamma)$. Equivalently, $-\Dd^{2}\acal(x)[h,h]\in K$
for any $x\in\inter(\Gamma)$ and $h\in\R^{m}$.
\item For any $x\in\inter(\Gamma)$, $y\in\Gamma\cap(2x-\Gamma)$, and $h=y-x$,
it holds that 
\begin{align*}
\beta\Dd^{2}\acal(x)[h,h] & \leq_{K}\Dd^{3}\acal(x)[h,h,h]\leq_{K}-\beta\Dd^{2}\acal(x)[h,h]\,,\\
\gamma\Dd^{2}\acal(x)[h,h] & \leq_{K}\Dd^{4}\acal(x)[h,h,h,h]\leq_{K}-\gamma\Dd^{2}\acal(x)[h,h]\,.
\end{align*}
\end{itemize}
\end{defn}

\begin{example}
\label{exa:useful-criteria} An affine mapping is $(\{0\},0,0)$-compatible
with any closed convex domain. We note that a function that is $(\R_{+},\beta,\gamma)$-compatible
with $\R_{+}$ is a $C^{4}$-smooth concave real-valued function $f:(0,\infty)\to\R$
such that for any $t>0$,
\begin{align*}
|f'''(t)| & \leq-\frac{\beta}{t}\,f''(t)\quad\text{and}\quad|f^{(4)}(t)|\leq-\frac{\gamma}{t^{2}}\,f''(t)\,.
\end{align*}
\begin{itemize}
\item Let $0<p\leq1$. Then the function of $f(t)=t^{p}$ is $(\R_{+},2-p,(2-p)\,(3-p))$-compatible
with $\R_{+}$.
\item $f(t)=\log t$ is $(\R_{+},2,6)$-compatible with $\R_{+}$.
\end{itemize}
The following lemma is an extension of \citet[Lemma 5.1.3]{nesterov1994interior}
to our fourth-order compatibility.
\end{example}

\begin{lem}
\label{lem:extension-compatibility} Let $K,K_{1},K_{2}$ be convex
cones in $\R^{m},\R^{m_{1}},\R^{m_{2}}$ respectively.
\begin{itemize}
\item If $\acal:\inter(\Gamma)\to\R^{m}$ is $(K,\beta,\gamma)$-compatible
with $\Gamma$ and $K\subset K'$ is a closed convex cone in $\R^{m}$,
then $\acal$ is $(K',\beta,\gamma)$-compatible with $\Gamma$.
\item If $\acal_{i}:\inter(\Gamma_{i})\to\R^{m_{i}}$ is $(K_{i},\beta_{i},\gamma_{i})$-compatible
with $\Gamma_{i}$ for $i=1,2$, then $\acal:\inter(\Gamma_{1}\times\Gamma_{2})\to\R^{m_{1}}\times\R^{m_{2}}$
mapping $(x,y)\to(\acal_{1}(x),\acal_{2}(y))$ is $(K_{1}\times K_{2},\max(\beta_{1},\beta_{2}),\max(\gamma_{1},\gamma_{2}))$-compatible
with $\Gamma_{1}\times\Gamma_{2}$.
\end{itemize}
\end{lem}

We now introduce a main result in this section (see \S\ref{proof:inverse-non-linear}).
To begin with, we recall that for a closed convex domain $G\subset\Rd$
the \emph{recessive cone} $R(G)$ of $G$ is $\{h\in\Rd:x+th\in G\ \text{for all }x\in G\text{ and }t>0\}$.
\begin{lem}
\label{lem:compatible} Let $G$ be a closed convex domain in $\R^{m}$,
$F$ be a highly $\theta$-self-concordant barrier for $G$, $\Gamma$
be a closed convex domain in $\Rd$, and $\Pi$ be a highly $\nu$-self-concordant
barrier for $\Gamma$. Let $\acal$ be a $(K,\beta,\gamma)$-compatible
with $\Gamma$, where $K$ is a ray contained in the recessive cone
$R(G)$. Assume that $\acal(\inter(\Gamma))\cap G\neq\emptyset$.
\begin{itemize}
\item The set $G^{+}=\overline{\inter(\Gamma)\cap\acal^{-1}\bpar{\inter(G)}}$
is a closed convex domain in $\Rd$.
\item For $\delta=\max\Par{\beta,\gamma,2}$, the function $\Psi(x)=F(\acal(x))+\delta^{2}\,\Pi(x)$
is a $(\theta+\delta^{2}\nu)$-self-concordant barrier for $G^{+}$.
\item $\Psi$ is highly self-concordant.
\end{itemize}
\end{lem}

Using this result, we can obtain a useful tool in establishing lower
trace self-concordance of a barrier for the direct product of structured
sets.
\begin{lem}
\label{lem:tool-concave} Let $f$ be a $C^{4}$ concave function
on $\{t>0\}$ such that $|f'''(t)|\leq\frac{\beta}{t}\,|f''(t)|$
and $|f^{(4)}(t)|\leq\frac{\gamma}{t^{2}}\,|f''(t)|$ for $t>0$.
Then the function 
\[
F(t,x)=-\log\bpar{f(t)-x}-\max(4,\beta^{2},\gamma^{2})\,\log t
\]
is a highly $(1+\max(4,\beta^{2},\gamma^{2}))$-self-concordant barrier
for the two dimensional convex domain
\[
G_{f}=\overline{\{(t,x)\in\R^{2}:t>0,\,x\leq f(t)\}}\,.
\]
\end{lem}

\begin{proof}
From the discussion in Example~\ref{exa:useful-criteria}, the map
$f(t):(0,\infty)\to\R$ is $(\R_{+},\beta,\gamma)$-compatible with
$\R_{+}$. Clearly, the identity map from $\R$ to $\R$ is $(\{0\},0,0)$-compatible
with $\R$. Hence by Lemma~\ref{lem:extension-compatibility}-(2)
implies that the map $\acal:\R_{+}\times\R\to\R^{2}$ defined by $\acal(t,x)=(f(t),x)$
is $(\{0\}\times\R_{+},\beta,\gamma)$-compatible with $\R_{+}\times\R$. 

Now observe that $G_{f}$ can be written as $\acal^{-1}\bpar{\{(t,x):x\leq t\}}$
and that $K=\{0\}\times\R_{+}$ is a ray contained in the recessive
cone $R(G)$ for $G:=\{(t,x):x\leq t\}$. By applying Lemma~\ref{lem:compatible}
to the highly $1$-self-concordant barriers $F(t,x)=-\log(t-x)$ for
$G$ and $\Phi(t,x)=-\log t$ for $\R_{+}\times\R$, it follows that
$F$ is is a highly $(1+\max(4,\beta^{2},\gamma^{2}))$-self-concordant
barrier for $G_{f}$.
\end{proof}
We can prove a similar result for a convex $f$ as follows:
\begin{lem}
\label{lem:tool-convex} Let $f$ be a $C^{4}$ convex function on
$\{x>0\}$ such that $|f'''(x)|\leq\frac{\beta}{x}\,f''(x)$ and $|f^{(4)}(x)|\leq\frac{\gamma}{x^{2}}\,f''(x)$
for $x>0$. Then the function 
\[
F(t,x)=-\log\bpar{t-f(x)}-\max(4,\beta^{2},\gamma^{2})\,\log x
\]
is a highly $(1+\max(4,\beta^{2},\gamma^{2}))$-self-concordant barrier
for the two dimensional convex domain
\[
G_{f}=\overline{\{(t,x)\in\R^{2}:x>0,\,t\geq f(x)\}}\,.
\]
\end{lem}

Its proof follows from applying Lemma~\ref{lem:tool-concave} to
the image of $G_{f}$ under the map $(t,x)\to(-x,t)$.

\global\long\def\vec{\textup{\textsf{vec}}}%
\global\long\def\svec{\textup{\textsf{svec}}}%

\section{Structured densities and constraint families \label{sec:handbook-barrier}}

In order to obtain a mixing-time bound of the $\dw$ for the reduced
problem, a concrete understanding of properties and parameters of
barriers for $K_{i}$ and $K_{j}$ is essential. To this end, we revisit
self-concordant barriers for structured convex constraints and level
sets, examining the required scaling factors which ensure those properties.

\subsection{Linear constraints}

Consider a set of linear constraints: $K=\{x\in\Rd:Ax\geq b\}$ for
$A\in\R^{m\times d}$ and $b\in\R^{m}$, where $A$ has no all-zero
rows. We use $s_{x}:=Ax-b$ to denote the slack at $x$, and $A_{x}:=S_{x}^{-1}A$
to denote the constraints normalized by the slack, where $S_{x}:=\Diag(s_{x})$
is the diagonalization of the slack.

We now introduce three barriers (and metrics) for handling the linear
constraints.

\paragraph{Logarithmic barrier.}

The logarithmic barrier $\phi_{\log}(x):=-\sum_{i=1}^{m}\log(a_{i}^{\T}x-b_{i})$
is the simplest self-concordant barrier for linear constraints. We
refer readers to \S\ref{proof:linear-log-barrier} for gentle introduction
to the log-barriers. As seen below, we demonstrate that the metric
induced by the logarithmic barrier has $\nu,\onu=m$ and requires
no scaling to achieve SSC, SLTSC, and SASC.
\begin{lem}
[Logarithmic barrier]\label{lem:log-barrier} For a closed convex
$K=\{x\in\Rd:Ax\geq b\}$ with $A\in\R^{m\times d}$ and $b\in\R^{m}$,
let $\phi_{\log}(x)=-\sum_{i=1}^{m}\log(a_{i}^{\T}x-b_{i})$ and define
$g(x):=\hess\phi_{\log}(x)=A_{x}^{\T}A_{x}$.
\begin{itemize}
\item $\nu=m$ (\citet{nesterov1994interior}).
\item SSC along $\rowspace(A)$ and $\onu=m$ (Lemma~\ref{lem:paramsBarrier}).
\item $\Dd^{2}g(x)[h,h]\succeq0$ for any $h\in\Rd$ (so SLTSC) (Claim~\ref{claim:diffLogBarrier}).
\item SASC (Lemma~\ref{lem:logBarrier-SASC}).
\end{itemize}
\end{lem}

\paragraph{Vaidya metric.}

In sampling over a polytope $K$, the number $m$ of constraints is
assumed to be greater than the ambient dimension $d$. Given that
the mixing time of the $\dw$ for uniform sampling is $\otilde{d\onu}=\otilde{dm}$,
a larger $m$ leads to a worse mixing time. Is there a self-concordant
barrier that has a better dependence on $m$ for its self-concordance
and symmetry parameters, without compromising SSC, SLTSC, and SASC?

Let us recall the \emph{leverage score} first and move onto such improved
self-concordant barriers. For a full-rank matrix $A\in\R^{m\times d}$
with $m\geq d$, we recall that $P(A)=A(A^{\T}A)^{-1}A^{\T}$ is the
orthogonal projection matrix onto the column space of $A$, and the
leverage scores of $A$ is $\sigma(A)=\diag(P(A))\in\R^{m}$. We let
$\Sigma(A):=\Diag(\sigma(A))=\Diag(P(A))$ and $P^{(2)}(A)=P(A)\circ P(A)$,
where $P(A)\circ P(A)$ is the Hadamard product of size $d\times d$
defined by $(P(A)\circ P(A))_{ij}=[P(A)]_{ij}^{2}$.

\citet{vaidya1996new} introduced the \emph{volumetric barrier} for
$K$ defined by
\[
\phi_{\vol}=\half\,\log\det(\hess\phi_{\log})=\half\,\log\det(A_{x}^{\T}A_{x})\,.
\]
Then the Hessian of $\phi_{\vol}$ can be written as
\[
\hess\phi_{\vol}=A_{x}^{\T}(3\Sigma_{x}-2P_{x}^{(2)})A_{x}\,,
\]
where $\Sigma_{x}=\Diag(\sigma(A_{x}))$ is the diagonalized leverage
scores, and this Hessian satisfies 
\[
A_{x}^{\T}\Sigma_{x}A_{x}\preceq\hess\phi_{\vol}(x)\preceq3A_{x}^{\T}\Sigma_{x}A_{x}\,.
\]
We refer readers to \S\ref{proof:linear-volumetric} for details.
In other words, the \emph{approximate} volumetric metric $A_{x}^{\T}\Sigma_{x}A_{x}$
serves as an $\mc O(1)$-approximation of the local metric $\hess\phi_{\vol}$
(i.e., $A_{x}^{\T}\Sigma_{x}A_{x}\asymp\hess\phi_{\vol}(x)$). We
find in Lemma~\ref{lem:paramsBarrier} that the local metric $40\sqrt{m}A_{x}^{\T}\Sigma_{x}A_{x}$
is SSC with $\nu,\,\onu=\mc O(\sqrt{m}d)$, but in some regime of
$d$ this parameter leads to worse mixing of the $\dw$. In the same
paper, \citet{vaidya1996new} introduced a \emph{regularized} volumetric
metric by adding $\O\bpar{\hess\phi_{\log}}$, which we call the \emph{Vaidya
metric}:

\[
g(x):=\sqrt{\frac{m}{d}}\,A_{x}^{\T}\bpar{\Sigma_{x}+\frac{d}{m}I_{m}}A_{x}\,.
\]
Note that $g(x)\asymp\hess\bpar{\sqrt{\frac{m}{d}}\bpar{\phi_{\vol}+\frac{d}{m}\text{\ensuremath{\phi_{\log}}}}}$.
We show that the Vaidya metric is also SSC, SLTSC, and SASC without
additional scaling, while it has a better $\nu$ and $\onu$ than
the logarithmic barrier.
\begin{lem}
[Vaidya metric]\label{lem:vaidya} For a closed convex $K=\{x\in\Rd:Ax\geq b\}$
with $A\in\R^{m\times d}$ and $b\in\R^{m}$, let $g(x)=\sqrt{\frac{m}{d}}A_{x}^{\T}\bpar{\Sigma_{x}+\frac{d}{m}I_{m}}A_{x}$.
\begin{itemize}
\item $\nu=\mc O(\sqrt{md})$ \citet[Theorem 5.2]{anstreicher1997volumetric}.
\item SSC and $\onu=\mc O(\sqrt{md})$ (Lemma~\ref{lem:paramsBarrier}).
\item SLTSC (Lemma~\ref{lem:vaidya-SLTSC}) and SASC (Lemma~\ref{lem:vaidya-SASC}).
\end{itemize}
\end{lem}

\paragraph{Lewis weights metric.}

Self-concordance and symmetry parameters of $\mc O(\sqrt{md})$ is
certainly better than $\mc O(m)$, but can we even achieve an $\mc O(d\log^{\mc O(1)}m)$
bound on those parameters?

Let us recall the $\ell_{p}$-\emph{Lewis weights}. The $\ell_{p}$-Lewis
weight of $A$ is denoted by $w(A)$, the solution $w$ to the equation
$w(A)=\diag\bpar{W^{\nicefrac{1}{2}-\nicefrac{1}{p}}A(A^{\T}W^{1-\nicefrac{2}{p}}A)^{-1}A^{\T}W^{\nicefrac{1}{2}-\nicefrac{1}{p}}}\in\R^{m}$
for $W:=\Diag(w)$. For $W_{x}=\Diag(w(A_{x}))$ and $p\geq2$, the
Lewis weight barrier function is defined by
\[
\phi_{\lw}(x):=\log\det(A_{x}^{\T}W_{x}^{1-\nicefrac{2}{p}}A_{x})\,.
\]
Note that the leverage score and volumetric barrier can be recovered
as a special case of the Lewis weight and barrier by setting $p=2$.
As done for the Vaidya metric, it is natural to consider the Lewis
weight metric with $p=\Theta(\log^{\mc O(1)}m)$, defined as 
\[
g(x):=\mc O(\log^{\mc O(1)}m)\,A_{x}^{\T}W_{x}A_{x}\,.
\]
In fact, this metric serves as an $\mc O(\log^{\mc O(1)}m)$-approximation
of $\hess\phi_{\lw}$, as demonstrated in the following relation proven
in \citet[Lemma 31]{lee2019solving}:
\[
A_{x}^{\T}\Sigma_{x}A_{x}\preceq\hess\phi_{\lw}\preceq(1+p)\,A_{x}^{\T}\Sigma_{x}A_{x}\,.
\]
Ignoring the logarithmic factors we have $\hess\phi_{\lw}\asymp g$.
Notably, the Lewis-weight metric needs an additional $\sqrt{d}$-scaling
for SLTSC and SASC, unlike the logarithmic barrier and Vaidya metric.
Hence, when combining this with other metrics, one should use $\sqrt{d}g$,
which leads to $\nu,\,\onu=\mc O(d^{3/2}\,\log^{\mc O(1)}m)$. 
\begin{lem}
[Lewis weight metric]\label{lem:Lewis-weight} For a closed convex
$K=\{x\in\Rd:Ax\geq b\}$ with $A\in\R^{m\times d}$ and $b\in\R^{m}$,
let $g(x)=\mc O(\log^{\mc O(1)}m)\,A_{x}^{\T}W_{x}A_{x}$.
\begin{itemize}
\item $\nu=\mc O(d\log^{5}m)$ \citet[Theorem 30]{lee2019solving}.
\item SSC and $\onu=\mc O(d\log^{\mc O(1)}m)$ (Lemma~\ref{lem:paramsBarrier}).
\item $\sqrt{d}g$ is SLTSC (Lemma~\ref{lem:Lw-SLTSC}) and SASC (Lemma~\ref{lem:Lw-SASC}).
\end{itemize}
\end{lem}

\subsubsection{Analysis of self-concordant metrics for linear constraints \label{subsec:analysis-linear-metric}}

\paragraph{Strong self-concordance and symmetry.}

We defer the proofs of two lemmas below to \S\ref{proof:linear-SSC-symm}.
We study SSC and symmetry of the metrics of the form $A_{x}^{\T}D_{x}A_{x}$
in Lemma~\ref{lem:helper4Diagonal}, where $D_{x}\in\R^{m\times m}$
is a diagonal matrix used to address the constraints of the form $Ax\geq b$
for $A\in\R^{m\times d}$ and $b\in\R^{m}$. Specifically, we relate
the notions of SSC and symmetry to well-studied terms in the field
of optimization, namely $\max_{i}\,[\sigma(\sqrt{D_{x}}A_{x})]_{i}/[D_{x}]_{ii}$
and $\snorm{\Dd D_{x}[h]}_{D_{x}^{-1}}^{2}$. 
\begin{lem}
\label{lem:helper4Diagonal} For $x\in\inter(K)$, let $g(x)=A_{x}^{\T}D_{x}A_{x}\in\Rdd$
for a diagonal matrix $0\prec D_{x}\in\R^{m\times m}$.
\begin{itemize}
\item For any PSD matrix function $g'$ such that $g'+g$ is invertible
on the domain,
\begin{align*}
 & \snorm{(g'(x)+g(x))^{-1/2}\Dd g(x)[h]\,(g'(x)+g(x))^{-1/2}}_{F}^{2}\\
 & \qquad\qquad\leq4\max_{i}\frac{[\sigma(\sqrt{D_{x}}A_{x})]_{i}}{[D_{x}]_{ii}}\cdot\bpar{\snorm h_{g(x)}^{2}+\sum_{i=1}^{m}\frac{(\Dd D_{x}[h])_{ii}^{2}}{[D_{x}]_{ii}}}\,.
\end{align*}
 
\item $\max_{h:\norm h_{g(x)}=1}\norm{A_{x}h}_{\infty}=\bpar{\max_{i\in[m]}\frac{[\sigma(\sqrt{D_{x}}A_{x})]_{i}}{[D_{x}]_{ii}}}^{1/2}$.
\item $K\cap(2x-K)\subset\dcal_{g}^{\sqrt{\tr(D_{x})}}(x)$.
\end{itemize}
\end{lem}

Then for each metric we refer to existing bounds on these terms, estimating
the smallest possible scaling required for SSC and symmetry. 
\begin{lem}
[Strong self-concordance and symmetry]\label{lem:paramsBarrier}
Let $A\in\R^{m\times d}$, $\Sigma_{x}=\Diag(\sigma(A_{x}))\in\R^{m\times m}$,
and $W_{x}=\Diag(w_{x})\in\R^{m\times m}$ for the $\ell_{p}$-Lewis
weight $w_{x}$ with $p=\mc O(\log m)$.
\begin{itemize}
\item Logarithmic metric: $g(x)=A_{x}^{\T}A_{x}$ with $D_{x}=I_{m}$ is
SSC along $\rowspace(A)$ with $\onu=m$.
\item Approximate volumetric metric: $g(x)=40\sqrt{m}A_{x}^{\T}\Sigma_{x}A_{x}$
with $D_{x}=40\sqrt{m}\Sigma_{x}$ is SSC with $\onu=\mc O(\sqrt{m}d)$.
\item Vaidya metric: $g(x)=22\sqrt{\frac{m}{d}}A_{x}^{\T}\bpar{\Sigma_{x}+\frac{d}{m}I_{m}}A_{x}$
with $D_{x}=22\sqrt{\frac{m}{d}}\bpar{\Sigma_{x}+\frac{d}{m}I_{m}}$
is SSC with $\onu=\mc O(\sqrt{md})$.
\item Lewis-weight metric: $\exists$ positive constants $c_{1}$ and $c_{2}$
such that $g(x)=c_{1}(\log m)^{c_{2}}A_{x}^{\T}W_{x}A_{x}$ is SSC
and $\onu$-symmetric with $\onu=\mc O^{*}(d)$.\label{lem:LSmetricStrongandSymmetry}
\end{itemize}
\end{lem}

\paragraph{Strongly lower trace self-concordance}

We show SLTSC of the Vaidya and Lewis-weight metric. Let $g_{2}$
be either Vaidya or Lewis-weight metric, and $g_{1}$ be an arbitrary
PSD matrix function on $K$ such that $g=g_{1}+g_{2}$ is PD on $\intk$.
Ensuring (S)LTSC of the Vaidya or Lewis-weight metrics is challenging,
as $\Dd^{2}g_{2}[h,h]\succeq0$ is difficult to verify due to complicated
expressions for $\Dd^{2}\Sigma_{x}[h,h]$ and $\Dd^{2}W_{x}[h,h]$.
As for the Vaidya metric, we compute higher-order derivatives of leverage
scores and other pertinent matrices in Lemma~\ref{lem:calculusLeverage},
finding succinct formulas by using algebraic properties of the Hadamard
product. We then show SLTSC of $g_{2}$ using these results (see \S\ref{proof:linear-vaidya-SLTSC}
for the proof):
\begin{lem}
[SLTSC of Vaidya]\label{lem:vaidya-SLTSC} $\tr\bpar{g^{-1}\Dd^{2}g_{2}(x)[h,h]}\geq-\snorm h_{g_{2}(x)}^{2}/2$
for the Vaidya metric $g_{2}$.
\end{lem}

For the Lewis-weights metric, analysis is more involved due to numerous
terms appearing in $\Dd^{2}W_{x}[h,h]$. In order to avoid dealing
with each of the terms, we employ existing bounds on derivatives of
$W_{x}$ and other relevant matrices in \S\ref{proof:linear-LW}.
This approach significantly simplifies the computation but comes at
the cost of an additional scaling of $\sqrt{d}$, which as far as
we can tell might be unavoidable. We refer readers to \S\ref{proof:linear-Lewis-SLTSC}
for the proof.
\begin{lem}
[SLTSC of Lewis-weight]\label{lem:Lw-SLTSC} $\tr\bpar{g(x)^{-1}\Dd^{2}g_{2}(x)[h,h]}\geq-\snorm h_{g_{2}(x)}^{2}$,
where $g_{2}(x)=cA_{x}^{\T}W_{x}A_{x}$ with $c=c_{1}(\log m)^{c_{2}}\sqrt{d}$
for some constants $c_{1},c_{2}>0$.
\end{lem}

\paragraph{Strongly average self-concordance.}

Typically, (S)ASC is the most challenging property to verify, often
requiring involved analysis in order to establish it \emph{without}
additional scalings. Since the three metrics are HSC (e.g., see Lemma~\ref{lem:Lw-hsc}
for Lewis-weight metrics), scaling by $d$ leads to SASC by Lemma~\ref{lem:hsc-to-sasc}.
However, for linear constraints one can still achieve SASC without
scaling (or with a smaller scaling) through more sophisticated concentration
techniques.

To sketch this idea, we recall that SASC requires showing that for
small enough $r$
\[
\snorm{z-x}_{g(z)}^{2}-\snorm{z-x}_{g(x)}^{2}\leq2\veps\frac{r^{2}}{d}\,.
\]
Taylor's expansion of $\snorm{z-x}_{g(z)}^{2}$ at $z=x$ up to second-order
necessitates bounds on
\[
\Dd g(x)[(z-x)^{\otimes3}]=\frac{r^{3}}{d^{3/2}}\Dd g(x)[h^{\otimes3}]\qquad\text{and}\qquad\Dd g(x')[(z-x)^{\otimes4}]=\frac{r^{4}}{d^{2}}\Dd^{2}g(x')[h^{\otimes4}]\,,
\]
for some $x'\in[x,z]$ and $h\sim\ncal(0,I_{d})$. Observe that the
first-order term $P(h):=\frac{r^{3}}{d^{3/2}}\Dd g(x)[h^{\otimes3}]$
is a Gaussian polynomial in $h$, and this is where we can invoke
the following concentration phenomenon:
\begin{lem}
[Concentration of Gaussian polynomials] \label{lem:conc-gaussian-poly}
For $d\geq1$, let $P:\Rd\to\R$ be a polynomial of degree $n$. For
any $t\geq(2e)^{n/2}$, 
\[
\P_{h\sim\ncal(0,I_{d})}\Bbrack{|P(h)|\geq t\sqrt{\E[P(h)^{2}]}}\leq\exp\bpar{-\frac{n}{2e}\,t^{2/n}}\,.
\]
\end{lem}

This concentration inequality necessitates bounding $\E[P(h)^{2}]$,
and this is where Stein's lemma comes into play:
\begin{lem}
\label{lem:stein} For $h=(h_{1},\dots,h_{d})\sim\ncal(0,I_{d})$,
it holds that $\E[h_{i}f(h)]=\E[\de_{i}f(h)]$.
\end{lem}

Unlike the first-order term, the second-order term is \emph{not} a
Gaussian polynomial due to $x'$ depending on $z$. To address this
issue, we derive an upper bound (in absolute value) of the quadratic
form. Using coordinate-wise closeness of slacks, leverage scores,
and Lewis weights at two nearby points, we replace every value estimated
at $z$ by those at $x$, removing dependence on $z$ in the quadratic
bound. The resulting quadratic bound is now a Gaussian polynomial,
so we follow the same proof approach as with the first-order term.

This approach was used by \citet{sachdeva2016mixing} for ASC of log-barriers
and by \citet{chen2018fast} for that of Vaidya and Lewis-weight metrics.
We further extend this approach to achieve SASC of those metrics,
going beyond ASC.
\begin{lem}
[SASC of logarithmic barrier] \label{lem:logBarrier-SASC} $g(x)=\hess\phi_{\log}(x)=A_{x}^{\T}A_{x}$
is SASC.
\end{lem}

See \S\ref{proof:linear-SASC-log} for the proof.
\begin{lem}
[SASC of Vaidya metric] \label{lem:vaidya-SASC} $g(x)=\mc O\bpar{\sqrt{\frac{m}{d}}}\,A_{x}^{\T}(\Sigma_{x}+\frac{d}{m}I_{m})A_{x}$
is SASC.
\end{lem}

See \S\ref{proof:linear-SASC-vaidya} for the proof.
\begin{lem}
[SASC of Lewis-weight metric] \label{lem:Lw-SASC} There exists constants
$c_{1}$ and $c_{2}$ such that $g(x)=c_{1}\sqrt{d}\log^{c_{2}}m\,A_{x}^{\T}W_{x}A_{x}=\mc O^{*}(\sqrt{d})\,A_{x}^{\T}W_{x}A_{x}$
is SASC.
\end{lem}

See \S\ref{proof:linear-SASC-Lw} for the proof.

\subsection{Quadratic potentials and constraints}

Suppose that in \eqref{eq:reduced-problem} we have either $f_{i}(x),\,h_{j}(x)=\snorm{x-\mu}_{\Sigma}^{2}$
or $\half x^{\T}Qx+p^{\T}x+l$ for $\mu,p\in\Rd$, $\Sigma\in\pd$,
and $0\neq Q\in\psd$.

\paragraph{Quadratic constraint.}

Consider a second-order region given by $K=\{x\in\Rd:\half x^{\T}Qx+p^{\T}x+l\leq0\}$.
\citet{nesterov1994interior} shows that $\phi:=-\log f$ is an $1$-self-concordant
barrier for $K$, when $f(x)=-\half\snorm{x-\mu}_{\Sigma}^{2}$ or
$-(\half x^{\T}Qx+p^{\T}x+l)$. Since $\onu=\mc O(\nu^{2})$ for a
self-concordant barrier due to Lemma~\ref{lem:bound-symmetry}, $\phi$
is $\mc O(1)$-symmetric. In case we consider $\snorm{x-\mu}_{\Sigma}^{2}$,
the trivial scaling by dimension $d$ implies that $d\phi$ is SSC
and $\mc O(d)$-symmetric.

Moreover, $d\phi$ is SASC by Lemma~\ref{lem:hsc-to-sasc} by HSC
of $\phi$. For HSC of $\phi$, we develop a handy tool for checking
HSC. See \S\ref{proof:quadratic} for the proof.
\begin{lem}
\label{lem:4th-log} For a real-valued function $f$ on $K\subset\Rd$,
let $\psi=-\log f$ be a $\nu$-self-concordant barrier for $K$.
Then, 
\[
|\Dd^{4}\psi(x)[h^{\otimes4}]|\lesssim\nu^{2}\snorm h_{\hess\psi(x)}^{2}+\big|\frac{\Dd^{4}f(x)[h^{\otimes4}]}{f(x)}\big|\,.
\]
\end{lem}

Using this tool, we can study properties of the barrier for the quadratic
constraints. We provide the proof in \S\ref{proof:quadratic}.
\begin{lem}
[Quadratic constraint]\label{lem:quadratic-const} For a closed convex
$K=\{x\in\Rd:\half x^{\T}Qx+p^{\T}x+l\leq0\}$ with $p\in\Rd$ and
$0\neq Q\in\psd$, let $\phi(x)=-\log(-l-p^{\T}x-\half x^{\T}Qx)$
and $g=d\,\hess\phi$.
\begin{itemize}
\item $\nu,\,\onu=\mc O(d)$.
\item SSC when $Q\succ0$, and SASC.
\item $\Dd^{2}g(x)[h,h]\succeq0$ for any $x\in\inter(K)$ and $h\in\Rd$
(so SLTSC).
\end{itemize}
\end{lem}

\paragraph{Gaussian distribution ($f(x)=\protect\half\protect\snorm{x-\mu}_{\Sigma}^{2}$).}

Suppose the quadratic term $f(x)=\half\snorm{x-\mu}_{\Sigma}^{2}$
appears in a potential of a target distribution. Then its epigraph
is 
\[
\{(x,t)\in\R^{d+1}:\half\snorm{x-\mu}_{\Sigma}^{2}-t\leq0\}\,,
\]
and clearly $q(x,t)=\half\snorm{x-\mu}_{\Sigma}^{2}-t$ is a quadratic
function in $(x,t)$. Hence, this level set admits an $1$-self-concordant
barrier
\[
\phi(x,t)=-\log(t-\half\snorm{x-\mu}_{\Sigma}^{2})\,.
\]
Our earlier discussion immediately leads to the following result:
\begin{lem}
[Quadratic potential] \label{lem:Gaussian-potential}Consider a closed
convex $K=\{(x,t):\half\snorm{x-\mu}_{\Sigma}^{2}\leq t\}$ with $\mu\in\Rd$
and $\Sigma\in\pd$, and let $\phi(x)=-\log(t-\half\snorm{x-\mu}_{\Sigma}^{2})$
and $g=d\,\hess\phi$.
\begin{itemize}
\item $\nu_{g},\,\onu_{g}=\mc O(d)$.
\item SSC and SASC.
\item $\Dd^{2}g(x,t)[h,h]\succeq0$ for any $(x,t)\in\inter(K)$ and $h\in\R^{d+1}$.
\end{itemize}
\end{lem}

\paragraph{Second-order cone ($f(x)=\protect\half\protect\snorm{x-\mu}_{\Sigma}$).}

It is common that a potential includes a non-smooth term like $\norm{Ax-b}_{2}$
in many applications, and we can handle such potentials via our framework.
\citet[Lemma 4.3.3]{nesterov1994interior} shows that 
\[
\phi(x,t)=-\log(t^{2}-\snorm x^{2})
\]
is a $2$-self-concordant for a level set $K=\{(x,t)\in\Rd\times\R:\snorm x_{2}\leq t\}$
(here we may assume that $\mu=0$ and $\Sigma=I$ due to Lemma~\ref{lem:linear-trans}).
This level set is called a \emph{second-order cone} or Lorentz cone.

Applying Lemma~\ref{lem:4th-log} to $f(x,t)=t^{2}-\norm x^{2}$
with $\nu=2$, we immediately show HSC of $\phi$. Thus, $d\phi$
satisfies SLTSC and SASC by Lemma~\ref{lem:hsc-to-sltsc} and Lemma~\ref{lem:hsc-to-sasc},
respectively.
\begin{lem}
[Second-order cone] \label{lem:soc} Consider a closed convex $K=\{(x,t):\snorm{x-\mu}_{\Sigma}\leq t\}$
with $\mu\in\Rd$ and $\Sigma\in\pd$, and let $\phi(x,t)=-\log(t^{2}-\snorm{x-\mu}_{\Sigma}^{2})$
and $g=d\,\hess\phi$.
\begin{itemize}
\item $\nu_{g},\,\onu_{g}=\mc O(d)$.
\item SSC, SASC, and SLTSC.
\end{itemize}
\end{lem}

\subsection{PSD cone}

The function $\phi(X)=-\log\det X$ serves as an $d$-self-concordant
barrier for the PSD cone $\psd$. While achieving self-concordance
does not require additional scaling, it turns out that SSC requires
a scaling of $\Theta(d)$. Notably, this scaling is less than the
trivial dimension-based scaling of $d_{s}:=d(d+1)/2$. Also, direct
computation leads to $\Dd^{4}\phi(X)[H,H]\succeq0$ (so SLTSC).

As $\phi$ is HSC, scaling by $d_{s}$ ensures SASC. However, we can
achieve ASC with a smaller scaling by $\mc O(d)$ via the random matrix
theory.
\begin{lem}
[PSD cone] \label{lem:psd} On a closed convex $K=\psd$, let $\phi(X)=-\log\det X$
and define $g=d\,\hess\phi$.
\begin{itemize}
\item $\nu=d^{2}$ (\citet{nesterov1994interior}) and $\onu=d^{2}$ (Lemma~\ref{lem:logdet-symm}).
\item SSC (Corollary~\ref{cor:logdet-ssc}).
\item $\Dd^{2}g(X)[H,H]\succeq0$ for any $X\in\intk$ and $H\in\mbb S^{d}$
(Lemma~\ref{lem:logdet-sltsc}).
\item ASC (Lemma~\ref{lem:logdet-asc}), and $d_{s}\,\hess\phi$ is SASC.
\end{itemize}
\end{lem}

\subsubsection{Formalism via matrix-vector transformations \label{subsec:formalism}}

In analyzing $\phi$, we work in $\R^{d_{s}}=\R^{d(d+1)/2}$ and $\mbb S^{d}$
simultaneously in the sequel, moving back and forth between them implicitly.
We justify this identification as follows.

\paragraph{Measure on $\protect\mbb S^{d}$.}

We can define and work with the Lebesgue measure on $\mbb S^{d}$
by identifying it with the Lebesgue measure on $\R^{d_{s}}$, where
each component in the Lebesgue measure on $\mbb S^{d}$ corresponds
to each entry in the upper triangular part. Hence, with the Lebesgue
measure $\D X$ on $\mbb S^{d}$ it is straightforward to define a
probability distribution on $\mbb S^{d}$ whose probability density
function with respect to $\D X$ is proportional to $\exp(-f)$ for
a function $f:\mbb S^{d}\to\R$. For instance, the uniform distribution
over a region corresponds to $f$ being constant in the region and
infinity outside of the region, and an exponential distribution to
$f(X)=\inner{C,X}=\tr(C^{\T}X)$ for $C\in\mbb S^{d}$.

\paragraph{Directional derivatives.}

A function $\phi:\mbb S^{d}\to\R$ induces its counterpart $\psi:\R^{d_{s}}\to\R$
defined by $\psi(x)=\phi(X)$ for $x:=\svec(X)$. For symmetric matrices
$\{H_{i}\}_{i\leq k}$, the $k$-th directional derivative of $\phi$
in directions $H_{1},\dots,H_{k}$ is 
\[
\Dd^{k}\phi(X)[H_{1},\cdots,H_{k}]\defeq\frac{\D^{k}}{\D t_{k}\cdots\D t_{1}}\phi\Bpar{X+\sum_{i=1}^{k}t_{i}H_{i}}\bigg\vert_{t_{1},\dots,t_{k}=0}\,.
\]
For $h_{i}:=\svec(H_{i})$, it follows that $\phi(X+\sum_{i=1}^{k}t_{i}H_{i})=\psi(x+\sum_{i=1}^{k}t_{i}h_{i})$
and thus
\[
\Dd^{k}\phi(X)[H_{1},\cdots,H_{k}]=\Dd^{k}\psi(x)[h_{1},\cdots,h_{k}]\,.
\]
With this identification in hand, since the notion of (symmetric or
strong) self-concordance is formulated in terms of directional derivatives,
we can deal with both representations without having to specify one
of them.

\paragraph{Important operators.}

We introduce three linear operators that enable us to make smooth
transitions between $\mbb S^{d}$ and $\R^{d_{s}}$.
\begin{defn}
[\citet{magnus1980elimination}] \label{def:linearOperators} Let
$E_{ij}=e_{i}e_{j}^{\T}\in\Rdd$ be the matrix with a single $1$
in the $(i,j)$ position and zeros elsewhere.
\begin{itemize}
\item $M:\R^{d_{s}}\to\R^{d^{2}}$ is the linear operator that maps $\svec(\cdot)$
to $\vec(\cdot)$ (i.e., $M\circ\svec=\vec$). It can be written as
$M=\sum_{i\geq j}\vec(T_{ij})u_{ij}^{\T}$, where $T_{ij}\in\Rdd$
has all zero entries except for $1$ at $(i,j)$ and $(j,i)$ positions
(i.e., $T_{ij}=E_{ij}+E_{ji}$ if $i\neq j$ and $E_{ij}$ if $i=j$),
and $u_{ij}=\svec(E_{ij})$.
\item $N:\R^{d^{2}}\to\R^{d^{2}}$ is the linear operator that maps $\vec(A)$
to $\vec\bpar{\half(A+A^{\T})}$ for a matrix $A\in\Rdd$.
\item $L:\R^{d_{s}}\to\R^{d^{2}}$ is the linear operator that maps $\vec(A)$
to $\svec(A)$ for a matrix $A\in\Rdd$. It can be written as $L=\sum_{i\geq j}u_{ij}\vec(E_{ij})^{\T}$. 
\end{itemize}
\end{defn}

\begin{lem}
[\citet{magnus1980elimination}] \label{lem:MNL-properties} Let
$M,N,L$ be matrices in Definition~\ref{def:linearOperators}.
\begin{itemize}
\item (Lemma 2.1) $N=N^{\T}=N^{2}$ and $N(A\otimes A)=(A\otimes A)N$ for
any $d\times d$ matrix $A$.
\item (Lemma 3.5) $MLN=N$.
\end{itemize}
\end{lem}

\subsubsection{Analysis of a self-concordant metric for the PSD cone \label{subsec:scBasicMetric}}

We first examine properties of the metric defined by the Hessian of
self-concordant barrier $\phi(X)=-\log\det X$ (see \citet[Theorem 4.3.3]{nesterov2003introductory}
for self-concordance). In this case, its Hessian and inverse have
clean formulas. 
\begin{prop}
\label{prop:metricFormula} Let $\grad_{X}^{2}\phi(X)=-\grad_{x}^{2}\log\det(\svec^{-1}(x))\in\R^{d_{s}\times d_{s}}$
for $X\in\psd$. Then,
\begin{align*}
\hess\phi(X) & =M^{\T}(X^{-1}\otimes X^{-1})M=M^{\T}(X\otimes X)^{-1}M\,,\\
\bpar{\hess\phi(X)}^{-1} & =M^{\dagger}(X\otimes X)\bpar{M^{\dagger}}^{\T}=LN(X\otimes X)NL^{\T}\,,
\end{align*}
where $M^{\dagger}=(M^{\T}M)^{-1}M^{\T}\in\R^{d_{s}\times d^{2}}$
is the Moore-Penrose inverse of $M\in\R^{d^{2}\times d_{s}}$.
\end{prop}

We defer the proof to Appendix~\ref{app:matrixCalculus}. We remark
that as an immediate corollary to this, the local norm of $h\in\R^{d_{s}}$
with metric $\hess\phi(X)$ is 
\[
\snorm h_{X}^{2}=\svec(H)^{\T}M^{\T}(X^{-1}\otimes X^{-1})M\svec(H)\underset{\text{(i)}}{=}\tr(HX^{-1}HX^{-1})\eqqcolon\snorm H_{X}^{2}\,,
\]
where (i) follows from $\vec=M\circ\svec$ (Definition~\ref{def:linearOperators})
and $\tr(DB^{\T}A^{\T}C)=\vec(A)^{\T}(B\otimes C)\vec(D)$ (Lemma~\ref{lem:Kronecker}). 

\paragraph{Symmetry.}
\begin{lem}
[$\onu$-symmetry] \label{lem:logdet-symm}For $X\in K=\psd$, the
barrier $\phi(X)=-\log\det X$ is $d$-symmetric.
\end{lem}

\begin{proof}
For $X\in K$, pick any $Y\in K\cap(2X-K)$, and define a symmetric
matrix $H:=Y-X$. Since $Y\in K$ and $2X-Y\in K$, we have $X+H\in K$
and $X-H\in K$. Thus,
\[
-I\preceq X^{-1/2}HX^{-1/2}\preceq I\,,
\]
and the magnitude of each eigenvalue $\{\lda_{i}\}_{i=1}^{d}$ of
$X^{-1/2}HX^{-1/2}$ is bounded by $1$. Hence,
\[
\snorm H_{X}^{2}=\tr(X^{-1}HX^{-1}H)=\snorm{X^{-1/2}HX^{-1/2}}_{F}^{2}\leq\sum_{i=1}^{d}\lda_{i}^{2}\leq d\,.\qedhere
\]
\end{proof}

\paragraph{Convexity of log-determinant of Hessian and SSC.}

Next, the convexity of the log-determinant of $\hess\phi$ can be
checked via properties of Kronecker products. See \S\ref{proof:psd-convex-ssc}
for the proof.
\begin{prop}
[Convexity of log-determinant of Hessian] \label{prop:convex-logdet}
$\log\det(\hess\phi(\cdot))$ is convex.
\end{prop}

We move onto SSC of $d\phi(X)$.
\begin{lem}
\label{lem:logdet-scaling} For $\psi_{X}:=\sup_{H\in\mbb S^{d}}\snorm{(\hess\phi(X))^{-1/2}\Dd^{3}\phi(X)[H]\,(\hess\phi(X))^{-1/2}}_{F}/\snorm H_{X}$,
we have
\[
\sqrt{2(d+1)}\leq\psi_{X}\leq2\sqrt{d}\,.
\]
\end{lem}

We present the proof in \S\ref{proof:psd-convex-ssc}. This result
informs us of the best possible scaling of $\phi$ that ensures SSC.
Recall that if $g$ satisfies $\snorm{g^{-1/2}\Dd g[h]g^{-1/2}}_{F}\leq2\alpha\norm h_{g}$
for $\alpha>0$, then $\alpha^{2}g$ is SSC. We remark that the scaling
of $d$ is obviously better than the trivial scaling of $d_{s}=\Theta(d^{2})$.
\begin{cor}
[Strong self-concordance] \label{cor:logdet-ssc} A function $d\phi$
is a strongly self-concordant barrier for $\psd$. Moreover, the scaling
factor of $d$ cannot be further improved.
\end{cor}

\paragraph{Strongly lower trace self-concordance.}

SLTSC of $\phi$ can be easily checked by noting $g(X)[H,H]=\tr(X^{-1}HX^{-1}H)$
and using the chain rule. See the details in \S\ref{proof:psd-sltsc}.
\begin{lem}
[SLTSC] \label{lem:logdet-sltsc}$\Dd^{2}g(X)[H,H]\succeq0$ for
any $X\in\intk$ and $H\in\mathbb{S}^{d}$.
\end{lem}

\paragraph{Average self-concordance.}

In establishing ASC, we find an interesting connection to a \emph{Gaussian
orthogonal ensemble} (GOE), one of the main objects studied in the
random matrix theory. We prove the following lemmas and explain challenges
when extending our arguments to SASC in \S\ref{proof:psd-asc}.
\begin{lem}
\label{lem:conn-to-goe} For $d_{s}=\frac{d(d+1)}{2}$ and $\svec(H)\sim\ncal\bpar{0,\frac{r^{2}}{d_{s}}g(X)^{-1}}$,
$\frac{\sqrt{d_{s}d}}{r}X^{-1/2}HX^{-1/2}$ is a GOE.
\end{lem}

\begin{lem}
[ASC] \label{lem:logdet-asc} $-d\,\log\det X$ is ASC.
\end{lem}

\subsection{Logarithm, exponential, entropy, and $\ell_{p}$-norm (power function)}

\paragraph{Logarithm in potentials.}

Consider $Q_{1}=\{(x,t)\in\R^{2}:-\log x\leq t,x>0\}$. As $f(\cdot)=-\log(\cdot)$
is convex on $\R_{+}$ and satisfies the condition in Lemma~\ref{lem:tool-convex}
with $\beta=2$ and $\gamma=6$,
\[
F(x,t)=-\log(t+\log x)-36\log x
\]
is a highly $37$-self concordant barrier for $Q_{1}$. Therefore,
$2F$ is SSC and SLTSC with $\onu=\mc O(1)$.
\begin{lem}
[Logarithm] Consider the direct product of level sets
\[
K=\prod_{i=1}^{d}\{(x_{i},t_{i})\in\R^{2}:-\log x_{i}\leq t_{i},\,x_{i}>0\}\,,
\]
and let $\phi(x,t)=-\sum_{i=1}^{d}\bpar{\log(t_{i}+\log x_{i})+36\log x_{i}}$
and $g=2\hess\phi$.
\begin{itemize}
\item $\nu,\,\onu=\mc O(d)$.
\item SSC and SLTSC.
\item $d\,\hess\phi$ is SASC.
\end{itemize}
\end{lem}

\begin{proof}
For $i\in[d]$, let $Q_{i}=\{(x_{i},t_{i})\in\R^{2}:-\log x_{i}\leq t_{i},\,y_{i}>0\}$
and $F_{i}(x_{i},t_{i})$ be the self-concordant barrier above. Note
that $2F_{i}$ is SSC and SLTSC. By Lemma~\ref{lem:ssc-direct} and
\ref{lem:sltsc-direct}, the Hessian of $F(x,t):=2\sum_{i=1}^{d}F_{i}(x_{i},t_{i})$
is SSC and SLTSC. The last item on SASC follows from Lemma~\ref{lem:hsc-to-sasc}.
\end{proof}

\paragraph{Exponent in potentials.}

Consider $Q_{2}=\{(x,t)\in\R^{2}:e^{x}\leq t\}=\{(x,t)\in\R^{2}:t>0,\,x\leq\log t\}$.
As $f(t)=\log t$ is concave and satisfies the condition in Lemma~\ref{lem:tool-concave}
with $\beta=2$ and $\gamma=6$,
\[
F(x,t)=-\log(\log t-x)-36\log t
\]
is a highly $37$-self concordant barrier for $Q_{2}$. Therefore,
$2F$ is SSC and SLTSC with $\onu=\mc O(1)$.
\begin{lem}
[Exponential] Consider the direct product of level sets
\[
K=\prod_{i=1}^{d}\{x_{i},t_{i})\in\R^{2}:\exp(x_{i})\leq t_{i}\}\,,
\]
and let $\phi(x,t)=-\sum_{i=1}^{d}(\log(\log t_{i}-x_{i})+36\log t_{i})$
and $g=2\hess\phi$.
\begin{itemize}
\item $\nu,\,\onu=\mc O(d)$.
\item SSC and SLTSC.
\item $d\,\hess\phi$ is SASC.
\end{itemize}
\end{lem}

\begin{proof}
For $i\in[d]$, let $Q_{i}=\{(x_{i},t_{i})\in\R^{2}:e^{x_{i}}\leq t_{i}\}$
and $F_{i}(x_{i},t_{i})$ be the self-concordant barrier above. Note
that $2F_{i}$ is SSC and SLTSC. By Lemma~\ref{lem:ssc-direct} and~\ref{lem:sltsc-direct},
the Hessian of $F(x,t):=2\sum_{i=1}^{d}F_{i}(x_{i},t_{i})$ is SSC
and SLTSC. The last item on SASC follows from Lemma~\ref{lem:hsc-to-sasc}.
\end{proof}

\paragraph{Entropy in potentials.}

Consider $Q_{3}=\{(x,t)\in\R^{2}:x\geq0,\,t\geq x\log x\}$. Note
that $f(x)=x\log x$ is convex on $\{x>0\}$ and satisfies the condition
in Lemma~\ref{lem:tool-convex} with $\beta=1$ and $\gamma=2$.
Hence,
\[
F(x,t)=-\log(t-x\log x)-36\log x
\]
is a highly $5$-self concordant barrier for $Q_{3}$. Therefore,
$2F$ is SSC and SLTSC with $\onu=\mc O(1)$.
\begin{lem}
[Entropy] Consider the direct product of level sets
\[
K=\prod_{i=1}^{d}\{(x_{i},t_{i})\in\R^{2}:x_{i}\geq0,\,t_{i}\geq x_{i}\log x_{i}\}\,,
\]
and let $\phi(x,t)=-\sum_{i=1}^{d}\bpar{\log(t_{i}-x_{i}\log x_{i})+36\log x_{i}}$
and $g=2\hess\phi$.
\begin{itemize}
\item $\nu,\,\onu=\mc O(d)$.
\item SSC and SLTSC.
\item $d\,\hess\phi$ is SASC.
\end{itemize}
\end{lem}

\begin{proof}
For $i\in[d]$, let $Q_{i}=\{(x_{i},t_{i})\in\R^{2}:x_{i}\geq0,\,t_{i}\geq x_{i}\log x_{i}\}$
and $F_{i}(x_{i},t_{i})$ be the self-concordant barrier above. Note
that $2F_{i}$ is SSC and SLTSC. By Lemma~\ref{lem:ssc-direct} and~\ref{lem:sltsc-direct},
the Hessian of $F(x,t):=2\sum_{i=1}^{d}F_{i}(x_{i},t_{i})$ is SSC
and SLTSC. The last item on SASC follows from Lemma~\ref{lem:hsc-to-sasc}.
\end{proof}

\paragraph{$\ell_{p}$-norm (power function).}

We start with the power functions. For $p\geq1$, consider $Q_{4}=\{(x,t)\in\R^{2}:t\geq\max(0,x)^{p}\}=\{(x,t)\in\R^{2}:t\geq0,\,x\leq t^{1/p}\}$.
Note that $f(t)=t^{1/p}$ is concave on $t>0$ and satisfies the condition
in Lemma~\ref{lem:tool-concave} with $\beta=2$ and $\gamma=6$.
Hence,
\[
F_{4}(x,t)=-\log(t^{1/p}-x)-36\log t
\]
is a highly $37$-self-concordant barrier for $Q_{4}$. Similarly,
$F_{5}(t,x)=-\log(t^{1/p}+x)-36\log t$ is a highly $37$-self concordant
barrier for the convex set $Q_{5}=\{(x,t)\in\R^{2}:t\geq\max(0,-x)^{p}\}$.
Since the convex set $Q_{6}=\{(x,t)\in\R^{2}:t\geq|x|^{p}\}$ is equal
to $Q_{4}\cap Q_{5}$, the sum of $F_{4}+F_{5}$, which is 
\[
F_{6}(x,t)=-\log(t^{2/p}-x^{2})-72\log t
\]
is a highly $72$-self-concordant barrier for $Q_{6}$. Hence, $2F$
is SSC and SLTSC with $\onu=\mc O(1)$.
\begin{lem}
[$\ell_p$-norm] Consider the direct product of level sets $K=\prod_{i=1}^{d}\{(x_{i},t_{i})\in\R^{2}:\Abs{x_{i}}^{p}\leq t_{i}\}$,
and let $\phi(x,t)=-\sum_{i=1}^{d}\bpar{\log(t_{i}^{2/p}-x_{i}^{2})+72\log t_{i}}$
and $g=2\hess\phi$.
\begin{itemize}
\item $\nu,\,\onu=\mc O(d)$.
\item SSC and SLTSC.
\item $d\,\hess\phi$ is SASC.
\end{itemize}
\end{lem}

\begin{proof}
Consider a highly $72$-self-concordant barrier $F_{i}$ above for
$\{(x_{i},t_{i}):|x_{i}|^{p}\leq t_{i}\}$ for $i\in[d]$. Note that
$2F_{i}$ is SSC and SLTSC. By Lemma~\ref{lem:ssc-direct} and~\ref{lem:sltsc-direct},
the Hessian of $F(x,t):=2\sum_{i=1}^{d}F_{i}(x_{i},t_{i})$ is SSC
and SLTSC. The last item on SASC follows from Lemma~\ref{lem:hsc-to-sasc}.
\end{proof}

\global\long\def\vec{\textup{\textsf{vec}}}%
\global\long\def\svec{\textup{\textsf{svec}}}%

\section{Examples \label{sec:examples}}

For given constraints and epigraphs, combining metrics for them (according
to the self-concordance theory for sampling developed in \S\ref{sec:sc-theory-rules})
and employing $\gcdw$ with the combined metric lead to a poly-time
mixing sampling algorithm. Compared to the state-of-the-art poly-time
mixing algorithm, the $\bw$, $\gcdw$ offers several advantages.
First, it does not require any preprocessing (e.g., rounding) due
to affine invariance. Also, it achieves faster mixing by leveraging
inherent geometric information in sampling problems.

The per-step complexity of $\dws$, however, is in general higher
than that of the $\bw$. The primary computational bottleneck lies
in computing the inverse of a local metric. Nevertheless, efficient
implementation of inverse maintenance can significantly reduce the
per-step complexity, improving the total complexity ($\#\,$iterations
needed for mixing times the per-step complexity).

In this section, we illustrate how our framework recovers theoretical
guarantees of previous work on $\dws$ for uniform sampling and extends
beyond uniform sampling. In particular, we show that $\gcdw$ is a
poly-time mixing algorithm capable of sampling uniform, exponential,
or Gaussian distributions on second-order cones or truncated PSD cones.
Additionally, we illustrate an efficient per-step implementation that
yields a faster total complexity when compared to general-purpose
samplers such as the $\bw$.

\subsection{Polytope sampling}

Consider a set of linear constraints given by $K=\{x\in\Rd:Ax\geq b\}$
with $A\in\R^{m\times d}$ and $b\in\R^{m}$. 

\paragraph{Uniform sampling.}

\citet{kannan2012random} first studied the $\dw$ for uniformly sampling
a polytope, where a local metric is set to be the Hessian of the logarithmic
barrier, $g=\hess\phi_{\textsf{log}}=A_{(\cdot)}^{\T}A_{(\cdot)}$.
They showed that the $\dw$ with the log-barrier mixes in $\mc O\bpar{md\log\frac{M}{\veps}}$
iterations with a warmness parameter $M$. An immediate consequence
of our work is that $\gcdw$ achieves the mixing time of $\otilde{md}$
\emph{without a warmness assumption}, as $\onu,\nu=m$ and $g$ is
SSC, LTSC, and ASC by Lemma~\ref{lem:log-barrier}.

\citet{chen2018fast} introduced the $\textsf{Vaidya walk}$ and the
$\textsf{Approximate John walk}$, which are essentially $\dws$ with
the Vaidya metric $\hess\phi_{\textsf{Vaidya}}$ and a version of
the Lewis-weight metric $\sqrt{d}\,\hess\phi_{\textsf{Lw}}$. Their
work showed that both walks achieves mixing times of $\mc O\bpar{\sqrt{m}d^{3/2}\log\frac{M}{\veps}}$
and $\mc O\bpar{d^{5/2}\log^{\mc O(1)}m\,\log\frac{M}{\veps}}$, respectively.
Building upon our analysis of the Vaidya metric and Lewis-weight metric
in Lemma~\ref{lem:vaidya} and \ref{lem:Lewis-weight}, we find that
$\gcdw$ with these metrics achieves the same mixing but without any
warmness assumption.

We note that for the same task the $\bw$ without a warm start requires
$\otilde{d^{3}}$ membership queries due to \citet{kannan1997random,jia2021reducing}.
Given that a membership query involves $\mc O(md)$ arithmetic operations,
the total complexity of the $\bw$ is $\otilde{md^{4}}$. In contrast,
the per-step of the $\dw$ with the log-barrier can be run in $\mc O(md^{\omega-1})$
operations through the fast matrix multiplication, so the total number
of arithmetic operations is $\otilde{m^{2}d^{\omega}}$. Thus, for
$m$ close to $d$ $\gcdw$ is provably faster than the $\bw$. When
an efficient inverse maintenance proposed in \citet{laddha2020strong}
is employed, the per-step complexity can be improved to $\mc O(d^{2}+\nnz(A))=\mc O(md)$.
In such cases $\gcdw$ is faster in a broader range of $m$. In particular,
if $A$ is as sparse as $\nnz(A)=\mc O(d^{2})$, then $\gcdw$ is
always faster than the $\bw$. Moreover, $\gcdw$ with the Lewis-weight
metric mixes in $\otilde{d^{2.5}}$ steps with the per-step complexity
of $\otilde{md^{\omega-1}}$, so it is always faster than the $\bw$
for any $m$.

\paragraph{Exponential and Gaussian sampling.}

The current mixing bound of the $\bw$ for general log-concave sampling
is $\otilde{d^{4}}$ due to \citet{lovasz2007geometry}. On the other
hand, the $\dw$ employed with any metric above for exponential sampling
converges in the same iterations as the $\dw$ for uniform sampling.
Since only difference between two sapling is the additional term of
$\exp\bpar{-(f(z)-f(x))}$ in the Metropolis filter, the fast implementation
techniques mentioned earlier can be applied to the context of exponential
sampling. As a result, for the exponential sampling each of the $\dws$
described above surpasses the $\bw$ by a larger margin.

For Gaussian sampling over a polytope, we first reduce it to the exponential
sampling as in \eqref{eq:reduced-problem}: for $y=(x,t)\in\R^{d+1}$
\begin{align*}
\text{sample } & y\sim\tilde{\pi}\propto\exp(-t)\\
\text{s.t. } & Ax\geq b,\ \half\snorm{x-\mu}_{\Sigma}^{2}\leq t\,.
\end{align*}
According to our theory, it is natural to use the metric given by
\[
g(x,t)=2\left[\begin{array}{cc}
\hess_{x}\phi_{\text{log}}(x)\\
 & 0
\end{array}\right]+2(d+1)\,\hess_{(x,t)}\phi_{\text{Gauss}}(x,t)\,,
\]
which is $\bpar{\mc O(m+d),\mc O(m+d)}$-Dikin-amenable due to Lemma~\ref{lem:Gaussian-potential}.
Thus, $\gcdw$ needs $\otilde{d(m+d)}$ iterations of the $\dw$.
We note that the log-barrier can be replaced by the Vaidya or Lewis-weight
metrics, and in such cases one can obtain provable guarantees on the
mixing time by computing $\nu$ and $\onu$, referring to \S\ref{sec:handbook-barrier}
or Table~\ref{tab:scaling-table}.

\subsection{Second-order cone sampling}

We consider a region given by $\snorm{x-\mu}_{\Sigma}\leq t$ and
$A\left[\begin{array}{cc}
x & t\end{array}\right]^{\T}\leq b$ for $A\in\R^{m\times(d+1)},b\in\R^{m},$ $\mu\in\Rd$, and $\Sigma\in\pd$.

\paragraph{Uniform and exponential sampling.}

In this case, our self-concordance theory suggests using
\[
\hess(2\sqrt{d+1}\phi_{\text{Lw}}+2(d+1)\phi_{\text{SOC}})\quad\text{or}\quad\hess(2\phi_{*}+2(d+1)\,\phi_{\text{SOC}})\ \text{for }*=\text{log, Vaidya}\,,
\]
to deal with the truncated SOC constraint. For the log-barrier case,
this yields an $\bpar{\mc O(m+d),\mc O(m+d)}$-Dikin-amenable metric
due to Lemma~\ref{lem:soc}, with which $\gcdw$ requires $\otilde{d(m+d)}$
iterations of the $\dw$.

\paragraph{Gaussian sampling.}

Following the reduction as in the polytope sampling, we should use
\[
g(x,t,t')=3\left[\begin{array}{cc}
\hess_{(x,t)}\phi_{\textup{log}}(x,t)+(d+1)\,\hess_{(x,t)}\phi_{\textup{SOC}}(x,t)\\
 & 0
\end{array}\right]+3(d+2)\,\hess_{(x,t,t')}\phi_{\textup{Gauss}}(x,t,t'),
\]
which is $\bpar{\mc O(m+d),\mc O(m+d)}$-Dikin-amenable, and thus
$\gcdw$ needs $\otilde{d(m+d)}$ iterations of the $\dw$.

\subsection{PSD cone sampling\label{subsec:PSD-cone-sampling}}

For a matrix $X\in\Rdd$, recall that $\vec(X)\in\R^{d^{2}}$ denotes
the vector obtained by stacking columns of $X$ vertically. Additionally,
we define $A\in\R^{m\times d^{2}}$, $S_{X}\in\R^{m\times m}$, and
$A_{X}\in\R^{m\times d^{2}}$ by 
\[
A:=\left[\begin{array}{ccc}
\vec(A_{1}) & \cdots & \vec(A_{m})\end{array}\right]^{\T},\quad S_{X}:=\Diag(\inner{A_{i},X}-b_{i}),\quad A_{X}:=S_{X}^{-1}A\,,
\]
where we assume $A$ has no all-zero rows and $(S_{X})_{ii}>0$ for
$i\in[m]$. 

\paragraph{Uniform and exponential sampling.}

The metric below comes from the Hessian of 
\[
-2d^{2}\,\log\det X-2\sum_{i=1}^{m}\log(\inner{A_{i},X}-b_{i})\,.
\]
Here the first term, the log-determinant, serves as a barrier for
the PSD cone while the second term is the standard logarithmic barrier
for linear constraints. We note that the $-\log\det X$ is strictly
convex on $x\in\intk$ for $K$ the truncated PSD cone, so all metrics
$g$ introduced in our main results are positive definite. Thus, the
$\dw$ with those $g$ is well-defined.
\begin{prop}
\label{thm:basicPSD} Let $K$ be the truncated PSD cone and $g$
be the local metric such that at each $X\in\intk$, for symmetric
matrices $H_{1},H_{2}$,
\[
g_{X}(H_{1},H_{2})=2d^{2}\tr(X^{-1}H_{1}X^{-1}H_{2})+2\,\vec(H_{1})^{\T}A_{X}^{\T}A_{X}\vec(H_{2})\,.
\]
Then $\gcdw$ needs $\otilde{(d^{3}+m)d^{2}}$ steps of the $\dw$
with the local metric $g$, where each step runs in $\mc O\bpar{(md^{\omega}+m^{2}d^{2})\wedge(d^{2\omega}+md^{2(\omega-1)})}$
time\footnote{Here $\omega<2.373$ is the current matrix multiplication complexity
exponent (\citet{le2014powers}).}.
\end{prop}

Since $g_{X}$ is $\bpar{\mc O(m+d^{3}),\mc O(m+d^{3})}$-Dikin-amenable
by Lemma~\ref{lem:psd}, $\gcdw$ requires ${\Otilde(d^{2}\,(d^{3}+m))}$
iterations of the $\dw$. As mentioned earlier, efficient maintenance
of the inverse of a metric function could lead to a faster per-step
complexity. As an example, we provide such an implementation of Proposition~\ref{thm:basicPSD}
in \S\ref{subsec:oracle-implementation}. Putting these together,
for an interesting regime of $m=\mc O(1)$, $\gcdw$ is faster than
the $\bw$ by a factor of $d$ in terms of the total complexity.

If we replace the log-barrier by the Vaidya metric, then the dependence
on $m$ is improved to $\sqrt{m}$ as in the polytope sampling. See
\S\ref{proof:Algorithms-for-PSD} for the proofs of the two claims
below.
\begin{prop}
\label{thm:hybridPSD} Let $K$ be the truncated PSD cone and $g$
be the local metric such that at each $X\in\intk$, for symmetric
matrices $H_{1},H_{2}$,
\[
g_{X}(H_{1},H_{2})=2d^{2}\tr(X^{-1}H_{1}X^{-1}H_{2})+44\sqrt{\frac{m}{d}}\,\vec(H_{1})^{\T}A_{X}^{\T}\bpar{\Sigma_{X}+\frac{d}{m}I_{m}}A_{X}\vec(H_{2})\,.
\]
Then $\gcdw$ needs $\otilde{(d^{2}+\sqrt{m})d^{3}}$ steps of the
$\dw$ with the local metric $g$, with each step running in $\otilde{md^{2(\omega-1)}}$
amortized time.
\end{prop}

Lastly, the dependence on $m$ can be made poly-logarithmic by working
with the Lewis-weight metric. We remark that for uniform sampling
the total complexity of $\gcdw$ is less than that of the $\bw$ by
the order of $d^{5-2\omega}$.
\begin{prop}
\label{thm:LSPSD} Let $K$ be the truncated PSD cone and $g$ be
the local metric such that at each $X\in\intk$, for symmetric matrices
$H_{1},H_{2}$, 
\[
g_{X}(H_{1},H_{2})=2d^{2}\tr(X^{-1}H_{1}X^{-1}H_{2})+dc_{1}(\log m)^{c_{2}}\,\vec(H_{1})^{\T}A_{X}^{\T}W_{X}A_{X}\vec(H_{2})\,,
\]
where $W_{X}$ is the diagonalized $\ell_{p}$-Lewis weight of $A_{X}$
with $p=\mc O(\log m)$, and $c_{1},c_{2}>0$ are universal constants.
Then $\gcdw$ requires $\otilde{d^{5}}$ steps of the $\dw$, with
each step running in $\otilde{md^{2(\omega-1)}}$ amortized time.
\end{prop}

\paragraph{Gaussian sampling.}

Just as in polytope or second-order cone sampling, we introduce a
new variable $t$ by replacing a quadratic term in the potential.
This reduces the Gaussian sampling problem to an exponential sampling
problem. We then work with a local metric 
\[
g(X,t)=3\Par{d\left[\begin{array}{cc}
\hess_{X}\phi_{\textup{Lw}}(X)\\
 & 0
\end{array}\right]+d^{2}\left[\begin{array}{cc}
\hess_{X}\phi_{\textup{PSD}}(X)\\
 & 0
\end{array}\right]+d^{2}\hess_{(X,t)}\phi_{\textup{Gauss}}(X,t)}\,,
\]
which is $\bpar{\mc O^{*}(d^{3}),\mc O^{*}(d^{3})}$-Dikin-amenable.
Thus, $\gcdw$ needs $\otilde{d^{5}}$ iterations of the $\dw$ with
the local metric $g$, and the per-step complexity remains $\otilde{md^{2(\omega-1)}}$
in amortized time.

\subsubsection{Per-step implementation \label{subsec:oracle-implementation}}

Now we design an oracle that implements each iteration of the $\dw$
(Algorithm~\ref{alg:DikinWalk}). This can be implemented as follows:
when the current point is $x$,
\begin{itemize}
\item Sample $z\sim\ncal\bpar{0,\frac{r^{2}}{d}g(x)^{-1}}$.
\item Compute $y=x+g(x)^{-1/2}z$ and propose it.
\item Accept $y$ with probability $1\wedge\bpar{\sqrt{\frac{\det g(y)}{\det g(x)}}\,\frac{\exp f(x)}{\exp f(y)}}$.
\end{itemize}
We provide two algorithms with the complexity of $\mc O(md^{\omega}+m^{2}d^{2})$
and $\mc O(d^{2\omega}+md^{2(\omega-1)})$. We can implement each
iteration in $\mc O\bpar{(md^{\omega}+m^{2}d^{2})\wedge(d^{2\omega}+md^{2(\omega-1)})}$
time by using the former for small $m$ and the latter for large $m$.
This completes the second half of Theorem~\ref{thm:basicPSD}. 

\paragraph{Algorithm for small $m$.}

For simplicity here, we ignore the constant factors of $g=g_{1}+g_{2}$,
where 
\[
g_{1}(X)=M^{\T}(X\kro X)^{-1}M=:BB^{\T}\qquad\text{and}\qquad g_{2}(X)=M^{\T}A^{\T}S_{X}^{-2}AM=:UU^{\T}\,.
\]
where $B:=M^{\T}(X\kro X)^{-1/2}\in\R^{d_{s}\times d^{2}}$ and $U:=M^{\T}A^{\T}S_{X}^{-1}\in\R^{d_{s}\times m}$.
Letting $u_{i}$ be the $i$-th column of $U$ for $i\in[m]$, we
note that $g_{2}=\sum_{i=1}^{m}u_{i}u_{i}^{\T}$.

We start with a subroutine for computing $g(X)^{-1}v$ for given $v\in\R^{d_{s}}$
in $\mc O(md^{\omega}+m^{2}d^{2})$ time.

\begin{algorithm2e}[t]

\caption{Computation of $g(X)^{-1}v$}\label{alg:subroutine}

\SetAlgoLined

\textbf{Input:} $X\in\psd$, vector $v\in\R^{d_{s}}$, local metric
$g$.

\textbf{Output:} $g(X)^{-1}v$

Prepare the column vectors $u_{i}$ of $U=M^{\T}A^{\T}S_{X}^{-1}$.

For $\bar{g}_{0}:=g_{1}(X)$, compute $\bar{g}_{0}^{-1}v$ and $\bar{g}_{0}^{-1}u_{i}$
for $i\in[m]$.

\For{$i=1,\cdots,m$}{

Compute $\bar{g}_{i}^{-1}v$ and $\bar{g}_{i}^{-1}u_{j}$ for $j\in[m]$,
according to 

\[
\bar{g}_{i}^{-1}w=\bar{g}_{i-1}^{-1}w-\frac{\bar{g}_{i-1}^{-1}u_{i}\cdot u_{i}^{\T}\bar{g}_{i-1}^{-1}w}{1+u_{i}^{\T}\bar{g}_{i-1}^{-1}u_{i}}\,.
\]

}

Output $\bar{g}_{m}^{-1}v$.

\end{algorithm2e}
\begin{prop}
\label{prop:oracle} Algorithm~\ref{alg:subroutine} computes $g(X)^{-1}v$
in $\mc O(md^{\omega}+m^{2}d^{2})$ time for a query $v\in\R^{d_{s}}$.
\end{prop}

See \S\ref{proof:eff_implement} for the proof. With this subroutine
in hand, we proceed to an efficient implementation of two tasks --
computation of (1) $g(x)^{-\half}z$ for a given vector $z\in\R^{d_{s}}$
and (2) $\sqrt{\frac{\det g(y)}{\det g(x)}}\,\frac{\exp f(x)}{\exp f(y)}$.

\begin{algorithm2e}[t]

\caption{Implementation of the $\dw$}\label{alg:perStep-small-m}

\SetAlgoLined

\textbf{Input:} current point $X\in\psd$, local metric $g$

\tcp{Step 1: Sampling from $\ncal\bpar{0,\frac{r^{2}}{d}g(X)^{-1}}$}

Draw $w\sim\ncal(0,I_{d^{2}+m})$ and $v\gets g(X)^{-1}\left[\begin{array}{cc}
B & U\end{array}\right]w$ by Algorithm~\ref{alg:subroutine}.\

Propose $y\gets\svec(X)+\frac{r}{\sqrt{d}}v$.

\

\tcp{Step 2: Computation of acceptance probability}

Use Algorithm~\ref{alg:subroutine} to prepare $\{\bar{g}_{i}^{-1}u_{1},\dots,\bar{g}_{i}^{-1}u_{m}\}_{i=0}^{m}$
at $X$ and $Y:=\svec^{-1}(y)$.\

$\det\bar{g}_{0}(\cdot)\gets2^{d(d-1)/2}(\det(\cdot))^{-(d+1)}$ ($\because$
Lemma~\ref{lem:Kronecker}-7)

\For{$i=1,\cdots,m$}{

$\det(\bar{g}_{i+1})\gets\det\bar{g}_{i}\cdot(1+u_{i+1}^{\T}\bar{g}_{i}^{-1}u_{i+1})$.

}

Accept $Y$ with probability $1\wedge\bpar{\sqrt{\frac{\det\bar{g}_{m}(Y)}{\det\bar{g}_{m}(X)}}\,\frac{\exp f(X)}{\exp f(Y)}}$.

\end{algorithm2e}
\begin{lem}
\label{lem:perStep-small-m} Algorithm~\ref{alg:perStep-small-m}
implements the $\dw$ with per-step complexity of $\mc O(md^{\omega}+m^{2}d^{2})$.
\end{lem}

\paragraph{Algorithm for large $m$.}

The algorithm right above has quadratic dependence on the number $m$
of constraints, which could become expensive for large $m$. In this
regime, we just fully compute the whole matrix function of size $\R^{d_{s}\times d_{s}}$,
which takes $\mc O(d^{2\omega}+md^{2(\omega-1)})$ time, and computing
its inverse, square-root, and determinant takes $\mc O(d^{2\omega})$
time. 

\subsubsection{Handling approximate Lewis weights \label{subsec:app-lewis-weight}}

When implementing the $\dw$ with the Lewis-weights metric, we use
an approximation algorithm presented in \citet{lee2019solving} for
computing and updating the Lewis weight, which ensures 
\[
(1-\delta)\wtilde_{X}\preceq W_{X}\preceq(1+\delta)\wt W_{X}
\]
for the approximate Lewis weights $\wtilde_{X}$ and a target accuracy
parameter $\delta$ (note that the initialization and update times
of the Lewis weight above hide poly-logarithmic dependence on $\log(1/\delta)$).
Strictly speaking, we should check that these approximate Lewis weights
do not affect the theoretical guarantees above.

To see this, let us define $\widetilde{g}=2(dg_{1}+\widetilde{g}_{2})$,
where for some constants $c_{1},c_{2}>0$
\[
g_{1}(X)=d^{2}M^{\T}(X\kro X)^{-1}M\qquad\text{and}\qquad\widetilde{g}_{2}=dc_{1}\Par{\log m}^{c_{2}}M^{\T}A_{X}^{\T}\widetilde{W}_{X}A_{X}M\,.
\]
First of all, the $\dw$ with $\widetilde{g}$ still converges to
a target distribution, since the approximation algorithm in \citet{lee2019solving}
is deterministic and thus the condition of detailed balance still
holds under the acceptance probability of $1\wedge\bpar{\sqrt{\frac{\det\tilde{g}(Y)}{\det\tilde{g}(X)}}\,\frac{\exp f(X)}{\exp f(Y)}}$.
For $\widetilde{P}_{X}$ the one-step distribution of the $\dw$ started
at $X$ with $\widetilde{g}$, we can show one-step coupling similar
to Lemma~\ref{lem:one-step}, following the overall proof therein
and taking $\delta=1/\text{poly}(d)$ small enough. See \S\ref{proof:Handling-approximate-Lewis}
for the proof. 
\begin{lem}
[One-step coupling] \label{lem:onestep-app-Lewis} For convex $K\subset\Rd$,
let $g:\intk\to\pd$ be SSC, ASC, LTSC, and $\phi:\intk\to\R$ be
its function counterpart. Suppose that the potential $f$ of the target
distribution $\pi$ is $\beta$-relatively smooth in $\phi$. Then
there exist constants $s_{1},s_{2}>0$ such that if $\snorm{x-y}_{g(x)}\leq s_{1}r/\sqrt{d}$
with $r=s_{2}\,(1\wedge1/\sqrt{\beta})$ for $x,y\in\intk$, then
$\dtv(\widetilde{P}_{x},\widetilde{P}_{y})\leq\frac{3}{4}+0.01$.
\end{lem}

\global\long\def\vec{\textup{\textsf{vec}}}%
\global\long\def\svec{\textup{\textsf{svec}}}%

\section{Proofs \label{sec:proofs}}

We collect deferred proofs in this section.

\subsection{Mixing of \texorpdfstring{$\msf{Dikin\ walk}$}{Dikin walk} ($\S$\ref{sec:mixing-Dikin})}

\subsubsection{One-step coupling \label{proof:onestep}}

We start with the one-step coupling of the $\dw$ under the setting
$\alpha\hess\phi\preceq\hess f\preceq\beta\hess\phi$ on $\intk$.
Roughly speaking, if $\snorm{x-y}_{x}\leq r/\sqrt{d}$ with $r\lesssim1\wedge\nicefrac{1}{\sqrt{\beta}}$,
then $\dtv(P_{x},P_{y})\leq0.99$.
\begin{proof}
[Proof of Lemma~\ref{lem:one-step}] For $\pi\propto\exp(-f)\cdot\mathbf{1}_{K}$
and $z\sim\ncal(x,\frac{r^{2}}{d}g(x)^{-1})$, let us denote
\[
p_{x}=\ncal\Bpar{x,\frac{r^{2}}{d}g(x)^{-1}},\qquad R(x,z)=\frac{p_{z}(x)}{p_{x}(z)}\frac{\pi(z)}{\pi(x)},\qquad A(x,z)=\min\bpar{1,R(x,z)\,\mathbf{1}_{K}(z)}\,.
\]
The transition kernel $P(x,\cdot)$ of the $\dw$ started at $x$
can be written as 
\[
P(x,dz)=\underbrace{\bpar{1-\E_{p_{x}}[A(x,\cdot)]}}_{\eqqcolon r_{x}}\,\delta_{x}(\D z)+A(x,z)\,p_{x}(\D z)\,.
\]
Thus, for $x,y\in\intk$,
\begin{align}
\dtv(P_{x},P_{y}) & =\underbrace{\half(r_{x}+r_{y})}_{\msf I}+\underbrace{\half\int|A(x,z)\,p_{x}(z)-A(y,z)\,p_{y}(z)|\,\D z}_{\msf{II}}\,.\tag{\ensuremath{\msf{TV\text{-}decomposition}}}\label{eq:tv-formula}
\end{align}

Let $h\sim\ncal(0,I_{d})$ and denote a bad event $B_{0}=\{z\in\Rd:\snorm{z-x}_{x}\ge cr\}$
with $c$ determined later. Due to $\snorm{z-x}_{x}=\frac{r}{\sqrt{d}}\snorm h$
(in law) and concentration of the standard Gaussian in a thin shell
of radius $\sqrt{d}$ with annulus $\mc O(1)$\footnote{A standard Gaussian $h\sim\ncal(0,I_{d})$ is concentrated around
a thin sell of radius $\sqrt{d}$ with annulus $\mc O(1)$: For $t>0$,
\[
\P_{h}(\norm h_{2}\geq\sqrt{d}+t)\leq\exp(-t^{2}/2)\,.
\]
}, we have $\P_{z}(B_{0})=\P_{h}(\snorm h\geq c\sqrt{d})\leq\exp\bpar{-(c-1)\sqrt{d}/2}$.
Hence, $\P(B_{0})\leq\veps$ for $c\geq1+\sqrt{\frac{2}{d}\,\log\frac{1}{\veps}}$. 

\paragraph{Rejection probability $r_{x}$ and $r_{y}$ (Term $\protect\msf I$).}

Note that
\[
r_{x}=1-\E_{p_{x}}[A(x,z)]=1-\int\min\Bpar{1,\,\underbrace{\mathbf{1}_{K}(z)\frac{\exp f(x)}{\exp f(z)}}_{\eqqcolon\textsf{A}}\underbrace{\frac{p_{z}(x)}{p_{x}(z)}}_{\eqqcolon\textsf{B}}}p_{x}(\D z)\,.
\]

As for $\textsf{A}$, we let $\hess\phi\preceq c_{\phi}g$ for some
$c_{\phi}>0$ and use Taylor's expansion at $x\in K\cap B_{0}^{c}$
to show that for some $x^{*}\in[x,z]$, 
\begin{align*}
f(x)-f(z) & +\nabla f(x)^{\T}(z-x)=-\snorm{z-x}_{\hess f(x^{*})}^{2}\geq-c_{\phi}\beta\,\snorm{z-x}_{g(x^{*})}^{2}\\
 & \underset{\text{(i)}}{\geq}-c_{\phi}\beta\,\snorm{z-x}_{x}^{2}\cdot(1+2\snorm{x-z}_{x})^{2}\geq-c_{\phi}\beta c^{2}r^{2}(1+2cr)^{2}\underset{\text{(ii)}}{\geq}-\veps\,,
\end{align*}
where we used Lemma~\ref{lem:scCloseness} in (i) and took $r\leq r_{1}(\veps)$
in (ii), which is defined so that $\beta c_{\phi}c^{2}r^{2}(1+cr)^{2}\leq\veps$
for any $r\leq r_{1}(\veps)$. It follows from $\dcal_{g}^{1}(x)\subset K$
and symmetry of $\ncal_{g}^{r}(x)$ that there exists a half-ellipsoid
$G\subset\dcal_{g}^{1}(x)$ in which $\inner{\grad f(x),z-x}\leq0$.
Thus, $f(x)-f(z)\geq-\veps$ holds on $z\in G$.

For a bad event $B_{1}:=G^{c}$, it holds that 
\[
\P_{z}(B_{1})\leq\half+\P_{z}\bpar{\dcal_{g}^{1}(x)^{c}}=\half+\P_{z}(\snorm{z-x}_{x}\geq1)=\half+\P_{h}\Bpar{\norm h\geq\frac{\sqrt{d}}{r}}\leq\half+\veps\,,
\]
where the last inequality follows from concentration of $h$ for any
$r\leq r_{2}(\veps):=\bpar{1+\frac{2}{\sqrt{d}}\,\log\frac{1}{\veps}}^{-1}$. 

As for $\textsf{B}$, for $\vphi(x):=\half\log\det g(x)$ we have
\[
\log\text{\textsf{B}}=-\frac{d}{2r^{2}}\bpar{\snorm{z-x}_{z}^{2}-\snorm{z-x}_{x}^{2}}+\bpar{\vphi(z)-\vphi(x)}\,.
\]
Invoking ASC of $\phi$, we can take $r_{3}(\veps)$ so that $\P_{z}\bpar{\snorm{z-x}_{z}^{2}-\snorm{z-x}_{x}^{2}\leq2\veps r^{2}/d}\geq1-\veps$
for any $r\leq r_{3}(\veps)$ and control the first term. Let the
complement of this event be our second bad event $B_{2}$.

For $\vphi(z)-\vphi(x)$, Taylor's expansion of $\vphi$ at $x$ leads
to 
\[
\vphi(z)-\vphi(x)=\underbrace{\inner{\grad\vphi(x),z-x}}_{\eqqcolon\textsf{A'}}+\underbrace{\half\inner{\hess\vphi(x^{*})(z-x),z-x)}}_{\eqqcolon\textsf{B'}}\text{ for some }x^{*}\in[x,z]\,.
\]
As for $\textsf{A}'$, we have $\inner{\grad\vphi(x),z-x}=\frac{r}{\sqrt{d}}\inner{g(x)^{-1/2}\grad\vphi(x),h}$,
and a standard tail bound for $h$ leads to 
\[
\P_{z}\Bpar{\inner{\grad\vphi(x),z-x}\leq-\frac{r}{\sqrt{d}}\,\snorm{g(x)^{-1/2}\nabla\vphi(x)}_{2}\cdot2\log\frac{1}{\veps}}\leq\veps\,.
\]
We call this event $B_{3}$ and bound $\snorm{g(x)^{-1/2}\nabla\vphi(x)}_{2}$
via SSC of $g$ as follows: omitting $x$ for simplicity,
\begin{align*}
\snorm{g^{-\half}\nabla\vphi}_{2} & =\sup_{v:\norm v_{2}=1}\inner{\grad\vphi,g^{-\half}v}\underset{\text{(i)}}{=}\sup_{v}\tr(g^{-1}\Dd g[g^{-\half}v])=\sup_{v}\tr(g^{-\half}\Dd g[g^{-\half}v]\,g^{-\half})\\
 & \underset{\text{(ii)}}{\leq}\sup_{v}\sqrt{d}\,\snorm{g^{-\half}\Dd g[g^{-\half}v]\,g^{-\half}}_{F}\underset{\text{(iii)}}{\leq}\sup_{v}2\sqrt{d}\snorm{g^{-\half}v}_{x}=2\sqrt{d}\,,
\end{align*}
where (i) follows from \eqref{eq:gradLogDet}, (ii) is due to $\tr(A)\leq\sqrt{d}\snorm A_{F}$
for $A\in\Rdd$, and (iii) is due to SSC. Conditioned on $B_{3}^{c}$,
taking $r\leq r_{4}(\veps):=\veps(4\log\frac{1}{\veps})^{-1}$, we
have
\[
\msf{A'}=\inner{\grad\vphi(x),z-x}\geq-4r\,\log\frac{1}{\veps}\geq-\veps\,.
\]

As for $\textsf{B}'$, denoting $u=z-x$ for $z\in B_{0}^{c}$ 
\begin{align}
\Dd^{2}\vphi(x^{*})[u,u] & \underset{\eqref{eq:hessLogDet}}{=}\tr\bpar{g(x^{*})^{-1}\Dd^{2}g(x^{*})[u,u]}-\snorm{g(x^{*})^{-\half}\Dd g(x^{*})[u]\,g(x^{*})^{-\half}}_{F}^{2}\nonumber \\
 & \underset{\text{(i)}}{\geq}-\snorm u_{x^{*}}^{2}-\snorm{g(x^{*})^{-\half}\Dd g(x^{*})[u]\,g(x^{*})^{-\half}}_{F}^{2}\ge-\snorm u_{x^{*}}^{2}-4\snorm u_{x^{*}}^{2}\nonumber \\
 & \underset{\text{(ii)}}{\geq}-5(1-\snorm{x-x^{*}}_{x})^{-2}\snorm u_{x}^{2}\label{eq:so-taylor-logdet}\\
 & \geq-5(1+2cr)^{2}c^{2}r^{2}\,,\nonumber 
\end{align}
where (i) follows from LTSC and (ii) follows from Lemma~\ref{lem:scCloseness}.
Hence, $\msf{B'}\geq-\veps/2$ by taking $r\leq r_{5}(\veps)$ so
that $5(1+2cr_{5})^{2}c^{2}r_{5}^{2}=\veps$. 

In summary, conditioned on $G:=\bigcap_{i=0}^{3}B_{i}^{c}$ with $\P_{z}(G)\geq\half-4\veps$
due to the union bound, we have
\begin{align}
\textsf{A}: & \,\frac{\exp f(x)}{\exp f(z)}\geq\exp(-\veps)\,,\label{eq:fx-ratio}\\
\textsf{B}: & \,\frac{p_{z}(x)}{p_{x}(z)}\geq\exp(-3\veps)\,,\label{eq:prop-ratio}\\
 & \,\vphi(z)-\vphi(x)\geq-2\veps\,.\label{eq:vphi-z-x}
\end{align}
Combining these together,
\begin{align*}
r_{x} & =1-\int\min\Bpar{1,\mathbf{1}_{K}(z)\frac{\exp f(x)}{\exp f(z)}\,\frac{p_{z}(x)}{p_{x}(z)}}p_{x}(\D z)\leq1-\int_{G}(1\wedge e^{-\veps}e^{-3\veps})\,\P_{z}(G)\leq\half+5\veps\,.
\end{align*}
Bounding $r_{y}$ in the same way, we conclude that $\textsf{I}\leq\half+5\veps$
in \eqref{eq:tv-formula}.

\paragraph{Overlapping part (Term $\protect\msf{II}$).}

WLOG, assume $f(y)\geq f(x)$. We denote good events by $G_{x}=\cap_{i=0,2,3}B_{x,i}^{c}$
and $G_{y}=\cap_{i=0,2,3}B_{y,i}^{c}$ such that $\P_{p_{x}}(G_{x}^{c})\leq3\veps$
and $\P_{p_{y}}(G_{y}^{c})\leq3\veps$, where 
\begin{align*}
B_{x,0} & =\{\snorm{z-x}_{x}\geq cr\}\ \text{with }c\geq1+\frac{2}{\sqrt{d}}\,\log\frac{1}{\veps},\ \text{and}\ B_{x,2}=\Bbrace{\snorm{z-x}_{z}^{2}-\snorm{z-x}_{x}^{2}>\frac{2\veps r^{2}}{d}}\\
B_{x,3} & =\Bbrace{\grad\vphi(x)^{\T}(z-x)\leq-\frac{2r\log\frac{1}{\veps}}{\sqrt{d}}\,\snorm{g(x)^{-\half}\nabla\vphi(x)}_{2}}\,.
\end{align*}
Let $G:=G_{x}\cup G_{y}$, and define a partition of $G$ by
\[
G_{x\backslash y}:=G_{x}\backslash G_{y},\qquad G_{x,y}:=G_{x}\cap G_{y},\qquad G_{y\backslash x}:=G_{y}\backslash G_{x}\,.
\]
Now we decompose the term $\textsf{II}$ as follows: for $Q:=|A(x,z)p_{x}(z)-A(y,z)p_{y}(z)|$,
\begin{align*}
\textsf{II} & =\half\int_{K\backslash G}Q\,\D z+\underbrace{\half\int_{G_{x\backslash y}}Q\,\D z}_{\eqqcolon\acal}+\underbrace{\half\int_{G_{y\backslash x}}Q\,\D z}_{\eqqcolon\bcal}+\underbrace{\half\int_{G_{x,y}}Q\,\D z}_{\eqqcolon\ccal}\\
 & \leq\half\bpar{\P_{p_{x}}(K\backslash G)+\P_{p_{y}}(K\backslash G)}+\acal+\bcal+\ccal\leq\half\bpar{\P_{p_{x}}(G_{x}^{c})+\P_{p_{y}}(G_{y}^{c})}+\acal+\bcal+\ccal\\
 & \leq3\veps+\acal+\bcal+\ccal\,.
\end{align*}
The term $\acal$ can be further decomposed by
\begin{align*}
2\mc A & \leq\int_{G_{x\backslash y}}A(x,z)\,|p_{x}(z)-p_{y}(z)|\,\D z+\int_{G_{x\backslash y}}|A(x,z)-A(y,z)|\,p_{y}(\D z)\\
 & \leq\int_{G_{x\backslash y}}|p_{x}(z)-p_{y}(z)|\,\D z+\P_{p_{y}}(G_{x\backslash y})\leq\int_{G_{x\backslash y}}|p_{x}(z)-p_{y}(z)|\,\D z+\underbrace{\P_{p_{y}}(G_{y}^{c})}_{\leq3\veps}\,,
\end{align*}
and in a similar way $\mc B\leq\half\int_{G_{y\backslash x}}|p_{x}(z)-p_{y}(z)|\,\D z+3\veps/2$.
Combining these together,
\begin{align*}
\mc A+\mc B & \leq3\veps+\half\int_{G_{x\backslash y}\cup G_{y\backslash x}}|p_{x}(z)-p_{y}(z)|\,\D z\leq3\veps+\dtv(p_{x},p_{y})\leq4\veps\,,
\end{align*}
where we used $\dtv(p_{x},p_{y})\leq\veps$; to see this, recall Pinsker's
inequality and a formula for the $\KL$ divergences between two Gaussians:
\[
2[\dtv(p_{x},p_{y})]^{2}\leq\KL(p_{y}\mmid p_{x})=\half\,\Bpar{\tr\bpar{g(y)^{-1}g(x)}-d+\log\det\bpar{g(y)g(x)^{-1}}+\frac{d}{r^{2}}\,\snorm{y-x}_{x}^{2}}\,.
\]
Let $\{\lda_{i}\}_{i\in[d]}$ be the eigenvalues of $g(x)^{-\half}g(y)g(x)^{-\half}$
and $\snorm{x-y}_{x}\leq\frac{sr}{\sqrt{d}}$ with $s>0$ to be determined.
Then, $\half\leq\lda_{i}\leq1+8\norm{x-y}_{x}$ by Lemma~\ref{lem:scCloseness}.
Using this and $\log x\leq x-1$ for $x>0$,
\begin{align*}
2\,\KL(p_{y}\mmid p_{x}) & =\sum_{i=1}^{d}\Bpar{\lda_{i}-1+\log\frac{1}{\lda_{i}}}+\frac{d}{r^{2}}\,\snorm{y-x}_{x}^{2}\le\sum_{i=1}^{d}\frac{(\lda_{i}-1)^{2}}{\lda_{i}}+s^{2}\leq s^{2}\,(128r^{2}+1)\,,
\end{align*}
Taking $s\leq s_{1}(\veps):=\veps$ and $r\leq r_{6}(\veps)$ so that
$\sqrt{128r_{6}^{2}+1}\leq2$, we obtain 
\begin{align}
\dtv(p_{x},p_{y})\leq\sqrt{\half\,\KL(p_{y}\mmid p_{x})} & \leq\frac{s}{2}\sqrt{128r^{2}+1}\leq\veps\,,\label{eq:TV-by-KL}
\end{align}

We now bound $\mc C$. Recall $B_{x,1}=\{\inner{\grad f(x),z-x}\ge0\}$
and $\P_{p_{x}}(B_{x,1})\leq\half+\O(\veps)$. Then,
\begin{align*}
2\mc C & =\int_{(G_{x}\cap G_{y})\backslash B_{x,1}^{c}}Q\,\D z+\int_{G_{x}\cap G_{y}\cap B_{x,1}^{c}}Q\,\D z\leq\int_{B_{x,1}}\underbrace{Q}_{\text{The traingle inequality}}\,\D z+\int_{G_{x}\cap G_{y}\cap B_{x,1}^{c}}Q\,\D z\\
 & \leq\int_{B_{x,1}}|A(x,z)-A(y,z)|\,p_{x}(\D z)+\int_{B_{x,1}}A(y,z)\:|p_{x}(z)-p_{y}(z)|\,\D z+\int_{G_{x}\cap G_{y}\cap B_{x,1}^{c}}Q\,\D z\\
 & \le\underbrace{\P_{p_{z}}(B_{x,1})}_{\leq\half+\veps}+2\underbrace{\dtv(p_{x},p_{y})}_{\leq\veps\ (\ref{eq:TV-by-KL})}+\int_{G_{x}\cap G_{y}\cap B_{x,1}^{c}}|A(x,z)\,p_{x}(z)-A(y,z)\,p_{y}(z)|\,\D z\\
 & \leq\half+2\veps+\int_{G_{x}\cap G_{y}\cap B_{x,1}^{c}}|A(x,z)\,p_{x}(z)-A(y,z)\,p_{y}(z)|\,\D z\,.
\end{align*}
One can check that
\[
|A(x,z)\,p_{x}(z)-A(y,z)\,p_{y}(z)|\,\D z=\Big|\min\Bpar{1,\underbrace{\frac{\exp f(x)}{\exp f(z)}\,\frac{p_{z}(x)}{p_{x}(z)}}_{\eqqcolon\msf U}}-\min\Bpar{\underbrace{\frac{p_{y}(z)}{p_{x}(z)}}_{\eqqcolon\msf V},\underbrace{\frac{\exp f(y)}{\exp f(z)}\,\frac{p_{z}(y)}{p_{x}(z)}}_{\eqqcolon\msf W}}\Big|\,p_{x}(\D z)\,.
\]
Here we note that $\msf U\geq e^{-4\veps}$ due to $\frac{\exp f(x)}{\exp f(z)}\geq e^{-\veps}$
and $\frac{p_{z}(x)}{p_{x}(z)}\geq e^{-3\veps}$ from \eqref{eq:fx-ratio}
and \eqref{eq:prop-ratio}. 

We now show that under additional conditioning, $|\log\textsf{V}|\lesssim\veps$
and $\log\textsf{W}\gtrsim-\veps$ on $z\in G_{x}\cap G_{y}\cap B_{x,1}^{c}$.
For $\vphi(\cdot)=\half\log\det g(\cdot)$ and $\msf L:=-\frac{d}{2r^{2}}(\snorm{z-y}_{y}^{2}-\snorm{z-x}_{x}^{2})$,
\begin{align}
\log\msf V & =-\frac{d}{2r^{2}}(\snorm{z-y}_{y}^{2}-\snorm{z-x}_{x}^{2})+\vphi(y)-\vphi(x)\nonumber \\
 & =\textsf{L}+\inner{\grad\vphi(x),y-x}+\underbrace{\half\inner{\hess\vphi(x^{*})(y-x),y-x}}_{\text{Use }\eqref{eq:so-taylor-logdet}}\quad\text{for some }x^{*}\in[x,y]\label{eq:bound-vphi}\\
 & \geq\textsf{L}-\snorm{g(x)^{-1/2}\nabla\vphi(x)}_{2}\snorm{y-x}_{x}-5\underbrace{(1+2\snorm{x-y}_{x})^{2}}_{\leq2}\snorm{y-x}_{x}^{2}\nonumber \\
 & \geq\textsf{L}-2\sqrt{d}\cdot s\frac{r}{\sqrt{d}}-10s^{2}\frac{r^{2}}{d}\geq\textsf{L}-\veps\,,\label{eq:logV-lower}
\end{align}
where the inequality follows from $s\leq\frac{\veps}{10}$ and $r\leq r_{7}(\veps):=1$. 

As for $\textsf{W}$, due to $f(y)\geq f(x)$ and $\exp(f(x)-f(z))\geq\exp(-\veps)$,
\begin{align}
\log\msf W & \geq\log\Bpar{\frac{\exp f(x)}{\exp f(z)}\frac{p_{z}(y)}{p_{x}(z)}}\geq-\veps-\frac{d}{2r^{2}}(\snorm{z-y}_{z}^{2}-\snorm{z-x}_{x}^{2})+\vphi(z)-\vphi(x)\nonumber \\
 & \underset{\text{(i)}}{\geq}-\veps-\frac{d}{2r^{2}}\Bpar{\snorm{z-y}_{y}^{2}+2\veps\frac{r^{2}}{d}-\snorm{z-x}_{x}^{2}}-2\veps=\textsf{L}-4\veps\,,\label{eq:logW-lower}
\end{align}
where (i) follows from $\snorm{z-y}_{z}^{2}-\snorm{z-y}_{y}^{2}\leq2\veps r^{2}/d$
on $z\in B_{y,2}^{c}$, and $\vphi(z)-\vphi(x)\geq-2\veps$ on $z\in B_{x,3}^{c}$
from \eqref{eq:vphi-z-x}.

Lastly, we show that $|\textsf{L}|$ is bounded by $\mc O(\veps)$
with high probability (w.r.t. $p_{x}$). Due to affine invariance
of the algorithm, we may assume that $x=0$ and $g(x)=I_{d}$ (so
$p_{x}=\ncal(0,I_{d})$). Therefore, 
\begin{align*}
\snorm{z-y}_{y}^{2}-\snorm{z-x}_{x}^{2} & =\snorm{z-y}_{y}^{2}-\snorm z^{2}=\snorm z_{g(y)-I_{d}}^{2}-2\inner{z,y}_{y}+\snorm y_{y}^{2}\,.
\end{align*}
The last term is bounded by $2\norm y^{2}$ due to SC of $g$. Using
a tail bound for Gaussians, we have $\P_{p_{x}}\bpar{|\inner{z,y}_{y}|\geq\frac{r}{\sqrt{d}}\,\snorm{g(y)y}_{2}\cdot2\log\frac{1}{\veps}}\leq\veps$
and call this event $C_{1}$. In addition, SC of $g$ leads to $g(y)\preceq2I_{d}$,
so $\snorm{g(y)y}\leq2\snorm y$.

To bound $\snorm z_{g(y)-I_{d}}^{2}$, we note that $\snorm y=\snorm{y-x}_{x}\leq1/\sqrt{2}$
and so
\begin{align*}
\snorm{g(y)-I_{d}}_{F} & \leq(1+2\snorm y)^{2}\snorm y\leq2s\frac{r}{\sqrt{d}}\,,\quad\text{(Lemma \ref{lem:strongSC-closeness})}\\
\E[\snorm z_{g(y)-I_{d}}^{2}] & =\frac{r^{2}}{d}\tr(g(y)-I_{d})\leq\frac{r^{2}}{d}\sqrt{d}\,\snorm{g(y)-I_{d}}_{F}\leq\frac{r^{2}}{d}\cdot2rs\,.
\end{align*}
By the Hanson-Wright inequality\footnote{\begin{lem*}
[Hanson-Wright; Adapted to Gaussian] Let $h\sim\ncal(0,\sigma^{2}I_{d})$
and $M\in\Rdd$. Then there exists universal constants $c,K>0$ such
that for $t\geq0$
\[
\P\bpar{|\norm h_{A}^{2}-\E[\norm h_{A}^{2}]|>t}\leq2\exp\Bpar{-c\min\Bpar{\frac{t^{2}}{K^{4}\sigma^{4}\norm M_{F}^{2}},\frac{t}{K^{2}\sigma^{2}\norm M_{2}}}}\,.
\]
\end{lem*}
}, for universal constants $K_{1},K_{2}>0$ and $t\geq0$ it holds
that
\[
\P_{z\sim\ncal(0,I_{d})}\bpar{|\norm z_{g(y)-I_{d}}^{2}-\E[\norm z_{g(y)-I_{d}}^{2}]|\geq t}\leq2\exp\Bpar{-K_{1}\Bpar{\frac{t^{2}}{K_{2}^{4}\frac{r^{4}}{d^{2}}\snorm{g(y)-I_{d}}_{F}^{2}}\wedge\frac{t}{K_{2}^{2}\frac{r^{2}}{d}\snorm{g(y)}_{2}}}}\,.
\]
By taking $r\leq r_{8}(\veps):=\frac{\sqrt{K_{1}}}{2K_{2}^{2}}$ and
$s\leq s_{2}(\veps):=\veps(1+\sqrt{\log\frac{2}{\veps}})^{-1}$, it
follows that $\norm z_{g(y)-I_{d}}^{2}\leq\frac{2\veps r^{2}}{d}$
with probability at least $1-\veps$. Denote the complement of this
event by $C_{2}$.

Conditioned on $z\in C_{1}^{c}\cap C_{2}^{c}$, we conclude that
\begin{align*}
|\snorm{z-y}_{y}^{2}-\snorm{z-x}_{x}^{2}| & \leq\snorm z_{g(y)-I_{d}}^{2}+2|\inner{z,y}_{y}|+2\snorm y^{2}\leq\frac{2r^{2}\veps}{d}+\frac{8r\norm y}{\sqrt{d}}\,\log\frac{1}{\veps}+2\norm y^{2}\leq\frac{2r^{2}}{d}\cdot3\veps\,,
\end{align*}
where the last inequality follows from $\norm y\leq\frac{sr}{\sqrt{d}}$
when $s\leq s_{3}(\veps):=\veps\,(4\log\frac{1}{\veps})^{-1}$. Hence,
$|\textsf{L}|\leq3\veps$ on $C_{1}^{c}\cap C_{2}^{c}$. Putting this
into \eqref{eq:logV-lower} and \eqref{eq:logW-lower},
\[
\log\msf V\geq\exp(-4\veps)\qquad\text{and}\qquad\log\msf W\geq\exp(-7\veps)\,.
\]

We can also show $\log\textsf{V}\leq5\veps$. Conditioned on $z\in C_{1}^{c}\cap C_{2}^{c}$,
\[
-\log\msf V=-\textsf{L}+\vphi(x)-\vphi(y)\geq-3\veps+\vphi(x)-\vphi(y)\geq-5\veps\,,
\]
since $\vphi(x)-\vphi(y)$ can be lowered bounded by $-2\veps$ as
in \eqref{eq:bound-vphi}. Hence, $\log\textsf{V}\leq5\veps$.

For $F:=G_{x}\cap G_{y}\cap B_{x,1}^{c}$ and $C:=(C_{1}\cup C_{2})^{c}$,
since $e^{-4\veps}\leq\msf V\leq e^{5\veps}$, $e^{-7\veps}\leq\msf W$,
and $\msf U\geq e^{-4\veps}$,
\begin{align*}
 & \int_{F}|A(x,z)\,p_{x}(z)-A(y,z)\,p_{y}(z)|\,\D z\leq\int_{C^{c}}(\cdot)\,\D z+\int_{F\cap C}(\cdot)\,\D z\\
\leq & \underbrace{\P_{p_{x}}(C^{c})}_{\leq2\veps}+2\underbrace{\dtv(p_{x},p_{y})}_{\leq\veps}+\int_{F\cap C}(\cdot)\,\D z\leq4\veps+\int_{F\cap C}|1\wedge\msf U-\msf V\wedge\msf W|\,p_{x}(\D z)\leq4\veps+(e^{5\veps}-e^{-4\veps})\\
\leq & 18\veps\,.
\end{align*}
Using this, we can bound $\ccal$ by
\begin{align*}
\ccal & \leq\frac{1}{4}+\veps+\half\int_{F}|A(x,z)\,p_{x}(z)-A(y,z)\,p_{y}(z)|\,\D z\leq\frac{1}{4}+10\veps\,.
\end{align*}
Therefore, $\textsf{II}\leq3\veps+\acal+\bcal+\ccal\leq3\veps+4\veps+\frac{1}{4}+10\veps\leq\frac{1}{4}+17\veps.$
Along with $\textsf{I}\leq\half+5\veps$, we can conclude that if
$r\leq\min_{i}r_{i}(\veps)$ and $s\leq\min_{i}s_{i}(\veps)$, then
$\dtv(P_{x},P_{y})\leq\frac{3}{4}+23\veps$.
\end{proof}

\subsubsection{Isoperimetric inequality \label{proof:isoperimetry}}

We now prove an isoperimetric inequality arising from the a SC barrier.
Recall the \emph{cross-ratio} \emph{distance} $d_{K}$ defined on
a convex body $K$: for $x,y\in\intk$, suppose that the chord passing
through $x,y$ has endpoints $p$ and $q$ in the boundary $\de K$
(so the order of points is $p,x,y,q$), then the cross-ratio distance
between $x$ and $y$ is defined by
\[
d_{K}(x,y)\defeq\frac{\snorm{x-y}_{2}\snorm{p-q}_{2}}{\snorm{p-x}_{2}\snorm{y-q}_{2}}\,.
\]
The first type of isoperimetric inequalities says $\psi_{\pi}\gtrsim1/\sqrt{\onu}$.
\begin{proof}
[Proof of Lemma~\ref{lem:symmetry-iso}] For a ball $B_{r}(0)$
of radius $r>0$ centered at the origin, we define a convex body $K_{r}:=K\cap B_{r}(0)$
and use $\pi_{r}$ to denote the truncated distribution of $\pi$
over $K_{r}$. Let $\{S_{1},S_{2},S_{3}\}$ be a partition of $K$
and define $S_{i}^{r}:=S_{i}\cap K_{r}$ for $i\in[3]$. By \citet[Theorem 2.5]{lovasz2007geometry},
we have 
\[
\pi_{r}(S_{3}^{r})\geq d_{K_{r}}(S_{1}^{r},S_{2}^{r})\,\pi_{r}(S_{1}^{r})\,\pi_{r}(S_{2}^{r})\,,
\]
where $d_{K_{r}}(S_{1}^{r},S_{2}^{r})=\inf_{x\in S_{1}^{r},y\in S_{2}^{r}}d_{K_{r}}(x,y)$.
Due to $d_{K_{r}}(x,y)\geq\norm{x-y}_{x}/\sqrt{\onu}$ for any $x,y\in K_{r}$
(see \citet[Lemma 2.3]{laddha2020strong}), 
\[
\pi_{r}(S_{3}^{r})\geq\inf_{x\in S_{1}^{r},\,y\in S_{2}^{r}}\frac{\snorm{x-y}_{x}}{\sqrt{\onu}}\,\pi_{r}(S_{1}^{r})\,\pi_{r}(S_{2}^{r})\geq\frac{1}{\sqrt{\onu}}\inf_{x\in S_{1},\,y\in S_{2}}\snorm{x-y}_{x}\,\pi_{r}(S_{1}^{r})\,\pi_{r}(S_{2}^{r})\,.
\]
As $r\to\infty$, the bounded convergence theorem implies $\pi_{r}(S_{i}^{r})\to\pi(S_{i})$
for $i\in[3]$, completing the proof.
\end{proof}
We provide the deferred proof for another isoperimetric inequality,
$\psi_{\pi}\gtrsim\sqrt{\alpha}$, originating from $\alpha$-relatively
strong-convexity of the potential with respect to $\hess\phi$.
\begin{proof}
[Proof of Lemma~\ref{lem:sc-iso}] The proof essentially follows
\citet{gopi2023algorithmic}. Their first proof ingredient is a modified
localization lemma \citet[Lemma 8]{gopi2023algorithmic}; let $f_{1},f_{2},f_{3},f_{4}$
be non-negative functions on $\Rd$ such that $f_{1}$ and $f_{2}$
are upper semicontinuous, and $f_{3}$ and $f_{4}$ are lower semicontinuous,
and $\phi:\Rd\to\R$ be convex. Then the following are equivalent:
\begin{itemize}
\item For any density $\pi:\Rd\to\R$ which is $1$-relatively strongly
logconcave in $\phi$,
\[
\int f_{1}\,\D\pi\cdot\int f_{2}\,\D\pi\leq\int f_{3}\,\D\pi\cdot\int f_{4}\,\D\pi\,.
\]
\item Let $\int_{E}h:=\int_{0}^{1}h((1-t)\,a+tb)e^{-\gamma t}\,\D t$. Then
$\int_{E}f_{1}e^{-\phi}\cdot\int_{E}f_{2}e^{-\phi}\leq\int_{E}f_{3}e^{-\phi}\cdot\int_{E}f_{4}e^{-\phi}$
for any $a,b\in\Rd$ and $\gamma\in\R$.
\end{itemize}
First of all, this can be generalized to an extended convex function
$f$ and $\phi$, whose values outside of $\intk$ are set to $\infty$.
Since the density $\pi$ and a needle $\exp\Par{\gamma t-\phi((1-t)a+tb)}$
for $\gamma\in\R$ and $a,b\in\Rd$ (induced by the extended $f$
and $\phi$) vanish outside of $\intk$, integrands above become zero
on $\intk^{c}$, and thus the integrals above remain the same.

As in \citet[Lemma 9]{gopi2023algorithmic}, the proof boils down
to the case of $\alpha=1$, and it suffices to show that there exists
a constant $C>0$ such that
\[
C\cdot d_{\phi}(S_{1},S_{2})\int_{S_{1}}e^{-f}\cdot\int_{S_{2}}e^{-f}\leq\int e^{-f}\int_{S_{3}}e^{-f}\,.
\]
We can replace $S_{i}\gets$ its closure $\bar{S_{i}}$ for $i\in[2]$,
which only increases the LHS. Also, we can replace $S_{3}\gets$ an
open set $\intk\backslash\bar{S_{1}}\backslash\bar{S_{2}}$, which
does not change the RHS since the boundary of a convex set is a null
set \citet[Theorem 1]{lang1986note}. By taking $f_{i}=\mathbf{1}_{S_{i}}$
for $i\in[3]$ and $f_{4}=(C\,d_{\phi}(S_{1},S_{2}))^{-1}$, we only
need to show that for some $0\leq c<d\leq1$,
\begin{align*}
 & C\cdot d_{\phi}(S_{1},S_{2})\int_{c}^{d}e^{\gamma t-\phi((1-t)\,a+tb)}\mathbf{1}_{S_{1}}((1-t)\,a+tb)\,\D t\cdot\int_{c}^{d}e^{\gamma t-\phi((1-t)\,a+tb)}\mathbf{1}_{S_{2}}((1-t)\,a+tb)\,\D t\\
\leq & \int_{c}^{d}e^{\gamma t-\phi((1-t)\,a+tb)}\,\D t\cdot\int_{c}^{d}e^{\gamma t-\phi((1-t)\,a+tb)}\mathbf{1}_{S_{3}}((1-t)\,a+tb)\,\D t\,,
\end{align*}
where $\phi((1-t)\,a+b)<\infty$ for $t\in(c,d)$. The rest of the
proof is similar to \citet[Lemma 9]{gopi2023algorithmic}.
\end{proof}

\subsection{Sampling IPM ($\S$\ref{sec:IPM-framework})}

\subsubsection{Well-definedness of sampling IPM \label{proof:IPM-welldefined}}
\begin{prop}
\label{prop:density-bounded} Let $p:\Rd\to\R$ be a log-concave density
with finite second moment. Then $p$ is bounded on $\Rd$.
\end{prop}

\begin{proof}
Let $X\sim p$ and denote the mean and covariance of the distribution
$p$ by $\mu:=\E[X]$ and $\Sigma:=\E[(X-\mu)(X-\mu)^{\T}]$. Then
the pushforward $T_{\#}p$ of $p$ via the map $T:x\mapsto y:=\Sigma^{-1/2}(x-\mu)$
is an isotropic log-concave, and satisfy $(T_{\#}p)(y)=\frac{p(x)}{|\det T|}$.
Since $T_{\#}p$ is bounded on $\Rd$ \citet[Theorem 5.14 (e)]{lovasz2007geometry},
$p$ is bounded as well.
\end{proof}
Next, we show that every measure appearing within the sampling IPM
is integrable.
\begin{proof}
[Proof of Proposition~\ref{prop:annealing-welldefined}] Recall
that we may assume $\phi\geq0$. Hence, all $\mu_{i}$'s in Phase
3 and 4 are well-defined
\[
\int_{K}\exp\Bpar{-\bpar{f(x)+\frac{\phi(x)}{\sigma_{i}^{2}}}}\,\D x\leq\int_{K}\exp(-f(x))\,\D x<\infty\,.
\]
In particular, $\exp\bpar{-(f+\frac{\phi}{\nu/d})}$ is integrable
with finite second moment. By Proposition~\ref{prop:density-bounded},
$f(x)+\frac{\phi(x)}{\nu/d}$ achieves a global minimum $m$ in $K$.
As $\sigma_{i}^{2}\leq\sigma_{i_{0}}^{2}=\nu/d$ in Phase 2, we have
\begin{align*}
 & \int_{K}\exp\Bpar{-\frac{\sigma_{i_{0}}^{2}f+\phi}{\sigma_{i_{0}}^{2}}}=\int_{K}\exp\Bpar{-\frac{\sigma_{i_{0}}^{2}f+\phi-\min(\sigma_{i_{0}}^{2}f+\phi)}{\sigma_{i_{0}}^{2}}-\frac{\min(\sigma_{i_{0}}^{2}f+\phi)}{\sigma_{i_{0}}^{2}}}\\
 & \geq\int_{K}\exp\Bpar{-\frac{\bar{f}+\phi-\sigma_{i_{0}}^{2}m}{\sigma_{i}^{2}}-m}=\exp\Bpar{m\bpar{\frac{\sigma_{i_{0}}^{2}}{\sigma_{i}^{2}}-1}}\int_{K}\exp\Bpar{-\frac{\bar{f}+\phi}{\sigma_{i}^{2}}}\,,
\end{align*}
where the inequality holds due to $\min(\sigma_{i_{0}}^{2}f+\phi)=\sigma_{i_{0}}^{2}m$
and $\bar{f}=\sigma_{i_{0}}^{2}f$. Therefore, $\mu_{i}$'s in Phase
2 are also well-defined.
\end{proof}

\subsubsection{Closeness of distributions in sampling IPM \label{proof:IPM-closeness}}

We begin with closeness between $\ncal\bpar{x^{*},\frac{\sigma_{0}^{2}}{1+\nu\beta d^{-1}}g(x^{*})^{-1}}\cdot\mathbf{1}_{\dcal_{g}^{3\sigma_{0}\sqrt{d}}(x^{*})}$
and $\exp\bpar{-\frac{\bar{f}+\phi}{\sigma_{0}^{2}}}$ in Phase 1.
\begin{proof}
[Proof of Lemma~\ref{lem:phase1}] Let $\gamma=9$, $r=(\gamma\sigma_{0}^{2}d)^{1/2}<0.01$,
$\psi:=\bar{f}+\phi$, and $S=\{x\in K:\psi(x)\leq\psi(x^{*})+r^{2}/4\}$.
For $\widetilde{\mu}_{0}=\exp(-\psi/\sigma_{0}^{2})\cdot\mathbf{1}_{K}\propto\mu_{0}$
and $x\in S$, we have $\mu_{0}(x)\geq e^{-\gamma d}\mu_{0}(x^{*}).$
Due to $\mu_{0}(S^{c})\leq\exp(-\gamma d/3)$ (Lemma~\ref{lem:mostMass-logconcave}),
it follows that $1=\mu_{0}(S)+\mu_{0}(S^{c})\leq\mu_{0}(S)+\exp(-\gamma d/3)$
and 
\begin{equation}
1\leq\bpar{1+2\exp(-\gamma d/3)}\,\mu_{0}(S)=\bpar{1+2\exp(-\gamma d/3)}\,\widetilde{\mu}_{0}(S)/\widetilde{\mu}_{0}(\Rd)\,.\label{eq:intp-intSp}
\end{equation}

We show $S\subset D=\dcal_{g}^{3\sigma_{0}\sqrt{d}}(x^{*})$. For
$x\in S$, use Taylor's expansion of $\psi$ at $x^{*}$: for some
$\bar{x}\in[x^{*},x]$
\begin{align}
\psi(x)-\psi(x^{*}) & =\half(x-x^{*})^{\T}\hess\psi(\bar{x})(x-x^{*})\geq\half(x-x^{*})^{\T}\hess\phi(\bar{x})(x-x^{*})\,.\label{eq:psi-taylor}
\end{align}
As $\psi(x)-\psi(x^{*})\leq r^{2}/4$ on $x\in S$, we have $\snorm{\bar{x}-x^{*}}_{\bar{x}}^{2}\leq\snorm{x-x^{*}}_{\bar{x}}^{2}\leq2(\psi(x)-\psi(x^{*}))\leq r^{2}/2$.
Thus, by self-concordance of $\phi$ 
\begin{equation}
\exp(-3r)\,\snorm{x-x^{*}}_{x^{*}}^{2}\leq\snorm{x-x^{*}}_{\bar{x}}^{2}\leq\exp(3r)\,\snorm{x-x^{*}}_{x^{*}}^{2}\,,\label{eq:closenss-initial}
\end{equation}
and it follows that $\snorm{x-x^{*}}_{x^{*}}^{2}\leq r^{2}$, showing
$S\subset D$.

Combining \eqref{eq:psi-taylor}, \eqref{eq:closenss-initial}, and
$(1+\nu\alpha d^{-1})\,\hess\phi\preceq\hess\psi\preceq(1+\nu\beta d^{-1})\,\hess\phi$,
we have 
\begin{equation}
\frac{\exp(-3r)}{2}\Bpar{1+\frac{\nu\alpha}{d}}\,\snorm{x-x^{*}}_{x^{*}}^{2}\underset{(*)}{\leq}\psi(x)-\psi(x^{*})\underset{(\#)}{\leq}\frac{\exp(3r)}{2}\Bpar{1+\frac{\nu\beta}{d}}\,\snorm{x-x^{*}}_{x^{*}}^{2}\,,\label{eq:approx-psigap}
\end{equation}
and thus for a constant $c:=1+\nu\beta d^{-1}$ and function $h(x):=-(2\sigma_{0}^{2})^{-1}\snorm{x-x^{*}}_{x^{*}}^{2}$,
\begin{align*}
 & \norm{\mu/\mu_{0}}=\E_{\mu}\bbrack{\deriv{\mu}{\mu_{0}}}=\frac{\int_{D}\exp\Bpar{-\frac{c}{\sigma_{0}^{2}}\snorm{x-x^{*}}_{x^{*}}^{2}+\frac{\psi}{\sigma_{0}^{2}}}\cdot\widetilde{\mu}_{0}(\Rd)}{\Bbrack{\int_{D}\exp\Bpar{-\frac{c}{2\sigma_{0}^{2}}\,\snorm{x-x^{*}}_{x^{*}}^{2}}}^{2}}\\
 & \underset{\text{(\ref{eq:intp-intSp})}}{\leq}\frac{1}{\Bbrack{\int_{D}\exp(c\cdot h)}^{2}}\int_{D}\exp\Bpar{-\frac{c}{\sigma_{0}^{2}}\snorm{x-x^{*}}_{x^{*}}^{2}+\underbrace{\frac{\psi}{\sigma_{0}^{2}}}_{\text{Use }(\#)\text{ in (\ref{eq:approx-psigap})}}}\bpar{1+2\exp(-\gamma n/3)}\underbrace{\widetilde{\mu}_{0}(S)}_{\text{Use }(*)}\\
 & \lesssim\frac{\int_{D}\exp\Bpar{-\frac{1}{2\sigma_{0}^{2}}\bpar{2c-e^{3r}(1+\nu\beta d^{-1})}\,\snorm{x-x^{*}}_{x^{*}}^{2}}\int_{D}\exp\bpar{-\frac{1}{2\sigma_{0}^{2}}e^{-3r}(1+\nu\alpha d^{-1})\,\snorm{x-x^{*}}_{x^{*}}^{2}}}{\Bbrack{\int_{D}\exp(c\cdot h)}^{2}}\\
 & =\underbrace{\frac{\int_{D}\exp\Bpar{\bpar{2c-c\,e^{3r}}\,h(x)}\cdot\int_{D}\exp\bpar{c\,e^{3r}h(x)}}{\Bbrack{\int_{D}\exp(c\cdot h)}^{2}}}_{=:\text{\textsf{A}}}\,\underbrace{\frac{\int_{D}\exp\bpar{e^{-3r}(1+\nu\alpha d^{-1})\,h(x)}}{\int_{D}\exp\bpar{c\,e^{3r}h(x)}}}_{=:\text{\textsf{B}}}\,.
\end{align*}
As for $\textsf{A}$, Lemma~\ref{lem:adam-logconcave} leads to 
\begin{align*}
\textsf{A} & \leq\Bpar{\frac{c^{2}}{(2c-c\,e^{3r})\,ce^{3r}}}^{d}=\Bpar{\frac{1}{(2-e^{3r})e^{3r}}}^{d}=(1+\mc O(r^{2}))^{d}=\mc O(1)\,.
\end{align*}
As for $\textsf{B}$, let $c_{1}=e^{-3r}\,(1+\nu\alpha d^{-1})$ and
$c_{2}=e^{3r}\,(1+\nu\beta d^{-1})$. With the change of variable
$y=\sigma_{0}^{-1}\sqrt{c_{i}}g(x^{*})^{1/2}(x-x^{*})$ for $i\in[2]$,
it follows that for $r_{i}:=r\sigma_{0}^{-1}\sqrt{c_{i}}(\geq3\sqrt{d})$
\begin{align*}
\textsf{B} & =\Bpar{\frac{c_{2}}{c_{1}}}^{d/2}\frac{\int_{B_{r_{1}}}\exp\bpar{-\half\snorm y^{2}}\,\D y}{\int_{B_{r_{2}}}\exp\bpar{-\half\snorm y^{2}}\,\D y}\leq\Bpar{\frac{c_{2}}{c_{1}}}^{d/2}\lesssim\Bpar{\frac{\nu\beta+d}{\nu\alpha+d}}^{d}\,e^{3rd}\lesssim\Bpar{\frac{\nu\beta+d}{\nu\alpha+d}}^{d}\,.\qedhere
\end{align*}
\end{proof}
Now we show closeness of two consecutive distributions in Phase 2,
i.e., $\sigma_{i+1}^{2}=\sigma_{i}^{2}\bpar{1+\frac{1}{\sqrt{d}}}$.
\begin{proof}
[Proof of Lemma~\ref{lem:phase2}] Observe that for $\psi=\bar{f}+\phi=\frac{\nu}{d}f+\phi$
on $K$ and $F(\sigma^{2})=\int_{K}\exp(-\psi/\sigma^{2})$, 
\begin{align*}
\snorm{\mu_{i}/\mu_{i+1}} & =\E_{\mu_{i}}\bbrack{\deriv{\mu_{i}}{\mu_{i+1}}}=\frac{\int_{K}\exp\bpar{-2\frac{\psi}{\sigma_{i}^{2}}+\frac{\psi}{\sigma_{i+1}^{2}}}\cdot\int_{K}\exp\bpar{-\frac{\psi}{\sigma_{i+1}^{2}}}}{\Par{\int_{K}\exp\bpar{-\frac{\psi}{\sigma_{i}^{2}}}}^{2}}=\frac{F\bpar{\bpar{\frac{2}{\sigma_{i}^{2}}-\frac{1}{\sigma_{i+1}^{2}}}^{-1}}\,F(\sigma_{i+1}^{2})}{F(\sigma_{i}^{2})^{2}}\,.
\end{align*}
By Lemma~\ref{lem:adam-logconcave}, the function $a^{d}F\bpar{\frac{\sigma^{2}}{a}}$
is log-concave in $a$. Using the definition with endpoints $\frac{2}{\sigma_{i}^{2}}-\frac{1}{\sigma_{i+1}^{2}}$
and $\frac{1}{\sigma_{i+1}^{2}}$, and the middle point $\frac{1}{\sigma_{i}^{2}}$,
we obtain
\[
\frac{F\bpar{\bpar{\frac{2}{\sigma_{i}^{2}}-\frac{1}{\sigma_{i+1}^{2}}}^{-1}}\,F(\sigma_{i+1}^{2})}{F(\sigma_{i}^{2})^{2}}\le\Biggl(\frac{\bpar{\frac{1}{\sigma_{i}^{2}}}^{2}}{\bpar{\frac{2}{\sigma_{i}^{2}}-\frac{1}{\sigma_{i+1}^{2}}}\,\frac{1}{\sigma_{i+1}^{2}}}\Biggr)^{d}=\Biggl(\frac{\bpar{1+\frac{1}{\sqrt{d}}}^{2}}{1+\frac{2}{\sqrt{d}}}\Biggr)^{d}\leq\Bpar{1+\frac{1}{d}}^{d}\leq e\,.\qedhere
\]
\end{proof}
We now establish closeness in Phase 3, during which we use the update
of $\sigma_{i+1}^{2}=\sigma_{i}^{2}\bpar{1+\frac{\sigma_{i}}{\sqrt{\nu}}}$.
\begin{proof}
[Proof of Lemma~\ref{lem:phase34}] The update is $\sigma_{i+1}^{2}=\sigma_{i}^{2}\Par{1+r}$
for $r=\frac{\sigma_{i}}{\sqrt{\nu}}$. For $s:=\frac{r}{1+r}$, $\sigma:=\sigma_{i}$,
and $F(\sigma^{2})=\int\exp(-f-\phi/\sigma^{2})$, we have
\begin{align*}
\snorm{\mu_{i}/\mu_{i+1}} & =\frac{F\bpar{\bpar{\frac{2}{\sigma_{i}^{2}}-\frac{1}{\sigma_{i+1}^{2}}}^{-1}}\,F(\sigma_{i+1}^{2})}{F(\sigma_{i}^{2})^{2}}=\frac{F\bpar{\frac{\sigma^{2}}{1+s}}\,F\bpar{\frac{\sigma^{2}}{1-s}}}{F(\sigma^{2})^{2}}\,.
\end{align*}
Let $g(t):=\log F\bpar{\frac{\sigma^{2}}{t}}$ for $t>0$. Then,
\begin{align}
\log\snorm{\mu_{i}/\mu_{i+1}} & =g(1+s)+g(1-s)-2g(1)=\int_{0}^{s}\bpar{g'(1+t)-g'(1-t)}\,\D t=\int_{0}^{s}\int_{1-t}^{1+t}g''(q)\,\D q\,\D t\label{eq:L2-bound-phase3}
\end{align}
and for a probability measure $\nu_{q}\propto\exp\bpar{-f-\frac{q\phi}{\sigma^{2}}}$,
\begin{align*}
g''(q) & =\frac{\D^{2}}{\D q^{2}}\Bbrack{\log\int_{K}\exp\Bpar{-f-\frac{q\phi}{\sigma^{2}}}}=-\frac{1}{\sigma^{2}}\,\frac{\D}{\D q}\Bigg[\frac{\int_{K}\phi\cdot\exp\Bpar{-f-\frac{q\phi}{\sigma^{2}}}}{\int_{K}\exp\Bpar{-f-\frac{q\phi}{\sigma^{2}}}}\Biggr]\\
 & =-\frac{1}{\sigma^{2}}\,\Bigg(-\frac{1}{\sigma^{2}}\,\frac{\int_{K}\phi^{2}\cdot\exp\Bpar{-f-\frac{q\phi}{\sigma^{2}}}}{\int_{K}\exp\Bpar{-f-\frac{q\phi}{\sigma^{2}}}}+\frac{1}{\sigma^{2}}\,\frac{\Bbrack{\int_{K}\phi\cdot\exp\Bpar{-f-\frac{q\phi}{\sigma^{2}}}}^{2}}{\Bbrack{\int_{K}\exp\Bpar{-f-\frac{q\phi}{\sigma^{2}}}}^{2}}\Biggr)\\
 & =\frac{1}{\sigma^{4}}\,\Bpar{\E_{\nu_{q}}[\phi^{2}]-(\E_{\nu_{q}}\phi)^{2}}=\frac{1}{\sigma^{4}}\,\var_{\nu_{q}}\phi\,.
\end{align*}
By the Brascamp-Lieb inequality with $V(\cdot):=f(\cdot)+\frac{q\phi(\cdot)}{\sigma^{2}}$,
\begin{align*}
\var_{\nu_{q}}\phi & \leq\E_{\nu_{q}}\bbrack{(\grad\phi)^{\T}\bpar{\hess V}^{-1}\grad\phi}\leq\frac{\sigma^{2}}{q}\,\E_{\nu_{q}}\snorm{\nabla\phi}_{(\hess\phi)^{-1}}^{2}\leq\frac{\sigma^{2}\nu}{q}\,,
\end{align*}
and thus $g''(q)\leq\frac{\nu}{q\sigma^{2}}.$ Putting this back to
\eqref{eq:L2-bound-phase3}, we acquire
\begin{align}
\log\norm{\mu_{i}/\mu_{i+1}} & \leq\frac{\nu}{\sigma^{2}}\int_{0}^{s}\int_{1-t}^{1+t}\frac{1}{q}\,\D q\,\D t=\frac{\nu}{\sigma^{2}}\int_{0}^{s}\bpar{\log(1+t)-\log(1-t)}\,\D t\nonumber \\
 & =\frac{\nu}{\sigma^{2}}\bpar{(1+s)\,\log(1+s)+(1-s)\,\log(1-s)}\lesssim\frac{\nu s^{2}}{\sigma^{2}}\,.\label{eq:bound-ph3}
\end{align}
It follows from $s=\frac{r}{1+r}$ and $r=\frac{\sigma}{\sqrt{\nu}}$
that $\mu_{i}$ is an $\mc O(1)$-warm start for $\mu_{i+1}$.

For Phase 4, observe that for $\mu\propto\exp(-f-\phi/\sigma^{2})$
with $\sigma^{2}=\nu$,
\begin{align*}
\snorm{\mu/\pi} & =\frac{\int_{K}\exp\bpar{-f-\frac{\phi}{\sigma^{2}/2}}\cdot\int_{K}\exp(-f)}{\Bbrack{\int_{K}\exp\Bpar{-f-\frac{\phi}{\sigma^{2}}}}^{2}}\underset{\text{(i)}}{=}\lim_{r\to1}\frac{F\bpar{\frac{\sigma^{2}}{1+r}}\cdot F\bpar{\frac{\sigma^{2}}{1-r}}}{F(\sigma^{2})}\\
 & \underset{\text{(ii)}}{\leq}\lim_{r\to1}\exp\Bpar{\mc O(1)\frac{\nu}{\sigma^{2}}\,\bpar{(1+r)\,\log(1+r)+(1-r)\,\log(1-r)}}=\exp\Bpar{\mc O(1)\frac{\nu}{\sigma^{2}}}=\exp(\mc O(1))\,.
\end{align*}
where (i) holds due to the monotone convergence theorem, and (ii)
follows from \eqref{eq:bound-ph3}. Therefore, $\mu$ serves as an
$\mc O(1)$-warm start for $\pi$.
\end{proof}
\begin{rem}
[Coupling argument] \label{rem:divine-intervention} The total number
of measures involved in Algorithm~\ref{alg:IPM-sampling} is $m:=\mc O(\sqrt{d})$.
Let $(X_{1},\dots,X_{m})$ be a sequence of samples provided by Algorithm~\ref{alg:IPM-sampling},
and $(\bar{X}_{1},\dots,\bar{X}_{m})$ be a sequence of samples where
each sample is drawn from the \emph{actual} target distributions $\{\mu_{\sigma^{2}}\}$.
Conditioned on events $X_{i}=\bar{X}_{i}$, Algorithm~\ref{alg:IPM-sampling}
ensures that there is a coupling such that $\P(X_{i+1}=\bar{X}_{i+1}\mid X_{i}=\bar{X}_{i})\geq1-\frac{\veps}{\sqrt{d}}$
due to $\veps/\sqrt{d}$ TV-distance guarantee. Combining these couplings,
\[
\P\Par{X_{i}=\bar{X_{i}}\ \forall i\in[m]}=\P(X_{1}=\bar{X}_{1})\cdot\prod_{i=2}^{m}\P(X_{i}=\bar{X}_{i}\mid X_{i-1}=\bar{X}_{i-1})\geq1-\veps\,.
\]
Thus, it leads to a coupling between $X_{m}$ and $\bar{X}_{m}$ such
that $\P(X_{m}=\bar{X}_{m})\geq1-\veps$, so $\law(X_{m})$ is within
$\veps$-TV distance to $\pi=\law(\bar{X}_{m})$.
\end{rem}

\subsection{Self-concordance theory ($\S$\ref{sec:sc-theory-rules})}

\subsubsection{Basic properties: strong self-concordance \label{proof:ssc-basic}}

We show that $2(g_{1}+g_{2})$ is SSC if $g_{1}$ and $g_{2}$ are
SSC.
\begin{proof}
[Proof of Lemma~\ref{lem:ssc-sum}] For fixed $x\in K_{1}\cap K_{2}$
and $h\in\Rd$, let $\Dd g_{i}:=\Dd g_{i}(x)[h]$ for $i=1,2$. Note
that 
\begin{align*}
 & \snorm{(g_{1}+g_{2})^{-\half}\Dd(g_{1}+g_{2})\,(g_{1}+g_{2})^{-\half}}_{F}\\
 & \leq\sum_{i=1}^{2}\snorm{(g_{1}+g_{2})^{-\half}\Dd g_{i}\,(g_{1}+g_{2})^{-\half}}_{F}=\sum_{i=1}^{2}\sqrt{\tr\bpar{(g_{1}+g_{2})^{-1}\Dd g_{i}\,(g_{1}+g_{2})^{-1}\Dd g_{i}}}\\
 & =\Bbrack{\tr\Bpar{\bpar{\underbrace{I+g_{1}^{-\half}g_{2}g_{1}^{-\half}}_{=:E_{1}}}^{-1}\underbrace{g_{1}^{-\half}\Dd g_{1}\,g_{1}^{-\half}}_{=:T_{1}}\bpar{I+g_{1}^{-\half}g_{2}g_{1}^{-\half}}^{-1}g_{1}^{-\half}\Dd g_{1}\,g_{1}^{-\half}}}^{1/2}\\
 & \qquad+\Bbrack{\tr\Bpar{\bpar{\underbrace{I+g_{2}^{-\half}g_{1}g_{2}^{-\half}}_{=:E_{2}}}^{-1}\underbrace{g_{2}^{-\half}\Dd g_{2}\,g_{2}^{-\half}}_{=:T_{2}}\bpar{I+g_{2}^{-\half}g_{1}g_{2}^{-\half}}^{-1}g_{2}^{-\half}\Dd g_{2}\,g_{2}^{-\half}\bigg)}}^{1/2}\\
 & =\sum_{i=1}^{2}\sqrt{\tr(E_{i}^{-1}T_{i}E_{i}^{-1}T_{i})}\leq\sum_{i=1}^{2}\sqrt{\tr(T_{i}E_{i}^{-2}T_{i})}\,,
\end{align*}
where we used the Cauchy-Schwarz inequality $\tr(A^{2})\leq\tr(A^{\T}A)$
in the last line. It follows from $I\preceq E_{i}$ that $I\preceq E_{i}^{2}$
and $I\succeq E_{i}^{-2}\succ0$. Therefore, 
\begin{align*}
\sum_{i=1}^{2}\sqrt{\tr(T_{i}E_{i}^{-2}T_{i})} & \leq\sum_{i=1}^{2}\snorm{T_{i}}_{F}\leq2\sum_{i=1}^{2}\snorm h_{g_{i}(x)}^{2}\leq2\sqrt{2}\norm h_{(g_{1}+g_{2})(x)}\,.
\end{align*}
Putting these together completes the proof.
\end{proof}

\subsubsection{Basic properties: lower trace self-concordance \label{proof:ltsc-basic}}

We now show that if $g$ is HSC, then $dg$ is SLTSC.
\begin{proof}
[Proof of Lemma~\ref{lem:hsc-to-sltsc}] We first consider when
$\bar{g}$ is positive definite on $K$. By HSC of $\bar{g}$, it
holds that $-\snorm h_{\bar{g}}^{2}\,\bar{g}\lesssim\Dd^{2}\bar{g}[h,h]$,
and thus
\[
-\frac{1}{d}\,\snorm h_{g}^{2}\,(g'+g)^{-\half}g\,(g'+g)^{-\half}\lesssim(g'+g)^{-\half}\Dd^{2}g[h,h]\,(g'+g)^{-\half}\,.
\]
Hence,
\begin{align*}
\tr\bpar{(g'+g)^{-1}\Dd^{2}g[h,h]} & \gtrsim-\frac{1}{d}\,\snorm h_{g}^{2}\,\tr\Bpar{(g'+g)^{-\half}g\,(g'+g)^{-\half}}=-\frac{1}{d}\,\snorm h_{g}^{2}\,\tr\bpar{g^{\half}(g'+g)^{-1}g^{\half}}\\
 & \geq-\frac{1}{d}\,\snorm h_{g}^{2}\,\tr(g^{\half}g^{-1}g^{\half})=-\snorm h_{g}^{2}\,.
\end{align*}

When $g$ is singular, we consider $\bar{g}_{\veps}=\bar{g}+\frac{\veps}{d}I\in\pd$
for $\veps>0$. Then $\bar{g}_{\veps}$ is HSC, so for $g_{\veps}=d\bar{g}_{\veps}$
\[
\tr\bpar{(g'+g_{\veps})^{-1}\Dd^{2}g[h,h]}\gtrsim-\snorm h_{g_{\veps}}^{2}\,.
\]
From $(g'+g_{\veps})^{-1}=\frac{1}{\det(g'+g_{\veps})}\,\text{adj}(g'+g_{\veps})$,
the LHS is continuous in $\veps$, and the RHS is too clearly. Sending
$\veps\to0$ completes the proof.
\end{proof}

\subsubsection{Basic properties: strongly average self-concordance \label{proof:sasc-basic}}

To prove Lemma \ref{lem:hsc-to-sasc}, we first recall a concentration
bound.
\begin{lem}
[\citet{narayanan2016randomized}, Lemma 4] \label{lem:odd-order-concen}Let
$h$ be drawn from $\mathbb{S}^{d-1}$ uniformly at random. For any
odd $k$, $C^{k}$-smooth $F:\Rd\to\R$, and $\veps>0$,
\[
\P_{h}\Bpar{|\Dd^{k}F(x)[h^{\otimes k}]|>k\veps\cdot\sup_{\snorm v\leq1}\Dd^{k}F(x)[v^{\otimes k}]}\leq\exp\Bpar{-\frac{d\veps^{2}}{2}}\,.
\]
\end{lem}

We show that if $g$ is HSC, then $dg$ is SASC, using this lemma
and following \citet{narayanan2016randomized}.
\begin{proof}
[Proof of Lemma~\ref{lem:hsc-to-sasc}] Let $g=d\,\hess\phi$ and
consider $g':\intk\to\psd$ such that $\bar{g}=g+g'$ is PD. For fixed
$w\in\Rd$, apply Taylor's expansion to $\vphi(z):=\norm w_{g(z)}^{2}$
at $z=x$, so there exists $p_{w}\in[x,z]$ such that $w^{\T}g(z)w=w^{\T}g(x)w+\Dd g(x)[z,w,w]+\half\,\Dd^{2}g(p_{w})[z,z,w,w].$
Putting $z=w$ here,
\[
|\snorm z_{g(z)}^{2}-\snorm z_{g(x)}^{2}|\leq|\Dd^{3}g(x)[z^{\otimes3}]|+\half|\Dd^{2}g(p_{z})[z^{\otimes4}]|\,.
\]

Going forward, we can assume that $x=0$ and $\bar{g}(x)=I$ due to
affine invariance, and then $z$ equals $rh/\sqrt{d}$ for $h\sim\ncal(0,I_{d})$
in law. Using a standard tail bound on the standard Gaussian, we have
$\P_{h}(\norm h\geq-\sqrt{d}\cdot2\log\veps)\leq\veps.$ Call this
event $B_{1}$. In addition, Lemma~\ref{lem:odd-order-concen} implies
that 
\[
\P\Bpar{\Big|\Dd^{3}\phi(x)\Bbrack{\frac{h^{\otimes3}}{\norm h^{3}}}\Big|\geq3\frac{\veps}{\sqrt{d}}\cdot\sup_{\norm v\leq1}\Dd^{3}\phi(x)[v^{\otimes3}]}\leq\veps\,,
\]
and call this event $B_{2}$. Conditioned on $B_{2}^{c}$,
\begin{align*}
\Big|\Dd^{3}\phi(x)\Bbrack{\frac{h^{\otimes3}}{\norm h^{3}}}\Big| & \leq\frac{3\veps}{\sqrt{d}}\,\sup_{\norm v\leq1}\Dd^{3}\phi(x)[v^{\otimes3}]\leq\frac{6\veps}{\sqrt{d}}\,\sup_{\norm v\leq1}\snorm v_{g(x)/d}^{3}\leq\frac{6\veps}{d^{2}}\,\sup_{\norm v\leq1}\snorm v_{g(x)}^{3}\underbrace{\leq}_{g(x)\preceq I_{d}}\frac{6\veps}{d^{2}}\,.
\end{align*}
Hence, conditioned on $z\in B_{1}^{c}\cap B_{2}^{c}$
\begin{align*}
|\Dd^{3}g(x)[z^{\otimes3}]| & =\frac{r^{3}}{\sqrt{d}}\,\Dd^{3}\phi(x)[h^{\otimes3}]\leq\frac{r^{3}}{\sqrt{d}}\,\frac{6\veps}{d^{2}}\,\snorm h^{3}\leq\frac{r^{2}}{d}\cdot48r\veps\Bpar{\log\frac{1}{\veps}}^{3}\,.
\end{align*}
By taking $r_{1}(\veps)$ so that $-48r_{1}\veps\,(\log\veps)^{3}\leq\veps$,
we can ensure $|\Dd^{3}g(x)[z^{\otimes3}]|\leq\veps r^{2}/d$ for
any $r\leq r_{1}(\veps)$.

As for $|\Dd^{2}g(p_{z})[z^{\otimes4}]|$, HSC of $\phi$ and Lemma~\ref{lem:scCloseness}
lead to 
\begin{align*}
\half\,|\Dd^{2}g(p_{z})[z^{\otimes4}]| & \leq3d\,\snorm z_{\hess\phi(p_{z})}^{4}\le\frac{3}{d}\snorm z_{\hess\phi(x)}^{4}\,(1+2\,\snorm z_{\hess\phi(x)}^{2})^{2}=\frac{3}{d}\,\snorm z_{g(x)}^{4}\,\bpar{1+\frac{2}{d}\,\snorm z_{g(x)}^{2}}^{2}\\
 & \underset{g\preceq I_{d}}{\leq}\frac{3}{d}\snorm z^{4}\bpar{1+\frac{2}{d}\,\snorm z^{2}}^{2}=\frac{3}{d}\,\frac{r^{4}}{d^{2}}\,\snorm h^{4}\Bpar{1+\frac{2r^{2}}{d^{2}}\norm h^{2}}^{2}\\
 & \leq\frac{r^{2}}{d}\cdot3r^{2}\,\bpar{2\log\frac{1}{\veps}}^{4}\Bpar{1+2r^{2}\bpar{2\log\frac{1}{\veps}}^{4}}^{2}\,.
\end{align*}
By taking $r_{2}(\veps)$ and $r_{3}(\veps)$ so that $\Bpar{1+2r_{2}^{2}\bpar{2\log\frac{1}{\veps}}^{4}}^{2}\leq2$
and $2^{2}\cdot3r_{3}^{2}\bpar{2\log\frac{1}{\veps}}^{4}\leq\veps$
respectively, it holds that on $B_{1}^{c}\cap B_{2}^{c}$
\[
\half\,|\Dd^{2}g(p_{z})[z^{\otimes4}]|\leq\veps\frac{r^{2}}{d}\ \text{for any }r\leq\min r_{i}(\veps).
\]
Putting all these together, it follows that $|\norm z_{g(z)}^{2}-\norm z_{g(x)}^{2}|\leq2\veps r^{2}/d$
with probability at least $1-2\veps$. By replacing $2\veps\gets\veps$,
the claim follows.
\end{proof}

\subsubsection{Collapse and embedding: well-definedness \label{proof:collapse-embedding-welldefined}}

We start with well-definedness of the notions of collapse and embedding
(Definition~\ref{def:sc-along-subspace}).
\begin{proof}
[Proof of Proposition~\ref{prop:collapse-well-defined}] Let $k:=\dim(W)$,
and $U$ and $V$ be matrices in $\R^{d\times k}$, where the columns
of each matrix form an orthonormal basis of $W$. Let us denote by
$g_{1}:=U^{\T}gU$ and $g_{2}:=V^{\T}gV$ matrices represented with
respect to $U$ and $V$, and define the invertible matrix $M=V^{-1}U\in\R^{k\times k}$.
Since $U$ and $V$ are full-column rank, if $g_{1}$ is PD, so is
$g_{2}$.

Suppose $g$ is SSC along $W$. Then, 
\begin{align*}
4\norm h_{g}^{2} & \geq\tr(g_{1}^{-1}\Dd g_{1}[h]\,g_{1}^{-1}\Dd g_{1})=\tr\bpar{(U^{\T}gU)^{-1}\cdot U^{\T}\Dd g[h]\,U\cdot(U^{\T}gU)^{-1}\cdot U^{\T}\Dd g[h]\,U}\\
 & =\tr\Bpar{(M^{\T}V^{\T}gVM)^{-1}\cdot M^{\T}V^{\T}\Dd g[h]\,VM\cdot(M^{\T}V^{\T}gVM)^{-1}\cdot M^{\T}V^{\T}\Dd g[h]\,VM}\\
 & =\tr\Bpar{(V^{\T}gV)^{-1}V^{\T}\Dd g[h]\,V\,(V^{\T}gV)^{-1}V^{\T}\Dd g[h]\,V}=\norm{g_{2}^{-\half}\Dd g_{2}[h]\,g_{2}^{-\half}}_{F}^{2}\,,
\end{align*}
and thus $g_{2}$ also satisfies the definition.
\end{proof}

\subsubsection{Collapse and embedding: affine transformation \label{proof:collap-affine}}

We begin with a barrier version.
\begin{proof}
[Proof of Lemma~\ref{lem:linear-trans}] For the first part, $\psi$
is a $\nu$-self-concordant barrier for $\bar{K}$ by \citet[Theorem 4.2.3]{nesterov2003introductory},
so $\dcal_{\bar{g}}^{1}(x)\subset\bar{K}\cap(2x-\bar{K})$ for $\bar{g}(\cdot):=\hess\psi(\cdot)$
by Lemma~\ref{lem:symmetricLeftpart}. Now let $z\in\bar{K}\cap(2x-\bar{K})$.
Then $Tz\in K$ and $T(2x-z)\in K$, and the latter implies $2y-Tz\in K$.
Thus $Tz\in K\cap(2y-K)$ and $Tz\in\dcal_{g}^{\sqrt{\onu}}(y)$.
Due to 
\begin{align*}
\Dd^{2}\psi(x)[(z-x)^{\otimes2}] & =\Dd^{2}\phi(y)[\bpar{A(z-x)}^{\otimes2}]=\Dd^{2}\phi(y)[(Tz-y)^{\otimes2}]\leq\onu\,,
\end{align*}
it follows that $\psi$ is also $\onu$-symmetric. 

For the second part, observe that $\Dd^{4}\psi(x)[v,v,h,h]=\Dd^{4}\phi(y)[Av,Av,Ah,Ah]\geq0$
for any $v,h\in\Rd$. The third part can be proven similarly.
\end{proof}
Next is a matrix version.
\begin{proof}
[Proof of Lemma~\ref{lem:linear-trans-matrix}] Let $\phi$ be a
$\nu$-self-concordant function counterpart of $g$. Then $\psi(x):=\phi(Tx)$
defined on $\inter(\bar{K})$ is $\nu$-self-concordant by Lemma~\ref{lem:linear-trans}.
For any $h\in\Rd$ and $y:=Tx$, we have 
\[
\Dd\bar{g}(x)[h]=A^{\T}\Dd g(y)[Ah]\,A\preceq2\norm{Ah}_{g(y)}\,A^{\T}g(y)A=2\norm h_{\bar{g}(x)}\,\bar{g}(x)\,.
\]

Consider a sequence $\{x_{n}\}\subset\bar{K}$ converging to a boundary
point $x\in\de\bar{K}$. If $Tx\notin\de K$, then $Tx\in\inter(K)$,
and the continuity of $T$ implies $x$ is also in $\inter(\bar{K})$.
Thus, $Tx\in\de K$ and $\psi(x_{n})=\phi(Tx_{n})\to\phi(Tx)=\infty$.
Lastly, $\hess\phi\asymp g$ leads to $\hess\psi=A^{\T}\hess\phi\,A\asymp A^{\T}gA=\bar{g}$,
and $\bar{g}$ is $\nu$-self-concordant for $\bar{K}$.

As for symmetry, since $\bar{g}$ is self-concordant, $\mc D_{\bar{g}}^{1}(x)\subset\bar{K}\cap(2x-\bar{K})$
for $x\in\inter(\bar{K})$ by Lemma~\ref{lem:dikin-in-body}. For
$z\in\bar{K}\cap(2x-\bar{K})$, as $Tz\in K\cap(2Tx-K)$ holds, it
follows that 
\[
\onu\geq\norm{Tz-Tx}_{g(y)}^{2}=\norm{z-y}_{A^{\T}g(y)A}^{2}=\norm{z-y}_{\bar{g}(x)}^{2}\,,
\]
and thus $\bar{g}$ is $\onu$-symmetric.

As for the second item, we first show that $\bar{g}$ is collapsed
onto $W=\rowspace(A)$ (i.e., $\bar{g}=P_{W}\bar{g}P_{W}$ for the
orthogonal projection $P_{W}$ onto $W$). To see this, observe that
\begin{align*}
P_{W}\bar{g}P_{W} & =P_{W}A^{\T}gAP_{W}=A^{\T}(AA^{\T})^{\dagger}A\cdot A^{\T}gA\cdot A^{\T}(AA^{\T})^{\dagger}A\,,
\end{align*}
and due to $AA^{\T}(AA^{\T})^{\dagger}A=AA^{\T}(A^{\T})^{\dagger}A^{\dagger}A=AA^{\dagger}A=A$,
we have $P_{W}\bar{g}P_{W}=A^{\T}gA=\bar{g}$. 

We now show that $\bar{g}$ is SSC along $W$. For $k:=\dim(W)$,
take $U\in\R^{d\times k}$ whose columns form an orthonormal basis
of $W$. It suffices to show that $g_{W}:=U^{\T}\bar{g}U=U^{\T}A^{\T}gAU=M^{\T}gM$
for $M:=AU\in\R^{m\times k}$ is SSC. First of all, we can check PDness
of $g_{W}$ as follows: Suppose $g_{W}v=0$ for some $v\in\R^{k}$.
Then $0=\norm v_{g_{W}}=\norm{g^{1/2}Mv}_{2}$ and $AUv=Mv=0$. Since
$Uv\in\rowspace(A)\cap\textsf{ker}(A)$ and $U$ is full-rank, we
have $v=0$. Next, for $h\in\R^{k}$ and $x\in\inter(\bar{K})$
\begin{align*}
 & \tr\bpar{g_{W}(x)^{-1}\Dd g_{W}(x)[h]\,g_{W}(x)^{-1}\Dd g_{W}(x)[h]}=\tr\Bpar{\bpar{g^{\half}M(M^{\T}gM)^{-1}M^{\T}g^{\half}\cdot g^{-\half}\Dd g(Tx)[Ah]\,g^{-\half}}^{2}}\\
\underset{\text{(i)}}{\leq} & \tr\Bpar{\bpar{g^{-\half}\Dd g(Tx)[Ah]\,g^{-\half}}^{2}}\leq\norm{g^{-\half}\Dd g(Tx)[Ah]\,g^{-\half}}_{F}^{2}\leq4\norm{Ah}_{g(Tx)}^{2}=4\norm h_{\bar{g}(x)}^{2}\,,
\end{align*}
where in (i) we used $P(g^{\half}M)=g^{\half}M(M^{\T}gM)^{-1}M^{\T}g^{\half}\preceq I$.
Thus, $\bar{g}$ is SSC along $W=\rowspace(A)$.

The third item immediately follows from $\Dd^{2}\bar{g}(x)[h,h]=A^{\T}\Dd^{2}g(y)[Ah,Ah]\,A\succeq0$
for any $h\in\Rd$. 

As for the fourth item, for any PSD matrix function $g'$ on $\bar{K}$
we have
\begin{align*}
 & \tr\bpar{(g'+\bar{g})^{-1}\Dd^{2}\bar{g}[h,h]}=\tr\Bpar{(g'+A^{\T}gA)^{-1}A^{\T}\Dd^{2}g[Ah,Ah]\,A}\\
= & \tr\Bpar{(A^{-\T}g'A^{-1}+g)^{-1}\Dd^{2}g[Ah,Ah]}\geq-\norm{Ah}_{g}^{2}=-\norm h_{\bar{g}}^{2}\,.
\end{align*}

The last item is straightforward to check by the change of variable.
\end{proof}

\subsubsection{Collapse and embedding: lifting up SSC, SLTSC, and SASC \label{proof:lifting-ssc}}

In passing SSC to an augmented space, the Woodbury matrix identity
is a main technical tool used: for matrices with compatible sizes
\[
(I+UV)^{-1}=I-U\,(I+VU)^{-1}V\,.
\]
Using this, we show that if $g\in\pd$ is SSC, then $\bar{g}+\veps I_{m}$
is SSC.
\begin{proof}
[Proof of Lemma~\ref{lem:embedding-ssc}] Fix $\veps>0,y\in\inter(K')$,
and $h\in\R^{m}$. Take a projection matrix $P\in\{0,1\}^{d\times m}$
such that $PP^{\T}=I_{d}$ and $\bar{g}(y)=P^{\T}g(Py)P$ for $x=Py\in\intk$.
Also for $k:=\dim(W)$, take a matrix $U\in\R^{d\times k}$ whose
columns form an orthonormal basis of $W$. Then $\bar{g}(y)=P^{\T}g(Py)P$
and $g(x)=Ug_{W}(x)U$, so for $M:=U^{\T}P\in\R^{k\times m}$,
\[
\bar{g}(y)=P^{\T}Ug_{W}(Py)U^{\T}P=M^{\T}g_{W}(Py)M\,.
\]
 Note that $MM^{\T}=I_{k}$. Thus,
\begin{align*}
 & \norm{(\bar{g}(y)+\veps I)^{-\half}\Dd(\bar{g}+\veps I)(y)[h]\,(\bar{g}(y)+\veps I)^{-\half}}_{F}^{2}=\tr\Bpar{\bpar{(\bar{g}(y)+\veps I)^{-1}\Dd\bar{g}(y)[h]}^{2}}\\
= & \tr\Bpar{\bpar{M(M^{\T}g_{W}(x)\,M+\veps I)^{-1}M^{\T}\cdot\Dd g_{W}(x)[Ph]}^{2}}\underset{\text{(i)}}{=}\tr\Bpar{\bpar{(g_{W}(x)+\veps I_{k})^{-1}\Dd g_{W}(x)[Ph]}^{2}}\\
\leq & \norm{g_{W}(x)^{-\half}\Dd g_{W}(x)[Ph]\,g_{W}(x)^{-\half}}_{F}^{2}\leq4\norm{Ph}_{g(x)}^{2}=4\norm h_{\bar{g}(y)}^{2}\,,
\end{align*}
where in (i) we used the identity $M\bpar{M^{\T}g_{W}(x)\,M+\veps I}^{-1}M^{\T}=(g_{W}(x)+\veps I_{k})^{-1}$.
To see this, we use the Woodbury matrix identity to get
\[
(\veps I_{m}+M^{\T}g_{W}M)^{-1}=\frac{1}{\veps}I_{m}-\frac{1}{\veps^{2}}M^{\T}g_{W}^{\half}\bpar{I_{k}+\frac{1}{\veps}g_{W}}^{-1}g_{W}^{\half}M\,,
\]
and thus conjugating both sides by $M$ results in 
\begin{align*}
M\bpar{M^{\T}g_{W}M+\veps I_{m}}^{-1}M^{\T} & =\frac{1}{\veps}I_{k}-\frac{1}{\veps}g_{W}^{\half}(g_{W}+\veps I_{k})^{-1}g_{W}^{\half}=\frac{1}{\veps}I_{k}-\frac{1}{\veps}(g_{W}+\veps I_{k})^{-1}g_{W}\,.
\end{align*}
Then, the identity follows from
\begin{align*}
(g_{W}+\veps I_{k})\cdot\bpar{\frac{1}{\veps}I_{k}-\frac{1}{\veps}\,(g_{W}+\veps I_{k})^{-1}g_{W}} & =\frac{1}{\veps}(g_{W}+\veps I_{k})-\frac{1}{\veps}g_{W}=I_{k}\,.\qedhere
\end{align*}
\end{proof}
In extending SLTSC and SASC, we need two technical lemmas: the inverse
of a block matrix and connection between P(S)Dness and Schur complements.
\begin{lem}
\label{lem:block-inverse} If $D$ and its Schur complement $A-BD^{-1}C$
are invertible, then 
\[
\left[\begin{array}{cc}
A & B\\
C & D
\end{array}\right]^{-1}=\left[\begin{array}{cc}
(A-BD^{-1}C)^{-1} & *\\
* & *
\end{array}\right]\,.
\]
\end{lem}

\begin{lem}
[Schur complement] \label{lem:schur} Let $A\in\Rdd,B\in\R^{d\times m},C\in\R^{m\times m}$
and define a matrix $M\in\R^{(m+d)\times(m+d)}$ by 
\[
M=\left[\begin{array}{cc}
A & B\\
B^{\T} & D
\end{array}\right]\,.
\]
Then $M\succ0$ if and only if $A\succ0$ and $C-BA^{-1}B^{\T}\succ0$
if and only $C\succ0$ and $A-B^{\T}C^{-1}B\succ0$.
\end{lem}

Using these, we show that if $g$ is SLTSC and SASC, then $\bar{g}$
is SLTSC and SASC.
\begin{proof}
[Proof of Lemma~\ref{lem:embedding-sltsc}] Take a full row-rank
projection matrix $P\in\{0,1\}^{d\times m}$ such that $\bar{g}(y)=P^{\T}g(Py)P$,
where the rows of $P$ forms a subset of the canonical basis $\{e_{1},\dots,e_{m}\}$.
We can augment the rows of $P$ with the rest of the canonical basis
so that the augmented matrix $\bar{P}\in\R^{m\times m}$ is an orthonormal
matrix. Then we can represent $\bar{g}$ by 
\[
\bar{g}(y)=\bar{P}^{\T}\left[\begin{array}{cc}
g(Py) & 0\\
0 & 0
\end{array}\right]\bar{P}\,.
\]

Consider a PSD matrix function $g':\inter(K')\to\mbb S_{+}^{m}$ such
that $g'+\bar{g}$ is PD on $K'$. Representing them in the block
form with $g_{A}\in\Rdd,g_{B}\in\R^{d\times(m-d)},$ and $g_{C}\in\R^{(m-d)\times(m-d)}$
\[
\bar{g}+g'=\bar{P}^{\T}\Par{\left[\begin{array}{cc}
g & 0\\
0 & 0
\end{array}\right]+\left[\begin{array}{cc}
g_{A} & g_{B}\\
g_{B}^{\T} & g_{C}
\end{array}\right]}\bar{P}=\bar{P}^{\T}\underbrace{\left[\begin{array}{cc}
g+g_{A} & g_{B}\\
g_{B}^{\T} & g_{C}
\end{array}\right]}_{\eqqcolon g^{*}}\bar{P}\,.
\]
Since $g^{*}$ is PD, $g_{C}$ and its Schur complement $(g+g_{A})-g_{B}g_{C}^{-1}g_{B}^{\T}$
are PD. Thus by Lemma~\ref{lem:block-inverse},
\[
\left[\begin{array}{cc}
g+g_{A} & g_{B}\\
g_{B}^{\T} & g_{C}
\end{array}\right]^{-1}=\left[\begin{array}{cc}
(g+g_{A}-g_{B}g_{C}^{-1}g_{B}^{\T})^{-1} & *\\
* & *
\end{array}\right]\,.
\]
Hence,
\begin{align*}
 & \tr\bpar{(\bar{g}+g')^{-1}\Dd^{2}\bar{g}(y)[h,h]}=\tr\Biggl(\bar{P}^{\T}\left[\begin{array}{cc}
g+g_{A} & g_{B}\\
g_{B}^{\T} & g_{C}
\end{array}\right]^{-1}\bar{P}\bar{P}^{\T}\left[\begin{array}{cc}
\Dd^{2}g(Py)[Ph,Ph] & 0\\
0 & 0
\end{array}\right]\bar{P}\Biggr)\\
= & \tr\Biggl(\left[\begin{array}{cc}
g+g_{A} & g_{B}\\
g_{B}^{\T} & g_{C}
\end{array}\right]^{-1}\left[\begin{array}{cc}
\Dd^{2}g(Py)[Ph,Ph] & 0\\
0 & 0
\end{array}\right]\Biggr)=\tr\bpar{(g+\underbrace{g_{A}-g_{B}g_{C}^{-1}g_{B}^{\T}}_{\succeq0})^{-1}\,\Dd^{2}g(Py)[Ph,Ph]}\\
\ge & -\norm{Ph}_{g(Py)}^{2}=-\norm h_{\bar{g}(y)}^{2}\,,
\end{align*}
where in the last inequality we used STLSC of $g$, since $g'\succeq0$
ensures that its Schur complement satisfies $g_{A}-g_{B}g_{C}^{-1}g_{B}^{\T}\succeq0$
by Lemma~\ref{lem:schur}.

For SASC, consider any PSD matrix function $g':\inter(K')\to\mbb S_{+}^{m}$.
For $x=Py$ and $z_{x}=Pz_{y}\in\Rd$ with $z_{y}\sim\ncal\bpar{y,\frac{r^{2}}{m}\,(\bar{g}+g)(y)^{-1}}$,
we have
\[
\norm{z_{y}-y}_{\bar{g}(z_{y})}^{2}-\norm{z_{y}-y}_{\bar{g}(y)}^{2}=\norm{z_{x}-x}_{g(z_{x})}^{2}-\norm{z_{x}-x}_{g(x)}^{2}\,.
\]
Also, $z_{x}-x=P\,(z_{y}-y)$ is a Gaussian with zero mean and covariance
\begin{align*}
 & \frac{r^{2}}{m}\,P\,(\bar{g}+g')(y)^{-1}P^{\T}=\frac{r^{2}}{m}\,P\bar{P}^{\T}\Par{\left[\begin{array}{cc}
g & 0\\
0 & 0
\end{array}\right]+\left[\begin{array}{cc}
g_{A} & g_{B}\\
g_{B}^{\T} & g_{C}
\end{array}\right]}^{-1}\bar{P}\bar{P}^{\T}\\
= & \frac{r^{2}}{m}\,\left[\begin{array}{cc}
I_{d} & 0_{d\times(m-d)}\end{array}\right]\Par{\left[\begin{array}{cc}
g & 0\\
0 & 0
\end{array}\right]+\left[\begin{array}{cc}
g_{A} & g_{B}\\
g_{B}^{\T} & g_{C}
\end{array}\right]}^{-1}\left[\begin{array}{c}
I_{d}\\
0_{d\times(m-d)}
\end{array}\right]=\frac{r^{2}}{m}\,(g+g_{A}-g_{B}g_{C}^{-1}g_{B}^{\T})^{-1}\,.
\end{align*}
Since $g_{A}-g_{B}g_{C}^{-1}g_{B}^{\T}\succeq0$ due to  $g'\succeq0$,
it holds that $g_{0}:=\frac{m-d}{d}g+\frac{m}{d}(g_{A}-g_{B}g_{C}^{-1}g_{B}^{\T})$
on $\intk$ is PSD. Now, it suffices to check that the covariance
matrix above is equal to $\frac{r^{2}}{d}(g+g_{0})^{-1}$:
\[
\frac{d}{r^{2}}\,(g+g_{0})=\frac{d}{r^{2}}\Bpar{g+\frac{m-d}{d}\,g+\frac{m}{d}\,(g_{A}-g_{B}g_{C}^{-1}g_{B}^{\T})}\frac{m}{r^{2}}\,(g+g_{A}-g_{B}g_{C}^{-1}g_{B}^{\T})\,.\qedhere
\]
\end{proof}

\subsubsection{Direct product: SSC and SLTSC \label{proof:direct-ssc-sltsc}}

We show that if $g_{i}\in\mbb S_{++}^{d_{i}}$ is SC, then $g=\sum d_{i}\bar{g}_{i}$
is SSC.
\begin{proof}
[Proof of Lemma~\ref{lem:ssc-direct}] Note that $d_{i}g_{i}$ is
SSC for $i=1,\dots,m$. For $x\in\prod E_{i}$ and $h=(h_{1},\dots,h_{m})\in\R^{l}$
with $h_{i}\in\R^{d_{i}}$, we have 
\begin{align*}
 & \norm{g(x)^{-\half}\Dd g(x)[h]\,g(x)^{-\half}}_{F}^{2}\\
 & =\left\Vert \left[\begin{array}{ccc}
g_{1}(x_{1})^{-\half}\Dd g_{1}(x_{1})[h_{1}]\,g_{1}(x_{1})^{-\half}\\
 & \ddots\\
 &  & g_{m}(x_{m})^{-\half}\Dd g_{m}(x_{m})[h_{m}]\,g_{m}(x_{m})^{-\half}
\end{array}\right]\right\Vert _{F}^{2}\\
 & =\sum_{i}\norm{g_{i}(x_{i})^{-\half}\Dd g_{i}(x_{i})[h_{i}]\,g_{i}(x_{i})^{-\half}}_{F}^{2}\leq4\sum_{i}\norm{h_{i}}_{d_{i}g_{i}(x_{i})}^{2}=4\norm h_{g(x)}^{2}\,.\qedhere
\end{align*}
\end{proof}
Next, we show that if $g_{i}\in\mbb S_{++}^{d_{i}}$ is HSC, then
$g=\sum d_{i}\bar{g}_{i}$ is SLTSC.
\begin{proof}
[Proof of Lemma~\ref{lem:sltsc-direct}] For $h=(h_{1},\dots,h_{m})$
and any PSD matrix function $g'$, we have
\begin{align*}
\tr\bpar{(g'+g)^{-1}\Dd^{2}g[h^{\otimes2}]} & =\sum_{i}\tr\bpar{(g'+(g-d_{i}\bar{g}_{i})+d_{i}\bar{g}_{i})^{-1}\Dd^{2}(d_{i}\bar{g}_{i})[h^{\otimes2}]}\gtrsim-\sum_{i}\norm h_{d_{i}\bar{g}_{i}}^{2}=-\norm h_{g}^{2}\,,
\end{align*}
where we used Lemma~\ref{lem:hsc-to-sltsc} in the inequality.
\end{proof}

\subsubsection{Inverse images under non-linear mappings \label{proof:inverse-non-linear}}
\begin{proof}
[Proof of Lemma~\ref{lem:compatible}] Since $\acal$ is $(R(G),\beta),\gamma)$-compatible
with $\Gamma$, the first two claims immediately follow from \citet[Proposition 5.1.7]{nesterov1994interior}.
Let $x\in G^{+}$ and $h\in\Rd$. Define the following notations:
\begin{align*}
u=\Dd\acal(x)[h], & \quad v=\Dd^{2}\acal(x)[h^{\otimes2}],\quad w=\Dd^{3}\acal(x)[h^{\otimes3}],\quad z=\Dd^{4}\acal(x)[h^{\otimes4}],\\
s=\sqrt{\Dd F(y)[v]}, & \quad\rho=\sqrt{\Dd^{2}\Pi(x)[h^{\otimes2}]},\quad r=\sqrt{\Dd^{2}F(y)[u^{\otimes2}]}\,.
\end{align*}
From direct computations, we have 
\begin{align*}
\Dd^{2}\Psi(x)[h^{\otimes2}] & =\Dd F(y)[v]+\Dd^{2}F(y)[u^{\otimes2}]+\delta^{2}\Dd^{2}\Pi(x)[h^{\otimes2}]=s^{2}+r^{2}+\delta^{2}\rho^{2}\,,\\
\Dd^{3}\Psi(x)[h^{\otimes3}] & =\Dd F(y)[w]+3\Dd^{2}F(y)[u,v]+\Dd^{3}F(y)[u^{\otimes3}]+\delta^{2}\Dd^{3}\Pi(x)[h^{\otimes3}]\,,\\
\Dd^{4}\Psi(x)[h^{\otimes4}] & =\Dd^{2}F(y)[w,u]+\Dd F(y)[z]+3\Dd^{3}F(y)[u,u,v]+3\Dd^{2}F(y)[v^{\otimes2}]\\
 & \qquad+3\Dd^{2}F(y)[u,w]+\Dd^{4}F(y)[u^{\otimes4}]+3\Dd^{3}F(y)[u,u,v]+\delta^{2}\Dd^{4}\Pi(x)[h^{\otimes4}]\\
 & =\Dd F(y)[z]+3\Dd^{2}F(y)[v^{\otimes2}]+4\Dd^{2}F(y)[u,w]\\
 & \qquad+6\Dd^{3}F(y)[u,u,v]+\Dd^{4}F(y)[u^{\otimes4}]+\delta^{2}\Dd^{4}\Pi(x)[h^{\otimes4}]\,.
\end{align*}
HSC of $F$ and $\Pi$ implies that 
\[
|\Dd^{4}\Pi(x)[h^{\otimes4}]|\leq6\rho^{4}\,,\qquad\text{and}\qquad|\Dd^{4}F(y)[u^{\otimes4}]|\leq6r^{4}\,.
\]
Since $\acal$ is $(K,\beta,\gamma)$-compatible and $K\subset R(G)$,
Lemma~\ref{lem:extension-compatibility}-1 implies concavity of $\acal$
with respect to $R(G)$, which means $-v\geq_{R(G)}0$. Then, \citet[Corollary 2.3.1]{nesterov1994interior}
ensures 
\[
\sqrt{\Dd^{2}F(y)[v^{\otimes2}]}\leq\Dd F(y)[v]=s^{2}\,.
\]
Hence, $|3\Dd^{2}F(y)[v,v]|\leq3(\Dd F(y)[v])^{2}=3s^{4}$, and self-concordance
of $F$ results in
\[
|6\Dd^{3}F(y)[u,u,v]|\leq12r^{2}\sqrt{\Dd^{2}F(y)[v,v]}\leq12r^{2}s^{2}\,.
\]
Since $\{h:h^{\T}\Pi(x)h\leq1\}$ is contained in $\Gamma\cap(2x-\Gamma)$,
compatibility of $\acal$ leads to 
\[
\beta\Dd^{2}\acal(x)\Bbrack{\Bpar{\frac{h}{\norm h_{\Pi(x)}}}^{\otimes2}}\leq_{K}\Dd^{3}\acal(x)\Bbrack{\Bpar{\frac{h}{\norm h_{\Pi(x)}}}^{\otimes3}}\leq_{K}-\beta\Dd^{2}\acal(x)\Bbrack{\Bpar{\frac{h}{\norm h_{\Pi(x)}}}^{\otimes2}}\,,
\]
and thus $\beta\rho v\leq_{K}w\leq_{K}-\beta\rho v$. As $K$ is a
ray, $\Dd^{2}F(y)[w,w]\leq\beta^{2}\rho^{2}\Dd^{2}F(y)[v,v]\leq\beta^{2}\rho^{2}s^{4}$.
Thus,
\[
|4\Dd^{2}F(y)[u,w]|\leq4\sqrt{\Dd^{2}F(y)[u,u]}\sqrt{\Dd^{2}F(y)[w,w]}\leq4r\beta\rho s^{2}\,.
\]
Lastly, since $\gamma v\rho^{2}\leq_{K}z\leq_{K}-\gamma v\rho^{2}$
and $K$ is a ray, we have 
\[
|\Dd F(y)[z]|\leq3\gamma\rho^{2}|\Dd F(y)[v]|=3\gamma\rho^{2}s^{2}\,.
\]
Putting these together,
\begin{align*}
 & \Abs{\Dd^{4}\Psi(x)[h^{\otimes4}]}\leq3\gamma\rho^{2}s^{2}+4r\beta\rho s^{2}+12r^{2}s^{2}+3s^{4}+6\delta^{2}\rho^{4}+6r^{4}\\
\leq & 6(\delta^{2}\rho^{4}+r^{4}+s^{4}+r^{2}s^{2}+\delta\rho^{2}s^{2}+\delta r\rho s^{2})\leq6\bpar{(\delta\rho)^{4}+r^{4}+s^{4}+r^{2}s^{2}+(\delta\rho)^{2}s^{2}+r^{2}s^{2}+(\delta\rho)^{2}s^{2}}\\
\leq & 6\bpar{(\delta\rho)^{2}+r^{2}+s^{2}}^{2}=6\bpar{\Dd^{2}\Psi(x)[h,h]}^{2}\,.\qedhere
\end{align*}
\end{proof}

\subsection{Main constraints and epigraphs ($\S$\ref{sec:handbook-barrier})}

\subsubsection{Linear constraints: strong self-concordance and symmetry \label{proof:linear-SSC-symm}}

We relate SSC and symmetry to well-studied terms in the field of optimization,
such as $\max_{i}\frac{[\sigma(\sqrt{D_{x}}A_{x})]_{i}}{[D_{x}]_{ii}}$
and $\norm{D_{x,h}'}_{D_{x}^{-1}}^{2}$.
\begin{proof}
[Proof of Lemma~\ref{lem:helper4Diagonal}] Let us write $g(x)=A_{x}^{\T}D_{x}A_{x}=A^{\T}V_{x}A$
for $V_{x}:=S_{x}^{-1}D_{x}S_{x}^{-1}$. By Claim~\ref{claim:diffLogBarrier},
\begin{align}
\Dd g(x)[h] & =A^{\T}(-2S_{x}^{-1}S_{x,h}S_{x}^{-1}D_{x}+S_{x}^{-1}\Dd D_{x}[h]\,S_{x}^{-1})A=A^{\T}V_{x}^{1/2}\overline{D}_{x}V_{x}^{1/2}A\,,\label{eq:Dgh}
\end{align}
where $\overline{D}_{x}:=-2S_{x,h}+D_{x}^{-1}\Dd D_{x}[h]$. Using
this,
\begin{align*}
\norm{(g'+g)^{-\half}\Dd g[h]\,(g'+g)^{-\half}}_{F}^{2} & =\tr\bpar{(g'+g)^{-1}A^{\T}V_{x}^{1/2}\overline{D}_{x}\underbrace{V_{x}^{1/2}A(g'+g)^{-1}A^{\T}V_{x}^{1/2}}_{=:P_{x}'}\overline{D}_{x}V_{x}^{1/2}A}\\
 & =\tr(P_{x}'\overline{D}_{x}P_{x}'\overline{D}_{x})\,.
\end{align*}
By Lemma~\ref{lem:matrix-projection}, we have $P_{x}'\preceq P_{x}=P(V_{x}^{1/2}A)=P(D_{x}^{1/2}A_{x})$,
and thus 
\begin{align*}
\tr(P_{x}'\overline{D}_{x}P_{x}'\overline{D}_{x}) & \leq\tr(P_{x}\overline{D}_{x}P_{x}\overline{D}_{x})\underset{\text{(i)}}{=}\diag(\overline{D}_{x})^{\T}P_{x}^{(2)}\,\diag(\overline{D}_{x})\underset{\text{(ii)}}{\leq}\diag(\overline{D}_{x})^{\T}\Sigma_{x}\,\diag(\overline{D}_{x})\\
 & \underset{\text{(iii)}}{\leq}4\sum_{i=1}^{m}[\sigma(D_{x}^{1/2}A_{x})]_{i}\,\bpar{(A_{x}h)_{i}^{2}+(D_{x}^{-1}\Dd D_{x}[h])_{i}^{2}}\\
 & \leq4\max_{i}\frac{[\sigma(D_{x}^{1/2}A_{x})]_{i}}{[D_{x}]_{ii}}\cdot\sum_{i=1}^{m}[D_{x}]_{ii}\,\bpar{(A_{x}h)_{i}^{2}+(D_{x}^{-1}\Dd D_{x}[h])_{i}^{2}}\\
 & \underset{\text{(iv)}}{=}4\max_{i}\frac{[\sigma(D_{x}^{1/2}A_{x})]_{i}}{[D_{x}]_{ii}}\cdot\bpar{\norm h_{g(x)}^{2}+\sum_{i=1}^{m}[D_{x}^{-1}]_{ii}(\Dd D_{x}[h])_{i}^{2}}\,,
\end{align*}
where (i) holds due to $x^{\T}(A\hada B)y=\tr\bpar{\Diag(x)A\Diag(y)B^{\T}}$
(Lemma~\ref{lem:Hadamard}), (ii) follows from $P_{x}^{(2)}\preceq\Sigma_{x}$
(Claim~\ref{claim:schurProjection})\footnote{Even though this lemma is proven for leverage scores, the proof there
can be extended to any orthogonal projection matrices.}, (iii) uses $(a+b)^{2}\leq2\Par{a^{2}+b^{2}}$ for $a,b\in\R$ and
$\Sigma_{x}=\Diag(P_{x})=\sigma(D_{x}^{1/2}A_{x})$, and (iv) holds
due to $\sum_{i=1}^{m}[D_{x}]_{ii}\,(A_{x}h)_{i}^{2}=h^{\T}A_{x}^{\T}D_{x}A_{x}h=h^{\T}g(x)h$.

As for the second claim,
\begin{align*}
 & \max_{h:\norm h_{g(x)}=1}\norm{A_{x}h}_{\infty}=\max_{h}\max_{i\in[m]}\Abs{\frac{a_{i}^{\T}h}{s_{i}}}=\max_{i\in[m]}\max_{u:\norm u_{2}=1}\Abs{\frac{a_{i}^{\T}g(x)^{-1/2}u}{s_{i}}}\\
= & \max_{i\in[m]}\left\Vert g(x)^{-1/2}\frac{a_{i}}{s_{i}}\right\Vert _{2}=\max_{i\in[m]}\sqrt{\frac{1}{s_{i}^{2}}a_{i}^{\T}g(x)^{-1}a_{i}}=\sqrt{\max_{i\in[m]}e_{i}^{\T}A_{x}g(x)^{-1}A_{x}^{\T}e_{i}}=\sqrt{\max_{i\in[m]}\frac{[\sigma(D_{x}^{1/2}A_{x})]_{i}}{[D_{x}]_{ii}}}\,.
\end{align*}

As for the last claim, for $h\in\Rd$ such that $\norm{A_{x}h}_{\infty}\leq1$
(i.e., $h\in K\cap(2x-K)$ for $K=\{Ax\geq b\}$ due to Lemma~\ref{lem:symmforPolytope})
we have
\begin{align*}
h^{\T}g(x)h & =h^{\T}A_{x}^{\T}D_{x}A_{x}h=\sum_{i=1}^{m}(D_{x})_{ii}(A_{x}h)_{i}^{2}\leq\norm{A_{x}h}_{\infty}^{2}\sum_{i=1}^{m}(D_{x})_{ii}\leq\tr(D_{x})\,.\qedhere
\end{align*}
\end{proof}
Now we establish SSC and compute the symmetry parameters of metrics
of the form $A_{x}^{\T}D_{x}A_{x}$:
\begin{proof}
[Proof of Lemma~\ref{lem:paramsBarrier}] \textbf{Logarithmic barrier}:
To show that $g$ is SSC along $\rowspace(A)$, consider a self-concordant
matrix $g(y)=S_{y}^{-2}=-\nabla_{y}^{2}(\sum_{i=1}^{m}\log y_{i})$
defined on $\{y\in\R^{m}:y\geq0\}$. By putting $D_{x}=I_{m}$ and
$A_{x}=S_{x}^{-1}$ into Lemma~\ref{lem:helper4Diagonal}-1, since
$\sigma(A_{x})\leq1$
\[
\norm{g(x)^{-\half}\Dd g(x)[h]\,g(x)^{-\half}}_{F}\leq2\Bpar{\max_{i\in[m]}\sigma(A_{x})_{i}}^{1/2}\,\norm h_{g(x)}\leq2\norm h_{g(x)}\,.
\]
Through the linear map $Tx=Ax-b=y$, we recover $g(x)=\hess\phi_{\log}(x)=A^{\T}S_{y}^{-2}A=A_{x}^{\T}A_{x}$,
which is SSC along $\rowspace(A)$ by Lemma~\ref{lem:linear-trans-matrix}.
For the $\onu$-symmetry, the first part (i.e., $\dcal_{g}^{1}(x)\subset K\cap(2x-K)$)
follows from Lemma~\ref{lem:symmetricLeftpart}. The second part
is immediate from $\onu=\tr(I_{m})=m$ and Lemma~\ref{lem:helper4Diagonal}-3.

\textbf{Approximate volumetric barrier}: For $D_{x}=\Sigma_{x}=\Sigma(A_{x})$,
by Lemma~\ref{lem:usefulFactLewis}-1 and 3 with $p=2$,
\begin{align*}
\max_{i}\frac{[\sigma(D_{x}^{1/2}A_{x})]_{i}}{[D_{x}]_{ii}} & \leq2\sqrt{m}\,,\quad\text{and}\quad\sum_{i=1}^{m}[D_{x}^{-1}]_{ii}\,(\Dd D_{x}[h])_{i}^{2}=\norm{\Sigma_{x}^{-1}\diag(\Dd\Sigma_{x}[h])}_{\Sigma_{x}}^{2}\leq4\norm h_{g(x)}^{2}\,.
\end{align*}
Using Lemma~\ref{lem:helper4Diagonal}-1,
\begin{align*}
\norm{g(x)^{-\half}\Dd g(x)[h]\,g(x)^{-\half}}_{F}^{2} & \leq4\max_{i}\frac{[\sigma(D_{x}^{1/2}A_{x})]_{i}}{[D_{x}]_{ii}}\,\bpar{\norm h_{g(x)}^{2}+\sum_{i=1}^{m}[D_{x}^{-1}]_{ii}(\Dd D_{x}[h])_{i}^{2}}\leq40\sqrt{m}\norm h_{g(x)}^{2}\,.
\end{align*}
For the $\onu$-symmetry, $\norm{A_{x}(y-x)}_{\infty}^{2}\leq\max_{i\in[m]}\frac{[\sigma(D_{x}^{1/2}A_{x})]_{i}}{[D_{x}]_{ii}}\leq2m^{1/2}$
for $y\in\dcal_{g}^{1}(x)$ by Lemma~\ref{lem:helper4Diagonal}-2.
Also, Lemma~\ref{lem:helper4Diagonal}-3 implies that $y$ with $\norm{A_{x}(y-x)}_{\infty}\leq1$
is contained in $\dcal_{g}^{\sqrt{\tr(D_{x})}}(x)$, where $\tr(D_{x})=\tr(P_{x})\leq d$.
Therefore, $\tilde{g}(x):=40\sqrt{m}g(x)=40\sqrt{m}A_{x}^{\T}\Sigma_{x}A_{x}$
is SSC with the symmetry parameter $\onu=\mc O(\sqrt{m}d)$.

\textbf{Vaidya metric}: Consider the metric without scaling: $g(x):=A_{x}^{\T}D_{x}A_{x}$
with $D_{x}=\Sigma_{x}+\frac{d}{m}I_{m}$. Then, using \citet[(4.5)]{anstreicher1997volumetric}
in (i) below
\begin{align}
\max_{i}\frac{[\sigma(D_{x}^{1/2}A_{x})]_{i}}{[D_{x}]_{ii}} & \underset{\text{Lemma \ref{lem:helper4Diagonal}-2}}{=}\Bpar{\max_{h\in\Rd}\frac{\norm{A_{x}h}_{\infty}}{\norm h_{g(x)}}}^{2}\underset{\text{(i)}}{\leq}\sqrt{\frac{m}{d}}\,,\label{eq:28-1}\\
\sum_{i=1}^{m}[D_{x}^{-1}]_{ii}\,(\Dd D_{x}[h])_{i}^{2} & \underset{\text{(ii)}}{\leq}\sum_{i=1}^{m}[\Sigma_{x}^{-1}]_{ii}(\Dd\Sigma_{x}[h])_{i}^{2}\underset{\text{Lemma \ref{lem:usefulFactLewis}-3}}{\leq}4h^{\T}A_{x}^{\T}\Sigma_{x}A_{x}h\leq4\norm h_{g(x)}^{2}\,.\nonumber 
\end{align}
Putting these back to Lemma~\ref{lem:helper4Diagonal}-1,
\begin{align*}
\norm{g(x)^{-\half}\Dd g(x)[h]\,g(x)^{-\half}}_{F}^{2} & \leq4\max_{i}\frac{[\sigma(D_{x}^{1/2}A_{x})]_{i}}{[D_{x}]_{ii}}\,\bpar{\norm h_{g(x)}^{2}+\sum_{i=1}^{m}[D_{x}^{-1}]_{ii}(\Dd D_{x}[h])_{i}^{2}}\leq20\sqrt{\frac{m}{d}}\norm h_{g(x)}^{2}\,.
\end{align*}
Thus, $\tilde{g}(x):=22\sqrt{\frac{m}{d}}g(x)=22\sqrt{\frac{m}{d}}A_{x}^{\T}\bpar{\Sigma_{x}+\frac{d}{m}I_{m}}A_{x}$
is SSC. For the $\onu$-symmetry, Lemma~\ref{lem:helper4Diagonal}-2
implies that for $y\in\dcal_{g}^{1}(x)$,
\[
\norm{A_{x}(y-x)}_{\infty}^{2}\leq\max_{i}\frac{[\sigma(D_{x}^{1/2}A_{x})]_{i}}{[D_{x}]_{ii}}\underset{\text{\eqref{eq:28-1}}}{\leq}\sqrt{\frac{m}{d}}\,.
\]
Also, Lemma~\ref{lem:helper4Diagonal}-3 implies that $y$ with $\norm{A_{x}(y-x)}_{\infty}\leq1$
is contained in $\dcal_{g}^{\sqrt{\tr(D_{x})}}(x)$, where
\[
\tr(D_{x})=\tr\bpar{\Sigma_{x}+\frac{d}{m}I_{m}}=\tr(\Sigma_{x})+d\leq2d\,.
\]
Therefore, $\tilde{g}(x)$ satisfies $\dcal_{\tilde{g}}^{1}(x)\subset K\cap(2x-K)\subset\dcal_{\tilde{g}}^{\sqrt{44(md)^{1/2}}}(x)$,
so $\tilde{g}$ is $\mc O(\sqrt{md})$-symmetric.

\textbf{Lewis-weight metric}: Consider the unscaled version first:
$g(x)=A_{x}^{\T}W_{x}A_{x}$. By Lemma~\ref{lem:helper4Diagonal}-1
\begin{align*}
\norm{g(x)^{-\half}\Dd g(x)[h]\,g(x)^{-\half}}_{F}^{2} & \leq4\max_{i}\frac{[\sigma(W_{x}^{1/2}A_{x})]_{i}}{[W_{x}]_{ii}}\,\bpar{\norm h_{g(x)}^{2}+\sum_{i=1}^{m}[W_{x}^{-1}]_{ii}(\Dd W_{x}[h])_{i}^{2}}\\
 & \underset{\text{(i)}}{\leq}8m^{\frac{2}{p+2}}\bpar{\norm h_{g(x)}^{2}+p^{2}\,\norm h_{g(x)}^{2}}\leq\bpar{8m^{\frac{2}{p+2}}(1+p^{2})}\,\norm h_{g(x)}^{2}\,,
\end{align*}
where in (i) we used Lemma~\ref{lem:usefulFactLewis}-1 and 3.

For the first part of the $\onu$-symmetry, Lemma~\ref{lem:helper4Diagonal}-2
implies that
\[
\max_{h:\norm h_{g(x)}=1}\norm{A_{x}h}_{\infty}=\sqrt{\max_{i}\frac{[\sigma(W_{x}^{1/2}A_{x})]_{i}}{[W_{x}]_{ii}}}\leq\sqrt{2m^{\frac{2}{p+2}}}\,,
\]
 and Lemma~\ref{lem:helper4Diagonal}-3 leads to $K\cap(2x-K)\subset\dcal_{g}^{\sqrt{d}}(x)$
due to
\[
\tr(W_{x})=\tr\bpar{W_{x}^{\half-\frac{1}{p}}A_{x}(A_{x}^{\T}W_{x}^{1-\frac{2}{p}}A_{x})^{-1}A_{x}^{\T}W_{x}^{\half-\frac{1}{p}}}=\tr\bpar{A_{x}^{\T}W_{x}^{1-\frac{2}{p}}A_{x}(A_{x}^{\T}W_{x}^{1-\frac{2}{p}}A_{x})^{-1}}=d\,.
\]
Therefore, $16p^{2}m^{\frac{2}{p+2}}A_{x}^{\T}W_{x}A_{x}$ is SSC
with $\mc O\bpar{dm^{\frac{2}{p+2}}}$-symmetry by Lemma~\ref{lem:symmforPolytope}.
By setting $p=\mc O(\log m)$, the claim follows.
\end{proof}

\subsubsection{Linear constraints: strongly lower trace self-concordance of Vaidya
\label{proof:linear-vaidya-SLTSC}}

Let $\theta_{1}(x):=A_{x}^{\T}\Sigma_{x}A_{x}$, $\theta_{2}(x):=A_{x}^{\T}A_{x}$,
and $\Gamma_{x}:=\Diag\bpar{A_{x}g(x)^{-1}A_{x}^{\T}}$. Recall $g=g_{1}+g_{2}$
for a PSD matrix function $g_{1}$ and the Vaidya metric $g_{2}$.
\begin{lem}
\label{lem:HybridGammaNorm} $\norm{\Gamma_{x}}_{\infty}\leq\frac{1}{44}$.
\end{lem}

\begin{proof}
For $\og_{2}:=\theta_{1}+\frac{d}{m}\theta_{2}=\frac{1}{44}\sqrt{\frac{d}{m}}g_{2}$,
it follows from $g^{-1}\preceq g_{2}^{-1}=\frac{1}{44}\sqrt{\frac{d}{m}}\og_{2}^{-1}$
that
\begin{align*}
44\norm{\Gamma_{x}}_{\infty} & \leq4\sqrt{\frac{d}{m}}\norm{\Diag(A_{x}\og_{2}^{-1}A_{x}^{\T})}_{\infty}=\sqrt{\frac{d}{m}}\max_{i\in[m]}\frac{\bbrack{\sigma\bpar{\sqrt{\Sigma_{x}+\frac{d}{m}I_{m}}A_{x}}}_{i}}{\bbrack{\Sigma_{x}+\frac{d}{m}I_{m}}_{ii}}\underset{\text{\eqref{eq:28-1}}}{\leq}1\,.\qedhere
\end{align*}
\end{proof}
Now we show SLTSC of the Vaidya metric:
\begin{proof}
[Proof of Lemma~\ref{lem:vaidya-SLTSC}]  As $\Dd^{2}\theta_{2}(x)[h,h]\succeq0$
by Claim~\ref{claim:diffLogBarrier}, we have 
\[
\tr\bpar{g^{-1}\Dd^{2}\theta_{2}(x)[h,h]}=\tr\bpar{g^{-\half}\Dd^{2}\theta_{2}(x)[h,h]g^{-\half}}\geq0\,.
\]
As for $\theta_{1}$, by Lemma~\ref{lem:calculusLeverage}-6 $\Dd^{2}\theta_{1}[h,h]\succeq-16A_{x}^{\T}\Diag(S_{x,h}P_{x}S_{x,h}P_{x})A_{x}-6A_{x}^{\T}\Diag(P_{x}S_{x,h}^{2}P_{x})A_{x}$,
so
\[
\tr\bpar{g^{-1}\Dd^{2}\theta_{1}(x)[h,h]}\geq-16\tr(\Gamma_{x}S_{x,h}P_{x}S_{x,h}P_{x})-6\tr(\Gamma_{x}P_{x}S_{x,h}^{2}P_{x})\,.
\]
We first note that $\tr(S_{x,h}P_{x}S_{x,h})=s_{x,h}^{\T}(P_{x}\circ I)s_{x,h}=s_{x,h}^{\T}\Sigma_{x}s_{x,h}=\norm h_{\theta_{1}}^{2}$.
Using this,
\begin{align*}
\tr(\Gamma_{x}S_{x,h}P_{x}S_{x,h}P_{x}) & =\tr(\Gamma_{x}^{1/2}S_{x,h}P_{x}\cdot S_{x,h}P_{x}\Gamma_{x}^{1/2})\leq\sqrt{\tr(\Gamma_{x}^{\half}S_{x,h}P_{x}^{2}S_{x,h}\Gamma_{x}^{\half})\,\tr(\Gamma_{x}^{\half}P_{x}S_{x,h}^{2}P_{x}\Gamma_{x}^{\half})}\\
 & =\sqrt{\tr(P_{x}S_{x,h}\Gamma_{x}S_{x,h}P_{x})}\sqrt{\tr(S_{x,h}P_{x}\Gamma_{x}P_{x}S_{x,h})}=\norm{\Gamma_{x}}_{\infty}\norm h_{\theta_{1}}^{2}\,,\\
\tr(\Gamma_{x}P_{x}S_{x,h}^{2}P_{x}) & =\tr(S_{x,h}P_{x}\Gamma_{x}P_{x}S_{x,h})\leq\norm{\Gamma_{x}}_{\infty}\tr(S_{x,h}P_{x}S_{x,h})\underset{\text{(i)}}{=}\norm{\Gamma_{x}}_{\infty}\norm h_{\theta_{1}}^{2}\,.
\end{align*}
Putting these together and using Lemma~\ref{lem:HybridGammaNorm},
\[
\tr\bpar{g^{-1}\Dd^{2}\theta_{1}(x)[h,h]}\geq-22\norm{\Gamma_{x}}_{\infty}\norm h_{\theta_{1}}^{2}\geq-\half\,\norm h_{\theta_{1}}^{2}\,,
\]
and it follows from $g_{2}=44\sqrt{\frac{m}{d}}\Par{\theta_{1}+\frac{d}{m}\theta_{2}}$
that $\tr\bpar{g^{-1}\Dd^{2}g_{2}(x)[h,h]}\geq-\half\,\norm h_{g_{2}}^{2}$.
\end{proof}

\subsubsection{Linear constraints: strongly lower trace self-concordance of Lewis-weight
\label{proof:linear-Lewis-SLTSC}}

For $\theta(x):=A_{x}^{\T}W_{x}A_{x}$ (i.e., the unscaled version
of $g_{2}$), we write $g_{2}=c\cdot\theta$ for a constant $c$,
which will be set to $c_{1}(\log m)^{c_{2}}\sqrt{d}$ for some constants
$c_{1},c_{2}>0$ later. Going forward, $P_{x}$ indicates the projection
matrix of $W_{x}^{\nicefrac{1}{2}-\nicefrac{1}{p}}A_{x}$ (i.e., $P_{x}=P(W_{x}^{\nicefrac{1}{2}-\nicefrac{1}{p}}A_{x})$).
\begin{lem}
\label{lem:GammaNormLSMetric}$\norm{\Gamma_{x}}_{\infty}\leq2c^{-1}m^{\frac{2}{p+2}}$.
\end{lem}

\begin{proof}
Note that $0\preceq\Gamma_{x}=\Diag(A_{x}g^{-1}A_{x}^{\T})\preceq c^{-1}\Diag(A_{x}\theta^{-1}A_{x}^{\T})$.
By Lemma~\ref{lem:usefulFactLewis}-1,
\[
\norm{\Diag(A_{x}\theta^{-1}A_{x}^{\T})}_{\infty}=\max_{i\in[m]}\frac{\bbrack{\sigma\bpar{W_{x}^{1/2}A_{x}}}_{i}}{\bbrack{W_{x}}_{ii}}\leq2m^{\frac{2}{p+2}}\,.\qedhere
\]
\end{proof}
Now we show SLTSC of the Lewis-weight metric:
\begin{proof}
[Proof of Lemma~\ref{lem:Lw-SLTSC}] From \eqref{eq:LW-second-derv},
$\Dd^{2}\theta[h,h]\succeq-4A_{x}^{\T}W_{x,h}'S_{x,h}A_{x}+A_{x}^{\T}W_{x,h}''A_{x}$.
Thus,
\[
\tr(g^{-1}\Dd^{2}\theta[h,h])\geq\tr\bpar{\Gamma_{x}(W_{x,h}''-4W_{x,h}'S_{x,h})}=-4\tr(\Gamma_{x}W_{x,h}'S_{x,h})+\tr(\Gamma_{x}W_{x,h}'')\,.
\]
As for the first term, $\tr(\Gamma_{x}W_{x,h}'S_{x,h})\leq p\,\norm{\Gamma_{x}}_{\infty}\norm h_{\theta}^{2}$
follows from \eqref{eq:trSW} with $\Gamma_{x}$ replacing $s_{x,h}^{2}$.

As for the second term $\tr(\Gamma_{x}W_{x,h}'')$ (i.e., \eqref{eq:trGamma}
with $\Gamma=\Gamma_{x}$), each term there is of the form $\tr(\Gamma_{x}\Diag(v))$
for $v\in\R^{m}$, which can be bounded as follows:
\begin{align*}
\big|\tr\bpar{\Gamma_{x}\Diag(v)}\big| & =|\tr(\Gamma_{x}W_{x}^{\half}W_{x}^{-\half}\Diag(v))|\leq\sqrt{\tr(W_{x}^{\half}\Gamma_{x}^{2}W_{x}^{\half})}\sqrt{\tr\bpar{\Diag(v)W_{x}^{-1}\Diag(v)}}\\
 & \leq\norm{\Gamma_{x}}_{\infty}\sqrt{\tr(W_{x})}\norm v_{W_{x}^{-1}}=\sqrt{d}\norm{\Gamma_{x}}_{\infty}\norm v_{W_{x}^{-1}}\,.
\end{align*}
Then, we obtain $|\tr(\Gamma_{x}W_{x,h}'')|\lesssim\sqrt{d}\norm{\Gamma_{x}}_{\infty}\norm h_{\theta}^{2}$
for $p=\O(\log m)$ by using this inequality together with the norm
bounds in Lemma~\ref{lem:second-deriv-Lewis}. 

Putting things together, we conclude that
\begin{align*}
\tr(g^{-1}\Dd^{2}\theta[h,h]) & \gtrsim-p\norm{\Gamma_{x}}_{\infty}\norm h_{\theta}^{2}-\sqrt{d}\norm{\Gamma_{x}}_{\infty}\norm h_{\theta}^{2}\gtrsim-c^{-1}\sqrt{d}\norm h_{\theta}^{2}\,,
\end{align*}
where the last line follows from Lemma~\ref{lem:GammaNormLSMetric}.
Therefore, there exists positive constants $d_{1}$ and $d_{2}$ such
that $\tr(g^{-1}\Dd^{2}\theta[h,h])\geq-c^{-1}d_{1}(\log m)^{d_{2}}\sqrt{d}\norm h_{\theta}^{2}$,
which implies
\[
\tr(g^{-1}\Dd^{2}g_{2}[h,h])\geq-c^{-1}d_{1}(\log m)^{d_{2}}\sqrt{d}\norm h_{g_{2}}^{2}\,.
\]
By taking $c=d_{1}(\log m)^{d_{2}}\sqrt{d}$, the metric $g_{2}=c\theta=d_{1}(\log m)^{d_{2}}\sqrt{d}A_{x}^{\T}W_{x}A_{x}$
is SLTSC. 
\end{proof}

\subsubsection{Linear constraints: strongly average self-concordance}

We proceed with a general form of the metric $g(x)=A_{x}^{\T}D_{x}A_{x}$
with a diagonal matrix $0\prec D_{x}\in\R^{m}$. Then we provide computational
lemmas used when proving SASC of barriers for the linear constraints.

We pick any $g':\intk\to\psd$ such that $\bar{g}:=g+g'\succ0$. By
affine invariance, we may assume $\bar{g}(x)=I$ and $x=0$. Note
that $g(x)\preceq I_{d}$, and $z$ equals $rh/\sqrt{d}$ for $h\sim\mc N(0,I_{d})$
in law. Applying Taylor's expansion to $\norm{z-x}_{g(z)}^{2}$ at
$z=x$ (as in the proof of Lemma~\ref{lem:hsc-to-sasc}), for some
$p_{z}\in[x,z]$
\begin{align*}
\big|\norm{z-x}_{g(z)}^{2}-\norm{z-x}_{g(x)}^{2}\big| & \leq\frac{r^{2}}{d}\,\Bpar{\frac{r}{\sqrt{d}}\underbrace{|\Dd g(x)[h^{\otimes3}]|}_{\eqqcolon\textsf{A}}+\frac{r^{2}}{2d}\underbrace{|\Dd^{2}g(p_{z})[h^{\otimes4}]|}_{\eqqcolon\textsf{B}}}\,.
\end{align*}
It suffices to show that $|\Dd g(x)[h^{\otimes3}]|=\mc O(d^{1/2})$
and $|\Dd^{2}g(p_{z})[h^{\otimes4}]|=\mc O(d)$ with high probability.

\paragraph{Term $\textsf{A}$.}

By \eqref{eq:Dgh}, we have $\Dd g(x)[h^{\otimes3}]=-2s_{x,h}^{\T}D_{x}S_{x,h}s_{x,h}+s_{x,h}^{\T}D_{x,h}'s_{x,h}$.
Let $a_{i}$ denote the $i$-th row of $A_{x}$ for $i\in[m]$, and
define two polynomials in $h$ as follows:
\begin{equation}
P_{1}(h):=s_{x,h}^{\T}D_{x}S_{x,h}s_{x,h}=\tr(D_{x}S_{x,h}^{3})=\sum_{i=1}^{m}d_{i}\,(a_{i}^{\T}h)^{3}\,,\quad\text{and}\quad P_{2}(h):=s_{x,h}^{\T}D_{x,h}'s_{x,h}\,.\label{eq:P12}
\end{equation}
By Lemma~\ref{lem:matrix-projection}, $D_{x}^{1/2}A_{x}A_{x}^{\T}D_{x}^{1/2}\preceq P(D_{x}^{1/2}A_{x})$
and thus
\begin{equation}
\max_{i\in[m]}\norm{a_{i}}^{2}=\norm{\Diag(A_{x}A_{x}^{\T})}_{\infty}\leq\max_{i}\frac{[\sigma(D_{x}^{1/2}A_{x})]_{i}}{[D_{x}]_{ii}}\,.\label{eq:max-ai}
\end{equation}
By Lemma~\ref{lem:variance-1},
\begin{align}
\E[P_{1}(h)^{2}] & =\E\Bbrack{\Bbrace{\sum_{i=1}^{m}d_{i}(a_{i}\cdot h)^{3}}^{2}}=9\sum_{i,j=1}^{m}\norm{d_{i}^{1/3}a_{i}}^{2}\norm{d_{j}^{1/3}a_{j}}^{2}\inner{d_{i}^{1/3}a_{i},d_{j}^{1/3}a_{j}}+6\sum_{i,j}\inner{d_{i}^{1/3}a_{i},d_{j}^{1/3}a_{j}}^{3}\nonumber \\
 & =9\cdot1^{\T}\Diag(A_{x}A_{x}^{\T})\,D_{x}^{1/2}\underbrace{D_{x}^{1/2}A_{x}A_{x}^{\T}D_{x}^{1/2}}_{\preceq P(D_{x}^{1/2}A_{x})\preceq I_{m}}D_{x}^{1/2}\,\Diag(A_{x}A_{x}^{\T})\,1+6\sum_{i,j}d_{i}d_{j}(a_{i}\cdot a_{j})^{3}\nonumber \\
 & \lesssim\norm{\Diag(A_{x}A_{x}^{\T})}_{\infty}\,\tr\bpar{\Diag(A_{x}A_{x}^{\T})\,D_{x}}+\max_{i}\norm{a_{i}}^{2}\cdot\sum_{i,j}d_{i}d_{j}(a_{i}\cdot a_{j})^{2}\nonumber \\
 & =\max_{i}\norm{a_{i}}^{2}\,\tr(A_{x}^{\T}D_{x}A_{x})+\max_{i}\norm{a_{i}}^{2}\cdot\sum_{j}\tr(d_{j}a_{j}^{\T}A_{x}^{\T}D_{x}A_{x}a_{j})\nonumber \\
 & \underset{\text{(i)}}{\leq}2\max_{i}\norm{a_{i}}^{2}\,\tr(A_{x}^{\T}D_{x}A_{x})\leq2d\,\max_{i}\norm{a_{i}}^{2}\,,\label{eq:P1_bound}
\end{align}
where (i) follows from $A_{x}^{\T}D_{x}A_{x}\preceq I_{d}$ and $\sum_{j}\tr(d_{j}a_{j}^{\T}A_{x}^{\T}D_{x}A_{x}a_{j})\leq\sum_{j}\tr(d_{j}a_{j}^{\T}a_{j})=\tr(A_{x}^{\T}D_{x}A_{x})$.

Another polynomial $P_{2}(h)$ requires a different strategy for bounding
$\E[P_{2}(h)^{2}]$ for each barrier. This polynomial vanishes for
the log-barrier, while the Vaidya and Lewis-weight metrics requires
rather involved tasks for bounding $\E[P_{2}(h)^{2}]$.

\paragraph{Term $\textsf{B}$.}

Due to \eqref{eq:LW-fourth-moment} (with $W_{x}$ replaced by $D_{x}$),
$|\Dd^{2}g(p_{z})[h^{\otimes4}]|$ consists of three polynomials:
\begin{equation}
\bar{P}_{3}(h):=\tr(D_{p_{z}}S_{p_{z},h}^{4})\,,\quad\bar{P}_{4}(h)=\tr(D_{p_{z},h}'S_{p_{z},h}^{2})\,,\quad\bar{P}_{5}(h)=\tr(D_{p_{z},h}''S_{p_{z},h}^{2})\,.\label{eq:P345}
\end{equation}
For each $i=3,4,5$, we define $P_{i}(h)$ by $\bar{P}_{i}(h)$ with
$p_{z}$ replaced by $x$. For the log-barrier, $\bar{P}_{3}(h)$
only matters since $D_{(\cdot)}=I_{m}$. For the Vaidya metric, $\bar{P}_{4}(h)$
and $\bar{P}_{5}(h)$ can be bounded by multiples of $\bar{P}_{3}(h)$.
For the Lewis-weight metric, each $\bar{P}_{i}$ requires a different
procedure for bounding $\E[\bar{P}_{i}(h)^{2}]$. Moreover, we can
show $\bar{P}_{i}(h)\lesssim P_{i}(h)$ and
\begin{align}
\E[P_{3}(h)^{2}] & =\sum_{i,j\in[m]}\E[d_{i}d_{j}\,(a_{i}\cdot h)^{4}(a_{j}\cdot h)^{4}]\underset{\textup{CS}}{\leq}\sum_{i,j}d_{i}d_{j}\sqrt{\E[(a_{i}\cdot h)^{8}]}\sqrt{\E[(a_{j}\cdot h)^{8}]}\nonumber \\
 & \underset{\text{(i)}}{\lesssim}\Bpar{\sum_{i}d_{i}\norm{a_{i}}^{4}}^{2}\leq\max_{i}\norm{a_{i}}^{4}\,\Bpar{\sum_{i}d_{i}\norm{a_{i}}^{2}}^{2}\underset{\text{(ii)}}{\leq}d^{2}\max_{i}\norm{a_{i}}^{4}\,,\label{eq:P3_bound}
\end{align}
where we used $a_{i}\cdot h\sim\ncal(0,\norm{a_{i}}^{2})$ in (i),
and $\sum_{i}d_{i}\norm{a_{i}}^{2}=\tr(A_{x}^{\T}D_{x}A_{x})\leq\tr(I_{d})$
in (ii).

We now show SASC of the three barriers for linear constraints, using
this proof outline.

\paragraph{SASC of log-barriers.\label{proof:linear-SASC-log}}
\begin{proof}
[Proof of Lemma~\ref{lem:logBarrier-SASC}] Set $g(x)=A_{x}^{\T}A_{x}$
(with $D_{x}=I_{m}$). By \eqref{eq:max-ai}, 
\[
\max_{i\in[m]}\norm{a_{i}}^{2}\leq\max[\sigma(A_{x}^{1/2})]_{i}\leq1\,.
\]

As for the term $\msf A$, it suffices to bound $P_{1}(h)=\tr(S_{x,h}^{3})$.
Since $\E[P_{1}(h)^{2}]\lesssim d$ by \eqref{eq:P1_bound}, by Lemma~\ref{lem:conc-gaussian-poly}
with $t=(2e)^{3/2}\vee\bpar{\frac{2e}{3}\log\frac{2}{\veps}}^{3/2}$
and $r_{1}(\veps):=\veps(2\sqrt{60}t)^{-1}$, we have that for any
$r\leq r_{1}(\veps)$,
\[
\text{Event }B_{1}:\quad\P_{h}\Bpar{\frac{r}{\sqrt{d}}\,|P_{1}(h)|\geq\veps}\leq\veps\,.
\]

As for the term $\msf B$, recall $\P_{z}\bpar{\norm z\geq-r\cdot2\log\veps}\leq\veps$
and call this event $B_{2}$. We take $r_{2}(\veps)$ so that $1-2r_{2}\log\veps\leq1.1$,
which ensures $\norm z\leq2r$ conditioned on $B_{2}^{c}$ for $r\leq r_{2}$.
Next, we establish coordinate-wise closeness of $s_{x}$ at close-by
points. Let $x_{t}=x+\frac{tr}{\sqrt{d}}h$, and $s_{t}=Ax_{t}-b$.
For $t\in[0,1]$,
\begin{align*}
\left\Vert S_{0}^{-1}\,\frac{\D s_{t}}{\D t}\right\Vert _{\infty} & =\frac{r}{\sqrt{d}}\,\norm{A_{x}h}_{\infty}\leq\frac{r}{\sqrt{d}}\,\norm h_{g(x)}\leq\frac{r}{\sqrt{d}}\,\norm h=\norm z\,,
\end{align*}
and conditioned on $z\in B_{2}^{c}$ we know $\norm z\leq2r\log\frac{1}{\veps}\leq0.1$
for $r\leq r_{2}$. Hence,
\[
\max_{i\in[m]}\Big|\frac{s_{p,i}-s_{x,i}}{s_{x,i}}\Big|\leq\int_{0}^{1}\left\Vert S_{0}^{-1}\,\frac{\D s_{t}}{\D t}\right\Vert _{\infty}\,\D t\leq0.1\,,
\]
and thus $1.2\geq s_{x,i}/s_{p,i}\geq0.9$ for all $i\in[m]$ (i.e.,
$S_{p}^{-1}\preceq1.2S_{x}^{-1}$).

Using this, we bound $\bar{P}_{3}(h)=\tr(S_{p,h}^{4})$ by a multiple
of $P_{3}(h)=\tr(S_{x,h}^{4})$ as follows:
\begin{align*}
\tr(S_{p,h}^{4}) & =\tr(h^{\T}A^{\T}S_{p,h}S_{p}^{-2}S_{p,h}Ah)\leq2\tr(h^{\T}A^{\T}S_{p,h}S_{x}^{-2}S_{p,h}Ah)=2\tr(S_{x,h}^{2}S_{p,h}^{2})\leq4\tr(S_{x,h}^{4})\,.
\end{align*}
Hence, $\E[\bar{P}_{3}(h)^{2}]\lesssim\E[P_{3}(h)^{2}]\lesssim d^{2}$
by \eqref{eq:P3_bound}. Using Lemma~\ref{lem:conc-gaussian-poly}
with $t=(2e)^{2}\vee\bpar{\frac{2e}{4}\log\frac{2}{\veps}}^{3/2}$
and taking $r_{3}(\veps):=(\nicefrac{\veps}{c_{1}t})^{1/2}$, we obtain
\[
\text{Event }B_{3}:\quad\P\Bpar{\frac{r^{2}}{2d}\cdot16\bar{P}_{3}(h)\geq\veps}\geq\veps\,,
\]

Combining bounds on $\msf A$ and $\msf B$ conditioned on $\cap_{i}B_{i}^{c}$,
we have with probability at least $1-3\veps$
\[
\big|\norm{z-x}_{g(z)}^{2}-\norm{z-x}_{g(x)}^{2}\big|\leq2\veps\frac{r^{2}}{d}\quad\text{for any }r\leq\min_{i}r_{i}(\veps)\,.
\]
By replacing $3\veps\gets\veps$, the claim follows.
\end{proof}

\paragraph{SASC of Vaidya metric.\label{proof:linear-SASC-vaidya}}
\begin{proof}
[Proof of Lemma~\ref{lem:vaidya-SASC}] Set $g(x)=A_{x}^{\T}D_{x}A_{x}$
with $D_{x}=\sqrt{\frac{m}{d}}(\Sigma_{x}+\frac{d}{m}I_{m})$. By
\eqref{eq:max-ai} and \eqref{eq:28-1},
\[
\max_{i\in[m]}\norm{a_{i}}^{2}\leq\max_{i}\frac{[\sigma(D_{x}^{1/2}A_{x})]_{i}}{[D_{x}]_{ii}}\leq1\,.
\]

\paragraph{Term $\textsf{A}$.}

As $\msf A$ consists of $P_{1}$ and $P_{2}$ (see \eqref{eq:P12}),
we show $\E[P_{i}(h)^{2}]\lesssim d$ for $i\in[2]$, which by Lemma~\ref{lem:conc-gaussian-poly}
implies $|\msf A|\leq\sqrt{d}$ w.h.p. As for $P_{1}(h)=\tr(D_{x}S_{x,h}^{3})$,
we have $\E[P_{1}(h)]^{2}\lesssim d$ from \eqref{eq:P1_bound}.

As for $P_{2}(h)=\tr(D_{x,h}'S_{x,h}^{2})$, our approach is similar
to \citet{chen2018fast}. By Lemma~\ref{lem:calculusLeverage}, 
\begin{align*}
|P_{2}(h)| & =\Big|\sqrt{\frac{m}{d}}\tr\Bpar{\Diag\bpar{(\Sigma_{x}-P_{x}^{(2)})\,s_{x,h}}\,S_{x,h}^{2}}\Big|\leq|P_{1}(h)|+|\tr(S_{x,h}^{3})|+\sqrt{\frac{m}{d}}\,\big|\tr\bpar{\Diag(P_{x}^{(2)}s_{x,h})\,S_{x,h}^{2}}\big|\,.
\end{align*}
Since we already established a high-probability bound for both $|P_{1}(h)|$
and $|\tr(S_{x,h}^{3})|$ (which is $P_{1}(h)$ for the log-barrier),
we focus on the third term in the RHS.

For $\sigma_{x}:=\diag\Par{P_{x}}$ and $\sigma_{x,i,j}:=(P_{x})_{ij}$,
it follows from $P_{x}^{2}=P_{x}$ that $\sigma_{x,i}=\sum_{j}\sigma_{x,i,j}^{2}$.
Hence,
\begin{align*}
\tr(\Sigma_{x}S_{x,h}^{3}) & =1^{\T}\Sigma_{x}s_{x,h}^{3}=\sum_{i}(s_{x,h})_{i}^{3}\sigma_{x,i}=\sum_{i,j=1}^{m}\sigma_{x,i,j}^{2}(s_{x,h})_{i}^{3}\,,\\
\tr\bpar{\Diag(P_{x}^{(2)}s_{x,h})\,S_{x,h}^{2}} & =\sum_{i,j=1}^{m}\sigma_{x,i,j}^{2}(s_{x,h})_{i}^{2}(s_{x,h})_{j}\underset{\text{symmetry}}{=}\sum_{i,j=1}^{m}\sigma_{x,i,j}^{2}(s_{x,h})_{j}^{2}(s_{x,h})_{i}\,.
\end{align*}
Combining these leads to
\begin{align*}
 & 2\,\tr(\Sigma_{x}S_{x,h}^{3})+6\,\tr\bpar{\Diag(P_{x}^{(2)}s_{x,h})S_{x,h}^{2}}\\
 & =\sum_{i,j=1}^{m}\sigma_{x,i,j}^{2}\bpar{(s_{x,h})_{i}^{3}+3(s_{x,h})_{i}^{2}(s_{x,h})_{j}+3(s_{x,h})_{i}(s_{x,h})_{j}^{2}+(s_{x,h})_{j}^{3}}=\sum_{i,j=1}^{m}\sigma_{x,i,j}^{2}\bpar{(s_{x,h})_{i}+(s_{x,h})_{j}}^{3}\,,
\end{align*}
so we handle $\sum_{i,j}\sigma_{x,i,j}^{2}\bpar{(s_{x,h})_{i}+(s_{x,h})_{j}}^{3}$
instead of $\tr\bpar{\Diag(P_{x}^{(2)}s_{x,h})S_{x,h}^{2}}$, as we
already bounded $\sqrt{\frac{m}{d}}\tr(\Sigma_{x}S_{x,h}^{3})=P_{1}(h)-\sqrt{\frac{d}{m}}\tr(S_{x,h}^{3})$.
Due to $(s_{x,h})_{i}+(s_{x,h})_{j}=(a_{i}+a_{j})^{\T}h$, for $c_{ij}:=a_{i}+a_{j}$
\begin{align}
 & \E\Bbrack{\Bbrace{\sum_{i,j\in[m]}\sigma_{x,i,j}^{2}\bpar{(s_{x,h})_{i}+(s_{x,h})_{j}}^{3}}^{2}}=\sum_{i,j,k,l}\sigma_{x,i,j}^{2}\sigma_{x,k,l}^{2}\E[(c_{ij}\cdot h)^{3}(c_{kl}\cdot h)^{3}]\nonumber \\
 & \underset{\text{Lemma \ref{lem:variance-1}}}{=}9\sum_{i,j,k,l}\sigma_{x,i,j}^{2}\sigma_{x,k,l}^{2}\norm{c_{ij}}^{2}\norm{c_{kl}}^{2}(c_{ij}\cdot c_{kl})+6\sum_{i,j,k,l}\sigma_{x,i,j}^{2}\sigma_{x,k,l}^{2}(c_{ij}\cdot c_{kl})^{3}\,.\label{eq:vaidya-cubic-expansion}
\end{align}

As for the first term in \eqref{eq:vaidya-cubic-expansion}, we denote
$z_{i}:=\sum_{j}\sigma_{x,i,j}^{2}\|c_{ij}\|^{2}$ and $Z:=\Diag\bpar{(z_{i})_{i\in[m]}}$.
Then,
\begin{align}
 & \sum_{i,j,k,l}\sigma_{x,i,j}^{2}\sigma_{x,k,l}^{2}\norm{c_{ij}}^{2}\norm{c_{kl}}^{2}(c_{ij}\cdot c_{kl})=\Bnorm{\sum_{ij}\sigma_{x,i,j}^{2}\norm{c_{ij}}^{2}c_{ij}}^{2}\nonumber \\
 & \leq2\Bnorm{\sum_{ij}\sigma_{x,i,j}^{2}\norm{c_{ij}}^{2}a_{i}}^{2}+2\Bnorm{\sum_{ij}\sigma_{x,i,j}^{2}\norm{c_{ij}}^{2}a_{j}}^{2}=4\Bnorm{\sum_{ij}\sigma_{x,i,j}^{2}\norm{c_{ij}}^{2}a_{i}}^{2}=\Bnorm{\sum_{i}z_{i}a_{i}}^{2}\nonumber \\
 & =1^{\T}ZA_{x}A_{x}^{\T}Z\,1\le1^{\T}ZD_{x}^{-1/2}P(D_{x}^{1/2}A_{x})\,D_{x}^{-1/2}Z\,1\le1^{\T}ZD_{x}^{-1}Z\,1\lesssim\sqrt{\frac{d}{m}}\,\tr(Z)\,,\label{eq:trZ-bound}
\end{align}
where the last inequality follows from $Z\precsim\Sigma_{x}\preceq\sqrt{\frac{d}{m}}D_{x}$
due to
\begin{align*}
z_{i} & \leq2\sum_{j}\sigma_{x,i,j}^{2}(\staticnorm{a_{i}}^{2}+\staticnorm{a_{j}}^{2})\lesssim\underbrace{\sigma_{x,i}\staticnorm{a_{i}}^{2}+\sum_{j}\sigma_{x,i,j}^{2}\staticnorm{a_{j}}^{2}}_{\eqqcolon\msf K_{i}}\leq\sigma_{x,i}\|a_{i}\|^{2}+\sigma_{x,i}\lesssim\sigma_{x,i}\,.
\end{align*}
Moreover, using the bound in $\msf K_{i}$ and $\sum_{i,j}\sigma_{x,i,j}^{2}\snorm{a_{j}}^{2}=\sum_{j}\sigma_{x,i}\norm{a_{i}}^{2}$
\[
\tr(Z)\lesssim\sum_{i}(\sigma_{x,i}\staticnorm{a_{i}}^{2}+\sum_{j}\sigma_{x,i,j}^{2}\staticnorm{a_{j}}^{2})=2\tr(A_{x}^{\T}\Sigma_{x}A_{x})\lesssim\sqrt{\frac{d}{m}}\,\tr(A_{x}^{\T}D_{x}A_{x})\leq d\sqrt{\frac{d}{m}}\,.
\]
Putting this into \eqref{eq:trZ-bound}, we obtain $\sum_{i,j,k,l}\sigma_{x,i,j}^{2}\sigma_{x,k,l}^{2}\norm{c_{ij}}^{2}\norm{c_{kl}}^{2}(c_{ij}\cdot c_{kl})\lesssim d^{2}/m$.

As for the second term in \eqref{eq:vaidya-cubic-expansion},
\begin{align*}
 & \sum_{i,j,k,l}\sigma_{x,i,j}^{2}\sigma_{x,k,l}^{2}\,(c_{ij}\cdot c_{kl})^{3}\lesssim\sum_{i,j,k,l}\sigma_{x,i,j}^{2}\sigma_{x,k,l}^{2}\,|c_{ij}\cdot c_{kl}|^{2}\\
 & \leq\sum_{i,j,k,l}\sigma_{x,i,j}^{2}\sigma_{x,k,l}^{2}\,(a_{i}\cdot a_{k}+a_{i}\cdot a_{l}+a_{j}\cdot a_{k}+a_{j}\cdot a_{l})^{2}\lesssim\sum_{i,j,k,l}\sigma_{x,i,j}^{2}\sigma_{x,k,l}^{2}\,(a_{i}\cdot a_{k})^{2}\\
 & =\sum_{ik}\sigma_{i}\sigma_{k}\,(a_{i}\cdot a_{k})^{2}=\sum_{k}\tr(\sigma_{k}a_{k}^{\T}A_{x}^{\T}\Sigma_{x}A_{x}a_{k})\leq\sqrt{\frac{d}{m}}\sum_{k}\tr(\sigma_{k}a_{k}^{\T}a_{k})\\
 & =\sqrt{\frac{d}{m}}\,\tr(A_{x}^{\T}\Sigma_{x}A_{x})\le\frac{d^{2}}{m}\,.
\end{align*}
This establish a high-probability bound of $\O(d^{2}/m)$ on \eqref{eq:vaidya-cubic-expansion},
implying an $\O(\sqrt{d})$-high-probability bound on $\sqrt{\frac{m}{d}}\big|\tr\bpar{\Diag(P_{x}^{(2)}s_{x,h})\,S_{x,h}^{2}}\big|$.

\paragraph{Term $\textsf{B}$.}

We show that $s_{x}$ and $s_{p_{z}}$ are close, and the same holds
for $\sigma_{x}$ and $\sigma_{p_{z}}$. For $s_{x}$, following the
argument for the log-barrier, we let $x_{t}:=x+th\frac{r}{\sqrt{d}}$
and $s_{t}:=Ax_{t}-b$. For $0\leq t\leq1$,
\begin{align*}
\Bnorm{S_{0}^{-1}\deriv{s_{t}}t}_{\infty} & =\frac{r}{\sqrt{d}}\,\norm{A_{x}h}_{\infty}\underset{\eqref{eq:28-1}}{\leq}\frac{r}{\sqrt{d}}\,\norm h_{A_{x}^{\T}D_{x}A_{x}}\leq\frac{r}{\sqrt{d}}\,\norm h=\norm z\,.
\end{align*}
Conditioned on the high-probability bound of $\norm z\leq2r\log\frac{1}{\veps}\leq0.1$
for any $r$ less than some $r(\veps)$,
\[
\max_{i\in[m]}\Big|\frac{s_{p,i}-s_{x,i}}{s_{x,i}}\Big|\leq\int_{0}^{1}\Bnorm{S_{0}^{-1}\deriv{s_{t}}t}_{\infty}\,\D t\leq0.1\,,
\]
and thus $1.2\geq s_{x,i}/s_{p,i}\geq0.9$ for all $i\in[m]$ (i.e.,
$S_{p}^{-1}\preceq1.2S_{x}^{-1}$). For $\sigma_{x}$, as we have
$\Sigma_{x}=\Diag(A_{x}(A_{x}^{\T}A_{x})^{-1}A_{x}^{\T})$, we have
the same closeness between $\sigma_{x,i}$ and $\sigma_{p,i}$ for
each $i\in[m]$.

Using the formulas in Lemma~\ref{lem:calculusLeverage},
\begin{align*}
|\Dd^{2}g(p)[h^{\otimes4}]| & \lesssim\sqrt{\frac{m}{d}}\,\Bigl(\tr\bpar{(\Sigma_{p}+\frac{d}{m}I_{m})S_{p,h}^{4}}+\underbrace{\tr(S_{p,h}^{2}P_{p}S_{p,h}P_{p}S_{p,h})}_{(*)}\\
 & \qquad\qquad\qquad+\tr(S_{p,h}^{2}P_{p}S_{p,h}^{2}P_{p})+\underbrace{\tr(S_{p,h}P_{p}S_{p,h}P_{p}S_{p,h}P_{p}S_{p,h})}_{\leq\tr(S_{p,h}^{2}P_{p}S_{p,h}^{2}P_{p})}\Bigr)\\
 & \underset{\text{(i)}}{\lesssim}\sqrt{\frac{m}{d}}\,\Bpar{\tr\Bpar{(\Sigma_{p}+\frac{d}{m}I_{m})S_{p,h}^{4}}+\tr(S_{p,h}^{2}\Sigma_{p}S_{p,h}^{2})+\underbrace{\tr(S_{p,h}^{2}P_{p}S_{p,h}^{2}P_{p})}_{\text{Use Lemma \ref{lem:Kronecker}}}}\\
 & \underset{\text{(ii)}}{\lesssim}\sqrt{\frac{m}{d}}\,\tr\Bpar{(\Sigma_{p}+\frac{d}{m}I_{m})S_{p,h}^{4}}\underset{\text{(iii)}}{\lesssim}\sqrt{\frac{m}{d}}\,\tr\Bpar{(\Sigma_{x}+\frac{d}{m}I_{m})S_{x,h}^{4}}=P_{3}(h)\,,
\end{align*}
where in (i) we used the Cauchy-Schwarz inequality on $(*)$:
\begin{align*}
 & \tr(S_{p,h}^{2}P_{p}S_{p,h}P_{p}S_{p,h})\leq\sqrt{\tr(S_{p,h}^{2}P_{p}^{2}S_{p,h}^{2})}\sqrt{\tr(S_{p,h}P_{p}S_{p,h}^{2}P_{p}S_{p,h})}\\
 & \underset{\text{AM-GM}}{\leq}\half\bpar{\tr(S_{p,h}^{2}P_{p}^{2}S_{p,h}^{2})+\tr(S_{p,h}P_{p}S_{p,h}^{2}P_{p}S_{p,h})}\leq\half\,\bpar{\tr(S_{p,h}^{2}\Sigma_{p}S_{p,h}^{2})+\tr(S_{p,h}^{2}P_{p}S_{p,h}^{2}P_{p})}\,,
\end{align*}
(ii) follows from $\tr(S_{p,h}^{2}P_{p}S_{p,h}^{2}P_{p})=s_{p,h}^{2}\cdot P_{p}^{(2)}s_{p,h}^{2}\preceq s_{p,h}^{2}\cdot\Sigma_{p}s_{p,h}^{2}\preceq s_{p,h}^{2}\cdot(\Sigma_{p}+\frac{d}{m}I_{m})s_{p,h}^{2}$,
and in (iii) we used coordinate-wise closeness of $s_{x}\leftrightarrow s_{p}$
and $\sigma_{x}\leftrightarrow\sigma_{p}$. By \eqref{eq:P3_bound},
$\E[P_{3}(h)^{2}]\lesssim d^{2}$, and an $\O(d)$-high-probability
bound on $|P_{3}(h)|$ (so on $\msf B$) follows from Lemma~\ref{lem:conc-gaussian-poly}.
\end{proof}

\paragraph{SASC of Lewis-weight. \label{proof:linear-SASC-Lw}}
\begin{proof}
[Proof of Lemma~\ref{lem:Lw-SASC}] Set $g(x)=\sqrt{d}A_{x}^{\T}W_{x}A_{x}$
(with $D_{x}=\sqrt{d}\,W_{x}$). By \eqref{eq:max-ai} and Lemma~\ref{lem:usefulFactLewis}-1,
\[
\max_{i\in[m]}\norm{a_{i}}^{2}\leq\max_{i}\frac{[\sigma(D_{x}^{1/2}A_{x})]_{i}}{[D_{x}]_{ii}}\leq\frac{2m^{\frac{2}{p+2}}}{\sqrt{d}}\lesssim\frac{1}{\sqrt{d}}\,.
\]

\paragraph{Term $\textsf{A}$.}

As done for the Vaidya metric, a high-probability bound on $\msf A$
requires $\E[P_{i}(h)^{2}]\lesssim d$ for $i=1,2$ (see \eqref{eq:P12}).
Note that $\E[P_{1}(h)^{2}]\lesssim\sqrt{d}$ by \eqref{eq:P1_bound}.

As for $P_{2}(h)=\sqrt{d}\,s_{x,h}^{\T}W_{x,h}'s_{x,h}$, we show
$\E[P_{2}(h)^{2}]\lesssim\sqrt{d}$. Due to $W_{x,h}'=-\Diag(W_{x}^{\half}N_{x}W_{x}^{\half}s_{x,h})$
(Lemma~\ref{lem:DWh}), $P_{2}(h)=-\sqrt{d}s_{x,h}^{\T}\Diag(W_{x}^{\half}N_{x}W_{x}^{\half}s_{x,h})s_{x,h}=-\sqrt{d}\tr\bpar{\Diag(W_{x}^{\half}N_{x}W_{x}^{\half}s_{x,h})S_{x,h}^{2}}$.
Thus,
\begin{align*}
P_{2}(h) & =\sqrt{d}\,\tr\bpar{\Diag(N_{x}W_{x}^{\half}s_{x,h})W_{x}^{\half}S_{x,h}^{2}}=\sqrt{d}\sum_{i=1}^{m}w_{i}^{1/2}(a_{i}\cdot h)^{2}(b_{i}\cdot h)\,,
\end{align*}
where $b_{i}$ is the $i$-th row of $B:=N_{x}W_{x}^{\half}A_{x}$
for $i=1,\dots,m$. By Lemma~\ref{lem:variance-2},
\begin{align*}
 & \E\Bbrack{\Bbrace{\sum_{i=1}^{m}w_{i}^{1/2}(a_{i}\cdot h)^{2}(b_{i}\cdot h)}^{2}}\\
 & =\sum_{i,j\in[m]}w_{i}^{1/2}w_{j}^{1/2}\|a_{i}\|^{2}\|a_{j}\|^{2}(b_{i}\cdot b_{j})\\
 & \quad+4\sum_{i,j}w_{i}^{1/2}w_{j}^{1/2}(a_{i}\cdot a_{j})(a_{i}\cdot b_{i})(a_{j}\cdot b_{j})+4\sum_{i,j}w_{i}^{1/2}w_{j}^{1/2}\|a_{i}\|^{2}(b_{i}\cdot a_{j})(a_{j}\cdot b_{j})\\
 & \quad+2\underbrace{\sum_{i,j}w_{i}^{1/2}w_{j}^{1/2}(a_{i}\cdot a_{j})^{2}(b_{i}\cdot b_{j})}_{=:T_{1}}+4\underbrace{\sum_{i,j}w_{i}^{1/2}w_{j}^{1/2}(a_{i}\cdot a_{j})(a_{i}\cdot b_{j})(a_{j}\cdot b_{i})}_{=:T_{2}}\\
 & =\underbrace{1^{\T}\Diag(A_{x}A_{x}^{\T})\,W^{\half}BB^{\T}W^{\half}\,\Diag(A_{x}A_{x}^{\T})\,1}_{\eqqcolon N_{1}}+4\cdot\underbrace{1^{\T}\Diag(A_{x}B^{\T})\,W^{\half}A_{x}A_{x}^{\T}W^{\half}\,\Diag(A_{x}B^{\T})\,1}_{\eqqcolon N_{2}}\\
 & \quad+4\cdot\underbrace{[1^{\T}\Diag(A_{x}A_{x}^{\T})\,W^{\half}B]\cdot[A_{x}^{\T}W^{\half}\,\Diag(A_{x}B^{\T})\,1]}_{\le N_{1}+N_{2}\text{ by Young's inequality}}+2T_{1}+4T_{2}\,.
\end{align*}
As for $N_{1}$, since $B^{\T}B=A_{x}^{\T}W_{x}^{\half}N_{x}^{2}W_{x}^{\half}A_{x}\leq p^{2}A_{x}^{\T}W_{x}A_{x}$
by Lemma~\ref{lem:LS-comp-tool}-1 and thus $B^{\T}B\precsim(d)^{-1/2}I_{d}$,
Lemma~\ref{lem:matrix-projection} ensures $BB^{\T}\precsim\frac{1}{\sqrt{d}}P(B)\preceq\frac{1}{\sqrt{d}}\,I_{m}$.
Hence,
\begin{align*}
N_{1} & \lesssim\frac{1}{\sqrt{d}}\,\tr\bpar{\Diag(A_{x}A_{x}^{\T})\,W\,\Diag(A_{x}A_{x}^{\T})}\leq\frac{1}{\sqrt{d}}\,\tr(A_{x}^{\T}WA_{x})\,\norm{\Diag(A_{x}A_{x}^{\T})}_{\infty}\lesssim\frac{1}{\sqrt{d}}\,.
\end{align*}
As for $N_{2}$, due to $A_{x}^{\T}W_{x}A_{x}\preceq\frac{1}{\sqrt{d}}I_{d}$
we have $W^{\half}A_{x}A_{x}^{\T}W^{\half}\preceq\frac{1}{\sqrt{d}}I_{m}$
by Lemma~\ref{lem:matrix-projection}. Thus,
\begin{align*}
N_{2} & \lesssim\frac{1}{\sqrt{d}}\,\tr\bpar{\{\Diag(A_{x}B^{\T})\}^{2}}=\frac{1}{\sqrt{d}}\sum_{i\in[m]}(a_{i}\cdot b_{i})^{2}\leq\frac{1}{\sqrt{d}}\sum_{i}\|a_{i}\|^{2}\|b_{i}\|^{2}\\
 & \leq\frac{1}{d}\tr(BB^{\T})\lesssim\frac{1}{d^{3/2}}\tr\bpar{P(B)}\le\frac{1}{\sqrt{d}}\,.
\end{align*}
As for $T_{1}$, by Young's inequality (i.e., $2(a\cdot b)\leq\snorm a^{2}+\norm b^{2}$)
\begin{align*}
T_{1} & =\sum_{i,j\in[m]}(a_{i}\cdot a_{j})^{2}\,\bpar{(w_{j}^{1/2}b_{i})\cdot(w_{i}^{1/2}b_{j})}\lesssim\sum_{i,j}(a_{i}\cdot a_{j})^{2}\,(w_{j}\norm{b_{i}}^{2}+w_{i}\norm{b_{j}}^{2})\\
 & =2\sum_{i,j}w_{j}(a_{i}\cdot a_{j})^{2}\norm{b_{i}}^{2}=\sum_{i}\norm{b_{i}}^{2}\cdot\tr\Bpar{a_{i}^{\T}\Bpar{\sum_{j}a_{j}w_{j}a_{j}^{\T}}a_{i}}\\
 & =\sum_{i}\norm{b_{i}}^{2}\tr(a_{i}^{\T}A_{x}^{\T}WA_{x}a_{i})\leq\frac{1}{\sqrt{d}}\sum_{i}\norm{b_{i}}^{2}\norm{a_{i}}^{2}\leq\frac{1}{d}\tr(BB^{\T})\leq\frac{1}{\sqrt{d}}\,.
\end{align*}
As for $T_{2}$, using $(a_{i}\cdot a_{j})\leq\norm{a_{i}}\norm{a_{j}}\lesssim\frac{1}{\sqrt{d}}$
\begin{align*}
T_{2} & =\sum_{i,j\in[m]}w_{i}^{1/2}w_{j}^{1/2}(a_{i}\cdot a_{j})(a_{i}\cdot b_{j})(a_{j}\cdot b_{i})\lesssim\frac{1}{\sqrt{d}}\sum_{i,j\in[m]}w_{i}^{1/2}w_{j}^{1/2}(a_{i}\cdot b_{j})(a_{j}\cdot b_{i})\\
 & =\frac{1}{\sqrt{d}}\sum_{i}w_{i}^{1/2}b_{i}^{\T}\sum_{j}a_{j}w_{j}^{1/2}b_{j}^{\textbackslash T}a_{i}=\frac{1}{\sqrt{d}}\sum_{i}\tr(a_{i}w_{i}^{1/2}b_{i}^{\T}A_{x}^{\T}W^{1/2}B)\\
 & =\frac{1}{\sqrt{d}}\tr\bpar{(A_{x}^{\T}W^{1/2}B)^{2}}\underset{\text{CS}}{\leq}\frac{1}{\sqrt{d}}\tr(B^{\T}W^{1/2}A_{x}A_{x}^{\T}W^{1/2}B)\leq\frac{1}{d}\tr(B^{\T}B)\leq\frac{1}{\sqrt{d}}\,.
\end{align*}
Putting all the bounds together, we have $\E[P_{2}(h)^{2}]\lesssim d\cdot\frac{1}{\sqrt{d}}=\sqrt{d}$.

\paragraph{Term $\textsf{B}$.}

We show that for any given $\alpha=\Theta(1)$, each coordinate of
$w_{x}/s_{x}^{\alpha}$ and $w_{p_{z}}/s_{p_{z}}^{\alpha}$ is close.
For $0\leq t\le1$, we define $x_{t}:=x+\frac{r}{\sqrt{d}}th$, and
$s_{t},$ $w_{t}$ in the same fashion. Then for $p=\O(\log m)$,
\begin{align*}
\max_{i\in[m]}\Big|\log\frac{(w_{p_{z},i})^{\alpha}}{s_{p_{z},i}}-\log\frac{(w_{x,i})^{\alpha}}{s_{x,i}}\Big| & \leq\int_{0}^{1}\Big|\frac{\D}{\D t}\log\frac{[w_{t,i}]^{\alpha}}{s_{t,i}}\Big|\,\D t\lesssim\frac{r}{\sqrt{d}}\,\norm h_{A_{x}^{\T}W_{x}A_{x}}\leq\frac{1}{d^{1/4}}\norm z\,.
\end{align*}
Just as in showing SASC of the Vaidya metric, we can make this bound
arbitrarily small (say $\delta\approx0$) by conditioning on the high-probability
region where $\norm z\leq r\log\frac{1}{\veps}\leq0.01$. Hence,
\begin{equation}
e^{-\delta}\frac{(w_{x,i})^{\alpha}}{s_{x,i}}\leq\frac{(w_{p_{z},i})^{\alpha}}{s_{p_{z},i}}\leq e^{\delta}\frac{(w_{x,i})^{\alpha}}{s_{x,i}}\,.\label{eq:closeness}
\end{equation}
We remark that this $\Theta(1)$-multiplicative closeness is still
valid without the $\sqrt{d}$-scaling of $A_{x}^{\T}W_{x}A_{x}$.

Using the formula for $\Dd^{2}(A_{x}^{\T}W_{x}A_{x})[h^{\otimes4}]$
in \eqref{eq:LW-fourth-moment},
\begin{align*}
 & |\Dd^{2}g(p)[h^{\otimes4}]|\lesssim\bpar{\bar{P}_{3}(h)+|\bar{P}_{4}(h)|+|\bar{P}_{5}(h)|}=\bar{P}_{3}(h)+\sqrt{d}\,\bpar{|\tr(W_{p,h}'S_{p,h}^{3})|+|\tr(W_{p,h}''S_{p,h}^{2})|}\\
 & =\bar{P}_{3}(h)+\sqrt{d}\underbrace{\big|\tr\bpar{S_{p,h}^{3}\Diag(W_{p}^{\half}N_{p}W_{p}^{\half}s_{p,h})}\big|}_{\eqqcolon T_{1}}+\sqrt{d}\underbrace{|\tr(S_{p,h}^{2}W_{p,h}'')|}_{\eqqcolon T_{2}}\,,
\end{align*}
where in the last line we used the formula for $W_{p,h}'$ (Lemma~\ref{lem:DWh}).

Now we show $\E[\bar{P}_{3}(h)^{2}]\lesssim d^{2}$ and $T_{i}\lesssim\sqrt{d}$
w.h.p. for $i=4,5$. As for $\bar{P}_{3}$, we have $\bar{P}_{3}(h)\lesssim P_{3}(h)$
from the closeness \eqref{eq:closeness} of $w_{i}/s_{i}^{4}$ for
each $i\in[m]$, so $\E[P_{3}(h)^{2}]\lesssim d^{2}\cdot d^{-1}=d$
from \eqref{eq:P3_bound}.

As for $T_{1}$, using the Cauchy-Schwarz 
\begin{align*}
T_{1} & =\big|\tr\bpar{S_{p,h}^{3}W_{p}^{\half}\,\Diag(N_{p}W_{p}^{\half}s_{p,h})}\big|\leq\sqrt{\tr(S_{p,h}^{3}W_{p}S_{p,h}^{3})}\sqrt{s_{p,h}^{\T}W_{p}^{1/2}N_{p}^{2}W_{p}^{1/2}s_{p,h}}\\
 & \underset{\text{(i)}}{\lesssim}\sqrt{s_{p,h}^{3}W_{p}s_{p,h}^{3}}\sqrt{s_{p,h}^{\T}W_{p}s_{p,h}}\underset{\text{(ii)}}{\lesssim}\sqrt{s_{x,h}^{3}W_{x}s_{x,h}^{3}}\sqrt{s_{x,h}^{\T}W_{x}s_{x,h}}=\sqrt{s_{x,h}^{3}W_{x}s_{x,h}^{3}}\cdot d^{-1/4}\norm h_{g(x)}\,,
\end{align*}
where in (i) we used $N_{x}\preceq p^{2}I$ (Lemma~\ref{lem:LS-comp-tool}),
and in (ii) the closeness of $w_{i}/s_{i}^{6}$ and $w_{i}/s_{i}^{2}$
established in \eqref{eq:closeness}. As for the first term in the
RHS, 
\begin{align*}
\E[(s_{x,h}^{3}W_{x}s_{x,h}^{3})^{2}] & \underset{\text{CS}}{\lesssim}\sum_{i,j\in[m]}w_{i}w_{j}\sqrt{\E[(a_{i}\cdot h)^{12}]}\sqrt{\E[(a_{j}\cdot h)^{12}]}=\Bpar{\sum_{i}w_{i}\,\bpar{\E[(a_{i}\cdot h)^{12}]}^{2}}^{2}\\
 & \lesssim\Bpar{\sum_{i}w_{i}\norm{a_{i}}^{6}}^{2}\leq\Bpar{\frac{1}{d^{3/2}}\sum_{i}w_{i}}^{2}=\frac{1}{d}\,.
\end{align*}
As for the second term, the concentration of the standard Gaussian
guarantees $\norm h_{g(x)}\leq\norm h\lesssim\sqrt{d}$ w.h.p. Therefore,
$T_{1}\lesssim\sqrt{d}$ w.h.p.

As for $T_{2}$, \eqref{eq:trGamma} with $\Gamma_{p}=S_{p,h}^{2}$
equals $T_{2}$. Following \eqref{eq:last-bound} with I, II, III,
IV defined in \eqref{eq:LW-second-derv},
\begin{align*}
T_{2} & \lesssim\sum_{v=\text{I,II,III,IV}}\sqrt{\tr(W_{p}S_{p,h}^{4})}\norm v_{W_{p}^{-1}}\underset{\text{(i)}}{\lesssim}\sqrt{\tr(W_{p}S_{p,h}^{4})}\,\bpar{\tr(S_{p,h}^{2}W_{p})+\tr(S_{p,h}^{4}W_{p})}\\
 & \underset{\text{(ii)}}{\lesssim}\sqrt{\tr(W_{x}S_{x,h}^{4})}\,\bpar{\tr(S_{x,h}^{2}W_{x})+\tr(S_{x,h}^{4}W_{x})}\,,
\end{align*}
where (i) follows from Lemma~\ref{lem:second-deriv-Lewis} (i.e.,
$\norm v_{W_{p}^{-1}}\lesssim\norm h_{A_{p}^{\T}W_{p}A_{p}}^{2}=\tr(S_{p,h}^{2}W_{p})$
for $v=$ I, II, III, and $\norm{\text{IV}}_{W_{p}^{-1}}\lesssim\tr(S_{p,h}^{4}W_{p})$),
and (ii) follows from the conditioned event where the closeness of
$w_{i}/s_{i}^{2}$ at $x$ and $z$ holds. Since we already established
the high-probability bounds of $d^{-1/2}P_{3}(h)=\tr(S_{x,h}^{4}W_{x})\lesssim1$
and $\tr(S_{x,h}^{2}W_{x})\lesssim\sqrt{d}$, combining these yield
$T_{2}\lesssim\sqrt{d}$ w.h.p.
\end{proof}

\subsubsection{Quadratic constraints \label{proof:quadratic}}

We show that a $\nu$-SC barrier $\psi(\cdot)=-\log f(\cdot)$ satisfies
\[
|\Dd^{4}\psi(x)[h^{\otimes4}]|\lesssim\nu^{2}\norm h_{\hess\psi(x)}^{2}+\Big|\frac{\Dd^{4}f(x)[h^{\otimes4}]}{f(x)}\Big|\,.
\]

\begin{proof}
[Proof of Lemma~\ref{lem:4th-log}] Fix $h\in\Rd$ and $x\in\inter(K)$,
define $\phi(t):=\psi(x+th)$. Then,
\begin{align*}
\phi' & =-\frac{f'}{f}\,,\\
\phi'' & =\Par{\frac{f'}{f}}^{2}-\frac{f''}{f}=(\phi')^{2}-\frac{f''}{f}\,,\\
\phi''' & =2\phi'\phi''-\frac{f'''f-f''f'}{f^{2}}=2\phi'\phi''-\frac{f'''}{f}+\frac{f''f'}{f^{2}}=2\phi'\phi''+\phi'(\phi''-(\phi')^{2})-\frac{f'''}{f}\\
 & =3\phi'\phi''-(\phi')^{3}-\frac{f'''}{f}\,,\\
\phi^{(4)} & =3(\phi'')^{2}+3\phi'\phi'''-3(\phi')^{2}\phi''-\frac{f^{(4)}f-f'''f'}{f^{2}}\\
 & =3(\phi'')^{2}+3\phi'\phi'''-3(\phi')^{2}\phi''+\phi'\Par{\phi'''-3\phi'\phi''+(\phi')^{3}}-\frac{f^{(4)}}{f}\\
 & =3(\phi'')^{2}+4\phi'\phi'''-6(\phi')^{2}\phi''+(\phi')^{4}-\frac{f^{(4)}}{f}\,.
\end{align*}
Since $|\phi'''|\leq2(\phi'')^{3/2}$ (SC of $\phi$) and $\phi''\geq\frac{1}{\nu}(\phi')^{2}$
(the definition of the barrier parameter), which is equivalent to
$|\phi'|\leq\sqrt{\nu}(\phi'')^{1/2}$, we can directly compute as
follows:
\begin{align*}
|\phi^{(4)}| & \leq4\,|\phi'\phi'''|+3\,|(\phi'')^{2}|+6|\,(\phi')^{2}\phi''|+|(\phi')^{4}|+\Big|\frac{f^{(4)}}{f}\Big|\\
 & \leq8\sqrt{\nu}\,|\phi''|^{2}+3\,|\phi''|^{2}+6\nu\,|\phi''|^{2}+\nu^{2}\,|\phi''|^{2}+\Big|\frac{f^{(4)}}{f}\Big|\lesssim\nu^{2}|\phi''|^{2}+\Big|\frac{f^{(4)}}{f}\Big|\,.\qedhere
\end{align*}
\end{proof}
Using this tool, we study Dikin-amenability of barriers for quadratic
constraints.
\begin{proof}
[Proof of Lemma~\ref{lem:quadratic-const}] Let us check the last
claim first. By Lemma~\ref{lem:linear-trans}, we may assume that
\[
\phi(x,y)=-\log(l+q^{\T}y-\half\norm x^{2})\,,
\]
and let $f(x,y)=l+q^{\T}y-\half\,\norm x^{2}$. For $z=(x,y)\in\intk$
and $u=(u_{x},u_{y})\in\Rd$, we have 
\begin{align}
\Dd\phi(z)[u] & =-\frac{1}{f}\,(q\cdot u_{y}-x\cdot u_{x})=\frac{x\cdot u_{x}-q\cdot u_{y}}{f}\,,\nonumber \\
\Dd^{2}\phi(z)[u,u] & =\frac{1}{f^{2}}\,(x\cdot u_{x}-q\cdot u_{y})^{2}+\frac{1}{f}\,\norm{u_{x}}^{2}\,.\label{eq:hessian-quadratic}
\end{align}

As for the first term in the RHS of \eqref{eq:hessian-quadratic},
it holds that for $v=(v_{x},v_{y})\in\Rd$ 
\begin{align*}
\Dd\Bpar{\frac{(x\cdot u_{x}-q\cdot u_{y})^{2}}{f^{2}}}[v] & =\frac{2\,(x\cdot u_{x}-q\cdot u_{y})(v_{x}\cdot u_{x})}{f^{2}}+2\,(x\cdot u_{x}-q\cdot u_{y})^{2}\cdot\frac{x\cdot v_{x}-q\cdot v_{y}}{f^{3}}\,,\\
\Dd^{2}\Bpar{\frac{(x\cdot u_{x}-q\cdot u_{y})^{2}}{f^{2}}}[v,v] & =\frac{2\,(v_{x}\cdot u_{x})^{2}}{f^{2}}+4\frac{(x\cdot u_{x}-q\cdot u_{y})(v_{x}\cdot u_{x})(x\cdot v_{x}-q\cdot v_{y})}{f^{3}}\\
 & \quad+\frac{4\,(x\cdot u_{x}-q\cdot u_{y})(v_{x}\cdot u_{x})(x\cdot v_{x}-q\cdot v_{y})+2\,(x\cdot u_{x}-q\cdot u_{y})^{2}\norm{v_{x}}^{2}}{f^{3}}\\
 & \quad+\frac{6\,(x\cdot u_{x}-q\cdot u_{y})^{2}(x\cdot v_{x}-q\cdot v_{y})^{2}}{f^{4}}\\
 & =\frac{2\,(v_{x}\cdot u_{x})^{2}}{f^{2}}+\frac{4\,(x_{q}\cdot u)(v_{x}\cdot u_{x})(x_{q}\cdot v)}{f^{3}}\\
 & \quad+\frac{4\,(x_{q}\cdot u)(v_{x}\cdot u_{x})(x_{q}\cdot v)+2(x_{q}\cdot u)^{2}\norm{v_{x}}^{2}}{f^{3}}+\frac{6\,(x_{q}\cdot u)^{2}(x_{q}\cdot v)^{2}}{f^{4}}\,,
\end{align*}
where $x_{q}:=(x,-q)\in\Rd$.

As for the second term, direct computations lead to 
\begin{align*}
\Dd\Bpar{\frac{\norm{u_{x}}^{2}}{f}}[v] & =\frac{1}{f^{2}}\,\norm{u_{x}}^{2}(x\cdot v_{x}-q\cdot v_{y})\,,\\
\Dd^{2}\Bpar{\frac{\norm{u_{x}}^{2}}{f}}[v,v] & =\frac{2}{f^{3}}\,\norm{u_{x}}^{2}(x\cdot v_{x}-q\cdot v_{y})^{2}+\frac{1}{f^{2}}\,\norm{u_{x}}^{2}\norm{v_{x}}^{2}\\
 & =\frac{2}{f^{3}}\,\norm{u_{x}}^{2}(x_{q}\cdot v)^{2}+\frac{1}{f^{2}}\,\norm{u_{x}}^{2}\norm{v_{x}}^{2}\,.
\end{align*}
Putting these together, for $u,v\in\Rd$
\begin{align*}
 & \Dd^{4}\phi[u,u,v,v]\\
 & =\frac{1}{f^{2}}\,\norm{u_{x}}^{2}\norm{v_{x}}^{2}+\underbrace{\frac{2}{f^{2}}\,(v_{x}\cdot u_{x})^{2}}_{\geq0}+\frac{4}{f^{3}}\,\Bpar{\half\,\norm{u_{x}}^{2}(x_{q}\cdot v)^{2}+2\,(x_{q}\cdot u)(v_{x}\cdot u_{x})(x_{q}\cdot v)+\frac{(x_{q}\cdot u)^{2}}{2}\,\norm{v_{x}}^{2}}\\
 & \qquad+\frac{6}{f^{4}}\,(x_{q}\cdot u)^{2}(x_{q}\cdot v)^{2}\\
 & \geq\frac{4}{f^{3}}\,\bigg(\underbrace{\half\norm{u_{x}}^{2}(x_{q}\cdot v)^{2}+\frac{1}{2}\norm{v_{x}}^{2}(x_{q}\cdot u)^{2}}_{\text{Use AM-GM}}+2(x_{q}\cdot u)(v_{x}\cdot u_{x})(x_{q}\cdot v)\bigg)\\
 & \qquad+\underbrace{\frac{1}{f^{2}}\,\norm{u_{x}}^{2}\norm{v_{x}}^{2}+\frac{6}{f^{4}}\,(x_{q}\cdot u)^{2}(x_{q}\cdot v)^{2}}_{\text{Use AM-GM}}\\
 & \geq\frac{4}{f^{3}}\,\bpar{\norm{u_{x}}\,\norm{v_{x}}\,|x_{q}\cdot v|\,|x_{q}\cdot u|-2|x_{q}\cdot u|\,|x_{q}\cdot v|\,\norm{u_{x}}\,\norm{v_{x}}}+\frac{2\sqrt{6}}{f^{3}}\,|x_{q}\cdot u|\,|x_{q}\cdot v|\,\norm{u_{x}}\,\norm{v_{x}}\\
 & =\frac{4}{f^{3}}\,\norm{u_{x}}\,\norm{v_{x}}\,|x_{q}\cdot v|\,|x_{q}\cdot u|\,\Bpar{\frac{\sqrt{6}}{2}-1}\geq0\,.\qedhere
\end{align*}
\end{proof}

\subsubsection{PSD: convexity and strongly self-concordance \label{proof:psd-convex-ssc}}

We start with convexity of $\log\det(\hess\phi)$ for $\phi(X)=-\log\det X$.
\begin{proof}
[Proof of Proposition~\ref{prop:convex-logdet}] Using Lemma~\ref{prop:metricFormula}
and $\det\bpar{M^{\T}(A\otimes A)M}=2^{d(d-1)/2}\,(\det A)^{d+1}$
(Lemma~\ref{lem:Kronecker}) in the first and second equality below,
\begin{align*}
\log\det\bpar{\hess\phi(X)} & =\log\det\bpar{M^{\T}(X^{-1}\otimes X^{-1})M}=\frac{d(d-1)}{2}\,\log2-(d+1)\,\log\det X\,.
\end{align*}
Since $-\log\det X$ is convex in $X$ \eqref{eq:2ndDiffLogDet},
the convexity of $\log\det\bpar{\hess\phi(X)}$ also follows.
\end{proof}
Observe from the proof that $\log\det\bpar{\hess\phi(X)}=\text{const.}+(d+1)\,\phi(X)$.
Differentiating both sides in direction $H$, by \eqref{eq:gradLogDet}
$\tr\bpar{[\hess\phi(X)]^{-1}\Dd^{3}\phi(X)[H]}=(d+1)\,\Dd\phi(X)[H]$.
Hence,
\begin{align}
 & \tr\bpar{[\hess\phi(X)]^{-\half}\Dd^{3}\phi(X)[H]\,[\hess\phi(X)]^{-\half}}=-(d+1)\,\tr(X^{-1}H)\,.\label{eq:difflogdet}
\end{align}

We are ready to show SSC of $\phi$.
\begin{proof}
[Proof of Lemma~\ref{lem:logdet-scaling}] For $H\in\mbb S^{d}$
and $t\in\R$, denote $X_{t}:=X+tH$ and $g_{t}:=M^{\T}(X_{t}\otimes X_{t})^{-1}M$.
Note that
\[
\bnorm{[\hess\phi(X)]^{-\half}\Dd^{3}\phi(X)[H]\,[\hess\phi(X)]^{-\half}}_{F}^{2}=\tr(g^{-1}\del_{t}g_{t}\vert_{t=0}\,g^{-1}\del_{t}g_{t}\vert_{t=0})\,,
\]
and
\begin{align}
\del_{t}g_{t}\vert_{t=0} & \underset{\text{(i)}}{=}\del_{t}\bpar{M^{\T}(X_{t}\otimes X_{t})^{-1}M}\Big|_{t=0}\underset{\text{(ii)}}{=}-M^{\T}(X\otimes X)^{-1}\,\del_{t}(X_{t}\otimes X_{t})\vert_{t=0}\,(X\otimes X)^{-1}M\nonumber \\
 & =-M^{\T}(X^{-1}\otimes X^{-1})(H\otimes X+X\otimes H)(X^{-1}\otimes X^{-1})M\nonumber \\
 & \underset{\text{(iii)}}{=}-M^{\T}(X^{-1}HX^{-1}\otimes X^{-1}+X^{-1}\otimes X^{-1}HX^{-1})M\,,\label{eq:18-1}
\end{align}
where (i) follows from Lemma~\ref{prop:metricFormula}, (ii) is due
to \eqref{eq:diffInverse}, and (iii) follows from $(A\otimes B)(C\otimes D)=(AC)\otimes(BD)$
(Lemma~\ref{lem:Kronecker}-3).

Recall that positive semidefinite matrices have unique positive semidefinite
square roots, so $(X\otimes X)^{\half}=X^{\half}\otimes X^{\half}$
(due to $(X^{1/2}\otimes X^{1/2})\cdot(X^{1/2}\otimes X^{1/2})=X\otimes X$).
Since $g_{t}=M^{\T}(X_{t}\otimes X_{t})^{-1/2}(X_{t}\otimes X_{t})^{-1/2}M$,
the corresponding orthogonal projection is 
\[
P_{t}:=P\bpar{(X_{t}\otimes X_{t})^{-\half}M}=(X_{t}\otimes X_{t})^{-\half}Mg_{t}^{-1}M^{\T}(X_{t}\otimes X_{t})^{-\half}\,.
\]
 By substituting $\del_{t}g_{t}\big|_{t=0}$ with \eqref{eq:18-1},
\begin{align*}
 & \tr(g^{-1}\del_{t}g_{t}\vert_{t=0}\,g^{-1}\del_{t}g_{t}\vert_{t=0})\\
 & =\tr\bigl(g^{-1}M^{\T}(X^{-1}HX^{-1}\otimes X^{-1}+X^{-1}\otimes X^{-1}HX^{-1})M\\
 & \qquad\qquad\cdot g^{-1}M^{\T}(X^{-1}HX^{-1}\otimes X^{-1}+X^{-1}\otimes X^{-1}HX^{-1})\cblue M\bigr)\\
 & =\tr\bigl(\cblue Mg^{-1}M^{\T}(X^{-1}HX^{-1}\otimes X^{-1}+X^{-1}\otimes X^{-1}HX^{-1})M\\
 & \qquad\qquad\cdot g^{-1}M^{\T}(X^{-1}HX^{-1}\otimes X^{-1}+X^{-1}\otimes X^{-1}HX^{-1})\bigr)\\
 & =\tr\Bpar{\bbrack{\cred{Mg^{-1}M^{\T}}(X^{-1}HX^{-1}\otimes X^{-1}+X^{-1}\otimes X^{-1}HX^{-1})}^{2}}\\
 & =\tr\Bpar{\bbrack{\cred{(X\otimes X)^{\half}P(X\otimes X)^{\half}}(X^{-1}HX^{-1}\otimes X^{-1}+X^{-1}\otimes X^{-1}HX^{-1})}^{2}}\\
 & =\tr\Bpar{\bbrack{P\underbrace{(X\otimes X)^{\half}(X^{-1}HX^{-1}\otimes X^{-1}+X^{-1}\otimes X^{-1}HX^{-1})(X\otimes X)^{\half}}_{\eqqcolon S}}^{2}}\\
 & =\tr(PSPS)\,.
\end{align*}
Using Lemma~\ref{lem:Kronecker}-3,
\begin{align*}
S & =\underbrace{X^{-\half}HX^{-\half}\otimes I_{d}}_{\eqqcolon A}+\underbrace{I_{d}\otimes X^{-\half}HX^{-\half}}_{\eqqcolon B}\,.
\end{align*}
By the Cauchy-Schwarz inequality along with $P^{\T}P=P^{2}=P$ and
$P\preceq I_{d}$,
\begin{align*}
\tr(PSPS) & \leq\tr((PS)^{\T}PS)\leq\tr(S^{\T}S)=\norm S_{F}^{2}\leq(\norm A_{F}+\norm B_{F})^{2}\,.
\end{align*}
Using Lemma~\ref{lem:Kronecker}-3, 
\begin{align*}
\norm A_{F}^{2} & =\tr\bpar{(X^{-\half}HX^{-\half}\otimes I_{d})\cdot(X^{-\half}HX^{-\half}\otimes I_{d})}\\
 & =\tr(X^{-\half}HX^{-1}HX^{-\half}\otimes I_{d})=\tr(X^{-\half}HX^{-1}HX^{-\half})\,\tr(I_{d})=d\,\norm H_{X}^{2}\,,
\end{align*}
and similarly $\norm B_{F}^{2}=d\,\norm H_{X}^{2}$. Therefore, $\psi_{X}\leq2\sqrt{d}$
follows from
\[
\bnorm{[\hess\phi(X)]^{-\half}\Dd^{3}\phi(X)[H]\,[\hess\phi(X)]^{-\half}}_{F}\leq\sqrt{\tr(PSPS)}\leq2\sqrt{d}\,\norm H_{X}\,.
\]

To see the optimality of $\O(d^{1/2})$, we recall \eqref{eq:difflogdet}:
\[
\tr\bpar{[\hess\phi(X)]^{-\half}\Dd^{3}\phi(X)[H]\,[\hess\phi(X)]^{-\half}}=-(d+1)\,\tr(X^{-1}H)\,.
\]
Taking supremum on both sides,
\begin{align*}
\sup_{H:\norm H_{X}=1}\tr\bpar{[\hess\phi(X)]^{-\half}\Dd^{3}\phi(X)[H]\,[\hess\phi(X)]^{-\half}} & =\sup_{\substack{H\in\mbb S^{d}:\\
\snorm{X^{-1/2}HX^{-1/2}}_{F}=1
}
}-(d+1)\,\tr(X^{-\half}HX^{-\half})\\
 & =\sup_{S\in\mbb S^{d}:\norm S_{F}=1}(d+1)\,\tr(S)\,,
\end{align*}
and this objective achieves the maximum at $H=-d^{-1/2}X$, with the
supremum being $(d+1)\sqrt{d}$. On the other hand, due to $\tr(A)\leq d^{1/2}\,\norm A_{F}$
for $A\in\R^{d\times d}$,
\begin{align*}
 & \tr\bpar{[\hess\phi(X)]^{-\half}\Dd^{3}\phi(X)[H]\,[\hess\phi(X)]^{-\half}}\\
 & \leq\sqrt{\frac{d(d+1)}{2}}\cdot\bnorm{[\hess\phi(X)]^{-\half}\Dd^{3}\phi(X)[H]\,[\hess\phi(X)]^{-\half}}_{F}\leq\sqrt{\frac{d(d+1)}{2}}\cdot\psi_{X}\norm H_{X}\,,
\end{align*}
and thus by taking supremum on both sides over a symmetric matrix
$H$ with $\norm H_{X}=1$, it follows that $(d+1)\sqrt{d}\leq\sqrt{\frac{d(d+1)}{2}}\,\psi_{X}$
and 
\[
\sqrt{2(d+1)}\leq\psi_{X}\,.\qedhere
\]
\end{proof}

\subsubsection{PSD: strongly lower trace self-concordance \label{proof:psd-sltsc}}

Direct computation leads to $\Dd^{2}g(X)[H,H]\succeq0$ (so SLTSC).
\begin{proof}
[Proof of Lemma~\ref{lem:logdet-sltsc}] For $g(X)=-\hess\log\det X$,
recall that $g(X)[H,H]=\tr(X^{-1}HX^{-1}H)$. Thus for any $V\in\mbb S^{d}$,
\begin{align*}
\Dd g(X)[H,H,V] & =-\tr(X^{-1}VX^{-1}\cdot HX^{-1}H)-\tr(X^{-1}H\cdot X^{-1}VX^{-1}\cdot H)\\
 & =-2\,\tr(X^{-1}VX^{-1}HX^{-1}H)\,,
\end{align*}
and differentiating again,
\begin{align}
 & \Dd^{2}g(X)[H,H,V,V]\nonumber \\
 & =4\,\tr(X^{-1}VX^{-1}VX^{-1}HX^{-1}H)+2\,\tr(X^{-1}VX^{-1}HX^{-1}VX^{-1}H)\nonumber \\
 & =4\,\tr(X^{-\half}HX^{-1}VX^{-1}VX^{-1}HX^{-\half})+2\,\tr(X^{-\half}VX^{-1}HX^{-\half}\cdot X^{-\half}VX^{-1}HX^{-\half})\nonumber \\
 & \underset{\text{(i)}}{\geq}4\,\tr(X^{-\half}HX^{-1}VX^{-1}VX^{-1}HX^{-\half})-2\,\tr(X^{-\half}HX^{-1}VX^{-\half}\cdot X^{-\half}VX^{-1}HX^{-\half})\nonumber \\
 & =2\,\tr(X^{-\half}HX^{-1}VX^{-1}VX^{-1}HX^{-\half})\geq0\,,\label{eq:D4ph1}
\end{align}
where in (i) we used the Cauchy-Schwarz inequality. Therefore, $\Dd^{2}g(X)[H,H]\succeq0$.
\end{proof}

\subsubsection{PSD: average self-concordance \label{proof:psd-asc}}

We establish a connection to the Gaussian orthogonal ensemble (GOE):
for $d_{s}=d(d+1)/2$ and $\svec(H)\sim\ncal\bpar{0,\frac{r^{2}}{d_{s}}\,g(X)^{-1}}$,
we have $\frac{\sqrt{d_{s}d}}{r}X^{-\half}HX^{-\half}$ is the GOE.
\begin{proof}
[Proof of Lemma~\ref{lem:conn-to-goe}] Let $h_{X}:=\svec(X^{-1/2}HX^{-1/2})$
and $h:=\svec(H)$. It holds that
\[
h_{X}=L(X\otimes X)^{-\half}Mh
\]
due to $h_{X}=\svec(X^{-\half}HX^{-\half})=L\,\vec(X^{-\half}HX^{-\half})=L(X\otimes X)^{-\half}\vec(H)=L(X\otimes X)^{-\half}Mh$.
As $h\sim\ncal\bpar{0,\frac{r^{2}}{d_{s}}\,g(X)^{-1}}$, $h_{X}$
is a Gaussian with zero mean and covariance
\begin{align*}
 & \frac{r^{2}}{d_{s}}L(X\otimes X)^{-\half}Mg(X)^{-1}M^{\T}(X\otimes X)^{-\half}L^{\T}\\
\underset{\text{(i)}}{=} & \frac{r^{2}}{d_{s}d}L(X\otimes X)^{-\half}MLN(X\otimes X)N^{\T}L^{\T}M^{\T}(X\otimes X)^{-\half}L^{\T}\\
\underset{(*)}{=} & \frac{r^{2}}{d_{s}d}L(X\otimes X)^{-\half}N(X\otimes X)N^{\T}(X\otimes X)^{-\half}L^{\T}\\
\underset{(*)}{=} & \frac{r^{2}}{d_{s}d}L(X\otimes X)^{-\half}(X\otimes X)N(X\otimes X)^{-\half}L^{\T}\underset{\text{(*)}}{=}\frac{r^{2}}{d_{s}d}LNL^{\T}\\
\underset{\text{(ii)}}{=} & \frac{r^{2}}{d_{s}d}\,\left[\begin{array}{cc}
I_{d}\\
 & \half I_{d(d-1)/2}
\end{array}\right]\,,
\end{align*}
where (i) follows from Proposition~\ref{prop:metricFormula}, $(*)$
follows from Lemma~\ref{lem:MNL-properties}, and (ii) follows from
\citet[Page 427]{magnus1980elimination} that $LNL^{\T}$ is a $d_{s}\times d_{s}$
diagonal matrix with $d$ times $1$ and $\half d(d-1)$ times $1/2$.
Precisely, the entries of $h_{X}\in\R^{d_{s}}$ corresponding to the
diagonals of $X^{-1/2}HX^{-1/2}$ are $1$, and its entries corresponding
to off-diagonals is $1/2$. This is exactly the covariance matrix
of a $d_{s}$-dimensional GOE, so $X^{-\half}HX^{-\half}\sim\frac{r}{\sqrt{d_{s}d}}G$
for the GOE $G$.
\end{proof}
Now we show ASC of $d\phi$.
\begin{proof}
[Proof of Lemma~\ref{lem:logdet-asc}] Expand $\norm{Z-X}_{Z}^{2}:=\norm{Z-X}_{g(Z)}^{2}$
at $X$ for $Z=X+H$:
\[
\norm{Z-X}_{Z}^{2}-\norm{Z-X}_{X}^{2}=\sum_{k=1}^{\infty}\frac{1}{k!}\,\Dd^{k}g(X)[H^{\otimes k+2}]\,.
\]
It follows from induction that for $H_{X}:=X^{-\half}HX^{-\half}$
\begin{align*}
\Dd g(X)[H^{\otimes3}] & =-2d\,\tr(X^{-1}HX^{-1}HX^{-1}H)=-2\tr(H_{X}^{3})\,,\\
\Dd^{2}g(X)[H^{\otimes4}] & =3!\,d\,\tr(H_{X}^{4})\,,\\
\Dd^{k}g(X)[H^{\otimes(k+2)}] & =(-1)^{k}(k+1)!\,d\,\tr(H_{X}^{k+2})\,.
\end{align*}
Putting these back into the series expansion, for $H$ the GOE (see
Lemma~\ref{lem:conn-to-goe})
\begin{align*}
 & \norm{Z-X}_{Z}^{2}-\norm{Z-X}_{X}^{2}=\sum_{k=1}^{\infty}(-1)^{k}(k+1)d\,\tr(H_{X}^{k+2})\\
= & \sum_{k=1}^{\infty}(-1)^{k}(k+1)d\cdot\Bpar{\frac{r}{\sqrt{d_{s}d}}}^{k+2}\tr(H^{k+2})=\frac{r^{2}}{d_{s}}\sum_{k=1}^{\infty}(-1)^{k}(k+1)\,\Bpar{\frac{r}{\sqrt{d_{s}d}}}^{k}\tr(H^{k+2})\,.
\end{align*}

As for ASC, it suffices to show that $\sum_{k=1}^{\infty}(-1)^{k}(k+1)\bpar{\frac{r}{\sqrt{d_{s}d}}}^{k}\,\tr(H^{k+2})$
can be made arbitrarily small. We first control $\sum_{k\geq2}$:
\[
\Big|\sum_{k\geq2}(-1)^{k}(k+1)\Bpar{\frac{r}{\sqrt{d_{s}d}}}^{k}\tr(H^{k+2})\Big|\leq\sum_{k\geq2}(k+1)\Bpar{\frac{r}{\sqrt{d_{s}d}}}^{k}d\cdot\norm H_{\text{op}}^{k+2}\,.
\]
By \citet[Corollary 4.4.8]{vershynin2018high}, $\norm H_{\text{op}}\lesssim\sqrt{d}$
holds with high probability, and thus
\begin{align*}
\sum_{k\geq2}(k+1)\Bpar{\frac{r}{\sqrt{d_{s}d}}}^{k}d\cdot\norm H_{\text{op}}^{k+2} & \leq\sum_{k\geq2}(k+1)r^{k}\frac{1}{d^{3k/2}}d\cdot d^{\frac{k+2}{2}}\leq\sum_{k\geq2}(k+1)r^{k}d^{2-k}\,.
\end{align*}
By taking $r=\Omega(1)$ small enough, we can make this series arbitrarily
small.

Now we bound $\frac{r}{d^{3/2}}\tr(H^{3})$ ($k=1$ case). This is
a Gaussian polynomial in $\svec(H)$, so it suffices to show $\E[(\tr(H^{3}))^{2}]=\mc O(d^{3})$;
we then use Lemma~\ref{lem:conc-gaussian-poly} to obtain a high-probability
bound on the Gaussian polynomial $\frac{r}{d^{3/2}}\tr(H^{3})$. For
$H=(H_{ab})\in\mbb S^{d}$, 
\[
\bpar{\tr(H^{3})}^{2}=\sum_{ipq}H_{ip}H_{pq}H_{qi}\cdot\sum_{jrs}H_{jr}H_{rs}H_{sj}=\sum_{ipqjrs}H_{ip}H_{pq}H_{qi}H_{jr}H_{rs}H_{sj}\,,
\]
where each $H_{**}$ in the summand is an independent Gaussian with
zero mean and variance $1$ or $1/2$ (as $H$ is the GOE). We can
classify the indices $\{i,p,q,j,r,s\}$ into the following types:
\begin{align*}
6\text{ distinct indices } & \{a,b,c,d,e,f\}\,,\\
5\text{ distinct indices } & \{a,b,c,d,(e,e)\}\,,\\
4\text{ distinct indices } & \{a,b,c,(d,d,d)\},\{a,b,(c,c),(d,d)\}\,,\\
\text{Others } & \dots\,,
\end{align*}
where for example $\{a,b,c,d,e,f\}$ means all indices are different,
and $\{a,b,c,d,(e,e)\}$ means that there appear 5 different indices
$\{a,b,c,d,e\}$ but exists one pair $(e,e)$ of the same index. Note
that $\E H_{ip}H_{pq}H_{qi}H_{jr}H_{rs}H_{sj}=\mc O(1)$ is at most
the sixth moment of a standard Gaussian. It implies that toward our
goal of showing $\mc O(d^{3})$-bound on $\bpar{\tr(H^{3})}^{2}$,
it suffices to look into only three types of indices above. This is
because the terms from other types contribute at most $\mc O(d^{3})$
to $\bpar{\tr(H^{3})}^{2}$.
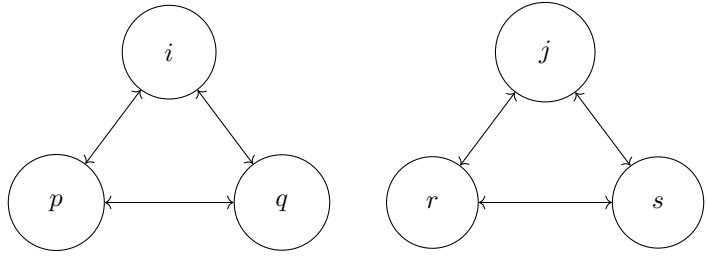
\begin{figure}[t]
\centering \begin{tikzpicture}% Left Diagram
\node[draw,circle,inner sep=10pt] (i) at (0,2) {$i$};
\node[draw,circle,inner sep=10pt] (p) at (-1.5,0) {$p$};
\node[draw,circle,inner sep=10pt] (q) at (1.5,0) {$q$};

\draw[<->] (i) -- (p);
\draw[<->] (i) -- (q);
\draw[<->] (p) -- (q);

% Right Diagram
\node[draw,circle,inner sep=10pt] (r) at (5,2) {$j$};
\node[draw,circle,inner sep=10pt] (s) at (6.5,0) {$s$};
\node[draw,circle,inner sep=10pt] (t) at (3.5,0) {$r$};

\draw[<->] (r) -- (s);
\draw[<->] (r) -- (t);
\draw[<->] (s) -- (t);\end{tikzpicture}\caption{A structure of indices of $H_{ip}H_{pq}H_{qi}\cdot H_{jr}H_{rs}H_{sj}$\label{fig:ipq-jrs}}
\end{figure}

For any term with 6 distinct indices, we can always find an `uncoupled'
$H_{**}$ (for example $H_{ab}$) in the summand that is independent
of all the others, so its expectation of the summand is $0$.

For the terms with $5$-distinct indices $\{a,b,c,d,(e,e)\}$, due
to symmetry (see Figure~\ref{fig:ipq-jrs}) we can further classify
the index $(i,p,q,j,r,s)$ into either $(a,b,c,d,e,e)$ or $(a,b,e,c,d,e)$.
In both cases , $H_{ab}$ has no coupled Gaussian, so the expectations
of the summand are also $0$. 

For $4$-distinct indices, let us first consider $\{a,b,c,(d,d,d)\}$-type
indices. In this case $(i,p,q,j,r,s)$ is of the form either $(a,a,a,b,c,d)$
or $(a,a,b,a,c,d)$ due to symmetry. In both cases, $H_{cd}$ has
no coupled Gaussian. Now consider $\{a,b,(c,c),(d,d)\}$-type indices.
Then $(i,p,q,j,r,s)$ is of the form either $(a,b,c,c,d,d)$ or $(a,c,c,b,d,d)$
or $(a,c,d,b,c,d)$. For each case, $H_{ab},H_{cc},H_{ac}$ are uncoupled
ones. Therefore, $\E[H_{ip}H_{pq}H_{qi}H_{jr}H_{rs}H_{sj}]=0$ whenever
there are at least $4$ distinct indices.
\end{proof}
\begin{rem}
\label{rem:challenge-extension-SASC}It seems challenging to show
that $\phi$ is SASC using the same technique. When $g$ is 
\[
g=d\,\hess(-\log\det X)+g'
\]
for other PSD matrix function $g'$, we know that $\svec(H_{X})=\svec(X^{-\half}HX^{-\half})$
follows a Gaussian distribution with zero mean and covariance matrix
$M$ satisfying
\[
M\preceq\left[\begin{array}{cc}
I_{d}\\
 & \half I_{d(d-1)/2}
\end{array}\right]\,.
\]
A main difference in the SASC setting is that the entries of $h=\svec(H_{X})$
might exhibit dependencies, making the previous approach infeasible.
This arises because many fundamental results in the random matrix
theory often presume independence of the entries of a random matrix.
Moreover, our combinatorial argument for the $k=1$ case is not feasible
in the presence of such dependencies.
\end{rem}

\subsection{Examples ($\S$\ref{sec:examples})}

\subsubsection{Algorithms for PSD sampling \label{proof:Algorithms-for-PSD}}
\begin{proof}
[Proof of Proposition~\ref{thm:hybridPSD}] We define $g_{X}=g=2(d^{2}g_{1}+g_{2})$,
where
\[
g_{1}(X)=M^{\T}(X\kro X)^{-1}M\qquad\text{and}\qquad g_{2}(X)=22\sqrt{\frac{m}{d}}\,M^{\T}A_{X}^{\T}\bpar{\Sigma_{X}+\frac{d}{m}I_{m}}A_{X}M\,.
\]
Since $d^{2}g_{1}$ and $g_{2}$ are SSC, $g$ is also SSC due to
Lemma~\ref{lem:ssc-sum} and $\mc O(d^{3}+\sqrt{md^{2}})$-symmetric\footnote{Since the dimension is $d_{s}$ in the PSD setting, we should replace
$d$ by $d_{s}=\O(d^{2})$ when applying Lemma~\ref{lem:paramsBarrier}.} due to Lemma~\ref{lem:symmetry-addition}. As $d^{2}g_{1}$ and
$g_{2}$ is SLTSC and SASC, $g$ is LTSC and ASC. Putting these together,
it follows that $g$ is $\bpar{\mc O(d^{3}+\sqrt{md^{2}}),\mc O(d^{3}+\sqrt{md^{2}})}$-Dikin-amenable.
Therefore, Theorem~\ref{thm:Dikin-annealing} implies that $\gcdw$
incurs $\otilde{d^{2}(d^{3}+\sqrt{md^{2}})}=\otilde{d^{3}(d^{2}+\sqrt{m})}$
total iterations of the $\dw$ with $g$.

Now we bound the per-step complexity of the $\dw$ (Algorithm~\ref{alg:DikinWalk}).
Recall that it requires (1) the update of the leverage scores, (2)
computation of the matrix function induced by the local metric $g$,
(3) the inverse of the matrix function and (4) its determinant. By
\citet[Theorem 46]{lee2019solving} (with $p=2$ and $d\gets d_{s}$
therein), the initialization of the leverage scores at the beginning
takes $\otilde{md^{2\omega}}$ and their updates takes $\otilde{md^{2(\omega-1)}}$
time. Since (1) takes $\otilde{md^{2(\omega-1)}}$, (2) takes $\otilde{d^{4}+md^{2(\omega-1)}}$,
and (3) and (4) take $\mc O\Par{d^{2\omega}}$, each iteration runs
in $\otilde{d^{2\omega}+md^{2(\omega-1)}}$ time. Even though the
initialization of leverage scores takes $\otilde{md^{2\omega}}$ time,
the amortized per-step time complexity becomes $\otilde{d^{2\omega}+md^{2(\omega-1)}}=\otilde{md^{2(\omega-1)}}$
time, as the mixing rate is $\otilde{d^{3}(d^{2}+\sqrt{m})}$.
\end{proof}
\begin{proof}
[Proof of Proposition~\ref{thm:LSPSD}] We define $g_{X}=g=2(d^{2}g_{1}+g_{2})$,
where for some constants $c_{1},c_{2}>0$,
\[
g_{1}(X)=M^{\T}(X\kro X)^{-1}M\qquad\text{and}\qquad g_{2}(X)=dc_{1}(\log m)^{c_{2}}M^{\T}A_{X}^{\T}W_{X}A_{X}M\,.
\]
Since $d^{2}g_{1}$ and $g_{2}$ are SSC, $g$ is also SSC due to
Lemma~\ref{lem:ssc-sum} and $\mc O^{*}(d^{3})$-symmetric due to
Lemma~\ref{lem:symmetry-addition}. As $d^{2}g_{1}$ and $g_{2}$
is SLTSC and SASC, $g$ is LTSC and ASC. Putting these together, it
follows that $g$ is $\bpar{\mc O^{*}(d^{3}),\mc O^{*}(d^{3})}$-Dikin-amenable.
Therefore, Theorem~\ref{thm:Dikin-annealing} implies that $\gcdw$
requires $\otilde{d^{5}}$ iterations of the $\dw$ with $g$. Since
the initialization and update of the Lewis weight takes $\otilde{md^{2\omega}}$
and $\otilde{md^{2(\omega-1)}}$ time \citet[Theorem 46]{lee2019solving},
the same implementation with Theorem~\ref{thm:hybridPSD} also has
the time complexity of $\otilde{md^{2(\omega-1)}}$.
\end{proof}

\subsubsection{Efficient implementation \label{proof:eff_implement}}
\begin{proof}
[Proof of Proposition~\ref{prop:oracle}] Let $v\in\R^{d_{s}}$
be a given vector, and denote $\bar{g}_{0}:=g_{1}$ and $\bar{g}_{i}:=\bar{g}_{i-1}+u_{i}u_{i}^{\T}$
for $i\in[m]$. We first prepare the column vectors $u_{i}$'s of
$U=M^{\T}A^{\T}S_{X}^{-1}$ in $\mc O(md^{2})$ time and then initialize
$\bar{g}_{0}^{-1}v$ and $\bar{g}_{0}^{-1}u_{i}$ for $i\in[m]$ in
$\mc O(md^{\omega})$ time. For $u_{i}$'s, note that $S_{X}$ can
be prepared in $\mc O(md^{2})$ time, and thus $A^{\T}S_{X}^{-1}$
takes $\mc O(md^{2})$ time due to $A\in\R^{d^{2}\times m}$. Since
each row of $M^{\T}\in\R^{d_{s}\times d^{2}}$ has at most two non-zero
entries, we can obtain $u_{i}$'s in $\mc O(md^{2})$ time.

For $\bar{g}_{0}^{-1}v$ and $\bar{g}_{0}^{-1}u_{i}$, we recall from
Lemma~\ref{prop:metricFormula} that for a vector $z\in\R^{d_{s}}$
\begin{align*}
g_{1}^{-1}z & =M^{\dagger}(X\kro X)(M^{\dagger})^{\T}z=LN(X\kro X)NL^{\T}z\,.
\end{align*}
Since each row of $L^{\T}\in\R^{d^{2}\times d_{s}}$ has at most two
non-zero entries, $w:=L^{\T}z\in\R^{d^{2}}$ can be computed in $\mc O(d^{2})$
time. From the definition of $N$, it follows that $Nw=\vec\bpar{\half(W+W^{\T})}$
for $W:=\vec^{-1}(w)\in\Rdd$, which also can be computed in $\mc O(d^{2})$
time. For $\overline{W}:=\half(W+W^{\T})$, it follows that
\[
(X\kro X)Nw=(X\kro X)\vec(\overline{W})\underset{\text{Lemma \ref{lem:Kronecker}-1}}{=}\vec(X\overline{W}X)\,,
\]
which can be computed in $\mc O(d^{\omega})$ time by the fast matrix
multiplication, and in a similar way we can compute $LN\,\vec(X\overline{W}X)$
in $\mc O(d^{2})$ time. Putting all these together, $\bar{g}_{0}^{-1}v$
can be computed in $\mc O(d^{\omega})$ time, and repeating this for
$u_{j}$'s yields $\{\bar{g}_{0}^{-1}v,\bar{g}_{0}^{-1}u_{1},\dots,\bar{g}_{0}^{-1}u_{m}\}$
in $\mc O(md^{\omega})$ time.

Starting with these initializations, we recursively use the Sherman--Morrison
formula: for $z\in\R^{d_{s}}$,
\begin{equation}
\bar{g}_{i}^{-1}z=\bar{g}_{i-1}^{-1}z-\frac{\bar{g}_{i-1}^{-1}u_{i}u_{i}^{\T}\bar{g}_{i-1}^{-1}z}{1+u_{i}^{\T}\bar{g}_{i-1}^{-1}u_{i}}\,.\label{eq:sherman-morrison}
\end{equation}
Using $\bar{g}_{i-1}^{-1}u_{j}$ and $\bar{g}_{i-1}^{-1}v$ from a
previous iteration, we can compute each of $\bar{g}_{i}^{-1}u_{j}$
and $\bar{g}_{i}^{-1}v$ in the current iteration in $\mc O(d^{2})$
time, and thus each round for update takes $\mc O(md^{2})$ time in
total. Since we iterate for $m$ rounds, Algorithm~\ref{alg:subroutine}
outputs $\bar{g}_{m}^{-1}v=g(X)^{-1}v$ in $\mc O(md^{\omega}+m^{2}d^{2})$
time.
\end{proof}
\begin{proof}
[Proof of Lemma~\ref{lem:perStep-small-m}] Here we provide details
of Algorithm~\ref{alg:perStep-small-m} in two stages -- (1) sampling
from $\ncal\bpar{0,\frac{r^{2}}{d}g(x)^{-1}}$ and (2) computation
of acceptance probability.

\paragraph{(1) Gaussian sampling:}

For simplicity, we ignore $r^{2}/d$ and illustrate how to draw $v\sim\ncal(0,g(X)^{-1})$
without full computation of $g(X)^{-1}$ in $\mc O(md^{\omega}+m^{2}d^{2})$
time.

Our approach is to compute $v:=g(X)^{-1}\left[\begin{array}{cc}
B & U\end{array}\right]w$ for $w\sim\ncal(0,I_{d^{2}+m})$, which follows the Gaussian distribution
with covariance
\begin{align*}
g(X)^{-1}\left[\begin{array}{cc}
B & U\end{array}\right]\Bpar{g(X)^{-1}\left[\begin{array}{cc}
B & U\end{array}\right]}^{\T} & =g(X)^{-1}(BB^{\T}+CC^{\T})g(X)^{-1}g(X)^{-1}\,,
\end{align*}
since $v$ is a linear transformation of the Gaussian random variable
$w$, and $BB^{\T}+CC^{\T}=g(X)$.

Denoting $w=(w_{b},w_{u})$ for $w_{b}\sim\ncal(0,I_{d^{2}})$ and
$w_{u}\sim\ncal(0,I_{m})$, we can show that $\left[\begin{array}{cc}
B & U\end{array}\right]w$ can be computed in $\mc O(d^{\omega}+md^{2})$ time as follows:
\begin{align*}
\left[\begin{array}{cc}
B & U\end{array}\right]w & =Bw_{b}+Uw_{c}=M^{\T}\underbrace{(X\kro X)^{-1/2}w_{b}}_{\text{Use Lemma \ref{eq:sherman-morrison}}}+M^{\T}A^{\T}S_{X}^{-1}w_{c}\\
 & =M^{\T}\Bpar{\vec\bpar{X^{-1/2}\vec^{-1}(w_{b})\,X^{-1/2}}+A^{\T}S_{X}^{-1}w_{c}}\,,
\end{align*}
where $\vec\bpar{X^{-1/2}\,\vec^{-1}(w_{b})\,X^{-1/2}}$ and $A^{\T}S_{X}^{-1}w_{u}$
can be computed in $\mc O(d^{\omega})$ and $\mc O(md^{2})$ time,
respectively. Since each row of $M^{\T}\in\R^{d_{s}\times d^{2}}$
has at most two non-zero entries, $\left[\begin{array}{cc}
B & U\end{array}\right]w$ can be computed in $\mc O(d^{\omega}+md^{2})$ time. Using Algorithm~\ref{alg:subroutine},
we obtain $v=g(X)^{-1}\left[\begin{array}{cc}
B & U\end{array}\right]w$ in $\mc O(md^{\omega}+m^{2}d^{2})$ time.

\paragraph{(2) Computation of acceptance probability. }

We show that this step also takes $\mc O(md^{\omega}+m^{2}d^{2})$
time. To compute $\det g(X)$, we use Algorithm~\ref{alg:subroutine}
to prepare $\{\bar{g}_{i}^{-1}u_{1},\dots,\bar{g}_{i}^{-1}u_{m}\}_{i=0}^{m}$
at $X$ and $Y=\svec^{-1}(y)$ in $\mc O(md^{\omega}+m^{2}d^{2})$
time. Recall the matrix determinant lemma:
\[
\det(A+uu^{\T})=(1+u^{\T}A^{-1}u)\,\det A\,.
\]
 Using the following recursive formula
\begin{align*}
\det(\bar{g}_{i+1}) & =\det(\bar{g}_{i}+u_{i+1}u_{i+1}^{\T})=(1+u_{i+1}^{\T}\bar{g}_{i}^{-1}u_{i+1})\,\det\bar{g}_{i}\,,
\end{align*}
we start with $\det\bar{g}_{0}=\det g_{1}=2^{d(d-1)/2}(\det X)^{-(d+1)}$
(see Lemma~\ref{lem:Kronecker}-7), which can be computed in $\mc O(d^{\omega})$
time, and compute $\det g(X)$ (and $\det g(Y)$ in the same way)
in $\mc O(md^{\omega}+m^{2}d^{2})$ time.
\end{proof}

\subsubsection{Handling approximate Lewis weights \label{proof:Handling-approximate-Lewis}}
\begin{proof}
[Proof of Lemma~\ref{lem:onestep-app-Lewis}] We just reproduce
the proof of Lemma~\ref{lem:one-step}. For $\pi\propto\exp(-f)\cdot\mathbf{1}_{K}$,
we denote 
\[
p_{x}=\ncal\Bpar{x,\frac{r^{2}}{d}g(x)^{-1}},\qquad R_{x}(z)=\frac{p_{z}(x)}{p_{x}(z)}\frac{\pi(z)}{\pi(x)},\qquad A_{x}(z)=\min\bpar{1,R_{x}(z)\,\mathbf{1}_{K}(z)}\,.
\]
Then the transition kernel of the $\dw$ started at $x$ can be written
as 
\begin{align*}
\widetilde{P}(x,dz) & =\underbrace{(1-\E_{p_{x}}[A_{x}(\cdot)])}_{=:r_{x}}\,\delta_{x}(\D z)+A_{x}(z)\,p_{x}(z)\,\D z\,.%\widetilde{P}(x,S)
\end{align*}
Thus, for $x,y\in\intk$ 
\begin{align*}
\dtv(P_{x},P_{y}) & =\underbrace{\frac{r_{x}+r_{y}}{2}}_{\textsf{I}}+\underbrace{\half\int|A_{x}(z)\,p_{x}(z)-A_{y}(z)\,p_{y}(z)|\,\D z}_{\textsf{II}}\,.
\end{align*}
\end{proof}
We note that $(1-\delta)\,\wt g_{2}\preceq g_{2}\preceq(1+\delta)\,\wt g_{2}$
and thus 
\begin{equation}
(1-\delta)\,\wt g\preceq g\preceq(1+\delta)\,\wt g\,,\label{eq:closeness-approx}
\end{equation}
and this implies $(1-\delta)\,I\preceq\wt g^{-1/2}g\wt g^{-1/2}\preceq(1+\delta)\,I$.
Hence, $(1-\delta)^{d^{2}/2}\leq\sqrt{\frac{\det g}{\det\wt g}}\leq(1+\delta)^{d^{2}/2}$
and 
\begin{align}
(1-\delta)^{d^{2}}\sqrt{\frac{\det\wt g(z)}{\det\wt g(x)}} & \leq\sqrt{\frac{\det g(z)}{\det g(x)}}\leq(1+\delta)^{d^{2}}\sqrt{\frac{\det\wt g(z)}{\det\wt g(x)}}\,.\label{eq:similar-ratio-approx}
\end{align}

With this in mind, recall that 
\[
r_{x}=1-\E_{p_{x}}[A_{x}(\cdot)]=1-\int\min\Bpar{1,\,\underbrace{\mathbf{1}_{K}(z)\frac{\exp(-f(z))}{\exp(-f(x))}}_{\eqqcolon\textsf{A}}\underbrace{\frac{p_{z}(x)}{p_{x}(z)}}_{\eqqcolon\textsf{B}}}\,p_{x}(z)\,\D z.
\]
We can bound $\textsf{A}$ in a similar way by using (\ref{eq:closeness-approx}).
As for $\textsf{B}$, 
\[
\log\text{\textsf{B}}=-\frac{d}{2r^{2}}(\snorm{z-x}_{z}^{2}-\snorm{z-x}_{x}^{2})+\half(\log\det\widetilde{g}(z)-\log\det\widetilde{g}(x))\,.
\]
As in Lemma~\ref{lem:one-step}, the second term can be bounded lower
by $\exp\Par{-3\veps}$ using (\ref{eq:similar-ratio-approx}). The
first term can be lower-bounded by invoking ASC of $g$. To see this,
ignoring the normalization constant of $g_{x}$ 
\begin{align*}
(*)= & \int\mathbf{1}\Bpar{\snorm{z-x}_{\widetilde{g}(z)}^{2}-\snorm{z-x}_{\widetilde{g}(x)}^{2}\leq2\veps\frac{r^{2}}{d}}\sqrt{\Abs{\widetilde{g}(x)}}\exp\bpar{-\half\snorm{z-x}_{\widetilde{g}(x)}^{2}}\,\D z\\
= & \int\mathbf{1}\Bpar{\snorm{z-x}_{\widetilde{g}(z)}^{2}-\snorm{z-x}_{\widetilde{g}(x)}^{2}\leq2\veps\frac{r^{2}}{d}}\sqrt{\Abs{g(x)}}\exp\bpar{-\half\snorm{z-x}_{g(x)}^{2}}\\
 & \qquad\cdot\sqrt{\Abs{\frac{\widetilde{g}(x)}{g(x)}}}\exp\bpar{-\half(\snorm{z-x}_{\widetilde{g}(x)}^{2}-\snorm{z-x}_{g(x)}^{2})}\,\D z\\
\leq & \int\mathbf{1}\Bpar{\snorm{z-x}_{\widetilde{g}(z)}^{2}-\snorm{z-x}_{\widetilde{g}(x)}^{2}\leq2\veps\frac{r^{2}}{d}}\sqrt{\Abs{g(x)}}\exp\bpar{-\half\snorm{z-x}_{g(x)}^{2}}\\
 & \qquad\cdot(1+\delta)^{d^{2}/2}\exp\bpar{\frac{\delta}{2}\snorm{z-x}_{g(x)}^{2}}\,\D z\,.
\end{align*}
Due to $\snorm{z-x}_{g(x)}^{2}\lesssim r^{2}$ w.h.p., taking $\delta=\veps/d^{10}$
leads to 
\[
(*)\leq2\int\mathbf{1}\Bpar{\snorm{z-x}_{\widetilde{g}(z)}^{2}-\snorm{z-x}_{\widetilde{g}(x)}^{2}\leq2\veps\frac{r^{2}}{d}}\sqrt{\Abs{g(x)}}\exp\bpar{-\half\snorm{z-x}_{g(x)}^{2}}\,\D z.
\]
Also, due to 
\begin{align*}
\snorm{z-x}_{\widetilde{g}(z)}^{2}-\snorm{z-x}_{\widetilde{g}(x)}^{2} & \geq(1-\delta)\,\snorm{z-x}_{g(z)}^{2}-(1+\delta)\,\snorm{z-x}_{g(x)}^{2}\\
 & =(1-\delta)\,(\snorm{z-x}_{g(z)}^{2}-\snorm{z-x}_{g(x)}^{2})-2\delta\,\snorm{z-x}_{g(x)}^{2}\,,
\end{align*}
we have 
\begin{align*}
(*) & \leq2\int\mathbf{1}\Bpar{\snorm{z-x}_{g(z)}^{2}-\snorm{z-x}_{g(x)}^{2}\leq(2\veps(1-\delta)^{-1}+\veps)\,\frac{r^{2}}{d}}\sqrt{\Abs{g(x)}}e^{-\half\snorm{z-x}_{g(x)}^{2}}\,\D z\leq6\veps
\end{align*}
by invoking ASC of $g$ in the last inequality. Putting these together,
$\msf I\leq\half+\mc O(\veps)$. For $\msf{II}$, we can follow the
proof of Lemma~\ref{lem:one-step} to show $\msf{II}\leq\frac{1}{4}+\mc O(\veps)$,
and every technical issue can be resolved by repeating the same techniques
above.

\begin{acknowledgement*}
This work was supported in part by NSF awards CCF-2007443 and CCF-2134105.
\end{acknowledgement*}
\bibliography{main}

\appendix
\addtocontents{toc}{\protect\setcounter{tocdepth}{1}} 

\section{Backgrounds on matrix algebra}

\subsection{Matrix identities}

We collect algebraic identities related to trace, vectorization, Kronecker
and Hadamard product. 
\global\long\def\vec{\textup{\textsf{vec}}}%
 
\begin{lem}
[Kronecker product] \label{lem:Kronecker} For $A,B,C,D\in\Rdd$
and $M$ in Definition~\ref{def:linearOperators},

\begin{multicols}{2}
\begin{enumerate}
\item $(A\otimes B)\,\vec(C)=\tr(BCA^{\T})$.
\item $\vec(A)^{\T}(B\otimes C)\vec(D)=\tr(DB^{\T}A^{\T}C)$.
\item $(A\otimes B)(C\otimes D)=AC\otimes BD$.
\item $(A\otimes B)^{-1}=A^{-1}\otimes B^{-1}$.
\item $(A\otimes B)^{\T}=A^{\T}\otimes B^{\T}$.
\item $\tr(A\otimes B)=\tr(A)\tr(B)$.
\item $\det\bpar{M^{\T}(A\otimes A)M}=2^{\nicefrac{d(d-1)}{2}}(\det A)^{d+1}$.
\item[]
\end{enumerate}
\end{multicols}
\end{lem}

\begin{lem}
[Hadamard product] \label{lem:Hadamard} Let $A,B,C,D\in\Rdd$,
$x,y\in\Rd$, and $D_{1},D_{2}\in\Rdd$ be diagonal matrices.

\begin{multicols}{2}
\begin{enumerate}
\item $(A\circ B)y=\diag(A\,\Diag(y)B^{\T})$.
\item $x^{\T}(A\circ B)y=\tr(\Diag(x)A\,\Diag(y)B^{\T})$.
\item $D_{1}(A\hada B)=(D_{1}A)\hada B=A\hada(D_{1}B)$.
\item $(A\hada B)D_{2}=(AD_{2})\hada B=A\hada(BD_{2})$.
\item $(A\otimes B)\circ(C\otimes D)=(A\circ C)\otimes(B\circ D)$. \item[]
\end{enumerate}
\end{multicols}
\end{lem}

\subsection{Matrix calculus \label{app:matrixCalculus}}

Let $g(x):\Rd\to\Rdd$ be a matrix function. Its gradient at $x$,
denoted by $\Dd g(x)$, is the third-order tensor defined by $(\Dd g(x))_{ijk}=\pderiv{g_{ij}(x)}{x_{k}}$.
Unless specified otherwise, the multiplication between higher-order
tensors and a matrix of size $d\times d$ is running over $(i,j)$-entries.
For instance, for a matrix $M\in\Rdd$ the product $\Dd g(x)\cdot M$
indicates the third-order tensor defined by
\[
(\Dd g(x)\,M)_{\cdot,\cdot,k}=(\Dd g(x))_{\cdot,\cdot,k}M\text{ for each }k\in[d]\,.
\]
In the same way, the trace is applied to a matrix spanned by $(i,j)$-entries,
i.e.,
\[
\bpar{\tr(\Dd g(x))}_{k}=\tr\Bpar{\bpar{\Dd g(x)}_{\cdot,\cdot,k}}\,.
\]

For $\vphi:\Rd\to\R$ with $\vphi(\cdot):=\log\det g(\cdot)$, its
gradient and the directional derivative in $h\in\Rd$ are
\begin{equation}
\grad\vphi(x)=\tr\bpar{g(x)^{-1}\Dd g(x)}\,,\qquad\text{and}\qquad\grad\vphi(x)\cdot h=\tr\bpar{g(x)^{-1}\Dd g(x)[h]}\,.\label{eq:gradLogDet}
\end{equation}
For the Hessian of $\vphi$, using the product rule and 
\begin{equation}
\Dd(g^{-1})(x)=-g(x)^{-1}\Dd g(x)\,g(x)^{-1}\,,\label{eq:diffInverse}
\end{equation}
we obtain
\begin{align}
\hess\vphi(x) & =\Dd\tr\bpar{g(x)^{-1}\Dd g(x)}=-\tr\bpar{g(x)^{-1}\Dd g(x)\,g(x)^{-1}\Dd g(x)}+\tr\bpar{g(x)^{-1}\Dd^{2}g(x)}\nonumber \\
 & =\tr\bpar{g(x)^{-1}\Dd^{2}g(x)}-\snorm{g(x)^{-\half}\Dd g(x)\,g(x)^{-\half}}_{F}^{2}\,,\label{eq:hessLogDet}
\end{align}
where $\Dd^{2}g(x)$ is the fourth-order tensor defined by $(\Dd^{2}g(x))_{ijkl}=\frac{\de[g(x)]_{ij}}{\de x_{k}\de x_{l}}$.

We now present formulas for the Hessian and its inverse of $\phi(\cdot)=-\log\det(\cdot)$
on $\pd$.
\begin{proof}
[Proof of Proposition \ref{prop:metricFormula}] By setting $g(X)=X$
and $\phi(X)=-\vphi(X)$ above, \eqref{eq:hessLogDet} implies that
for a symmetric matrix $H\in\mbb S^{d}$
\begin{align}
\hess\phi(X)[H,H] & =\snorm{X^{-\half}HX^{-\half}}_{F}^{2}=\tr(X^{-1}HX^{-1}H)\label{eq:2ndDiffLogDet}\\
 & =\vec(H)^{\T}(X^{-1}\otimes X^{-1})\vec(H)=\vec(H)^{\T}(X\otimes X)^{-1}\vec(H)\,,\nonumber 
\end{align}
where the last equality follows from Lemma~\ref{lem:Kronecker}.
When representing $X$ and $H$ in $\R^{d_{s}}$ space with notations
$x:=\svec(X)$ and $h:=\svec(H)$, the definition of $M$ (see Definition~\ref{def:linearOperators})
turns \eqref{eq:2ndDiffLogDet} into
\[
\hess\phi(x)[h,h]=h^{\T}M^{\T}(X\otimes X)^{-1}Mh\,,
\]
and thus $g_{X}:=\nabla_{x}^{2}\phi(x)=\nabla_{X}^{2}\phi(X)$ equals
$M^{\T}(X\otimes X)^{-1}M$. The formula for the inverse, $g_{X}^{-1}=M^{\dagger}(X\otimes X)(M^{\dagger})^{\T}$,
is immediate from \citet{magnus1980elimination}, and another part
follows from $M^{\dagger}=LN$ and $N^{\T}=N$ \citet[Lemma 3.6 and Lemma 2.1]{magnus1980elimination}.
\end{proof}

\section{Self-concordant barriers for linear constraints}

We collect details on self-concordant barriers for linear constraints,
$P=\{x\in\Rd:Ax\geq b\}$ with $A\in\R^{m\times d}$ and $b\in\R^{m}$:
the logarithmic, volumetric, and Lewis-weight barrier/metric. Recall
the notations used in the paper: $s_{x}=\diag(Ax-b)\in\R^{m}$, $S_{x}=\Diag(s_{x})\in\R^{m\times m}$,
and $A_{x}=S_{x}^{-1}A\in\R^{m\times d}$. Also, $s_{x,h}=A_{x}h\in\R^{m}$
and $S_{x,h}=\Diag(s_{x,h})\in\R^{m\times m}$. Let $h\in\Rd$.

\subsection{Logarithmic barriers \label{proof:linear-log-barrier}}

 For $x\in P$, the logarithmic barrier (or log-barrier) and the
Hessian metric are given by
\[
\phi_{\log}(x):=-\sum_{i=1}^{m}\log(a_{i}^{\T}x-b)\,,\qquad\text{and}\qquad g(x)=\hess\phi(x)=A_{x}^{T}A_{x}\,.
\]

\begin{claim}
\label{claim:diffLogBarrier} $\Dd S_{x}[h]=\Diag(Ah)$ and $\Dd S_{x}^{-1}[h]=-S_{x}^{-1}S_{x,h}$.
Also, $\Dd g(x)[h]=-2A_{x}^{\T}S_{x,h}A_{x}$ and $\Dd^{2}g(x)[h,h]=6A_{x}^{\T}S_{x,h}^{2}A_{x}\succeq0$.
\end{claim}

\begin{proof}
The first is obvious from differentiation of $S_{x}=\Diag(Ax-b)$
w.r.t. $x$. As for the second,
\begin{align*}
\Dd S_{x}^{-1}[h] & =-S_{x}^{-1}\Dd S_{x}[h]\,S_{x}^{-1}=-S_{x}^{-1}\Diag(Ah)S_{x}^{-1}=-S_{x}^{-1}\Diag(A_{x}h)=-S_{x}^{-1}S_{x,h}\,.
\end{align*}
As for the third and fourth, as $g(x)=A^{\T}S_{x}^{-2}A$,
\begin{align*}
\Dd g(x)[h] & =A^{\T}\Dd S_{x}^{-2}[h]\,A=-2A^{\T}S_{x}^{-3}\Dd S_{x}[h]A=-2A_{x}^{\T}S_{x}^{-1}\Diag(Ah)A_{x}=-2A_{x}^{\T}S_{x,h}A_{x}\,.\\
\Dd^{2}g(x)[h,h] & =-2A^{\T}\Dd S_{x}^{-3}[h]\,\Diag(Ah)A=6A^{\T}S_{x}^{-4}\Dd S_{x}[h]\,\Diag(Ah)A=6A_{x}^{\T}S_{x,h}^{2}A_{x}\,.\qedhere
\end{align*}
\end{proof}

\subsection{Volumetric barriers \label{proof:linear-volumetric}}

\citet{vaidya1996new} introduced the \emph{volumetric barrier} for
$P$, defined by 
\[
\phi_{\vol}(x)=\half\,\log\det\bpar{\hess\phi_{\log}(x)}=\half\,\log\det(A_{x}^{\T}A_{x})\,.
\]

\begin{claim}
$\grad\phi_{\vol}(x)=-A_{x}^{\T}\sigma_{x}$ and $\hess\phi_{\vol}(x)=A_{x}^{\T}(3\Sigma_{x}-2P_{x}^{(2)})A_{x}$.
\end{claim}

\begin{proof}
For $P_{x}:=P(A_{x})$, using \eqref{eq:gradLogDet} with Claim~\ref{claim:diffLogBarrier}
and apply Lemma~\ref{lem:Hadamard} in (i),
\begin{align*}
\grad\phi_{\vol}(x)[h] & =-\tr\bpar{(A_{x}^{\T}A_{x})^{-1}A_{x}^{\T}S_{x,h}A_{x}}=-\tr(P_{x}S_{x,h})\underset{\text{(i)}}{=}-1^{\T}(P_{x}\circ I_{m})s_{x,h}=-h^{\T}A_{x}^{\T}\sigma_{x}\,,
\end{align*}

For the Hessian of $\phi_{\vol}$, let $g(x)=A_{x}^{\T}A_{x}$ and
then by \eqref{eq:hessLogDet}, 
\[
\hess\phi_{\vol}(x)[h,h]=\half\,\bpar{\tr(g^{-1}\Dd^{2}g[h,h])-\tr(g^{-1}\Dd g[h]\,g^{-1}\Dd g[h])}\,.
\]
As for the first term, Claim~\ref{claim:diffLogBarrier} leads to
\begin{align*}
\tr(g^{-1}\Dd^{2}g[h,h]) & =6\tr(g^{-1}A_{x}^{\T}S_{x,h}^{2}A_{x})=6\tr(P_{x}S_{x,h}IS_{x,h})=6h^{\T}A_{x}^{\T}(P_{x}\circ I)A_{x}h=6h^{\T}A_{x}^{\T}\Sigma_{x}A_{x}h\,.
\end{align*}
As for the second term, 
\begin{align*}
\tr(g^{-1}\Dd g[h]\,g^{-1}\Dd g[h]) & =4\tr(P_{x}S_{x,h}P_{x}S_{x,h})=4(A_{x}h)^{\T}(P_{x}\circ P_{x})(A_{x}h)=4h^{\T}A_{x}^{\T}P_{x}^{(2)}A_{x}h\,.
\end{align*}
Hence, $\Dd^{2}\phi_{\vol}(x)[h,h]=h^{\T}A_{x}^{\T}(3\Sigma_{x}-2P_{x}^{(2)})A_{x}h$,
which completes the proof.
\end{proof}
\begin{claim}
\label{claim:schurProjection} $P_{x}^{(2)}\preceq\Sigma_{x}$, so
$A_{x}^{\T}\Sigma_{x}A_{x}\preceq\hess\phi_{\vol}(x)\preceq3A_{x}^{\T}\Sigma_{x}A_{x}$.
\end{claim}

\begin{proof}
Due to $\Sigma_{x}=P_{x}\circ I$, it suffices to show $h^{\T}P_{x}\circ(I-P_{x})\,h\geq0$
for any $h\in\Rd$. Since $P_{x}$ and $I-P_{x}$ are orthogonal projections,
for $H=\Diag(h)$ and $C:=P_{x}H(I-P_{x})$ ,
\begin{align*}
h^{\T}P_{x}\circ(I-P_{x})\,h & =\tr\bpar{HP_{x}H(I-P_{x})}=\tr\bpar{(I-P_{x})HP_{x}P_{x}H(I-P_{x})}=\tr(C^{\T}C)\geq0\,.\qedhere
\end{align*}
\end{proof}

\subsubsection{Derivatives of leverage scores and projection matrices}

 We derive formulas for derivatives of leverage scores, orthogonal
projections, and so on.
\begin{lem}
\label{lem:calculusLeverage} For $x,h\in\Rd$, let $P_{x}=A_{x}(A_{x}^{\T}A_{x})^{-1}A_{x}^{\T}$,
$\Sigma_{x}=\Diag(P_{x})$, and $\Lambda_{x}=\Sigma_{x}-P_{x}^{(2)}$.
Denote $\theta(x):=A_{x}^{\T}\Sigma_{x}A_{x}$.
\begin{itemize}
\item \textup{\citet[Lemma 24]{lee2019solving}} $\Sigma_{x,h}'=-2\Diag(\Lambda_{x}s_{x,h})=2\bpar{\Diag(P_{x}S_{x,h}P_{x})-\Sigma_{x}S_{x,h}}$.
\item \textup{\citet[Lemma 49]{lee2019solving}} $P_{x,h}'=-P_{x}S_{x,h}-S_{x,h}P_{x}+2P_{x}S_{x,h}P_{x}$.
\item $\Lambda_{x,h}'=-2\Diag(\Lambda_{x}s_{x,h})+2P_{x}\circ P_{x}S_{x,h}+2S_{x,h}P_{x}\circ P_{x}-2(P_{x}S_{x,h}P_{x})\circ P_{x}-2P_{x}\circ(P_{x}S_{x,h}P_{x})$.
\item $\Sigma_{x,h}''=6S_{x,h}\Sigma_{x}S_{x,h}+8\Diag(P_{x}S_{x,h}P_{x}S_{x,h}P_{x})-6\Diag(P_{x}S_{x,h}^{2}P_{x})-8\Diag(S_{x,h}P_{x}S_{x,h}P_{x})$.
\item $\Dd\theta(x)[h]=-2A_{x}^{\T}\Sigma_{x}S_{x,h}A_{x}+A_{x}^{\T}\Sigma_{x,h}'A_{x}$.
\item $\Dd^{2}\theta(x)[h,h]=6A_{x}^{\T}S_{x,h}\Sigma_{x}S_{x,h}A_{x}-4A_{x}^{\T}\Sigma_{x,h}'S_{x,h}A_{x}+A_{x}^{\T}\Sigma_{x,h}''A_{x}$.
Equivalently, 
\begin{align*}
\Dd^{2}\theta(x)[h,h] & =20A_{x}^{\T}S_{x,h}\Sigma_{x}S_{x,h}A_{x}-16A_{x}^{\T}\Diag(S_{x,h}P_{x}S_{x,h}P_{x})A_{x}\\
 & \qquad-6A_{x}^{\T}\Diag(P_{x}S_{x,h}^{2}P_{x})A_{x}+8A_{x}^{\T}\Diag(P_{x}S_{x,h}P_{x}S_{x,h}P_{x})A_{x}.
\end{align*}
\end{itemize}
\end{lem}

\begin{proof}
As for the third item,
\begin{align*}
 & \Lambda_{x,h}'=\Sigma_{x,h}'-P_{x,h}'\circ P_{x}-P_{x}\circ P_{x,h}'\\
 & =-2\Diag(\Lambda_{x}s_{x,h})-(-P_{x}S_{x,h}-S_{x,h}P_{x}+2P_{x}S_{x,h}P_{x})\circ P_{x}-P_{x}\circ(-P_{x}S_{x,h}-S_{x,h}P_{x}+2P_{x}S_{x,h}P_{x})\\
 & \underset{\text{(i)}}{=}-2\Diag(\Lambda_{x}s_{x,h})+2P_{x}\circ P_{x}S_{x,h}+2S_{x,h}P_{x}\circ P_{x}-2(P_{x}S_{x,h}P_{x})\circ P_{x}-2P_{x}\circ(P_{x}S_{x,h}P_{x})\,,
\end{align*}
where in (i) we used $D(A\hada B)=(DA)\circ B=A\hada(DB)$ and $(A\hada B)D=(AD)\hada B=A\circ(BD)$\footnote{This property allows us to write $DA\hada B$ without parenthesis.}
for a diagonal matrix $D\in\Rdd$ (Lemma~\ref{lem:Hadamard}). 

As for the fourth item,
\begin{align*}
 & \Sigma_{x,h}''=-2\Dd\bpar{\Diag(\Lambda_{x}s_{x,h})}[h]=-2\Diag(\Lambda_{x,h}'s_{x,h})+2\Diag(\Lambda_{x}S_{x,h}s_{x,h})\\
 & =-2\Diag\bpar{\bbrack{-2\Diag(\Lambda_{x}s_{x,h})+2P_{x}\circ P_{x}S_{x,h}+2S_{x,h}P_{x}\circ P_{x}-2(P_{x}S_{x,h}P_{x})\circ P_{x}-2P_{x}\circ(P_{x}S_{x,h}P_{x})}s_{x,h}}\\
 & \qquad+2\Diag(\Lambda_{x}S_{x,h}s_{x,h})\\
 & =4\Diag(\cred{\Lambda_{x}}s_{x,h})\cblue{S_{x,h}}-4\Diag(P_{x}\circ P_{x}S_{x,h}s_{x,h})-4\Diag(S_{x,h}P_{x}\circ P_{x}s_{x,h})\\
 & \qquad+4\Diag\bpar{(P_{x}S_{x,h}P_{x})\circ P_{x}s_{x,h}}+4\Diag\bpar{P_{x}\circ(P_{x}S_{x,h}P_{x})s_{x,h}}+2\Diag(\cred{\Lambda_{x}}S_{x,h}s_{x,h})\\
 & =4\Diag\bpar{\cblue{S_{x,h}}\cred{(\Sigma_{x}-P_{x}\circ P_{x})}s_{x,h}}-4\Diag(P_{x}\circ P_{x}S_{x,h}s_{x,h})-4\Diag(S_{x,h}P_{x}\circ P_{x}s_{x,h})\\
 & \qquad+4\Diag\bpar{(P_{x}S_{x,h}P_{x})\circ P_{x}s_{x,h}}+4\Diag\bpar{P_{x}\circ(P_{x}S_{x,h}P_{x})s_{x,h}}+2\Diag\bpar{\cred{(\Sigma_{x}-P_{x}\circ P_{x})}S_{x,h}s_{x,h}}\\
 & =\ccyan{4\Diag(S_{x,h}\Sigma_{x}s_{x,h})}-6\Diag(P_{x}\circ P_{x}S_{x,h}s_{x,h})-8\Diag(S_{x,h}P_{x}\circ P_{x}s_{x,h})\\
 & \qquad+4\Diag\bpar{(P_{x}S_{x,h}P_{x})\circ P_{x}s_{x,h}}+4\Diag\bpar{P_{x}\circ(P_{x}S_{x,h}P_{x})s_{x,h}}+\ccyan{2\Diag(\Sigma_{x}S_{x,h}s_{x,h})}\\
 & =\text{\ensuremath{\ccyan{6\Diag(S_{x,h}\Sigma_{x}s_{x,h})}}}-6\Diag(\cblue{P_{x}\circ P_{x}S_{x,h}s_{x,h}})-8\Diag(\cblue{S_{x,h}P_{x}\circ P_{x}s_{x,h}})\\
 & \qquad+4\Diag\bpar{\cblue{(P_{x}S_{x,h}P_{x})\circ P_{x}s_{x,h}}}+4\Diag\bpar{\cblue{P_{x}\circ(P_{x}S_{x,h}P_{x})s_{x,h}}}\\
 & \underset{\text{(i)}}{=}6S_{x,h}\Sigma_{x}\Diag(s_{x,h})-6\Diag\Bpar{\diag\bpar{P_{x}S_{x,h}(P_{x}S_{x,h})^{\T}}}-8\Diag\Bpar{\diag(S_{x,h}P_{x}S_{x,h}P_{x}^{\T})}\\
 & \qquad+4\Diag(P_{x}S_{x,h}P_{x}S_{x,h}P_{x})+4\Diag\bpar{P_{x}S_{x,h}(P_{x}S_{x,h}P_{x})^{\T}}\\
 & =6S_{x,h}\Sigma_{x}S_{x,h}-6\Diag(P_{x}S_{x,h}^{2}P_{x})-8\Diag(S_{x,h}P_{x}S_{x,h}P_{x})+8\Diag(P_{x}S_{x,h}P_{x}S_{x,h}P_{x})\,,
\end{align*}
where in (i) we applied Lemma~\ref{lem:Hadamard}-1 to the terms
with blue. 

Applying the product rule to $\theta(x)=A_{x}^{\T}\Sigma_{x}A_{x}=A^{\T}S_{x}^{-2}\Sigma_{x}A,$
\begin{align*}
\Dd\theta[h] & =-2A^{\T}S_{x}^{-3}\Sigma_{x}\Diag(Ah)A+A^{\T}S_{x}^{-2}\Sigma_{x,h}'A=-2A_{x}^{\T}\Sigma_{x}S_{x,h}A_{x}+A_{x}^{\T}\Sigma_{x,h}'A_{x}\,,\\
\Dd^{2}\theta[h,h] & =6A_{x}^{\T}S_{x,h}\Sigma_{x}S_{x,h}A_{x}-2A_{x}^{\T}\Sigma_{x,h}'S_{x,h}A_{x}-2A_{x}^{\T}S_{x,h}\Sigma_{x,h}'A_{x}+A_{x}^{\T}\Sigma_{x,h}''A_{x}\\
 & =6A_{x}^{\T}S_{x,h}\Sigma_{x}S_{x,h}A_{x}-4A_{x}^{\T}\Sigma_{x,h}'S_{x,h}A_{x}+A_{x}^{\T}\Sigma_{x,h}''A_{x}\,.
\end{align*}
By substituting $\Sigma_{x,h}'$ and $\Sigma_{x,h}''$ with our formulas
above, 
\begin{align*}
 & \Dd^{2}\theta[h,h]=6A_{x}^{\T}S_{x,h}\Sigma_{x}S_{x,h}A_{x}-4A_{x}^{\T}\Sigma_{x,h}'S_{x,h}A_{x}+A_{x}^{\T}\Sigma_{x,h}''A_{x}\\
 & =6A_{x}^{\T}S_{x,h}\Sigma_{x}S_{x,h}A_{x}+8A_{x}^{\T}\bpar{\Sigma_{x}S_{x,h}-\Diag(P_{x}S_{x,h}P_{x})}S_{x,h}A_{x}\\
 & \qquad+A_{x}^{\T}\Bpar{6S_{x,h}\Sigma_{x}S_{x,h}-6\Diag(P_{x}S_{x,h}^{2}P_{x})-8\Diag(S_{x,h}P_{x}S_{x,h}P_{x})+8\Diag(P_{x}S_{x,h}P_{x}S_{x,h}P_{x})}A_{x}\\
 & =20A_{x}^{\T}S_{x,h}\Sigma_{x}S_{x,h}A_{x}-16A_{x}^{\T}\Diag(S_{x,h}P_{x}S_{x,h}P_{x})A_{x}-6A_{x}^{\T}\Diag(P_{x}S_{x,h}^{2}P_{x})A_{x}\\
 & \qquad+8A_{x}^{\T}\Diag(P_{x}S_{x,h}P_{x}S_{x,h}P_{x})A_{x}\,.\qedhere
\end{align*}
\end{proof}

\subsection{Lewis-weight metric \label{proof:linear-LW}}

We recall preliminaries on the Lewis weights. Particularly, the leverage
scores are simply the $\ell_{2}$-Lewis weights.
\begin{lem}
[\citet{lee2019solving}] \label{lem:usefulFactLewis} Let $W_{x}=\Diag(w_{x}(A_{x}))\in\pd$
be the $\ell_{p}$-Lewis weights and $g(x)=A_{x}^{\T}W_{x}A_{x}$
the Lewis-weights metric, and $h\in\Rd$.
\begin{itemize}
\item \textup{(Lemma 26)} $\max_{i\in[m]}\frac{[\sigma(W_{x}^{1/2}A_{x})]_{i}}{(w_{x})_{i}}\leq2m^{\frac{2}{p+2}}$.
\item \textup{(Lemma 33)} $\norm{A_{x}h}_{W_{x}}=\norm h_{g(x)}$ and $\norm{A_{x}h}_{\infty}\leq\sqrt{2}m^{\frac{1}{p+2}}\norm h_{g(x)}$.
\item \textup{(Lemma 34)} $\norm{W_{x}^{-1}w_{x,h}'}_{W_{x}}\leq p\,\norm h_{g(x)}$.
\end{itemize}
\end{lem}

Next is a directional derivative of the $\ell_{p}$-Lewis weight of
$A_{x}$.
\begin{lem}
[\citet{lee2019solving}, Lemma 24] \label{lem:DWh} The directional
derivative of the $\ell_{p}$-Lewis weight $W_{x}$ in direction $h\in\Rd$
is
\[
W_{x,h}':=\Dd W_{x}[h]=-2\,\Diag(\Lambda_{x}G_{x}^{-1}W_{x}s_{x,h})=-\Diag(W_{x}^{\half}N_{x}W_{x}^{\half}s_{x,h})\,,
\]
where $\Lambda_{x}\defeq W_{x}-P_{x}^{(2)}$, $\bar{\Lambda}_{x}\defeq W_{x}^{-\half}\Lambda_{x}W_{x}^{-\half}$,
$G_{x}\defeq W_{x}-\bpar{1-\frac{2}{p}}\Lambda_{x}$, and $N_{x}\defeq2\bar{\Lambda}_{x}(I-c_{p}\bar{\Lambda}_{x})^{-1}$.
\end{lem}

It is known that these matrices satisfy
\begin{align}
P_{x}^{(2)}\preceq W_{x}\preceq I\,,\label{eq:lewisBasic-PWI}\\
\Lambda_{x}\preceq W_{x}\,,\label{eq:lewisBasic-LW}\\
\frac{2}{p}W_{x}\preceq G_{x}\preceq W_{x}\,, & \text{ which implies }W_{x}^{-1}\preceq G_{x}^{-1}\preceq\frac{p}{2}W_{x}^{-1}\text{ and }I\preceq W_{x}^{\half}G_{x}^{-1}W_{x}^{\half}\preceq\frac{p}{2}I\,.\label{eq:lewisBasic-WGW}
\end{align}
We can also compute the second-order directional derivative of $W_{x}$
in direction $h\in\Rd$.
\begin{lem}
[Second-order derivative of $W_x$] \label{lem:second-deriv-Lewis}
Let $w_{x}\in\R^{m}$ be the $\ell_{p}$-Lewis weight, $\Gamma\in\R_{\geq0}^{m\times m}$
a diagonal matrix, and $h\in\Rd$. Then,
\begin{align}
 & W_{x,h}''=-\Diag\bpar{\half W_{x}^{-\half}W_{x,h}'N_{x}W_{x}^{\half}s_{x,h}+W_{x}^{\half}N_{x,h}'W_{x}^{\half}s_{x,h}+\half W_{x}^{\half}N_{x}W_{x}^{-\half}W_{x,h}'s_{x,h}+2\Lambda_{x}G_{x}^{-1}W_{x}s_{x,h}^{2}}\,,\nonumber \\
 & \tr(\Gamma W_{x,h}'')=-\half\,\tr\bpar{\Gamma\,\Diag(\underbrace{W_{x}^{-\half}W_{x,h}'N_{x}W_{x}^{\half}s_{x,h}}_{\textup{\text{I}}})}-\tr\bpar{\Gamma\,\Diag(\underbrace{W_{x}^{\half}N_{x,h}'W_{x}^{\half}s_{x,h}}_{\textup{\text{II}}})}\nonumber \\
 & \qquad\qquad\qquad-\half\,\tr\bpar{\Gamma\,\Diag(\underbrace{W_{x}^{\half}N_{x}W_{x}^{-\half}W_{x,h}'s_{x,h}}_{\textup{\text{III}}})}-2\tr\bpar{\Gamma\,\Diag(\underbrace{\Lambda_{x}G_{x}^{-1}W_{x}S_{x,h}s_{x,h}}_{\textup{\text{IV}}})}\,,\label{eq:trGamma}\\
 & \Dd^{2}(A_{x}^{\T}W_{x}A_{x})[h,h]=6A_{x}^{\T}S_{x,h}W_{x}S_{x,h}A_{x}-4A_{x}^{\T}W_{x,h}'S_{x,h}A_{x}+A_{x}^{\T}W_{x,h}''A_{x}\label{eq:LW-second-derv}
\end{align}
where $\snorm{\textup{I}}_{W_{x}^{-1}}\lesssim p^{3}m^{\frac{1}{p+2}}\norm h_{\theta}^{2}$,
$\norm{\textup{\text{II}}}_{W_{x}^{-1}}\lesssim p^{3.5}\norm h_{\theta}^{2}$,
$\norm{\textup{\text{III}}}_{W_{x}^{-1}}\lesssim p^{3}m^{\frac{1}{p+2}}\,\norm h_{\theta}^{2}$,
and $\norm{\textup{\text{IV}}}_{W_{x}^{-1}}\lesssim pm^{\frac{1}{p+2}}\norm h_{\theta}^{2}$.
Here, $\lesssim$ hides universal constants and poly-logarithmic factors
in $m$.
\end{lem}

\begin{proof}
The formula for $W_{x,h}''$ follows from differentiating the formula
for $W_{x,h}'$ (Lemma~\ref{lem:DWh}). The dual local norms of I\textasciitilde IV
can be bounded as follows:
\begin{align*}
\norm{\text{I}}_{W_{x}^{-1}} & =\norm{W_{x}^{-1}W_{x,h}'N_{x}W_{x}^{\half}s_{x,h}}_{2}\leq\underbrace{\norm{W_{x}^{-1}W_{x,h}'}_{2}}_{\text{Lemma \ref{lem:LS-comp-tool}-2}}\underbrace{\norm{N_{x}}_{2}}_{\text{Lemma \ref{lem:LS-comp-tool}-1}}\norm{W_{x}^{\half}s_{x,h}}_{2}\lesssim p^{3}m^{\frac{1}{p+2}}\norm h_{\theta}^{2}\,,\\
\norm{\text{II}}_{W_{x}^{-1}} & =\norm{N_{x,h}'W_{x}^{\half}s_{x,h}}_{2}\leq\underbrace{\norm{I+N_{x}}_{2}}_{\text{Lemma \ref{lem:LS-comp-tool}-1}}\underbrace{\norm{(I+N_{x})^{-\half}N_{x,h}'(I+N_{x})^{-\half}}_{2}}_{\text{Lemma \ref{lem:LS-comp-tool}-3}}\norm{W_{x}^{\half}s_{x,h}}_{2}\lesssim p^{3.5}\norm h_{\theta}^{2}\,,\\
\norm{\text{III}}_{W_{x}^{-1}} & =\norm{N_{x}W_{x}^{-\half}W_{x,h}'s_{x,h}}_{2}\leq\underbrace{\norm{N_{x}}_{2}}_{\text{Lemma \ref{lem:LS-comp-tool}-1}}\underbrace{\norm{W_{x}^{-1}W_{x,h}'}_{2}}_{\text{Lemma \ref{lem:LS-comp-tool}-2}}\norm{W_{x}s_{x,h}}_{2}\lesssim p^{3}m^{\frac{1}{p+2}}\,\norm h_{\theta}^{2}\,,\\
\norm{\text{IV}}_{W_{x}^{-1}}^{2} & =s_{x,h}^{\T}S_{x,h}W_{x}G_{x}^{-1}\underbrace{\Lambda_{x}W_{x}^{-1}\Lambda_{x}}_{\preceq W_{x}\ \text{\eqref{eq:lewisBasic-LW}}}G_{x}^{-1}W_{x}S_{x,h}s_{x,h}\leq s_{x,h}^{\T}S_{x,h}W_{x}\underbrace{G_{x}^{-1}W_{x}G_{x}^{-1}}_{\preceq\frac{p^{2}}{4}W_{x}^{-1}\ \text{\eqref{eq:lewisBasic-WGW}}}W_{x}S_{x,h}s_{x,h}\\
 & \leq p^{2}s_{x,h}^{\T}W_{x}^{\half}S_{x,h}^{2}W_{x}^{\half}s_{x,h}\leq p^{2}\norm{s_{x,h}}_{\infty}^{2}\norm h_{\theta}^{2}\leq p^{2}m^{\frac{2}{p+2}}\norm h_{\theta}^{4}\,,
\end{align*}
where we used Lemma~\ref{lem:usefulFactLewis}-2 in the last inequality.
\end{proof}
Next, we recall bounds on the derivatives of matrices relevant to
Lewis weights.
\begin{lem}
[\citet{lee2019solving}] \label{lem:LS-comp-tool} Let $Ax\geq b$
and $h\in\Rd$. For $c_{p}=1-2/p$ with $p>2$, let $\bar{\Lambda}_{x}:=W_{x}^{-\half}\Lambda_{x}W_{x}^{-\half}=I-W_{x}^{-\half}P_{x}^{(2)}W_{x}^{-\half}$,
$N_{x}\defeq2\bar{\Lambda}_{x}(I-c_{p}\bar{\Lambda}_{x})^{-1}$ and
$\theta_{x}=A_{x}^{\T}W_{x}A_{x}$.
\begin{itemize}
\item \textup{(Lemma 31)} $N_{x}$ is symmetric and $0\preceq N_{x}\preceq pI$.
\item \textup{(Lemma 34)} $\norm{W_{x}^{-1}w_{x,h}}_{\infty}\leq p(\sqrt{2}m^{\frac{1}{p+2}}+p/2)\,\norm h_{\theta_{x}}$.
\item \textup{(Lemma 37)} $\norm{(I+N_{x})^{-\half}\Dd N_{x}[h]\,(I+N_{x})^{-\half}}_{2}\leq4p^{5/2}\norm h_{\theta_{x}}$.
\end{itemize}
\end{lem}

Lastly, we remind a result about closeness of the Lewis weights at
close-by points.
\begin{lem}
[\citet{lee2019solving}] \label{lem:weight-close} In the same
setting above, let $x_{t}=x+th$, $s_{t}=s_{x_{t}}$, $w_{t}=w_{x_{t}}$,
and $z_{t,\alpha}\in\R^{m}$ be a vector defined by $[z_{t,\alpha}]_{i}:=\frac{\D}{\D t}\log\Bpar{\frac{[w_{t,i}]^{\alpha}}{s_{t,i}}}$.
Then,
\[
\norm{z_{t}}_{\infty}\leq\bpar{\sqrt{2}(1+|\alpha|p)m^{\frac{1}{p+2}}+p\,|\alpha|\,\max(1,p/2)}\,\norm h_{A_{t}^{\T}W_{t}A_{t}}\,.
\]
\end{lem}

Now we present an auxiliary result showing HSC of the Lewis-weight
metric.
\begin{lem}
\label{lem:Lw-hsc} The metric $g(x)=cA_{x}^{\T}W_{x}A_{x}$ is HSC
for $c=c_{1}(\log m)^{c_{2}}d^{1/2}$ with some constants $c_{1},c_{2}>0$, 
\end{lem}

\begin{proof}
Let $\theta(x)=A_{x}^{\T}W_{x}A_{x}$ and $h\in\Rd$. From \eqref{eq:LW-second-derv},
\begin{align}
\Dd^{2}\theta[h,h,h,h] & =6s_{x,h}^{\T}S_{x,h}W_{x}S_{x,h}s_{x,h}-4s_{x,h}^{\T}W_{x,h}'S_{x,h}s_{x,h}+s_{x,h}^{\T}W_{x,h}''s_{x,h}\nonumber \\
 & =\tr(6S_{x,h}^{4}W_{x}-4S_{x,h}^{3}W_{x,h}'+S_{x,h}^{2}W_{x,h}'')\,.\label{eq:LW-fourth-moment}
\end{align}
As for the first term, $|\tr(S_{x,h}^{4}W_{x})|\leq\norm{s_{x,h}}_{\infty}^{2}\norm h_{\theta}^{2}$.
As for the second term,
\begin{align}
|\tr(S_{x,h}^{3}W_{x,h}')| & \leq\norm{s_{x,h}}_{\infty}^{2}\tr\bpar{\sqrt{S_{x,h}W_{x,h}'^{2}S_{x,h}}}=\norm{s_{x,h}}_{\infty}^{2}\tr\bpar{\sqrt{W_{x,h}'W_{x}^{-1}W_{x,h}'}\sqrt{S_{x,h}W_{x}S_{x,h}}}\nonumber \\
 & \underset{\text{(i)}}{\leq}\norm{s_{x,h}}_{\infty}^{2}\sqrt{\tr(W_{x,h}'W_{x}^{-1}W_{x,h}')}\sqrt{\tr(S_{x,h}W_{x}S_{x,h})}=\norm{s_{x,h}}_{\infty}^{2}\norm{W_{x}^{-1}w_{x,h}'}_{W_{x}}\norm h_{\theta}\nonumber \\
 & \underset{\text{(ii)}}{\leq}p\norm{s_{x,h}}_{\infty}^{2}\norm h_{\theta}^{2}\label{eq:trSW}
\end{align}
where we used the Cauchy-Schwarz in (i) and Lemma~\ref{lem:usefulFactLewis}-3
in (ii).  

As for the last term, we first use the formula for $\tr(S_{x,h}^{2}W_{x,h}'')$
with $\Gamma=S_{x,h}^{2}$ in Lemma~\ref{lem:second-deriv-Lewis}.
Each term there is of the form $\tr(S_{x,h}^{2}\Diag(v))$ for $v=\,$I
\textasciitilde{} IV, which can be bounded as follows:
\begin{align}
\big|\tr\bpar{S_{x,h}^{2}\Diag(v)}\big| & =\big|\tr\bpar{S_{x,h}^{2}W_{x}^{\half}W_{x}^{-\half}\Diag(v)}\big|\leq\sqrt{\tr(W_{x}^{\half}S_{x,h}^{4}W_{x}^{\half})}\sqrt{\tr\bpar{\Diag(v)W_{x}^{-1}\Diag(v)}}\label{eq:last-bound}\\
 & \leq\norm{s_{x,h}}_{\infty}\,\norm h_{\theta}\,\norm v_{W_{x}^{-1}}\,.\nonumber 
\end{align}
Using the norm bounds in Lemma~\ref{lem:second-deriv-Lewis},
it follows that $|\tr(S_{x,h}^{2}W_{x,h}'')|\lesssim\norm h_{\theta}^{4}$
for $p=\O(\log m)$. Putting everything together with $\norm{s_{x,h}}_{\infty}\leq\sqrt{2}m^{\frac{1}{p+2}}\norm h_{\theta}\lesssim\norm h_{\theta}$
(Lemma~\ref{lem:usefulFactLewis}-2),
\begin{align*}
|\Dd^{2}\theta[h,h,h,h]| & \lesssim\norm{s_{x,h}}_{\infty}^{2}\norm h_{\theta}^{2}+\norm{s_{x,h}}_{\infty}\norm h_{\theta}^{3}\lesssim\norm h_{\theta}^{4}\,.\qedhere
\end{align*}
\end{proof}

\section{Technical lemmas}
\begin{lem}
\label{lem:matrix-projection} For a matrix $M\in\R^{m\times d}$
and $E\in\Rdd$ such that $E+M^{\T}M\succ0$, it holds that
\[
M(E+M^{\T}M)^{-1}M^{\T}\preceq P(M)=M(M^{\T}M)^{\dagger}M^{\T}\,.
\]
\end{lem}

\begin{proof}
Let us denote the LHS by $P'$ and the RHS by $P$. We show $I-P'\succeq I-P$
instead. First, $(P')^{2}\preceq P'$ and $(I-P')^{2}\preceq I-P'$
follow from
\begin{align*}
P'P' & =M(E+M^{\T}M)^{-1}\underbrace{M^{\T}M}_{\preceq E+M^{\T}M}\,(E+M^{\T}M)^{-1}M^{\T}\preceq M(E+M^{\T}M)^{-1}M^{\T}=P'\,,\\
(I-P')^{2} & =I+P'P'-2P'\preceq I-P'\,.
\end{align*}
It follows from $(I-P')^{2}\preceq I-P'$ that for any $v\in\R^{m}$
\begin{align*}
v^{\T}(I-P')v & \geq\snorm{(I-P')v}^{2}\geq\norm{(I-P)v}_{2}^{2}=v^{\T}(I-P)v\,,
\end{align*}
where the inequality holds due to $P'v,Pv\in\text{range}(M)$ and
$Pv=\arg\min_{w\in\,\text{range}(M)}\norm{v-w}_{2}^{2}$.
\end{proof}
\begin{prop}
\label{prop:stein-comp} Let $v,w,p,q,r,s\in\Rd$ and $h\sim\ncal(0,I_{d})$.
\begin{itemize}
\item $\E[(v\cdot h)(w\cdot h)^{3}]=3\norm w^{2}(v\cdot w)$.
\item $\E[(v\cdot h)^{2}(w\cdot h)^{2}]=\norm v^{2}\norm w^{2}+2(v\cdot w)^{2}$.
\item $\E[(p\cdot h)^{2}(r\cdot h)(s\cdot h)]=\norm p^{2}(r\cdot s)+2(p\cdot s)(p\cdot r)$.
\end{itemize}
\end{prop}

\begin{proof}
Using Stein's lemma (Lemma~\ref{lem:stein}),
\begin{align*}
\E[(v\cdot h)(w\cdot h)^{3}] & \underset{\text{Stein}}{=}\sum_{i}w_{i}\E[h_{i}(v\cdot h)(w\cdot h)^{2}]=\sum_{i}w_{i}\bpar{v_{i}\E[(w\cdot h)^{2}]+2w_{i}\E[(v\cdot h)(w\cdot h)]}\\
 & =(v\cdot w)\norm w^{2}+2\norm w^{2}(v\cdot w)=3\norm w^{2}(v\cdot w)\,,\\
\E[(v\cdot h)^{2}(w\cdot h)^{2}] & =\sum_{i}v_{i}\E[h_{i}(v\cdot h)(w\cdot h)^{2}]\underset{\text{Stein}}{=}\sum_{i}v_{i}\Par{v_{i}\E[(w\cdot h)^{2}]+2w_{i}\E[(v\cdot h)(w\cdot h)]}\\
 & =\norm v^{2}\norm w^{2}+2(v\cdot w)^{2}\,,\\
\E[(p\cdot h)^{2}(r\cdot h)(s\cdot h)] & =\sum_{i}p_{i}\E[h_{i}(p\cdot h)(r\cdot h)(s\cdot h)]\\
 & \underset{\text{Stein}}{=}\sum p_{i}\Par{p_{i}\E[(r\cdot h)(s\cdot h)]+r_{i}\E[(p\cdot h)(s\cdot h)]+s_{i}\E[(p\cdot h)(r\cdot h)]}\\
 & =\norm p^{2}(r\cdot s)+(p\cdot r)(p\cdot s)+(p\cdot s)(p\cdot r)=\norm p^{2}(r\cdot s)+2(p\cdot s)(p\cdot r)\,.\qedhere
\end{align*}
\end{proof}
These estimations result in a useful lemma for establishing SASC of
barriers for linear constraints.
\begin{lem}
\label{lem:variance-1} For $v,w\in\Rd$ and $h\sim\ncal(0,I_{d})$,
$\E[(v\cdot h)^{3}(w\cdot h)^{3}]=9\norm v^{2}\norm w^{2}(v\cdot w)+6(v\cdot w)^{3}$.
\end{lem}

\begin{proof}
Using Stein's lemma,
\begin{align*}
\E[(v\cdot h)^{3}(w\cdot h)^{3}] & =\sum_{i}v_{i}\E[h_{i}(v\cdot h)^{2}(w\cdot h)^{3}]=\sum v_{i}\bpar{2v_{i}\E[(v\cdot h)(w\cdot h)^{3}]+3w_{i}\E[(v\cdot h)^{2}(w\cdot h)^{2}]}\\
 & \underset{\text{(i)}}{=}2\norm v^{2}\cdot3\norm w^{2}(v\cdot w)+3(v\cdot w)\bpar{\norm v^{2}\norm w^{2}+2(v\cdot w)^{2}}=9\norm v^{2}\norm w^{2}+6(v\cdot w)^{3}\,,
\end{align*}
where in (i) we used Proposition~\ref{prop:stein-comp}-1 and 2.
\end{proof}
\begin{lem}
\label{lem:variance-2} For $p,q,r,s\in\Rd$ and $h\sim\ncal(0,I_{d})$,
\begin{align*}
\E[(p\cdot h)^{2}(q\cdot h)(r\cdot h)^{2}(s\cdot h)] & =(q\cdot s)\norm p^{2}\norm r^{2}+4(p\cdot r)(p\cdot q)(r\cdot s)\\
+2\norm p^{2}(r\cdot q)(r\cdot s) & +2\norm r^{2}(p\cdot q)(p\cdot s)+2(p\cdot r)^{2}(q\cdot s)+4(p\cdot s)(p\cdot r)(r\cdot q)\,.
\end{align*}
\end{lem}

\begin{proof}
Using Stein's lemma,
\begin{align*}
 & \E[(p\cdot h)^{2}(q\cdot h)(r\cdot h)^{2}(s\cdot h)]=\sum_{i}q_{i}\E[h_{i}(p\cdot h)^{2}(r\cdot h)^{2}(s\cdot h)]\\
 & =\sum q_{i}\bpar{2p_{i}\E[(p\cdot h)(r\cdot h)^{2}(s\cdot h)]+2r_{i}\E[(p\cdot h)^{2}(r\cdot h)(s\cdot h)]+2s_{i}\E[(p\cdot h)^{2}(r\cdot h)^{2}]}\\
 & \underset{\text{(i)}}{=}2(p\cdot q)\bpar{\norm r^{2}(p\cdot s)+2(p\cdot r)(r\cdot s)}+2(r\cdot q)\bpar{\norm p^{2}(r\cdot s)+2(p\cdot s)(p\cdot r)}\\
 & \qquad+(q\cdot s)\bpar{\norm p^{2}\norm r^{2}+2(p\cdot r)^{2}}\\
 & =(q\cdot s)\norm p^{2}\norm r^{2}+4(p\cdot r)(p\cdot q)(r\cdot s)+2\norm p^{2}(r\cdot q)(r\cdot s)+2\norm r^{2}(p\cdot q)(p\cdot s)\\
 & \qquad+2(p\cdot r)^{2}(q\cdot s)+4(p\cdot s)(p\cdot r)(r\cdot q)\,.
\end{align*}
In (i), we used Proposition~\ref{prop:stein-comp}-3 to the first
two terms and Proposition~\ref{prop:stein-comp}-2 to the third term.
\end{proof}

\end{document}